\documentclass[aps,reprint,showkeys,pra,superscriptaddress,twocolumn,longbibliography,nofootinbib]{revtex4-2}

\usepackage[colorlinks=true,citecolor=blue,linkcolor=magenta]{hyperref}
\usepackage[caption=false]{subfig}
\usepackage{xcolor, colortbl}
\usepackage{amsmath,amsthm,amsfonts,amssymb}
\usepackage{thmtools,mathtools}
\usepackage{physics}
\usepackage{thm-restate}
\usepackage{bbm}
\usepackage{multirow}
\usepackage{booktabs,array}
\usepackage{float}
\usepackage{quantikz}
\usepackage{graphicx}
\usepackage{placeins}
\usepackage{enumitem}
\usepackage{relsize}

\usepackage{soul}
\usepackage{tikz}
\usetikzlibrary{positioning}
\usetikzlibrary{matrix}
\usepackage{forest}

\newcolumntype{A}{ >{$} r <{$} @{\,} >{${}} l <{$} }

\newtheorem{thm}{Theorem}
\newtheorem{lemma}[thm]{Lemma}
\newtheorem{definition}{Definition}
\newtheorem{observation}{Observation}
\newtheorem{prop}[thm]{Proposition}
\newtheorem{crllr}[thm]{Corollary}

\newcommand{\groupA}{\mathcal{A}}
\newcommand{\groupG}{\mathcal{G}}
\newcommand{\groupH}{\mathcal{H}}
\newcommand{\groupK}{\mathcal{K}}
\newcommand{\groupP}{\mathcal{P}}

\newcommand\gru{U}
\newcommand\grsu{SU}
\newcommand\gro{O}
\newcommand\grso{SO}
\newcommand\grsp{Sp}
\newcommand\grgl{GL}
\newcommand\grudiag{U_\mathrm{diag}}
\newcommand\grspdiag{Sp_\mathrm{diag}}
\newcommand\grschur{\mathrm{Schur}}
\newcommand\grcs{\mathrm{CS}}

\newcommand{\mfg}{\mathfrak{g}}
\newcommand{\mfk}{\mathfrak{k}}
\newcommand{\mfp}{\mathfrak{p}}
\newcommand{\mfa}{\mathfrak{a}}
\newcommand{\mfh}{\mathfrak{h}}

\newcommand{\mfs}{\mathfrak{s}} 
\newcommand{\mfm}{\mathfrak{m}} 

\newcommand{\mfu}{\mathfrak{u}}
\newcommand{\mfsu}{\mathfrak{su}}

\newcommand{\mfso}{\mathfrak{so}}
\newcommand{\mfsp}{\mathfrak{sp}}
\newcommand{\mfgl}{\mathfrak{gl}}

\newcommand{\basisB}{\mathcal{B}}

\newcommand{\paulix}{X}
\newcommand{\pauliy}{Y}
\newcommand{\pauliz}{Z}

\newcommand{\graph}[1]{\mathcal{#1}} 

\newcommand{\ad}{\operatorname{ad}}
\newcommand{\Ad}{\operatorname{Ad}}
\renewcommand{\det}[1]{\mathrm{det}(#1)}
\renewcommand{\real}{\operatorname{Re}}
\newcommand{\imag}{\operatorname{Im}}
\newcommand{\diag}{\operatorname{diag}}
\renewcommand{\rank}{\operatorname{rk}}

\newcommand{\spanR}{\operatorname{span}_{\mbr}}
\newcommand{\spaniR}{\operatorname{span}_{i\mbr}}
\newcommand{\spanC}{\operatorname{span}_{\mbc}}
\newcommand{\fsl}[1]{\ensuremath{\mathrlap{\not{\phantom{#1}}}#1}}

\newcommand*{\id}{\mathchoice
  {\openone}
  {\openone}
  {\scalebox{.7}{\openone}} 
  {\scalebox{.5}{\openone}} 
}

\newcommand{\kak}{\ifmmode \mathrm{KAK}\else KAK \fi}

\newcommand\mbn{\mathbb{N}}
\newcommand\mbr{\mathbb{R}}
\newcommand\mbc{\mathbb{C}}

\newcommand\Ztwo{\mathbb{Z}_2}

\newlength{\swapradius}
\newcommand{\swapangle}{70}
\newcommand{\swapsymbol}{%
    \,
    \begin{tikzpicture}%
        \draw (0,0.05) arc (-90+\swapangle:-90:\swapradius); \draw (0,0.05) arc (90+\swapangle:90:\swapradius);%
        \draw (0,0.05) arc (90-\swapangle:90:\swapradius); \draw (0,0.05) arc (-90-\swapangle:-90:\swapradius);%
    \end{tikzpicture}%
}%
\newcommand{\smallswap}{\!\!\scalebox{0.65}{\swapsymbol}}

\def\shrug{\texttt{\raisebox{0.75em}{\char`\_}\char`\\\char`\_\kern-0.5ex(\kern-0.25ex\raisebox{0.25ex}{\rotatebox{45}{\raisebox{-.75ex}"\kern-1.5ex\rotatebox{-90})}}\kern-0.5ex)\kern-0.5ex\char`\_/\raisebox{0.75em}{\char`\_}}}

\newcommand{\rowrule}{\rule[.5ex]{1.25em}{0.3pt}}
\newlength{\cellsize}
\definecolor{tabgrey}{HTML}{DDDDDD}

\definecolor{lightgrey}{HTML}{666666}
\definecolor{orange1}{HTML}{ffa82a}
\definecolor{orange2}{HTML}{ff7b03}
\definecolor{purple1}{HTML}{9411a7}
\definecolor{purple2}{HTML}{b47ed1}
\definecolor{green1}{HTML}{8ad75b}
\definecolor{salmon1}{HTML}{f27675}
\definecolor{green2}{HTML}{44aa00}

\definecolor{blue1}{HTML}{2fb6e6}

\colorlet{tabpurple}{purple1!25}
\colorlet{tabgreen}{green1!25}

\newcommand{\hamcolor}[1]{{\color{green2}{#1}}}
\newcommand{\rc}[1]{\rowcolor{#1}[\tabcolsep]}

\newcommand{\mr}[1]{\multirow{-2}{*}{#1}}

\usepackage{acro}
\acsetup{
    format=\emph,
    format/short, 
    format/long, 
}
\DeclareAcronym{CD}{short=CD, long=Cartan decomposition}
\DeclareAcronym{CSA}{short=CSA, long=Cartan subalgebra}
\DeclareAcronym{CSG}{short=CSG, long=Cartan subgroup}
\DeclareAcronym{CSD}{short=CSD, long=cosine-sine decomposition}
\DeclareAcronym{QSD}{short=QSD, long=quantum Shannon decomposition}
\DeclareAcronym{KGD}{short=KGD, long=Khaneja-Glaser decomposition}
\DeclareAcronym{EVD}{short=EVD, long=eigenvalue decomposition}
\DeclareAcronym{SVD}{short=SVD, long=singular value decomposition}
\DeclareAcronym{DLA}{short=DLA, long=dynamical Lie algebra}

\makeatletter 
\renewcommand\onecolumngrid{
    \do@columngrid{one}{\@ne}%
    \def\set@footnotewidth{\onecolumngrid}
    \def\footnoterule{\kern-6pt\hrule width 1.5in\kern6pt}%
}
\renewcommand\twocolumngrid{
    \def\footnoterule{
    \dimen@\skip\footins\divide\dimen@\thr@@
    \kern-\dimen@\hrule width.5in\kern\dimen@}
    \do@columngrid{mlt}{\tw@}
}%
\makeatother
\newcommand*{\belowrulesepcolor}[1]{%
  \noalign{%
    \kern-\belowrulesep
    \begingroup
      \color{#1}%
      \hrule height\belowrulesep
    \endgroup
  }%
}
\newcommand*{\aboverulesepcolor}[1]{%
  \noalign{%
    \begingroup
      \color{#1}%
      \hrule height\aboverulesep
    \endgroup
    \kern-\aboverulesep
  }%
}
\begin{document}

\title{Recursive Cartan decompositions for unitary synthesis}
\author{David Wierichs}
\affiliation{Xanadu, Toronto, ON, M5G 2C8, Canada}
\author{Maxwell West}
\affiliation{School of Physics, University of Melbourne, Parkville, VIC 3010, Australia}
\affiliation{Theoretical Division, Los Alamos National Laboratory, Los Alamos, NM 87545, USA}
\author{Roy~T. Forestano}
\affiliation{Institute for Fundamental Theory, Department of Physics, University of Florida, Gainesville, FL 32653, USA}
\affiliation{Theoretical Division, Los Alamos National Laboratory, Los Alamos, NM 87545, USA}
\author{M. Cerezo}
\affiliation{Information Sciences, Los Alamos National Laboratory, Los Alamos, NM 87545, USA}
\author{Nathan Killoran}
\affiliation{Xanadu, Toronto, ON, M5G 2C8, Canada}

\begin{abstract}
Recursive \acp{CD} provide a way to exactly factorize quantum circuits into smaller components, making them a central tool for unitary synthesis.
Here we present a detailed overview of recursive \acp{CD}, elucidating their mathematical structure, demonstrating their algorithmic utility, and implementing them numerically at large scales.
We adapt, extend, and unify existing mathematical frameworks for recursive \acp{CD}, allowing us to gain new insights and streamline the construction of new circuit decompositions. 
Based on this, we show that several leading synthesis techniques from the literature—the Quantum Shannon, Block-ZXZ, and Khaneja-Glaser decompositions—implement the same recursive \ac{CD}. We also present new recursive \acp{CD} based on the orthogonal and symplectic groups, and derive parameter-optimal decompositions.
Furthermore, we aggregate numerical tools for \acp{CD} from the literature, put them into a common context, and complete them to allow for numerical implementations of all possible classical \acp{CD} in canonical form.
As an application, we efficiently compile fast-forwardable Hamiltonian time evolution to fixed-depth circuits, compiling the transverse-field XY model on $10^3$ qubits into $2\times10^6$ gates in $22$ seconds on a laptop.
\end{abstract}

\maketitle

\section{Introduction}

Quantum compilation~\cite{venturelli2018compiling,khatri2019quantum,sharma2019noise} plays a critical role in the quantum software stack, performing the essential task of converting a user's code into machine code. Spanning multiple stages, quantum compilation routines are needed to break down programs into constituent building blocks, reduce resource requirements in circuits, lower overheads and runtime, and map to quantum hardware. Lowering quantum resources is important not only in the present era where noise limits the depth of circuits, but also in the approaching era of fault-tolerance, as the first generations of fault-tolerant hardware will be constrained by limited qubit counts and high operational overheads. Efficient compilation could potentially open up such resource-limited devices for more applications.

While designing more resource-efficient circuits could be done manually by algorithm experts, this process is slow, not scalable, and possibly sub-optimal for performance. It is better to have a suite of highly optimized passes and subroutines which automate the compilation process across multiple layers of abstraction, freeing up research capacity to focus on developing the quantum algorithms or applications. The more we can unite different compilation routines under similar structure, the easier the compilation process will be to automate. 

A core quantum compilation task is exactly decomposing an arbitrary unitary into universal building blocks. For this, a number of approaches have been proposed~\cite{tucci1999rudimentary, mottonen2004quantum, bergholm2005quantum, nakajima2005new, mottonen2006decompositions, iten2016quantum, khaneja2001cartan, vatan2004realization, bullock2004note, mansky2023near, 
shende2005synthesis, mottonen2006decompositions, drury2008constructive,
de2016block,de2018unified, fuhr2018note, krol2024beyond,
barenco1995elementary, cybenko2001reducing, aho2003compiling, vartiainen2004efficient, jiang2018quantum, rakyta2022approaching, arrazola2022universal}. 
Interestingly, many exact synthesis routines in the literature are based on the same mathematical structure, known as \acp{CD}~\cite{cartan1926sur}. 
A \ac{CD} is a method for breaking up a matrix $M$ into the form $M=KAK'$, where $K^{(\prime)}$ and $A$ come from two complementary matrix subgroups.
Quantum computing practitioners are likely familiar with well-known special cases: the \ac{EVD}, the \ac{SVD}, and the Bloch-Messiah decomposition are all examples of \acp{CD}~\cite{edelman2023fifty}. 

\Acp{CD} are closely related to a mathematical object called \emph{symmetric spaces}. The possible symmetric spaces—and hence the possible \acp{CD}—were completely classified 100 years ago into roughly twenty families~\cite{cartan1926sur}. Of these, ten are relevant for unitary matrices. Even with the families fixed, there remain sufficient degrees of freedom within each family to fine-tune the decompositions, e.g.,~to reduce gate counts or cast gates into local forms. Moreover, \acp{CD} have special symmetry properties, linking commutation relations to $\Ztwo$ symmetries, which can be further leveraged to optimize circuits. 
Finally, the classification also makes it evident how \acp{CD} can be chained recursively, providing recipes to construct multi-stage matrix decompositions. These properties make \acp{CD} a compelling framework for compilation.

Composing \acp{CD} recursively allows for $N$-qubit circuits to be progressively decomposed down to one-and two-qubit gates.
The number of such possible recursive \acp{CD} is exponentially large; however, only a handful of these are truly represented in the literature. Prominent examples include the \ac{CSD}~\cite{tucci1999rudimentary, mottonen2004quantum, bergholm2005quantum, nakajima2005new, mottonen2006decompositions, iten2016quantum}, the \ac{QSD}~\cite{shende2005synthesis, mottonen2006decompositions, drury2008constructive}, the \ac{KGD}~\cite{khaneja2001cartan, vatan2004realization, bullock2004note, mansky2023near}, and—though not previously recognized as such—the Block-ZXZ decomposition~\cite{de2016block,de2018unified, fuhr2018note, krol2024beyond}.
Genuinely new insights and approaches have been slow to come, as more recent works often tweak earlier methods, providing modest optimizations or adding new numerical implementations. 
On the software side, numerical linear algebra methods for computing \acp{CD} have been presented in several works~\cite{bullock2004canonical, nakajima2005new, sutton2009computing, fuhr2018note, hackl2021bosonic}, but these have generally been limited in scale, flexibility, or ease-of-use. Recently, some papers have proposed variational approaches~\cite{kokcu2022fixed} and lax dynamics~\cite{chu2024lax} to finding \acp{CD}, though these can be brittle and require heavy enough classical computing resources to limit their applicability to toy model scales. Finally, though \acp{CD} are a powerful tool in linear algebra, their mathematical underpinnings—while not innately complicated—can often appear obtuse and intimidating, impacting their wider adoption by the field.

In this work, we outline a unified framework for recursive \acp{CD} of unitary gates. We present a comprehensive overview of the topic, bringing together disparate threads from the quantum computing literature in one place. To make the subject more approachable, we update, extend, unify, and simplify many past works. Along the way, we remedy several mistakes, misconceptions, and omissions from the literature. We demonstrate the generality of the \ac{CD} perspective by presenting the \ac{QSD}, the \ac{KGD}, and the optimized Block-ZXZ decomposition as the same recursive sequence, differing only in the representations of certain subalgebras. Using our framework, we outline three new unitary synthesis techniques which, to our knowledge, have not appeared before. The first two are based on the orthogonal and symplectic subgroups of the unitary group while the third achieves a parameter-optimal decomposition, i.e.,~it uses one parameter per degree of freedom in the decomposed unitary. Finally, we demonstrate the benefits of our software-aligned approach by numerically decomposing the time-evolution unitary of a fast-forwardable Hamiltonian for arbitrary time values, compiling a 1000-qubit circuit into $10^6$ gates in 22 seconds.

This work is structured as follows. In Sec.~\ref{sec:groups}, we review basic facts about the \ac{CD} of Lie groups. We move to the related decompositions of Lie algebras in Sec.~\ref{sec:algebras}, presenting several mathematical results about building recursions. We take a more practical turn in Sec.~\ref{sec:numerical_decomps}, where we present a unified numerical approach to compute any classical compact \ac{CD} from a standard \ac{EVD}. In Sec.~\ref{sec:new_decomps}, we use our framework to identify two new unitary synthesis strategies based on the orthogonal and symplectic groups. We present a detailed numerical example of compiling a Hamiltonian simulation problem in Sec.~\ref{sec:hamsim}. Finally, we present some forward-looking perspectives in the concluding Sec.~\ref{sec:conclusion}.

\section{Cartan decompositions of Lie groups}\label{sec:groups}

\begin{figure}
    \centering
    \includegraphics[width=\linewidth]{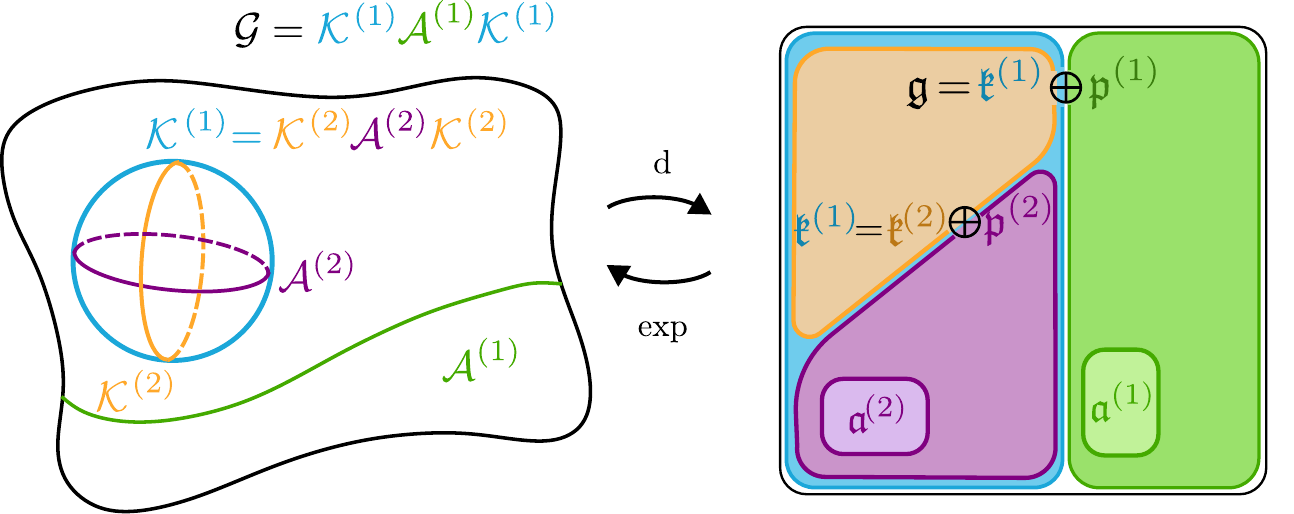}
    \caption{The core objects of this work are compact Lie groups $\groupG$ that are broken down by (recursive) \aclp{CD}. Here we depict a group being decomposed into two factors: a subgroup $\groupK^{(1)}$ and an abelian \acf{CSG} $\groupA^{(1)}$, with the former being decomposed further into $\groupK^{(2)}$ and $\groupA^{(2)}$. In the Lie algebra $\mfg$ of $\groupG$, the decomposition materializes as a recursive vector space sum $\mfg=\big(\mfk^{(2)}\oplus\mfp^{(2)}\big)\oplus\mfp^{(1)}$, where the \acs{CSG} generators stem from abelian subalgebras $\mfa^{(j)}\subset\mfp^{(j)}$.}
    \label{fig:main}
\end{figure}
In this section we discuss \acp{CD} of Lie groups, and the closely associated concept of symmetric spaces.
We will keep the mathematical contents brief and focus on the key technique for compilation, the \emph{\kak decomposition} of group elements.
For this, we also look at recursive \kak decompositions and review their appearance in the literature on quantum compilation.
We sketch the concept of (recursive) Cartan decompositions at the group and algebra level in Fig.~\ref{fig:main}.

\subsection{Symmetric spaces}\label{sec:symmetric_spaces}

\begin{table*}[ht]
    \centering
    \begin{tabular}{lcccccccc}
        Type & $\mathcal{G}$ & $\dim(\mathcal{G})$ & $\mathcal{K}$ & $\dim(\mathcal{K})$ & $\dim(\mathcal{G}/\mathcal{K})$ & $\rank(\mathcal{G}/\mathcal{K})$ & Overparametrization \\\midrule
        A & $\gru(n)\times \gru(n)$ & $2n^2$ & $\gru(n)$ & $n^2$ & $n^2$ & $n$ & $n$ \\
        AI & $\gru(n)$ & $n^2$ & $\grso(n)$ & $(n^2-n)/2$ & $(n^2+n)/2$ & $n$ & $0$ \\
        AII & $\gru(2n)$ & $4n^2$ & $\grsp(n)$ & $2n^2+n$ & $2n^2-n$ & $n$ & $3n$ \\
        AIII & $\gru(p+q)$ & $(p+q)^2$ & $\gru(p)\times \gru(q)$ & $p^2+q^2$ & $2pq$ & $r$ & $(p-q)^2 + r$ \\
        \midrule
        BD & $\grso(n)\times \grso(n)$ & $n^2-n$ & $\grso(n)$ & $(n^2-n)/2$ & $(n^2-n)/2$ & $\lfloor n/2\rfloor$ & $\lfloor n/2\rfloor$\\
        BDI & $\grso(p+q)$ & $[(p+q)^2-p-q]/2$ & $\grso(p)\times \grso(q)$ & $(p^2-p+q^2-q)/2$ & $pq$ & $r$ & $[(p-q)^2 -|p-q|]/2$ \\
        DIII & $\grso(2n)$ & $2n^2-n$ & $\gru(n)$ & $n^2$ & $n^2-n$ & $n/2$ & $3n/2$ \\
        \midrule
        C & $\grsp(n)\times \grsp(n)$ & $4n^2+2n$ & $\grsp(n)$ & $2n^2+n$ & $2n^2+n$ & $n$ & $n$\\
        CI & $\grsp(n)$ & $2n^2+n$ & $\gru(n)$ & $n^2$ & $n^2+n$ & $n$ & $0$ \\
        CII & $\grsp(p+q)$ & $2(p+q)^2+p+q$ & $\grsp(p)\times \grsp(q)$ & $2(p^2+q^2)+p+q$ & $4pq$ & $r$ & $2(p-q)^2+p+q+r$ \\
    \end{tabular}
    \caption{Classification of the (irreducible) symmetric spaces formed from Lie groups of classical real Lie algebras.
    Here $\groupG$ is the total group, $\groupK$ is the subgroup, $\groupG/\groupK$ is the symmetric space, and $\rank$ denotes its rank (an invariant quantity, given by the dimension of its \acs{CSG}), with $r\coloneqq \min(p,q)$.
    The final column shows the overparametrization introduced by the \kak decomposition of $\groupG$, i.e.,~the number of excess degrees of freedom in $\groupK\groupA\groupK$ compared to $\groupG$ itself.
    Note that this is a variant of the mathematical classification which has been adapted to our needs (we excluded non-compact spaces and exceptional Lie algebras, and added an abelian phase to $\gru(n)=\grsu(n)\times \gru(1)$); see App.~\ref{sec:symspace_classification}.
    Adapted from~\cite{edelman2023fifty,wiersema2025geometric}.}
    \label{tab:symm_classif}
\end{table*}

Symmetric spaces are a useful concept in a wide spectrum of physics and numerics.
They come up frequently when trying to find exact decompositions of matrix groups into more manageable subcomponents, as is done in quantum compilation~\cite{khaneja2001cartan,shende2005synthesis,krol2024beyond,edelman2023fifty}.
They also allow to understand symmetry-restricted quantum operations~\cite{wiersema2025geometric} and state spaces, e.g.,~of pure fermionic or bosonic Gaussian states~\cite{hackl2024average,hackl2021bosonic}.
We also refer to~\cite{gorodski2021introduction} and~\cite{magnea2002introduction} for introductions to symmetric spaces, with focus on mathematical clarity and application to quantum information, respectively.
There are a several ways to introduce (Riemannian) symmetric spaces, with somewhat different underlying intuitions. For our purposes, an algebraic definition will be most useful, in particular because we consider Lie groups arising from Lie algebras. 
For this, we use the notion of a \acf{CD}, which we will define rigorously in Def.~\ref{def:cartan_decomposition} in Sec.~\ref{sec:algebras}.
Intuitively, a \ac{CD} of a Lie algebra $\mfg$ splits it into a subalgebra $\mfk$ and the orthogonal complement $\mfp$ in such a way that the two subspaces satisfy special commutation relations.
We thus may informally define a symmetric space to be a quotient space for which the involved tangent spaces end up forming a \ac{CD}.

\begin{definition}
    Given a connected Lie group $\groupG$ and a Lie subgroup $\groupK$, the \emph{quotient space} $\groupP=\groupG/\groupK$ consists of the (left) cosets of $\groupK$, i.e., $\groupP=\{g\groupK\vert g\in \groupG\}$, and it forms a smooth manifold. $\groupP$ is called a \emph{symmetric space} if its tangent space $\mfp$ and the Lie algebra $\mfk$ of $\groupK$ form a \ac{CD} of the Lie algebra $\mfg$ of $\groupG$ via $\mfg=\mfk\oplus\mfp$.
\end{definition}
We note that if $\groupK$ is a normal subgroup, then $\groupG/\groupK$ is itself a group and $\mfp$ is an algebra. For symmetric spaces, we are interested in the complementary case, i.e.,~the quotient space may not form a group.

We will be particularly focused on compact Lie groups, and if $\groupK$ is compact, it turns out that $\groupP$ is a \emph{Riemannian} symmetric space, i.e.,~comes with a Riemannian metric. This is important to us not necessarily due to the geometric implications, but mostly because a complete classification of the Riemannian symmetric spaces was obtained by \'Elie Cartan a century ago~\cite{cartan1926sur}.
We sketch the classification in App.~\ref{sec:symspace_classification}, making a series of restrictions that adapt it to our scope, and arrive at Tab.~\ref{tab:symm_classif}.
From here on, we imply the restriction to (Riemannian) symmetric spaces arising from classical real compact Lie algebras, up to $\mfu(1)$ additions\footnote{The addition of abelian components breaks with the actual classification and with the underlying Lie algebra being a classical algebra. We accept this risk of confusion, for the benefit of obtaining self-consistent recursive decompositions.}.

The class labels for symmetric spaces in Tab.~\ref{tab:symm_classif} can also be encountered in different—but mathematically related—scenarios. For example, to label equivalence classes of so-called (Cartan) involutions; see Sec.~\ref{sec:cartan_involution}. As well, in physical contexts the classification translates to (anti)unitary symmetries~\cite{dalesandro2007quantum}, and in particular gives rise to the periodic table—or ``tenfold way"—of topological insulators~\cite{altland1997nonstandard,ryu2010topological}.

\subsection{\kak decomposition}\label{sec:kak_decomposition}
The \emph{\kak decomposition} (also called the \emph{global \acl{CD}})\footnote{We will largely use the terms ``\acl{CD}'' and ``\kak decomposition'' interchangeably throughout the text, and omit the ``global'' qualifier whenever brevity can be achieved without adding confusion.} is a methodology for exactly decomposing a Lie group $\groupG$ into the product of subgroups $\groupK$ and $\groupA$.
It is a very powerful tool, as it allows us to describe complex transformations using simpler transformations which may be easier to interpret or more efficient to implement. For instance, \acp{EVD} or \acp{SVD} of matrices are special cases of \kak decompositions. As we will see, \kak decompositions are also commonly used in the decomposition of quantum gates, i.e.,~unitary operators, into more manageable constituents.
A key strength of the method is its versatility; we have the freedom to choose any subgroup $\groupK$ of $\groupG$ such that $\groupP=\groupG/\groupK$ is a symmetric space, and any $\groupA$ that is a maximal abelian subgroup of $\groupP$.

To approach the \kak decomposition we look at the so-called KP decomposition first, which generalizes the polar decomposition of matrices. For proofs of the statements in this section, we refer to standard literature such as~\cite{knapp2013lie}.

\begin{thm}{\cite[Thm.~6.31]{knapp2013lie}}\label{thm:kp_decomp}
    Let $\groupG/\groupK$ be a symmetric space. Then $\groupG=\groupK\groupP$, that is, any $G\in\groupG$ can be written as $G=KP$ with $K\in\groupK$ and $P\in\groupP$.
\end{thm}

We will see below that we may equivalently write $\groupG=\groupP\groupK$.
This reproduces the standard polar decomposition for $\groupG=\grgl(n,\mbc)$ and $\groupK=\gru(n)$, allowing us to decompose any invertible square matrix $M$ into a unitary $K$ and a positive definite Hermitian matrix $P$ generated by a Hermitian matrix $x$, $M=KP=K\exp(x)$. Note that this example does not fall into the reduced scope of our classification above because the symmetric space $\grgl(n,\mbc)/\gru(n)$ is non-compact. We consider it an archetypical example nonetheless, given how widely used polar decompositions are. It will be the only non-compact example, and all other groups will be subgroups of some unitary group.

The symmetric space $\groupG/\groupK$ comes with additional structure that is crucial to obtain the \kak decomposition.

\begin{prop}{\cite[Prop.~7.29]{knapp2013lie}}\label{prop:horizontal_kak_decomp}
    Let $\groupP=\groupG/\groupK$ be a symmetric space with tangent space $\mfp$ and select a maximal abelian subalgebra $\mfa\subset\mfp$.
    Then for any $P=\exp(x)\in\groupP$, there is a $K\in\groupK$ and an $a\in\mfa$ such that $x=KaK^\dagger$ and thus $P=K\exp(a)K^\dagger$.
\end{prop}
A maximal abelian subalgebra like $\mfa$ also is called a \acf{CSA}\footnote{We emphasize that ``\acl{CSA}'' is sometimes used to mean a maximal abelian subalgebra consisting of elements whose adjoint action is diagonalizable; we do not make that final assumption here.}. 
We will see later that \emph{any} subalgebra of $\mfp$ is forced to be abelian, which is why this property is sometimes not mentioned explicitly in the definition of a \ac{CSA}.
Note that $\mfa$ here is a \emph{horizontal} \ac{CSA} in $\mfp$, not to be confused with a \ac{CSA} of the full Lie algebra $\mfg$ of $\groupG$.
We will call the abelian group $\groupA\coloneqq\exp(\mfa)$ the (horizontal) \acf{CSG} of $\groupP$.
The abelian group $\gru(1)^{m}$ is also referred to as a torus\footnote{Generalizing the ``standard" two-dimensional torus $\gru(1)^2$ that forms the hull of a bagel.}, so that the \ac{CSG} is called a \emph{maximal torus} in some contexts.
The map acting on a group element $G$ by conjugating it with another group element $K$ is called the \emph{adjoint action} of $\groupG$ on itself, and denoted as
\begin{align}\label{eq:def_Ad}
    \Ad_K:\groupG\to\groupG, \ \ G\mapsto \Ad_KG=KGK^\dagger.
\end{align}
Prop.~\ref{prop:horizontal_kak_decomp} can then be written as $\groupP=\Ad_\groupK\groupA$.

Continuing our example $\groupP=\grgl(n,\mbc)/\gru(n)$, note that $\mfp=\{x\in \mfgl(n,\mbc)| x^\dagger=x\}$, which contains real-valued diagonal matrices as a maximal abelian subalgebra.
The proposition above thus tells us that any Hermitian $x\in\mfp$ can be written as $K a K^\dagger$ for some unitary $K$ and a real-valued diagonal $a$. This is simply the \ac{EVD} of $x$. Similarly, $\exp(x)=K\exp(a)K^\dagger$, confirming that a Hermitian exponential matrix has positive eigenvalues only.

We may combine Thm.~\ref{thm:kp_decomp} and Prop.~\ref{prop:horizontal_kak_decomp} into the following corollary, which is an important, if not the most important, statement for our work.

\begin{crllr}\label{cor:kak}
    Let $\groupP=\groupG/\groupK$ be a symmetric space and $\mfa$ a \ac{CSA} of its tangent space $\mfp$.
    Then for any $G\in\groupG$, there exist $K_1, K_2\in\groupK$ and $A\in\groupA$ such that $G=K_1 A K_2$.
\end{crllr}

As mentioned before, we may reverse the ordering in Thm.~\ref{thm:kp_decomp} using this construction:
\begin{align}
     G=KP=KK_1AK_1^\dagger=K_2AK_2^\dagger K=P'K.
\end{align}
For our example, the corollary provides us with the \ac{SVD} of any invertible square matrix: $G=K_1 A K_2$ with unitary $K_{1,2}$ and a real positive diagonal matrix $A$. 
Similarly, we may start with $\groupG=\grgl(n,\mbr)$ and $\groupK=\gro(n)$ instead, to obtain the polar (Thm.~\ref{thm:kp_decomp}) and \ac{SVD} (Cor.~\ref{cor:kak}) of real invertible square matrices, as well as the orthogonal \ac{EVD} of symmetric matrices (Prop.~\ref{prop:horizontal_kak_decomp}).
For more details and $51$ other matrix factorizations based on \acp{CD}, see~\cite{edelman2023fifty}.
We also make here one crucial remark about the \kak decomposition: it is an \emph{existence} statement, which does not provide methods for obtaining the constituent $K$s and $A$s. Obtaining these in practice is one of the focal points of this work; see Sec.~\ref{sec:numerical_decomps}.
For a tutorial on some of these decompositions in the context of quantum optics, also see~\cite{houde2024matrix}.

Note that for a given group $\groupG$, the subgroup $\groupK$, which must be chosen such that $\groupG/\groupK$ is a symmetric space, fixes the structure of the \kak decomposition. The choice of the \ac{CSA} $\mfa\subset\mfp$, or equivalently of the \ac{CSG} $\groupA<\groupP$, simply changes the appearance of the decomposition in a given basis, and a change of basis via some $K\in\groupK$ takes us to any other choice $\mfa'$. While such a basis choice is not too exciting from a mathematical point of view, it can have important consequences when compiling quantum circuits, as we will see later in Sec.~\ref{sec:recursive_kak_literature}.

It is interesting to note that for Hamiltonian time evolution there is a key difference between $U(t)=\exp(-iHt)$ being decomposed via Prop.~\ref{prop:horizontal_kak_decomp} or Cor.~\ref{cor:kak}. In the former case, we obtain an \ac{EVD} of $H=KaK^\dagger$, and $U(t)=K\exp(at)K^\dagger$ only differs by a rescaling of the \ac{CSA} generator $a$. In the latter case, though the structure $U(t)=K_1(t)a(t)K_2(t)$ is the same for all $t$, we find different subgroup elements $K_{1,2}(t)$ for each evolution time.
In both scenarios, we find a quantum circuit with $t$-independent depth, in stark contrast to approximate product formula approaches, where longer time evolutions have less accuracy than shorter ones, and we must increase the depth to maintain the same level of accuracy.
We provide a detailed description of this compilation approach in Sec.~\ref{sec:hamsim}.

Along with these benefits, the \kak decomposition comes with some caveats. Chiefly, there will be systems for which we cannot efficiently obtain a \kak decomposition because the involved matrix groups have too large dimension, such as dense unstructured gates acting on a large number of qubits. This underlying constraint should always be kept in mind when working with this decomposition. Fortunately, there remain a number of situations where the \kak decomposition is still tractable to find numerically, e.g.,~due to a manageable number of qubits, or some sparse mathematical representation we can work with. One example for the former is decomposing intermediate-sized gates within a circuit as part of a quantum compilation pipeline. The most important example for the latter is the time evolution under free-fermionic Hamiltonians.

\subsection{Recursive \kak decompositions}\label{sec:recursive_kak}
One key feature of the \kak decomposition is that it can be applied recursively, allowing us to break down larger gates into successively smaller ones. The key point to recognize is that after the first decomposition, $K_1, K_2$ live in a Lie group $\groupK<\groupG$, so we can consult the list of symmetric spaces which have $\groupK$ as the base group. For our selection of symmetric spaces in Tab.~\ref{tab:symm_classif}, such a ``follow-up" symmetric space always exists, i.e.,~there are no combinations that lead to dead ends in the recursion. We can then decompose $K_1, K_2$ according to any valid \kak decomposition for $\groupK$ (we may even choose a different decomposition for $K_1$ than for $K_2$!). For example, starting from $\groupG=\gru(N)$, we can perform an AI decomposition, with $\groupK=\gro(N)$, giving the form $U=O^{(1)}_1 A^{(1)} O^{(1)}_2$ for any $U\in\gru(N)$. Then we can apply BDI decompositions to both of $O^{(1)}_1$ and $O^{(1)}_2$, yielding $U=O^{(2)}_{1} A^{(2)}_1 O^{(2)}_{2} A^{(1)} O^{(2)}_{3} A^{(2)}_2 O^{(2)}_{4}$, where $A^{(1)}$ and $A^{(2)}_{1,2}$ lie in the \ac{CSG} from the first and second decomposition, respectively.
In Fig.~\ref{fig:main} we illustrate a recursive \kak decomposition at the group level (left).

This structure yields a set of recursive \kak decompositions that grows exponentially with the recursion length.
Similar to how the classification of symmetric spaces allows us to enumerate all possible \acp{CD}, we can ask if there is a unified way to enumerate all possible recursive \acp{CD}. We will take a few steps in this direction in this work, providing insights into the structure of the recursion space in Sec.~\ref{sec:algebras}.
If we restrict decompositions of type AIII, BDI, and CII to (almost) equal $p$ and $q$, each recursion step roughly halves the dimension of the Lie group in question, so that recursions on matrices of dimension $n$ will roughly have length $\log_2(n)$ in that case.
The full decomposition of $\groupG$ will contain \ac{CSG} elements from all recursion levels but only subgroup elements from $\groupK$ of the very last \kak decomposition.

The space of recursive decompositions can be searched for elements with desired properties. In the context of unitary synthesis, this could be the overall gate count, the number of a specific gate type like CNOTs, or the number of gate parameters.
For each decomposition step, the basis and the \ac{CSA} need to be chosen in addition, which offers further customization options but also makes the search space continuous.
As we will see in the next section, recursive \kak decompositions in the quantum compilation literature are largely restricted to two choices.
Besides manual optimizations, it is mostly the refined level of basis choices at which well-known decompositions differ, interestingly with a significant impact on CNOT counts.
We present new recursive decompositions obtained from our framework in Sec.~\ref{sec:new_decomps}, optimizing parameter counts instead of CNOT counts or targeting less conventional groups.

\subsection{Recursive Cartan decompositions in the wild}\label{sec:recursive_kak_literature}

\begin{figure*}
    \centering
    \includegraphics[width=0.9\linewidth]{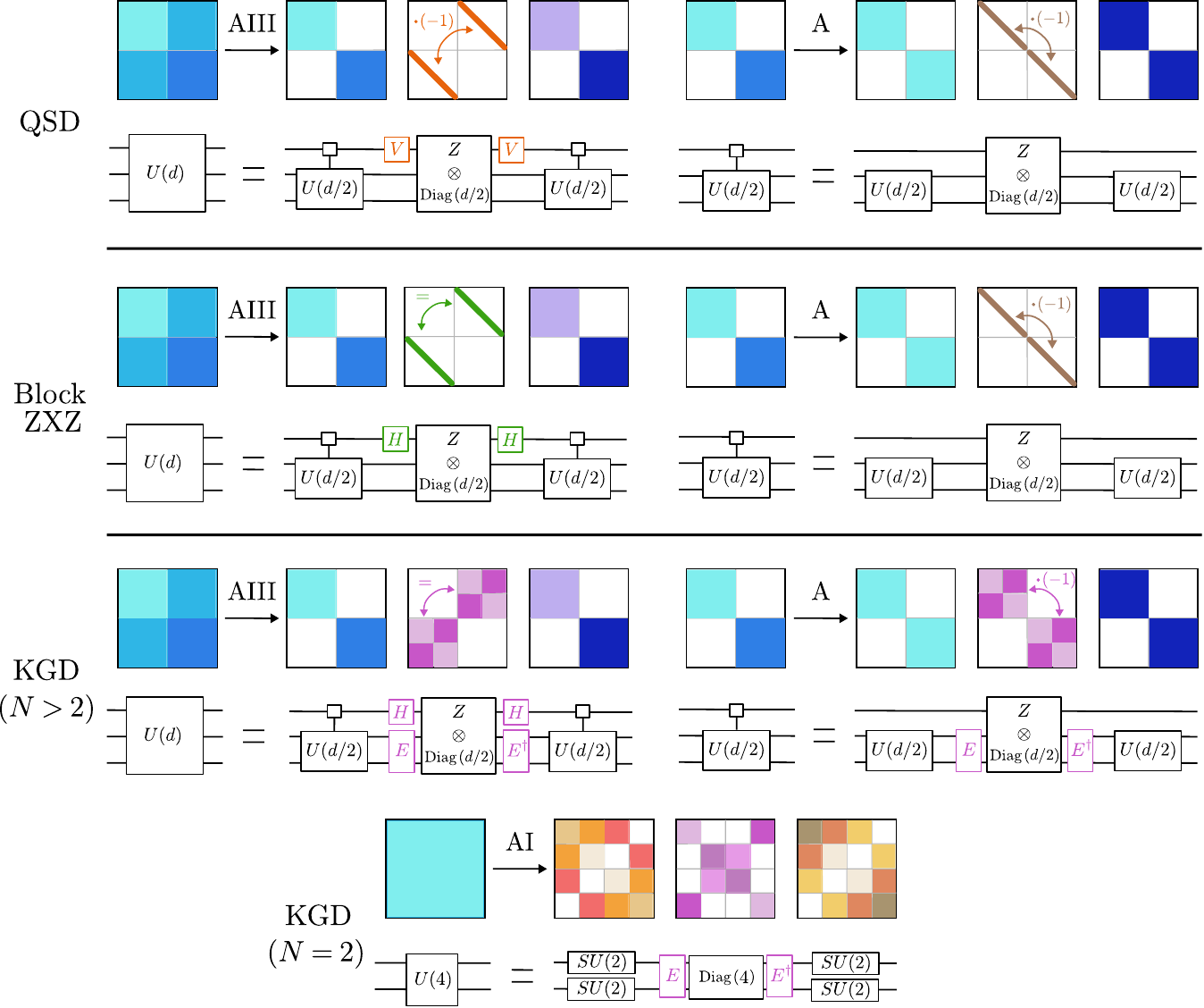}
    \caption{Established recursive \acp{CD} for unitary synthesis from the literature. The circuit diagrams represent the respective \kak decomposition at the group level, with the structure of their generators sketched by the colored matrices above them. The \ac{QSD}, Block-ZXZ, and \ac{KGD} all employ a recursion of alternating types AIII and A, with the \ac{KGD} switching to type AI for its last step on two qubits. Otherwise, the approaches differ in the choice of \ac{CSA}, encoded by basis rotations $V$, $H$ and $E$, and in subsequent optimizations to reduce the CNOT count. See App.~\ref{sec:unified_cartan_decomp_details} for details.}
    \label{fig:recursive_kak_literature}
\end{figure*}

In this section, we highlight that many well-known and top-performing exact unitary synthesis methods are ultimately recursive \acp{CD} of the same type. Specifically, the \acl{KGD}~\cite{khaneja2001cartan, vatan2004realization, bullock2004note, mansky2023near}, the \acl{QSD}~\cite{shende2005synthesis, mottonen2006decompositions, drury2008constructive}, and the optimized Block-ZXZ decomposition~\cite{krol2024beyond} can all be identified with a recursive AIII+A strategy, as we visualize in Fig.~\ref{fig:recursive_kak_literature}.
Alongside the unified framework we present here, the family of synthesis algorithms based on QR decompositions via Givens rotations~\cite{barenco1995elementary, cybenko2001reducing, aho2003compiling, vartiainen2004efficient, shende2005synthesis, jiang2018quantum, rakyta2022approaching, arrazola2022universal} may be considered—to our current understanding—as an independent approach.

It is perhaps surprising that the most well-known decompositions from the literature predominantly use only two types, AIII and A, despite the full classification of symmetric spaces providing many alternative matrix factorizations~\cite{edelman2023fifty}. 
When viewed through this lens, several circuit identities—which otherwise seem to be found via expert insight on independent problems—are revealed as applications of the same \ac{CD} techniques.
Furthermore, the widespread presence of AIII+A \acp{CD} seems to be recognized and acknowledged irregularly in the literature. 
We hope that by elucidating the shared Lie-algebraic structures that underlie these many methods, and highlighting places where different choices are possible, improved unitary synthesis techniques can be developed in the future.

Many synthesis methods explicitly use recursive applications of the \ac{CSD}~\cite{tucci1999rudimentary, mottonen2004quantum, bergholm2005quantum, nakajima2005new, mottonen2006decompositions, iten2016quantum}, which is a type-AIII \ac{CD}. For others, like the Block-ZXZ decomposition~\cite{de2016block,de2018unified, fuhr2018note}, the connection to the type-AIII \ac{CD} is less evident. The AIII decomposition is widely used for unitary synthesis because the $\mathrm{K}$s in $\kak$ end up having beneficial block-diagonal forms. With the choice $p=q=2^{N-1}$ (as is done for every technique mentioned here), the subgroup $\groupK$ can be naturally associated with $(N-1)$-qubit unitaries controlled on the basis states of the remaining qubit. Similarly, the matrix $\mathrm{A}$ has the form of single-qubit rotations controlled on the remaining $N-1$ qubits. Following this recipe recursively leads to decomposition schemes involving so-called ``uniformly controlled rotation"~\cite{mottonen2004quantum} or ``quantum multiplexor'' gates~\cite{shende2005synthesis}, i.e.,~rotations on a subset of qubits controlled by basis states of the remaining qubits. 

While AIII decompositions streamline gates into a convenient block-diagonal form, the gates still remain non-local across all qubits after each recursion step. This is undesirable because it can lead to suboptimal CNOT gate counts when compiling to a universal gate set.
To ameliorate this, a number of techniques combine the AIII decomposition with a second procedure which decouples one qubit completely from the rest at each recursion step. We highlight here that this decoupling step can be understood mathematically as a type-A \ac{CD}\footnote{This fact was noted in~\cite{drury2008constructive} and~\cite{dagli2008general}, though they don't explicitly identify the \ac{CD}'s type.}. 
Combining the type-AIII and type-A decompositions, we thus have the alternating chain of recursions
\begin{align}
    \gru(2^N)&\overset{\rm AIII}{\longrightarrow}\gru(2^{N-1})\times \gru(2^{N-1})\\
    &\overset{\rm A}{\longrightarrow}\gru(2^{N-1})\\
    &\overset{\rm AIII}{\longrightarrow}\gru(2^{N-2})\times \gru(2^{N-2})\\
    &\quad\vdots\nonumber
\end{align}
Techniques that use this recursion strategy differ typically only in two regards: through the choice of \ac{CSA}, and through unitary basis changes on the subgroup $\groupK$. We provide more details on the AIII+A decomposition strategy and outline the choices which differentiate the various families in App.~\ref{sec:unified_cartan_decomp_details}.

Readers familiar with symmetric spaces might notice that the AIII+A recipe could also be adapted to analogous BDI+BD or CII+C recursive strategies (potentially starting with an initial AI or AII step). Indeed, such a decomposition strategy works, though we are only aware of one work which uses this idea for gate synthesis. Specifically, Ref.~\cite{wei2012decomposition} proposes the BDI+BD recipe, covering the three-qubit case of $\grso(8)$\footnote{They also cover the two-qubit BDI case for $\grso(4)$.}, but does not recognize how it can be used recursively for larger systems. 
In Sec.~\ref{sec:new_decomps}, we will present the BDI+BD and CII+C decomposition strategies—which we call the \emph{completely orthogonal decomposition} and the \emph{completely symplectic decomposition}, respectively—in more detail.
Another recursive \ac{CD}, based on a pure AIII recipe, was presented in~\cite{deguise2018simple}.

Finally, we mention that (recursive) \acp{CD} also appear elsewhere in the literature for topics not focused on unitary synthesis, e.g.,~works investigating mathematical frameworks~\cite{dalesandro2007quantum, dagli2008general, guaita2024representation} or decompositions designed for entanglement theory~\cite{bullock2004canonical, bullock2005time, dalesandro2006decompositions}.

\section{Cartan decompositions of Lie algebras}\label{sec:algebras}
In this section we move from the Lie group level to the algebraic level.
We will start with the basic mathematical objects required to understand \acp{CD} at an algebraic, or local, level.
Then we present multiple higher-order \acp{CD}, including the recursive \acp{CD} omnipresent in quantum compilation literature, and extend the framework by Da{\u{g}}l{\i} et al.~\cite{dagli2008general,dalesandro2007quantum}, providing further insight into their structure.
This will also allow us to understand some common wisdom about $\Ztwo$ symmetries in physics from a mathematical perspective.

\begin{definition}
    A real Lie algebra is a vector space $\mfg$ over $\mbr$ together with a \emph{Lie bracket} $[\cdot, \cdot]:\mfg\times \mfg \to\mfg$ that is antisymmetric and satisfies the Jacobi identity $[x,[y,z]]+[z,[x,y]]+[y,[z,x]]=0$ for all $x,y,z\in\mfg$.
    A Lie subalgebra is a vector subspace that is closed under the Lie bracket.
\end{definition}

We will deal exclusively with real matrix Lie algebras\footnote{Note that a real Lie algebra may still contain complex-valued matrices, it will just not be a vector space over the field $\mbc$.} and the matrix commutator as Lie bracket.
A subalgebra $\mfh\subset\mfg$ of a Lie algebra $\mfg$ is called an \emph{ideal} if $[\mfg,\mfh]\subset\mfh$. The algebras $\{0\}$ and $\mfg$ are both always ideals of $\mfg$.

\begin{definition}
    A real Lie algebra $\mfg$ is called \emph{simple} if it is non-abelian and its only ideals are $\{0\}$ and $\mfg$ itself.
    A real Lie algebra is called \emph{semisimple} if it is a direct sum of real simple Lie algebras.
\end{definition}

Finite-dimensional semisimple Lie algebras are of particular interest because they represent natural symmetries of many physical systems, and they have been completely classified mathematically. In particular, we will care about \emph{real compact semisimple Lie algebras} whose simple components are isomorphic to one of the four classical Lie algebras 
\begin{alignat}{3}
    A_{n-1} &=\mfsu(n)&=&\left\{x\in \mbc^{n\times n} | x = -x^\dagger, \tr(x)=0\right\}, \label{eq:def_sun}\\
    B_n/D_n &=\mfso(d)&=&\left\{x\in \mbr^{d\times d} | x=-x^T\right\},\label{eq:def_son}\\
    C_n &=\mfsp(n)&=&\left\{x\in \mfsu(2n) | x=-J_nx^TJ_n^T\right\},\label{eq:def_spn}
\end{alignat}
given in the form of their defining matrix representation. Here, $d=2n+1$ for $B_n$ and $d=2n$ for $D_n$, and we denoted the commonly used symplectic form as
\begin{align}\label{eq:def_J_n}
    J_n=\begin{pmatrix} 0 & \id_n \\ -\id_n & 0 \end{pmatrix}.
\end{align}
The only other simple components that could appear in a real compact semisimple Lie algebra, but which we will not include in this work, are the compact real forms of one of the exceptional Lie algebras; also see App.~\ref{sec:symspace_classification}.
For unitary synthesis it will be convenient to allow for additional $\mfu(1)$ components. Even though they break the semisimplicity and depart from the standard classification of \acp{CD}, it will be easy enough to ensure that these components play nicely with the used mathematical tools.

For a specific quantum computing application, Lie algebras commonly enter the stage via the so-called \ac{DLA} of a Hamiltonian, gate generators, or other sets of preferred operators.
\begin{definition}\label{def:dla}
    Given a set of operators $\{g_j\}$, their \acf{DLA} is obtained by taking all linear combinations of nested commutators of the operators, $\langle\{g_j\}\rangle_\mathrm{Lie} = \spanR\{g_j, [g_j, g_k], [g_j, [g_k, g_\ell]], \dots\}$.
    We also define the \ac{DLA} of a Hamiltonian $H$ with respect to a specific operator basis by representing it in said basis as $H=\sum_j \hat{h}_j$ and considering the \ac{DLA} of the basis operators, $\langle iH \rangle_\mathrm{Lie} \coloneqq \langle\{i\hat{h}_j\}\rangle_\mathrm{Lie}$. Note that the basis dependence is implicit.
\end{definition}
For our purposes, it will suffice to think about unitary algebras and their subalgebras, so that the operators are always skew-Hermitian.

\subsection{Local Cartan decompositions}
Given a Lie algebra $\mfg$, we are interested in a subalgebra $\mfk\subset\mfg$ that comes with a special relationship between itself and its complement $\mfp$. For our purposes, we will consider ``complement" to mean the orthogonal complement with respect to the trace inner product on $\mfg$. We will denote the direct sum of orthogonal vector subspaces with $\oplus$, both if they commute or not, which will be implied by the context.

\begin{definition}\label{def:cartan_decomposition}
    Let $\mfg$ be a semisimple Lie algebra with subalgebra $\mfk\subset \mfg$.
    The vector space decomposition $\mfg=\mfk\oplus\mfp$ is called a \emph{(local) \acl{CD}} if the following commutation relations hold:
    \begin{align}
        [\mfk, \mfk] & \subseteq \mfk\quad\mathrm{(subalgebra~property)}, \label{eq:subalg_prop} \\
        [\mfk, \mfp] & \subseteq \mfp\quad\mathrm{(reductive~property)}, \label{eq:reduct_prop} \\
        [\mfp, \mfp] & \subseteq \mfk\quad\mathrm{(symmetric~property)}. \label{eq:symm_prop}
    \end{align}
\end{definition}

We include the first property for completeness, even though it is satisfied by assumption.
At this point, some names deserve a comment. The latter two properties in the definition can be understood from the Lie group perspective: the second property makes the quotient space $\groupP\coloneqq e^\mfg/e^\mfk$ a \emph{reductive homogeneous space}, and adding the third one turns it into a symmetric space.
While we focus on local \acp{CD} in this section, this connection allows us to leverage a key result about their global counterparts, namely the complete classification due to Cartan; see Sec.~\ref{sec:symmetric_spaces} and App.~\ref{sec:symspace_classification}.
In parts of the mathematical literature, $\mfg=\mfk\oplus\mfp$ is only called a \acl{CD} if the pair satisfies an additional property\footnote{Namely that the Killing form, which we will not introduced here, be negative (positive) definite on $\mfk$ ($\mfp$).}.
Our definition above deviates from this convention but is common in the field of quantum computation~\cite{kokcu2022fixed,wiersema2025geometric,dagli2008general}.
Let us reuse the notation of the Adjoint action $\Ad_K$ from Eq.~(\ref{eq:def_Ad}) for its own differential, which acts on Lie algebra elements instead:
\begin{align}\label{eq:def_Ad_on_algebra}
    \Ad_K:\mfg\to\mfg,\ \ x\mapsto KxK^{-1}.
\end{align}
The subalgebra and reductive properties of a \ac{CD} then can be expressed as
\begin{align}\label{eq:CD_properties_at_mixed_level}
    \Ad_\groupK \mfk\subset\mfk,\ \Ad_\groupK \mfp\subset\mfp.
\end{align}
Differentiating this action with respect to the subscript argument yields the adjoint action of $\mfg$ on itself, given by the Lie bracket (i.e.,~the matrix commutator for our purposes):
\begin{align}
    \ad_k:\mfg\to\mfg,\ \ x\mapsto \ad_k(x)=[k,x].
\end{align}
This allows us to rephrase Eq.~(\ref{eq:subalg_prop}-\ref{eq:symm_prop}) as $\ad_\mfk(\mfk)\subseteq\mfk$, $\ad_\mfk(\mfp)\subseteq\mfp$, and $\ad_\mfp(\mfp)\subseteq\mfk$, respectively.

We want to make sure that \acp{CD} still work if we allow for $\mfu(1)$ components, as anticipated above. To address this, we make the following observation.
\begin{observation}\label{obs:phases}
    A \ac{CD} $\mfg=\mfk\oplus\mfp$ of a semisimple Lie algebra can be extended to a \ac{CD} of $\mfg\oplus \mfu(1)^{\oplus m}$, where each of the $\mfu(1)$ components may be independently assigned to $\mfk$ or $\mfp$.
\end{observation}
The result in the previous observation works because the additional $\mfu(1)$ components cannot disrupt the defining commutation relations of the \ac{CD}.
They merely dilute the classification of possible \acp{CD} because their location in the decomposition is arbitrary.
We will implicitly use this observation, usually to decompose $\mfu(n)$ instead of $\mfsu(n)$.

The properties of a \ac{CD} allow us to understand the Lie algebra $\mfg$ (and its corresponding Lie group $\groupG$) based on basic commutation relations between the two preferred subspaces, which we will sometimes call the \emph{vertical subspace} ($\mfk$) and the \emph{horizontal subspace} ($\mfp$)\footnote{To the best of our knowledge there is no fundamental reason for this association, but it is widely used from gauge theory to particle physics.}. The strong structural guarantees of a \ac{CD} allow for rather general statements.

\begin{restatable}{prop}{uniquecdofH}\label{prop:unique_CD_of_H}
    Consider a set of operators $\basisB$.
    If there is a \ac{CD} of the \ac{DLA} $\langle \basisB\rangle_\text{Lie}$ for which $\basisB$ is horizontal, it is unique.
\end{restatable}
Note that the statement, which we prove in App.~\ref{sec:horizontal_proofs}, applies to Hamiltonians in a fixed operator basis as well.
It is a useful guardrail for methods that require—or prefer—a set of operators (or Hamiltonian) to lie in the horizontal space, as does our compilation algorithm in Sec.~\ref{sec:hamsim}.

The proposition does not extend to any $\basisB$ in a given algebra, but only holds for the unique \ac{DLA} generated by $\basisB$.
If we are thinking about the commonly used Pauli basis, we furthermore make the following existence observation.

\begin{restatable}{prop}{cdofHexistsPauli}\label{prop:CD_of_H_exists_Pauli}
    Consider a set $\basisB$ of Pauli words $\{iP_j\}_j$ that is a minimal generating set for its \ac{DLA} $\langle \basisB \rangle_\text{Lie}$. Then there is a unique \ac{CD} of the \ac{DLA} for which $\basisB$ is horizontal.
    The statement continues to hold if $\basisB$ is a minimal generating set extended by even-order commutators.
\end{restatable}

\subsection{Cartan involutions}\label{sec:cartan_involution}

\begin{table*}[ht]
    \centering
    \begin{tabular}{lclclc}
        Type & $\mfg$ & $\ \ \theta=\Theta$ & Group & Constraint on $G$ & Generic $\theta_0$ \\\midrule
        A & $\mfu(n)\oplus \mfu(n)$ & $\Ad_X\circ\swapsymbol$ & $\gru(n)^{\times 2}$ & $X^\dagger=X^{\smallswap}$ & $\swapsymbol$ \\
        AI & $\mfu(n)$ & $\Ad_V \circ \,\ast $ & $\gru(n)$ & $V^T=V$ & $\ast$\\
        AII & $\mfu(2n)$ & $\Ad_W \circ \,\ast $ & $\gru(2n)$ & $W^T=-W$ & $\Ad_{J_n}\circ\,\ast$\\
        AIII & $\mfu(p+q)$ & $\Ad_H $ & $\gru(p+q)$ & $H^\dagger=H$ & $\Ad_{I_{p,q}}$\\\midrule
        BD & $\mfso(n)\oplus \mfso(n)$ & $\Ad_X\circ\swapsymbol$ & $\grso(n)^{\times 2}$ & $X^T=X^{\smallswap}$ & $\swapsymbol$\\
        BDI & $\mfso(p+q)$ & $\Ad_R $ & $\gro(p+q)$ & $R^T=R$ & $\Ad_{I_{p,q}}$\\
        DIII & $\mfso(2n)$ & $\Ad_L $ & $\grso(2n)$ & $L^T=-L$ & $\Ad_{J_n}$\\\midrule
        C & $\mfsp(n)\oplus \mfsp(n)$ & $\Ad_X\circ\swapsymbol$ & $\grsp(n)^{\times 2}$ & $X^\dagger=X^{\smallswap}$ & $\swapsymbol$\\
        CI & $\mfsp(n)$ & $\Ad_S$ & $\grsp(n)$ & $S^\dagger=-S$ & $\Ad_{J_n}=\ast$\\
        CII & $\mfsp(p+q)$ & $\Ad_P $ & $\grsp(p+q)$ & $P^\dagger=P$ & $\Ad_{K_{p,q}}$\\
    \end{tabular}
    \caption{Functional form of all Cartan involutions on the considered Lie algebras $\mfg$ with respect to their ``defining" representations. Here we made the basis choice of the involution explicit, so that no implicit degrees of freedom are left. All involutions consist of conjugation with a matrix $G$, potentially combined with complex conjugation $\ast$ or a swap operation ${\protect\swapsymbol}$ (Eq.~(\ref{eq:def_swap})). The different involution types on the same algebra differ in the additional properties $G$ has to satisfy, and/or in the presence of $\ast$. 
    We also provide a common, or generic, choice $\theta_0$ for the involution, with $J_n$, $I_{p,q}$ and $K_{p,q}$ defined in Eqs.~(\ref{eq:def_J_n}), and (\ref{eq:def_I_pq_K_pq}).
    The involution $\Theta$ on the group $\groupG=\exp(\mfg)$ is equal to the involution $\theta$ on $\mfg$ in this representation.
    }
    \label{tab:all_cartan_involutions}
\end{table*}

Involutions on a Lie algebra are closely tied to \acp{CD}, and thus to symmetric spaces. They are a convenient tool for defining and classifying possible decompositions. In addition, we will use them to describe interactions of \acp{CD} and to relate \acp{CD} to $\mathbb{Z}_2$ symmetries of physical systems.

\begin{definition}
    Let $\mfg$ be a Lie algebra. A map $\theta:\mfg\to\mfg$ is called an \emph{involution} if it squares to the identity. Moreover, if this map is a Lie algebra automorphism, we will call it a \emph{Cartan involution}.
\end{definition}

Even more so than for the term \ac{CD}, a \emph{Cartan} involution is frequently required to have an additional property\footnote{Namely that $(x,y)\mapsto-B(x,\theta(y))$ be positive definite, where $B$ is the Killing form that we still will not introduce.}. For consistency,  we again deviate from this and use the ``quantum physicists'' convention. This should not cause confusion for our scope because the identity is the only map fulfilling this additional property for real compact Lie algebras, i.e.,~to mathematicians there would be only a single, trivial Cartan involution anyways.

Because any Cartan involution $\theta$ is a linear map which squares to the identity, it can have only two possible eigenvalues, namely $\pm 1$. This partitions the domain $\mfg$ into a direct sum $\mfg=\mfs\oplus\mfm$ whose components are fully symmetric ($\theta(x)=x~\forall\ x\in\mfs$) and fully antisymmetric ($\theta(x)=-x~\forall\ x\in\mfm$). Furthermore, since $\theta$ is an automorphism, we have
\begin{align}
    \theta([\mfs,\mfs])=[\theta(\mfs),\theta(\mfs)]=[\mfs,\mfs].
\end{align}
We conclude that $[\mfs, \mfs]$ lies within the $+1$ eigenspace of $\theta$, i.e., $[\mfs, \mfs]\subseteq \mfs$. This implies that $\mfs$ is itself a Lie (sub)algebra. Through similar arguments, we can conclude that $[\mfm, \mfm]\subset\mfs$ and $[\mfs, \mfm]\subset\mfm$. 

We notice the strong similarity between the three algebraic conditions on the eigenspaces $\mfs$/$\mfm$ and the definition of a \ac{CD}. Indeed, from any Cartan involution $\theta:\mfg\to\mfg$, we can construct a \ac{CD} by partitioning $\mfg$ into the $\theta$-symmetric and $\theta$-antisymmetric subspaces\footnote{Note that the tangent space $\mfm$ of the \emph{symmetric} space $e^\mfg/e^\mfs$ is the $\theta$-\emph{antisymmetric} subspace. \tiny\shrug}. Conversely, any \ac{CD} $\mfg=\mfk\oplus\mfp$ can be used to define an involution $\theta=\Pi_\mfk-\Pi_\mfp$ from the projectors onto $\mfk$ and $\mfp$. This binary split already gives away a close relationship with $\Ztwo$ symmetries, which we discuss in Sec.~\ref{sec:involutions_symmetries}.
Given their equivalence, we will implicitly switch between \acp{CD} and Cartan involutions.

This equivalence also allows us to classify all Cartan involutions, inheriting from the classification of symmetric spaces in Sec.~\ref{sec:symmetric_spaces}. Focusing on the classical simple Lie algebras $\mfsu$, $\mfso$ and $\mfsp$ again, we find 10 types which are assigned the same labels as the corresponding symmetric spaces from Tab.~\ref{tab:symm_classif}. We denote the type of an involution on a semisimple algebra $\mfg=\oplus_j\mfg_j$ as $\oplus_j T_j$, where $T_j$ is the type of $\theta$ restricted to the simple component $\mfg_j$.
For each type, there are infinitely many Cartan involutions, which only differ by a basis change generated by some element of the algebra. 
In Tab.~\ref{tab:all_cartan_involutions}, we make this basis change explicit and characterize the concrete functional form of any valid involution $\theta$ in the defining representation (see Eqs.~(\ref{eq:def_sun})-(\ref{eq:def_spn})) of the algebras. 
They all contain the Adjoint action $\Ad_G$ (c.f.~Eq.~(\ref{eq:def_Ad_on_algebra})) of some $G\in\exp(\mfg)$, with the two exceptions that $G\in\gro(p+q)$ for type BDI and that $G\in\gru(p+q), \det{G}=\pm 1$ is allowed for type AIII even if we restrict to $\mfg=\mfsu(p+q)$.
We also provide a generic form $\theta_0$ that arises from basis choices commonly used in the literature, with $J_n$ defined in Eq.~(\ref{eq:def_J_n}) and
\begin{align}\label{eq:def_I_pq_K_pq}
    I_{p,q}=\begin{pmatrix} \id_p & 0 \\ 0 & -\id_q \end{pmatrix},\ 
    K_{p,q}=\begin{pmatrix} \id_p & \!0 & \!0 & \!0 \\ 0 & \!-\id_q & \!0 & \!0 \\ 0 & \!0 & \!\id_p & \!0 \\ 0 & \!0 & \!0 & \!-\id_q\end{pmatrix}.
\end{align}
Intuitively, it can be very helpful to identify these as Pauli matrices for $p=q$\footnote{Strictly speaking, this works as soon as the system of interest has one two-level subsystem for $J_n$ and $I_{n/2,n/2}$, or two such subsystems for $K_{n/2,n/2}$.}: $J_n=iY_0\otimes \id_{n/2}$, $I_{n/2,n/2}=Z_0\otimes \id_{n/2}$ and $K_{n/2,n/2}=\id_2\otimes Z_1\otimes \id_{n/4}$.
On doubled classical algebras $\mfg=\mfh\oplus\mfh$ we also need a swap operation that exchanges the direct summands. We denote it as
\begin{align}\label{eq:def_swap}
    \swapsymbol(x\oplus y)\equiv \Ad_{\smallswap}(x\oplus y)=y\oplus x, 
\end{align}
where we denoted both the map and the conjugating operator as $\swapsymbol$. The latter is equal to $X_0$ if the doubled algebra is embedded as the diagonal into a larger ambient algebra.
Note that trivial decompositions with $\mfk=\mfg$ are valid \acp{CD} of type AIII on $\gru(n)$, BDI on $\grso(n)$, and CII on $\grsp(n)$, each with $p=0$ (or $q=0$).

The derivation of Tab.~\ref{tab:all_cartan_involutions} can be found in App.~\ref{sec:involutions_calculations:characterize}.
Even though the form of possible involutions depends on the representation of the algebra, this explicit form will enable us to derive the behaviour of Cartan involutions (and thus \acp{CD}) under composition and recursion.
In particular, the characterization plays a crucial role in our proofs of Lemma~\ref{lemma:composition_of_involutions} and Thm.~\ref{thm:grading_recursion_equivalence}. Furthermore, we expect it to be of use for practitioners and to help with identifying Cartan involutions in applications.

Given a Cartan involution $\theta$ on $\mfg$, we also have an involution $\Theta$ on the group $\groupG$—i.e., a map $\Theta:\groupG\rightarrow\groupG$ with $\Theta^2=\id$—which has $\theta$ as its differential: $\theta=\mathrm{d}\Theta$.
For the concrete functional forms on the defining representations in Tab.~\ref{tab:all_cartan_involutions}, we can use a linear $\Theta$, so that we have simply $\Theta=\theta$. We will use this in our derivation of numerical \kak decomposition routines in Sec.~\ref{sec:numerical_decomps}.

\subsection{Cartan gradings}\label{sec:gradings}
Here we discuss the notion of a Cartan grading and how it arises naturally when considering simultaneous \acp{CD} of the same algebra.
The additional structure captured by a Cartan grading will allow us to easily discover new subalgebras and \acp{CD} for a given algebra.
In this and the next section, we build on top of the framework by Da{\u{g}}l{\i} et al.~\cite{dagli2008general}, modifying and extending it.
To construct higher-order \acp{CD} like Cartan gradings from multiple \acp{CD}, we will require them to ``play nice" in the following sense.

\begin{definition}\label{def:compatible_set_of_cds}
    We call two \acp{CD} \emph{compatible} if their corresponding Cartan involutions $\theta_{1,2}$ commute or, equivalently, if their eigenspaces $\mfk_{1,2}$ and $\mfp_{1,2}$ satisfy
    \begin{align}\label{eq:compatibility_pair_of_cds}
        \mfg = \left(\mfk_1\cap\mfk_2\right)\oplus\left(\mfk_1\cap\mfp_2\right)\oplus
        \left(\mfp_1\cap\mfk_2\right)\oplus\left(\mfp_1\cap\mfp_2\right).
    \end{align}
    We denote a set of compatible \acp{CD} as $\{\mfg=\mfk_j\oplus \mfp_j\}_j$ with \ac{CD} types $\{T_j\}_j$ using integer subscripts.
\end{definition}

Next, we define Cartan gradings as a restriction of the semigroup-grading from~\cite{dagli2008general}.

\begin{definition}\label{def:grading}
    A \emph{Cartan $c$-grading} of a Lie algebra $\mfg$ is a direct sum decomposition~\cite[Def.~3.1]{dagli2008general}
    \begin{align}
        \mfg = \bigoplus_{s\in\Ztwo^c} \mfg_s, \quad\text{with}\quad [\mfg_s, \mfg_t]\subseteq \mfg_{s+t}\ \forall\ s,t\in\Ztwo^c.
    \end{align}
\end{definition}
Note that in contrast to Def.~\ref{def:compatible_set_of_cds}, subscripts of Cartan gradings are bit strings, not integers.

A set of $c$ compatible \acp{CD} allows us to split $\mfg$ orthogonally into all $2^c$ intersections of subalgebras $\mfk_j$ and horizontal subspaces $\mfp_j$. Intuitively, we can see this split giving rise to a Cartan $c$-grading, with each bit in $\Ztwo^c$ corresponding to one \ac{CD} and indicating whether the intersection was made with $\mfk_j$ or $\mfp_j$ for that \ac{CD}.
The following proposition makes this intuition precise.

\begin{restatable}{prop}{setofcdstograding}\label{prop:set_of_cds_to_grading}
    Consider $c$ mutually compatible \acp{CD} $\{\mfg=\mfk_j\oplus\mfp_j\}_{1\leq j\leq c}$. 
    Then the following spaces form a Cartan $c$-grading of $\mfg$~\cite[Prop.~3.1]{dagli2008general}:
    \begin{align}\label{eq:grading_from_set}
        \mfg_s = \bigcap_{\substack{j=1\\s_j=0}}^c \mfk_j\cap \bigcap_{\substack{j=1\\s_j=1}}^c \mfp_j,\quad \forall \ s=s_1s_2\dots s_c \in \Ztwo^c.
    \end{align}
\end{restatable}
The original statement in~\cite{dagli2008general} without the compatibility criterion is not correct, which we prove alongside Prop.~\ref{prop:set_of_cds_to_grading} in App.~\ref{sec:involutions_calculations:HOCD}.
The proposition connects \acp{CD} and Cartan gradings to the discrete group $\Ztwo^c$.
Intuitively, instead of reasoning directly about the Lie algebra $\mfg$ and its allowed Cartan involutions, this connection allows us to work with the simpler $\Ztwo^c$, its subgroups, and its homomorphisms to $\Ztwo$.
One consequence is that, just like there are naturally $c$ subgroups of $\Ztwo^{c}$ which are isomorphic to $\Ztwo^{c-1}$, any Cartan $c$-grading naturally contains $c$ Cartan $(c-1)$-gradings, which can be obtained by summing over one of the bits of $\Ztwo^c$. More generally, we may sum over $k$ different bits to obtain $\tbinom{c}{k}$ Cartan $(c-k)$-gradings, corresponding to the statement that $\Ztwo^c$ has $\tbinom{c}{k}$ subgroups isomorphic to $\Ztwo^{c-k}$ if we distinguish its bits.
We will formalize this intuition in the following proposition, depicted in Fig.~\ref{fig:ztwo_and_cds}.

\begin{figure}
    \centering
    \includegraphics[width=\linewidth]{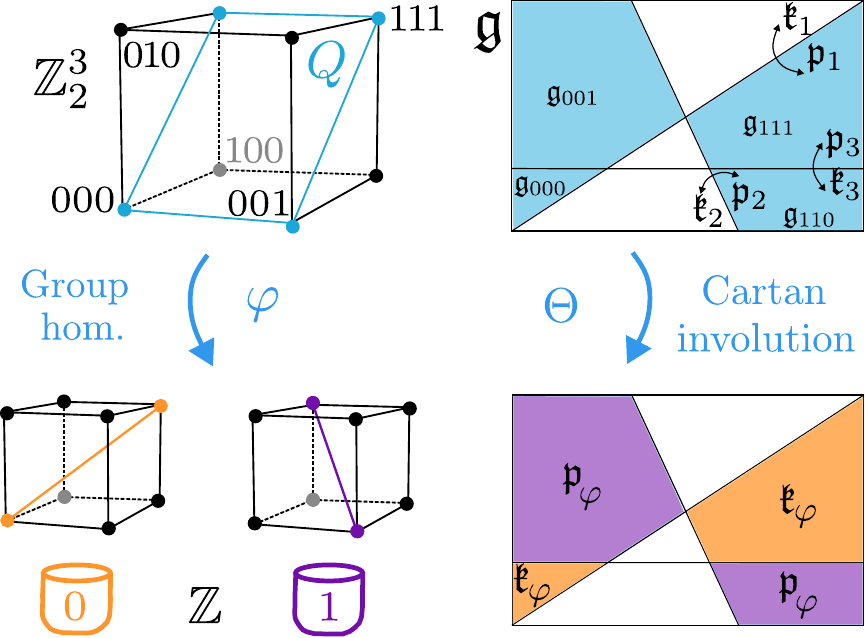}
    \caption{$\Ztwo^c$ captures the structure underlying a Cartan $c$-grading of an algebra $\mfg$ (Def.~\ref{def:grading}). Prop.~\ref{prop:grading_to_set_of_cds} uses this to connect subgroups $Q\subset\Ztwo^c$ to subalgebras $\mfg_Q$, and homomorphisms $\varphi:Q\to\Ztwo$ to Cartan decompositions $\mfg_Q=\mfk_\varphi\oplus\mfp_\varphi$.}
    \label{fig:ztwo_and_cds}
\end{figure}

\begin{restatable}{prop}{gradingtosetofcds}\label{prop:grading_to_set_of_cds}
    Consider a Cartan $c$-grading $\mfg=\bigoplus_{s\in\Ztwo^c}\mfg_s$ together with a subgroup $Q\subseteq\Ztwo^c$.
    Then $\mfg_Q\coloneqq \oplus_{s\in Q}\mfg_s$ is a Lie algebra.
    Further consider a group homomorphism $\varphi:Q\to\Ztwo$. Then 
    \begin{align}
        \mfg_Q=\mfk_\varphi\oplus \mfp_\varphi = \bigoplus_{s\in Q, \varphi(s)=0} \mfg_s\ \ \oplus\bigoplus_{s\in Q, \varphi(s)=1} \mfg_s
    \end{align}
    is a \ac{CD} of $\mfg_Q$.
    For $Q=\Ztwo^c$ and homomorphisms $\varphi_j$ that read out the $j$th bit of $s\in\Ztwo^c$, we obtain a set of compatible \acp{CD} of $\mfg$ that reproduce the grading via Prop.~\ref{prop:set_of_cds_to_grading}.
\end{restatable}

This proposition, which we prove in App.~\ref{sec:involutions_calculations:HOCD}, demonstrates the usefulness of Cartan gradings; discovering new subalgebras or \acp{CD} is as easy as choosing a subgroup of $\Ztwo^c$ or morphisms to $\Ztwo$, as we exemplify next with a few special cases.
First, for any $s\in\Ztwo^c$, choosing the group $Q=\{\mathbf{0}, s\}\cong\Ztwo$, where $\mathbf{0}$ is the identity in $\Ztwo^c$, and the homomorphism $\varphi(s)=1$, leads to a Lie algebra $\mfg_Q$ and its \ac{CD} $\mfg_{\mathbf{0}}\oplus\mfg_s$. This also implies that $\mfg_{\mathbf{0}}$ is itself a Lie algebra, as was also mentioned in~\cite{dagli2008general} and follows from the first statement of the proposition with $Q=\{\mathbf{0}\}$.
Second, for $Q=\{\mathbf{0}, s, t, s+t\}\cong\Ztwo^2, \varphi(s)=\varphi(t)=1$, we find a third \ac{CD} of $\mfg_Q$ from two compatible \acp{CD}. We discuss this scenario in detail in App.~\ref{sec:involutions_calculations:compose}.
Finally, if the homomorphism is not surjective, i.e.,~$\varphi(Q)=\{0\}$, the horizontal space $\mfp_\varphi$ is empty and $\mfk_\varphi$ is not a proper subalgebra of $\mfg_Q$.


We conclude this section by observing that combining \acp{CD} into Cartan gradings is a symmetric, or abelian, process, as changing the order of \acp{CD} is as simple as reordering the bits in $\Ztwo^c$.
This property will be used in the next section.

\subsection{Recursive Cartan decompositions}\label{sec:recursions}

Here we discuss another way for \acp{CD} to interact; instead of ``living" next to each other on the total space $\mfg$ as in Def.~\ref{def:compatible_set_of_cds}, they can be acting recursively, with one \ac{CD} decomposing the subalgebra of the previous \ac{CD}. As we saw in Sec.~\ref{sec:recursive_kak_literature}, this is a widely used concept in unitary synthesis. We start by defining it formally.

\begin{definition}\label{def:recursive_decomps}
    An $r$-recursive \ac{CD} of a Lie algebra $\mfg$ is a sequence of \acp{CD}
    \begin{align}
        \mfk^{(\ell)}=\mfk^{(\ell+1)} \oplus \mfp^{(\ell+1)}\,,\quad \forall\ 0\leq \ell < r,
    \end{align}
    with types $\{T_{\ell}\}_\ell$, where we set $\mfk^{(0)}=\mfg$~\cite[Def.~3.2]{dagli2008general}.
    We write the type of the recursive \ac{CD} as $T_1\to T_2\cdots \to T_{r}$.
\end{definition}

Let us compare Cartan gradings and recursive \acp{CD}.
Props.~\ref{prop:set_of_cds_to_grading} and~\ref{prop:grading_to_set_of_cds} showed that we can always think of a grading as combining $c$ \acp{CD} symmetrically, because computing the intersections in Eq.~(\ref{eq:grading_from_set}) is associative and commutative. This yields $2^c$ subspaces $\mfg_s$ that give rise to many new \acp{CD} of $\mfg$ and of subalgebras $\mfg_Q$ composed of multiple $\mfg_s$.
In contrast, an $r$-recursive \ac{CD} has a fixed hierarchy of decompositions, and in particular the horizontal spaces are left untouched by all following recursion steps, leading to $r+1$ subspaces only, namely $\left\{\mfp^{(\ell)}\right\}_{\ell=1}^r$ and $\mfk^{(r)}$. The associated involutions are defined on differing domains $\mfk^{(\ell)}$ and a priori cannot be extended to larger domains, let alone to all of $\mfg$.
This means that we can view recursive \acp{CD} as the main object for unitary synthesis, whereas Cartan gradings form a simpler and more flexible ``intermediate representation". In the following we will see that gradings allow us to discover new recursions and move between them.

Concretely, a Cartan grading not only implies a recursive \ac{CD}~\cite{dagli2008general}, but the converse (almost) holds as well.
To understand the ``almost", we need to introduce a homogeneity criterion for those \acp{CD} that produce semisimple subalgebras.
\begin{definition}\label{def:homogeneity}
    \Acp{CD} of type AIII, BDI or CII produce non-simple subalgebras from simple components.
    A recursive \ac{CD} is called \emph{homogeneous} if all simple components of each such non-simple subalgebra are decomposed with the same decomposition type at all subsequent recursion steps.
\end{definition}

We want to note a few things about this definition.
First, it can be possible to make an inhomogeneous recursive \ac{CD} homogeneous by simply rewriting it. This is because even if \acp{CD} of different types act at the same recursion level, making the recursion inhomogeneous, it might be possible to exchange some of them with trivial decomposition steps to obtain an equivalent homogeneous recursion.
Second, recursive involutions of type AIII, BDI, CII lead to a tree structure of simple components, and homogeneity has to be checked on the complete tree, implying homogeneity on subtrees. The converse is not true.
Third, the homogeneity criterion is a generalization of the requirement that an involution, or $\Ztwo$ symmetry, be either unitary or antiunitary. While this is always true on $\mfu(n)$ (Wigner's theorem), an involution on $\mfu(p)\oplus \mfu(q)$, arising from an AIII decomposition, might perform an (anti)unitary transformation on $\mfu(p)$ ($\mfu(q)$), leading to a ``unitarily indefinite" transformation such as $\theta(x_p\oplus x_q)=x^*_p\oplus x_q$ (of type AI$\oplus$AIII). 
A homogeneous recursive \ac{CD} is not allowed to contain such indefinite transformations, which generalizes the notion from ``unitarily definite" to ``type-definite". For example, even though (DIII$\oplus$BDI) has a unitarily definite involution, it is not type-definite and thus can make a recursive \ac{CD} inhomogeneous.

With the notion of a homogeneous recursive \ac{CD} in our hands, we now can state the mathematical main result of this section.

\begin{restatable}{thm}{gradingrecursionequivalence}\label{thm:grading_recursion_equivalence}
    Consider a Cartan $c$-grading $\mfg=\bigoplus_{s\in\Ztwo^c}\mfg_s$. Then the spaces
    \begin{align}
        \mfk_{\ell}=\bigoplus_{s\in\Ztwo^{c-\ell}} \mfg_{0^{\ell-1} 0\mathlarger{s}}\qquad
        \mfp_{\ell}=\bigoplus_{s\in\Ztwo^{c-\ell}} \mfg_{0^{\ell-1} 1 \mathlarger{s}},
    \end{align}
    for $1\leq \ell\leq c$ 
    and $\mfk_0=\mfg$, define a $c$-recursive \ac{CD}~\cite[Prop~3.2]{dagli2008general}.
    Conversely, consider a classical simple Lie algebra $\mfg$ with an $r$-recursive \ac{CD} that can be made homogeneous. Then there is at least one Cartan $r$-grading that reproduces this recursive \ac{CD}.
\end{restatable}
This proposition, which we prove in App.~\ref{sec:involutions_calculations:HOCD}, provides the promised translation between gradings and recursive \acp{CD}, allowing us to traverse the space of recursions by intermediately converting to gradings.
As mentioned in the last section, a Cartan $c$-grading remains intact if we reorder the \acp{CD}. Correspondingly, a single grading can be turned into $c!$ recursive \acp{CD}.
We say that we \emph{lift} an $r=2$-recursive \ac{CD} $T_1\to T_2$ to a set of two \acp{CD} with types $T_1$ and $T_2'$, denoted as $(T_1\to T_2)\nearrow \{T_1, T_2'\}$.
As indicated by saying ``at least'' in the proposition, the lifted type $T_2'$ is not necessarily unique; see the example below.
This is because there may be more than one way to extend the involution of $T_2$ from $\mfk^{(1)}$, where it is defined within the recursion, to all of $\mfg$.
As our proof is constructive, we obtain explicit lifts for all homogoneous $2$-recursive \acp{CD} in App.~\ref{sec:involutions_calculations:HOCD}, and list them in Tab.~\ref{tab:2_recursive_lifts} without claiming completeness.

Overall we proved in this section that recursive \acp{CD} and Cartan gradings induce each other, up to exceptions specified by the homogeneity condition on recursive decompositions.
More precisely, via the equivalence of Cartan gradings and sets of compatible \acp{CD}, we obtain the following structure:
\begin{center}
\begin{tikzcd}[column sep=1em, row sep=4em, wire types={n,n}]
\text{$c$ compatible CDs} 
\arrow[rr, shift left=0ex, Leftrightarrow, "\text{Prop.~\ref{prop:grading_to_set_of_cds}}" near start, "\text{Prop.~\ref{prop:set_of_cds_to_grading}}" near end, shorten=2mm]
& \text{ } &
\text{Cartan $c$-grading} 
\arrow[ld, shift left=0ex, Leftrightarrow, "\text{Thm.~\ref{thm:grading_recursion_equivalence}}", end anchor={[xshift=-1.5em]north east}, shorten=2mm] 
\\
& \text{$c$-recursive CD} & 
\end{tikzcd}
\end{center}
This tells us that we may use these higher-order \ac{CD} concepts interchangeably. In particular, if a relevant property can be evaluated on a Cartan grading, this can be expected to be easier and faster than considering all possible recursions that it produces.
In addition, gradings allow us to move between recursive \acp{CD}, which are crucial algebraic objects for unitary synthesis in quantum compilation.
We visualize the three concepts for $c=r=3$ in Fig.~\ref{fig:set_grading_recursion_tree} and illustrate the (non-)commutative nature of Cartan gradings (recursive \acp{CD}) in Fig.~\ref{fig:grading_recursion_block_mat}. There, changing the role of two \acp{CD} merely changes the labels of the subspaces in the grading, whereas the recursion changes its structure altogether.

\begin{figure}
    \centering
    \includegraphics[width=\linewidth]{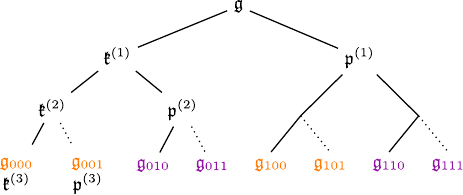}
    \caption{Visualization of higher-order \acp{CD}, here of order three. The $2^3=8$ leaf nodes of the tree form a Cartan $3$-grading of $\mfg$, indexed by bit strings $s\in\Ztwo^3$. This grading arises from three compatible \acp{CD} $\mfg=\mfk_j\oplus\mfp_j$, one of which is applied at each level of the tree. The respective vertical (horizontal) subspaces are composed of 1) the left (right) half of the leaf nodes, 2) the orange (purple) leaf nodes, and 3) the solid (dotted) edges.
    The grading also induces recursive decompositions, one of which is given by the \acp{CD} $\mfg=\mfk^{(1)}\oplus\mfp^{(1)}$, $\mfk^{(1)}=\mfk^{(2)}\oplus\mfp^{(2)}$, and $\mfk^{(2)}=\mfk^{(3)}\oplus\mfp^{(3)}$. Some subspaces of the different higher-order decompositions coincide, namely $\mfk^{(1)}=\mfk_1$, $\mfp^{(1)}=\mfp_1$, $\mfk^{(3)}=\mfg_{000}$, and $\mfp^{(3)}=\mfg_{001}$.
    }
    \label{fig:set_grading_recursion_tree}
\end{figure}

\begin{figure}
    \centering
    \includegraphics[width=\linewidth]{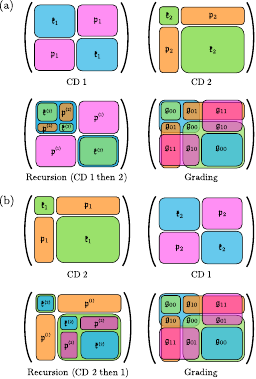}
    \caption{(a) A schematic visualization of the subspaces of Secs.~\ref{sec:gradings} and~\ref{sec:recursions} in the case of two compatible \acp{CD} of type AIII (top). These two decompositions induce a recursive decomposition of the Lie algebra (bottom left) and a $\Ztwo^2$-grading (bottom right). (b) The subspaces obtained when considering a different ordering of the decompositions (top). As discussed in Section~\ref{sec:recursions}, we obtain a different recursive decomposition of the Lie algebra compared to (a) (bottom left), but an evidently isomorphic $\Ztwo^2$-grading (bottom right).}
    \label{fig:grading_recursion_block_mat}
\end{figure}

\subsubsection{Example}
As an example,
consider $\mfg=\mfu(2n)$ and $\theta_1=\ast$. Then $\theta_1$, which is the generic involution of type AI, yields the subalgebra $\mfk^{(1)}=\mfso(2n)$, if we move the global phase to $\mfp^{(1)}$. We define $\theta_2$ on $\mfk^{(1)}$ as $\theta_2=\Ad_{J_n}$, with $J_n$ from Eq.~(\ref{eq:def_J_n}), which is the generic involution of type DIII and yields $\mfk^{(2)}=\mfu(n)\subset\mfso(2n)$.
We may extend $\theta_2$ to all of $\mfg$ as $\tilde{\theta}_2=\Ad_{J_n}$, and find that it is of type $T_2'=$AIII, because $\Ad_{J_n}=\Ad_{iJ_n}$ and $(iJ_n)^\dagger=iJ_n$ (c.f.~Tab.~\ref{tab:all_cartan_involutions}).
We thus found a new \ac{CD} of $\mfg$ of type AIII, and together with $\theta_1$ it gives rise to a $2$-grading that would reproduce the recursion we started out with.
As we discussed, the order within the bit strings does not have an impact on the grading beyond relabeling. In contrast, swapping $T_1$ and $T_2'$ changes the induced $2$-recursive \ac{CD} to $T_2'\to T'_1$ (c.f.~Fig.~\ref{fig:grading_recursion_block_mat}), where $T_1'$ is the restriction of $T_1$ to the subalgebra $\mfk^{(2)\,\prime}=\mfu(n)\oplus\mfu(n)$ of $T_2'$. We find $T'_1=$AI$^{\oplus 2}$ and $\mfk^{(1)\,\prime}=\mfso(n)\oplus\mfso(n)$.

As mentioned above, the lift is not unique. Here, we note that $\ast=\mathrm{id}$ on $\mfk^{(1)}$ by construction, so that we can rewrite $\theta_2=\Ad_{J_n}\circ\,\ast$. We find the corresponding extension $\tilde{\theta}_2=\Ad_{J_n}\circ\,\ast$ acting on $\mfg$, which is of type AII because $J_n^T=-J_n$.
The grading thus now arises from different \acp{CD} of $\mfg$, and the recursion obtained from reordering them is $T_2'\to T'_1=$AII$\to$CI ($\theta_1$ restricted to $\mfsp(n)$ is of type CI) with $\mfk^{(1)\,\prime}=\mfu(n)$.
Note how exchanging the lift carried us from orthogonal to symplectic subalgebras.

\subsubsection{Literature examples}
For the decompositions in the literature discussed in Sec.~\ref{sec:recursive_kak_literature}, we look to the recursion alternating between \acp{CD} of types AIII and A, as well as its siblings alternating BDI and BD or CII and C, respectively. They always merge two branches during their second (fourth, sixth, ...) decomposition that are created by their first (third, fifth,...) decomposition.
This clearly makes them homogeneous and thus allows us to lift the \ac{KGD}, the \ac{QSD}, and the Block-ZXZ decomposition to a grading on $\mfu(2^N)$. According to Tab.~\ref{tab:2_recursive_lifts} in App.~\ref{sec:involutions_calculations:HOCD}, we have the following lifts:
\begin{alignat}{6}
    &(\mathrm{AIII}&\to& \mathrm{A}&)\nearrow &\ \{\mathrm{AIII}&, \mathrm{AIII}\},\\
    &(\mathrm{A}&\to& \mathrm{AIII}&)\nearrow &\ \{\mathrm{A}&, \mathrm{AIII}\},\\
    &(\mathrm{AIII}&\to& \mathrm{AIII}&)\nearrow &\ \{\mathrm{AIII}&, \mathrm{AIII}\},
\end{alignat}
producing a type-AIII \ac{CD} at every iteration. 
If we ignore the AI decomposition in the last step of the \ac{KGD}, we thus can confirm the manual construction of the grading for the \ac{KGD} in~\cite[Sec.~4.1]{dagli2008general} and extend it to \ac{QSD} and Block-ZXZ. However, due to the final AI decomposition, the grading constructed in~\cite{dagli2008general} in fact does \emph{not} reproduce the recursive \ac{KGD}. We explain the details and give a suitable modification for the \ac{KGD} in App.~\ref{sec:grading_for_literature}.

\subsection{Cartan involutions and \texorpdfstring{$\Ztwo$}{Z2} symmetries}\label{sec:involutions_symmetries}
Before we conclude the mathematical discussion of Cartan involutions, we want to address their intimate relationship with $\Ztwo$ symmetries, which also was considered in~\cite{albertini2006analysis}. This allows us to use our tools for \acp{CD} to understand and classify such symmetries better, and to construct new ones.

$\Ztwo$ symmetries are inherently involutions that act on Lie algebra elements like a Hamiltonian or an observable, or on group elements like the time-evolution operator or a quantum circuit.
However, they do not have to be \emph{Cartan} involutions, i.e.,~they might not preserve the algebra/group structure.
In the following we restrict ourselves to symmetries that do have this additional property, allowing us to translate a number of common physical notions to the setting of (higher-order) \acp{CD}. We focus on the algebraic level, with Hamiltonians forming the prime example for an algebra element.
Symmetry-restricted subspaces are subalgebras, whereas antisymmetric spaces are horizontal spaces of a \ac{CD}.
It is then clear why restrictions to antisymmetric subspaces do not lead to mathematically well-defined descriptions of a physical system; such spaces are not closed under commutators, like the time-evolution generator under a Hamiltonian $H$.
Translating Prop.~\ref{prop:unique_CD_of_H}, we find that a Hamiltonian $H$ can be antisymmetric with respect to at most one involution on its \ac{DLA}, computed with respect to some operator basis.
If this basis is the Pauli basis and $H$ consists of a minimal generating set of terms for the \ac{DLA}, Prop.~\ref{prop:CD_of_H_exists_Pauli} tells us further that such an antisymmetry always exists.
To the best of our knowledge, these statements are not commonly known, and they follow immediately from the mentioned propositions, showcasing the use of our framework for basic physics.

As for the combination of multiple $\Ztwo$ symmetries, we (re)discover that they must be compatible in the sense of Def.~\ref{def:compatible_set_of_cds} in order to compose into a new symmetry. That is, they must commute, or their symmetry sectors must share a basis. Further, Cartan gradings of a \ac{DLA} reproduce the intuitive nested structure of symmetry sectors, which, e.g.,~allows us to restrict to the symmetric subspace of one symmetry at a time. 
Other results for \acp{CD} can be used constructively. Prop.~\ref{prop:grading_to_set_of_cds} tells us how to construct new subalgebras and/or symmetries simply by looking at $\Ztwo^c$, its subgroups, and homomorphisms to $\Ztwo$.
This can also be used to construct (polynomial) algebras of any desired type with ``almost Pauli words", i.e.,~low-rank Pauli sentences, as generators. This is contrast to algebras strictly based on Pauli words, which, for example, can not produce polynomially-sized symplectic algebras~\cite{aguilar2024full}.

Classifying the given symmetries then allows us to characterize the symmetry sectors of newly composed symmetries via Lemma~\ref{lemma:composition_of_involutions} and the proof of equivalence in Thm.~\ref{thm:grading_recursion_equivalence} enables us to switch between simultaneously defined symmetries and those that are defined iteratively on a given sector of previous symmetries.
Our finding that lifts of recursive \acp{CD} are not unique implies that symmetries defined in such an iterative manner do not necessarily fix global symmetries.

Finally, we note that the numerical methods for (recursive) \kak decompositions in Sec.~\ref{sec:numerical_decomps} below will allow a decomposition of group elements into purely symmetric ($K$) and purely antisymmetric blocks ($A$) with respect to one or multiple $\Ztwo$ symmetries.
This puts these methods into a very broad context for physics in general, and we hope that the unified framework we present for them will enable practitioners to exploit $\Ztwo$ symmetries in compilation more easily.

\section{Numerical \kak decompositions}\label{sec:numerical_decomps}

\begin{table}[h!]
    \centering
    \begin{tabular}{cA}
        Type & & \phantom{=}\text{Decomposition}\\\midrule
        A & U\oplus U' &= (U_1\oplus U_1)(D\oplus D^\dagger)(U_2\oplus U_2)\\
        AI & U&=O_1 D O_2 \\ 
        AII & U&=S_1(D\oplus D)S_2 \\
        AIII & U&=K_1 F K_2 \\\midrule
        BD & O\oplus O'&=(O_1\oplus O_1)(\mu\oplus \mu^T)(O_2\oplus O_2)\\
        BDI & O&=K_1 F K_2 \\
        DIII & O&= U_1(\mu\oplus\mu^T)U_2 \\\midrule
        C & S\oplus S'&=(S_1\oplus S_1)(\fsl{D}\oplus \fsl{D}^\dagger)(S_2\oplus S_2)\\
        CI & S&=U_1 \fsl{D} U_2\\
        CII & S&= K_1(F\oplus F^T) K_2
    \end{tabular}
    \caption{Generic numerical \kak decompositions obtained from Thm.~\ref{thm:abstract_numerical_non_general}, with the involution given by the generic choice in Tab.~\ref{tab:all_cartan_involutions} and the \ac{CSG} fixed by the middle element of the decomposition. Matrices can be identified by their letter, with $U\in\gru(n)$, $O\in\grso(n)$, $S\in\grsp(n)$, (c.f.~Eqs.~(\ref{eq:def_sun})-(\ref{eq:def_spn})) $D\in\grudiag(n)$, $\fsl{D}\in\grspdiag(n)$, $F\in\grcs(p,q)$, $\mu\in\grschur(n)$, (c.f.~Eqs.~(\ref{eq:def_matrix_csgs})) and $K_i\in\groupG(p)\times\groupG(q)$ with $\groupG=\gru$ (AIII), $\groupG=\grso$ (BDI), and $\groupG=\grsp$ (CII). For DIII and CI, $U_i\in\grso(2n)\cap\grsp(n)\cong\gru(n)$.}
    \label{tab:numerical_overview}
\end{table}

As discussed in Sec.~\ref{sec:groups}, for a given target $G\in \groupG$ the general structure theory of semisimple Lie groups implies the \textit{existence} of a \kak decomposition; the goal of this section is to actually construct these factorizations for the defining representations of all types in Tab.~\ref{tab:symm_classif}.
These numerical algorithms are of practical use in general, as they provide a large number of constructive matrix factorizations~\cite{edelman2023fifty}.
In the specific context of unitary synthesis, they will allow us to compile Hamiltonian time evolution via recursive \kak decompositions in Sec.~\ref{sec:hamsim}, and they guarantee that circuit templates like those discussed in Sec.~\ref{sec:recursive_kak} can be constructed explicitly.

We will combine insights about the general structure of global \acp{CD} to construct any of the \kak decompositions, in a generic form, from a general-purpose \ac{EVD} algorithm, a complex-valued \ac{CSD}, or a real-valued \ac{CSD}.
Our implementation of these numerical \kak decompositions in Python are available at~\cite{symmetrycompilationrepo}.
We will use four abelian subgroups, with $r=\min(p,q)$ and $n_0=n\!\!\mod 2$:
\begin{align}
    \grudiag(n)&=\left\{
        \diag(e^{i\alpha_1},\dots,e^{i\alpha_n})|\alpha_j\in \mbr
    \right\}, \label{eq:def_matrix_csgs}\\
    \grsp_\text{diag}(n)&=\left\{
        D\oplus D^\dagger|D\in\grudiag(n)
    \right\}, \nonumber\\
    \grcs(p, q) &= \left\{\left.
        \begin{pmatrix}
            C & 0 & S\\
            0 & \id_{n-2r} & 0\\
            -S & 0 & C
        \end{pmatrix}
    \right|
        \begin{array}{c}
            C_j=\cos(\alpha_j), \\
            S_j=\sin(\alpha_j), \\
            1\leq j \leq r
        \end{array}
    \right\}, \nonumber\\
    \grschur(n)&=\left\{
        \bigoplus_{j=1}^{\left\lfloor n/2\right\rfloor}\begin{pmatrix}
            c_j & s_j\\-s_j & c_j
        \end{pmatrix} \oplus \id_{n_0}
    \bigg |
        \begin{array}{c}
            c_j=\cos(\alpha_j), \\
            s_j=\sin(\alpha_j)
        \end{array}
    \right\}. \nonumber
\end{align}

\begin{restatable}{thm}{abstractnumericalnongeneral}\label{thm:abstract_numerical_non_general}
    Let $\groupP=\groupG/\groupK$ be a symmetric space from Tab.~\ref{tab:symm_classif}.
    The corresponding generic \kak decomposition in Tab.~\ref{tab:numerical_overview} can be implemented using a standard implementation of a unitary \ac{EVD} or \ac{CSD}.
\end{restatable}

Many key ideas for the constructive proof this theorem—which we leave to App.~\ref{sec:numerical_details}—already exist in the literature. For some symmetric spaces we mostly refer to such previous results, but we will provide some general auxiliary lemmas that help elucidate the underlying structure of these case-specific constructions.
A central object for this generalized structure is
\begin{equation}\label{eq:Delta}
    \Delta=G\Theta(G)^{-1},
\end{equation}
where $\Theta$ is the Cartan involution at the group level. 
In the context of K\"ahler structures, which we discuss in more detail in App.~\ref{sec:kahler}, $\Delta$ is called the \emph{relative complex structure}~\cite{hackl2021bosonic}.

Finally, we conjecture that a fully unified proof, which traces all \kak decompositions back to the eigenvalue problem, is possible, with many important generalization steps taken in App.~\ref{sec:numerical_details}, but we leave it for future work to complete such a proof.

\section{Constructing new circuit decompositions}\label{sec:new_decomps}
Here we construct three new recursive decompositions for arbitrary $N$-qubit unitaries, which emerge naturally from our mathematical framework. 
The first strategy uses an initial AI decomposition to break a unitary $U$ into gates generated by $\mfso$ and $\mfa_{\rm AI}\subset\mfso^\perp$, and then applies a recursive BDI+BD strategy to decompose the resulting orthogonal gates. The second strategy follows a similar route using the unitary symplectic group, i.e., beginning with an AII decomposition, then recursively applying CII+C decompositions. We call these the \emph{completely orthogonal} and \emph{completely symplectic} decompositions, respectively.
The third strategy provides a recursive \ac{CD} that does not introduce excess parameters into the circuit description. As such, it provides a unique path for parameter-optimal decompositions of arbitrary unitary, orthogonal, or symplectic matrices.

\subsection{Completely orthogonal decomposition}

\begin{figure*}[ht]
    \centering
    \includegraphics[width=\linewidth]{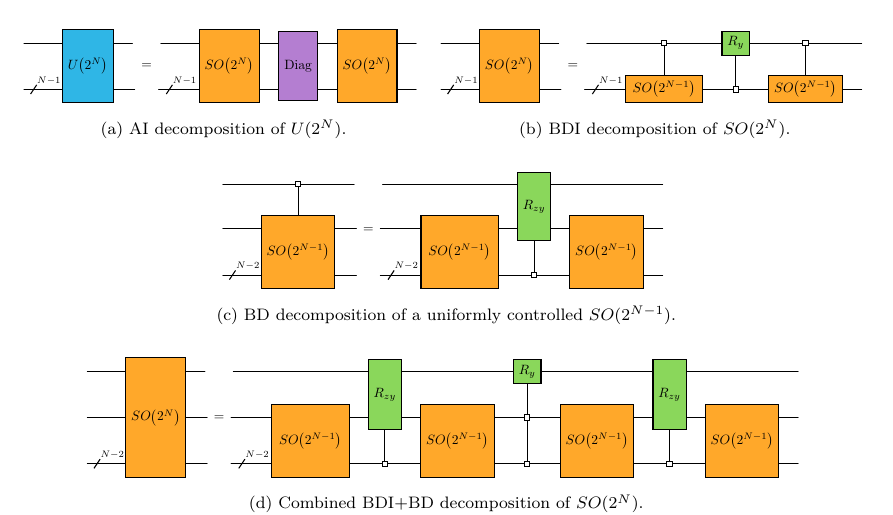}    
    \caption{\Ac{CD} steps used to construct the completely orthogonal decomposition. The $\scriptstyle\square$ symbol indicates a ``uniformly controlled rotation'' or ``quantum multiplexor'' gate, where a different gate is applied to the target for every bitstring value of the control qubits. For every step, we have fixed two degrees of freedom, associated with a choice of basis and a choice of \ac{CSA}. Using different choices would lead to different, unitarily equivalent, decompositions.}
    \label{fig:all_COD_decomps}
\end{figure*}

The completely orthogonal decomposition begins by breaking down an arbitrary $N$-qubit unitary using the subgroup of orthogonal gates. These orthogonal gates are then decoupled from each qubit one-by-one. The full sequence of recursions down to single-qubit gates is
\begin{align}
    \gru\!\left(2^N\right)&\overset{{\rm AI}}{\longrightarrow}\grso\!\left(2^{N}\right)\nonumber\\
    &\overset{{\rm BDI}}{\longrightarrow}\grso\!\left(2^{N-1}\right)\times \grso\!\left(2^{N-1}\right)\nonumber\\
    &\overset{{\rm BD}}{\longrightarrow}\grso\!\left(2^{N-1}\right)\\
    & \quad\vdots\nonumber \\
    &\overset{{\rm BDI}}{\longrightarrow}\grso(2)\times \grso(2)\nonumber\\
    &\overset{{\rm BD}}{\longrightarrow}\grso(2),\nonumber
\end{align}
where the BDI+BD subroutine takes place $N-1$ times.
We depict the constituent steps of this decomposition in Fig.~\ref{fig:all_COD_decomps}.

\subsubsection{AI decomposition of \texorpdfstring{$\gru\!\left(2^N\right)$}{U(2\^{}N)}}
The first step in this decomposition, depicted in Fig.~\ref{fig:all_COD_decomps}a, is to use an AI decomposition to recast a unitary gate as a product of gates $U=KAK'$, where $K^{(\prime)}$ are orthogonal matrices and $A$ is diagonal. More specifically, we partition the algebra $\mfg(N)=\mfu\!\left(2^N\right)=\mfk_{{\rm AI}}(N)\oplus \mfp_{{\rm AI}}(N)$\footnote{This is called the ODO decomposition in~\cite[Thm.~4.1]{edelman2023fifty}.}, with
\begin{align}
    \mfk_{{\rm AI}}(N)&=  \mfso(2^N) \nonumber\\
    &=  \spaniR\{P\in\{\id, \paulix, \pauliy, \pauliz\}^{\otimes N} \vert \#\pauliy~\mathrm{ odd}\}, \\
    \mfp_{{\rm AI}}(N)&=  \mfso(2^N)^\perp \nonumber \\
    &= \spaniR\{P\in\{\id, \paulix, \pauliy, \pauliz\}^{\otimes N}\vert \#\pauliy~\mathrm{ even}\}.     
\end{align}
For the \ac{CSA}, we are free to use the computational basis, 
\begin{align}
    \mfa_{{\rm AI}}(N)=\spaniR\{\ketbra{j}{j}\}_{j=0}^{2^{N}-1}.
\end{align}
This choice generates all diagonal unitaries, i.e., distinct phases applied to each computational basis state.

\subsubsection{BDI decomposition of \texorpdfstring{$\grso(2^N)$}{SO(2\^{}N)}}
For the next step, depicted in Fig.~\ref{fig:all_COD_decomps}b, we break apart $\mfk_{{\rm AI}}(N)$ using a BDI decomposition:
\begin{align}
    \mfk_{{\rm AI}}(N)=\ & \mfk_{{\rm BDI}}(N)\oplus\mfp_{{\rm BDI}}(N), \\
    \mfk_{{\rm BDI}}(N)=\ & \id\otimes\mfso(2^{N-1})\ \oplus\ \pauliz\otimes\mfso(2^{N-1}), \\
    \mfp_{{\rm BDI}}(N)=\ & \paulix\otimes\mfso(2^{N-1})\ \oplus\ \pauliy\otimes\mfso(2^{N-1})^\perp.
\end{align}
We then choose the \ac{CSA} to be 
\begin{align}
    \mfa_{{\rm BDI}}(N)=\spaniR\{\pauliy\otimes\ketbra{j}{j}\}_{j=0}^{2^{N-1}-1},
\end{align}
which is the same choice as the \ac{CSD}. Both $\mfk_{{\rm BDI}}(N)$ and $\mfa_{{\rm BDI}}(N)$
will generate uniformly controlled gates. From $\mfk_{{\rm BDI}}(N)$, we generate $\grso(2^{N-1})$ rotations on qubits $\{2,\dots,N\}$, controlled on the $\ket{0}$ and $\ket{1}$ states of the first qubit. From $\mfp_{{\rm BDI}}(N)$, we generate $\pauliy$ rotations on the first qubit, controlled by all computational-basis states of the remaining qubits.

\subsubsection{BD decomposition of uniformly controlled \texorpdfstring{$\grso\!\left(2^{N-1}\right)$}{SO(2\^{}(N-1))}}
Finally, we decouple the first qubit by using a BD decomposition, depicted in Fig.~\ref{fig:all_COD_decomps}c:
\begin{align}
    \mfk_{{\rm BDI}}(N)=\ & \mfk_{{\rm BD}}(N)\oplus\mfp_{{\rm BD}}(N), \\
    \mfk_{{\rm BD}}(N)=\ & \id\otimes\mfso(2^{N-1}), \\
    \mfp_{{\rm BD}}(N)=\ & \pauliz \otimes\mfso(2^{N-1}).
\end{align}
We choose a skew-symmetric basis for the \ac{CSA},
\begin{align}\label{eq:a_BD_cod}
    \mfa_{\rm BD}(N)=\spaniR\left\{\pauliz\otimes Y \otimes\ketbra{j}{j}\right\}_{j=0}^{2^{N-2}-1}\!\!.
\end{align}

The BDI and BD steps make up a basic recursion subroutine of the completely orthogonal decomposition. They can be repeated as many times as desired, successively decoupling one qubit each time. The combined BDI+BD decomposition is shown in Fig.~\ref{fig:all_COD_decomps}d.
The closest comparable to the decomposition above is the 3-qubit $\grso(8)$ decomposition outlined in~\cite{wei2012decomposition}, which also uses a BDI+BD strategy, but not recursively. 
Comparing our Fig.~\ref{fig:all_COD_decomps}d with Fig.~2 of~\cite{wei2012decomposition}, the circuits have both similarities and differences. These differences highlight the fact that, even with the same choices of \ac{CD} types, different choices of basis rotations or \acp{CSA} can lead to distinct circuit structures. For example,~\cite{wei2012decomposition} uses the two-qubit ``magic basis'', which makes certain gates local. This point is also discussed in detail for the case of AIII+A decompositions in App.~\ref{sec:unified_cartan_decomp_details}.

\subsection{Completely symplectic decomposition}

\begin{figure*}[ht]
   \centering
   \includegraphics[width=\linewidth]{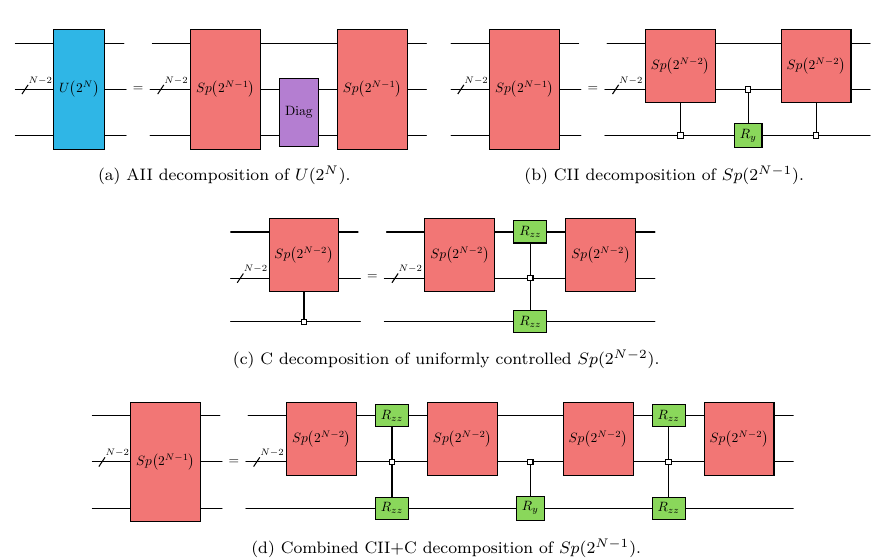}
    \caption{Recursive \ac{CD} steps used to construct the completely symplectic decomposition. It is structurally very similar to the completely orthogonal decomposition in Fig.~\ref{fig:all_COD_decomps}. The first qubit (by convention) plays a special role in the symplectic case, which is why we decouple qubits from the bottom instead of from the top.}
    \label{fig:all_CSD_decomps}
\end{figure*}

This decomposition mirrors the completely orthogonal decomposition, replacing the AI, BDI, and BD decompositions with AII, CII, and C, respectively. The complete set of constituent decompositions is depicted in Fig.~\ref{fig:all_CSD_decomps}, consisting of the following sequence of recursions:
\begin{align}
    \gru\!\left(2^N\right)&\overset{{\rm AII}}{\longrightarrow}\grsp\!\left(2^{N-1}\right)\nonumber\\
    &\overset{{\rm CII}}{\longrightarrow}\grsp\!\left(2^{N-2}\right)\times \grsp\!\left(2^{N-2}\right)\nonumber\\
    &\overset{{\rm C}}{\longrightarrow}\grsp\!\left(2^{N-2}\right)\\
    &\quad\vdots\nonumber \\
    &\overset{{\rm CII}}{\longrightarrow}\grsp(1)\times \grsp(1)\nonumber\\
    &\overset{{\rm C}}{\longrightarrow}\grsp(1),\nonumber
\end{align}
where the CII+C subroutine takes place $N-1$ times.

\subsubsection{AII decomposition of \texorpdfstring{$\gru\!\left(2^N\right)$}{U(2\^{}N)}}
The first step, depicted in Fig.~\ref{fig:all_CSD_decomps}a, is to decompose an arbitrary unitary using an AII decomposition, which splits the Lie algebra into the symplectic subalgebra $\mfg(N)=\mfu(2^N)$ and its complement: 
\begin{align}
    \mfg(N) &=  \mfk_{{\rm AII}}(N)\oplus\mfp_{{\rm AII}}(N), \\
    \mfk_{{\rm AII}}(N) &=  \mfsp\!\left(2^{N-1}\right), \\
    \mfp_{{\rm AII}}(N) &=  \mfsp\!\left(2^{N-1}\right)^\perp.
\end{align}
For our purposes, we represent the symplectic algebra $\mfsp(2^{N-1})$ as the subset of $N$-qubit traceless skew-hermitian matrices which have the following block form in the computational basis:
\begin{align}
    \label{eq:symplectic_block_form_}
    \mfsp\!\left(2^{N-1}\right) = 
    \left\{
    \begin{pmatrix}
        A & B \\
        -B^\dagger & -A^T
    \end{pmatrix} \Bigg\vert
    A^\dagger=-A,B^T=B\right\}.
\end{align}
Notice that the first qubit takes a special role, defining the block-partitioning of the matrix\footnote{Any qubit could play this role in principle, but the form of Eq.~(\ref{eq:symplectic_block_form_}) would change.}. 
We can reparametrize this algebra to the following form:
\begin{align}\label{eq:sp_pauli_blocks}
    \mfsp\!\left(2^{N-1}\right)
    &=\spanR\left\{\id\otimes \mfso\!\left(2^{N-1}\right), \right.\\
    &\left. \hspace{1.5cm}\{X, Y, Z\}\otimes \mfso\!\left(2^{N-1}\right)^\perp\right\},\nonumber
\end{align}
where $\id$ pairs with $\real(A)$, $Z$ with $i\imag(A)$, $X$ with $i\imag(B)$, and $Y$ with $i\real(B)$.
Accordingly, $\mfp_{\rm AII}(N)$ is spanned by $\id\otimes\mfso\!\left(2^{N-1}\right)^\perp$ and $\{X, Y, Z\}\otimes \mfso\!\left(2^{N-1}\right)$.
Hence, we may choose as the \ac{CSA} the $2^{N-1}$ diagonal operators within the first subspace,
\begin{align}
    \mfa_{{\rm AII}}(N)=\spaniR\{\id\otimes\ketbra{j}{j}\}_{j=0}^{2^{N-1}-1},
\end{align}
which are arbitrary diagonal generators on all but the first qubit.

\subsubsection{CII decomposition of \texorpdfstring{$\grsp\!\left(2^{N-1}\right)$.}{Sp(2\^{}(N-1))}}
For the next decomposition, depicted in Fig.~\ref{fig:all_CSD_decomps}b, we define a CII involution on $\mfsp\!\left(2^{N-1}\right)$ by conjugating the last qubit by $\pauliz$:
\begin{align}
    \theta_{\rm CII} = \Ad_{\id_{2^{(N-1)}}\otimes\pauliz},
\end{align}
The $+1$ ($-1$) eigenspace of $\theta_{\rm CII}$ consequently is given by generators with $\pauliz$ or $\id$ ($\paulix$ or $\pauliy$) on the last qubit. For the $+1$ eigenspace, the symplectic structure is maintained on all but the last qubit, while for the $-1$ eigenspace it is maintained for $\paulix$ on the last qubit but for $\pauliy$ the combinations of leading Pauli operators and (skew-)orthogonal subspaces in Eq.~(\ref{eq:sp_pauli_blocks}) are flipped. Overall we find
\begin{alignat}{5}
    \mfk_{\rm CII}(N)
    &=\mfsp\!\left(2^{N-2}\right) \otimes \id\ \oplus\ \mfsp\!\left(2^{N-2}\right)\!\otimes \pauliz,\label{eq:k_CII_}\\
    \mfp_{\rm CII}(N)
    &=\mfsp\!\left(2^{N-2}\right) \otimes \paulix\ \oplus\ \mfsp\!\left(2^{N-2}\right)^\perp\!\otimes \pauliy. \label{eq:p_CII_}
\end{alignat}
Looking at Eqs.~(\ref{eq:sp_pauli_blocks}) and (\ref{eq:p_CII_}), we also identify a suitable \ac{CSA} as 
\begin{align}
    \mfa_{{\rm CII}}(N)=\spaniR\{\id\otimes\ketbra{j}{j}\otimes\pauliy\}_{j=0}^{2^{N-2}-1}.
\end{align}
This \ac{CSA} generates uniformly controlled $\pauliy$ rotations on the last qubit, conditioned on the computational-basis states of the remaining qubits, excepting the first one.

\subsubsection{C decomposition of uniformly controlled \texorpdfstring{$\grsp\!\left(2^{N-2}\right)$}{Sp(2\^{}(N-2))}}
Our final step is to use a type-C decomposition, depicted in Fig.~\ref{fig:all_CSD_decomps}c, to separate the components in Eq.~(\ref{eq:k_CII_}) ending in $\id$ and $\pauliz$ from each other:
\begin{align}
    \mfk_{{\rm CII}}(N)& =  \mfk_{{\rm C}}(N)\oplus\mfp_C(N), \\
    \mfk_{{\rm C}}(N) &= \mfsp(2^{N-2})\otimes\id, \\
    \mfp_{{\rm C}}(N) &= \mfsp(2^{N-2})\otimes\pauliz.
\end{align}
As the \ac{CSA} for this stage we can choose (c.f.~Eq.~(\ref{eq:sp_pauli_blocks}))
\begin{align}
    \mfa_{{\rm C}}(N) &=  \spaniR\left\{Z\otimes\ketbra{j}{j}\otimes\pauliz\right\}_{j=0}^{2^{N-2}-1}\!\!.
\end{align}
This gives rise to uniformly controlled $\pauliz\pauliz$ rotations on the first and last active qubit, with all qubits inbetween acting as multiplexing controls.

The combined CII+C steps can be used recursively to decouple one qubit at a time; see Fig.~\ref{fig:all_CSD_decomps}d. Note that we decouple qubits in reverse order compared to the completely orthogonal decomposition due to the fact that the first qubit has a special role in defining the symplectic structure, which we want to preserve throughout the recursion. Other than the first one, the qubits can be decoupled in any desired order.
Beyond this difference in ordering, we observe that the structure of the completely orthogonal and completely symplectic \acp{CD} is very similar.

\subsection{Parameter-optimal recursive decompositions}\label{sec:overp_free_recursions}
When evaluating a (recursive) decomposition, the literature discussed in Sec.~\ref{sec:recursive_kak_literature} often focuses on CNOT counts, considering single-qubit gates to be significantly cheaper. While this is an adequate priority on noisy quantum computers, it can be expected that the converse will be true for error-corrected quantum computation.
There, CNOTs are commonly considered notably cheaper than a single-qubit rotation (about an arbitrary angle), because the former is a Clifford gate whereas the latter requires $T$ gates. The exact relative cost between Clifford and $T$ gates depends on architectural details and the used error correction code, and is subject to constant change under recent research.
Our summary of the excess parameters introduced by \kak decompositions in Tab.~\ref{tab:symm_classif} allows us to characterize recursions that guarantee optimal parameter counts, with the dimension of the decomposed group providing a lower bound. There is exactly one such recursion for each ``simple" Lie group (up to a global phase) we consider, and they only differ in the starting position within the following chain:
\begin{align}
    \grsp(n)&\overset{CI}{\longrightarrow}\gru(n)\nonumber\\
    &\overset{{\rm AI}}{\longrightarrow}\grso(n)\nonumber\\
    &\overset{{\rm BDI}}{\longrightarrow}\grso(\lceil n/2\rceil)\times \grso(\lfloor n/2\rfloor)\\
    &\overset{{\rm BDI}}{\longrightarrow}\cdots \overset{{\rm BDI}}{\longrightarrow}\grso(q_j+1)^{\times n_j} \times \grso(q_k)^{\times 2^j-n_j}\nonumber\\
    &\overset{{\rm BDI}}{\longrightarrow}\cdots \overset{{\rm BDI}}{\longrightarrow}\bigtimes_{i=1}^{n} \grso(2),\nonumber
\end{align}
with $q_j=\lfloor n / 2^j\rfloor$ and $n_j=n \!\!\mod \!2^j$ after $j$ BDI decompositions.
Note that this chain exists for any $n$, and that BDI decompositions are exactly parameter-optimal for the two scenarios $p=q$ and $p=q+1$ that we require.
The chain starting at $\gru(n)$ also was noted to be parameter-optimal in~\cite{fuhr2018note}.

Provided a suitable basis choice, this translates into a minimal number of rotation angles and we expect this to lead to good $T$ gate counts.
While other decompositions that introduce an overparametrization can still be optimized to remove the excess degrees of freedom again, this requires additional compilation efforts, and it is not clear whether it is possible to attain minimal parameter counts from a generic decomposition that was designed to minimize the CNOT count.

Another reason to turn to parameter-efficient recursive \acp{CD} will be showcased for the compilation of a Hamiltonian simulation problem in the next section. There, we perform the recursive decomposition in an entirely different representation than that on qubits, so that a CNOT-count minimizing decomposition does not provide any immediate value, whereas reduced parameter counts lead to lower compilation cost.
Indeed, in Sec.~\ref{sec:fdhs:example}, we will use the above parameter-optimal chain to decompose $\grso(2n)$.

\section{Compiling Hamiltonian simulation}\label{sec:hamsim}
In this section we will compile the time evolution of a Hamiltonian to a quantum circuit with time-independent depth, by combining a recursive \ac{CD} (Sec.~\ref{sec:recursive_kak}) and the introduced numerical decomposition techniques (Sec.~\ref{sec:numerical_decomps}). Such a compilation task and the same example Hamiltonian were tackled with a variational approach in~\cite{kokcu2022fixed}, and we will compare the two approaches further below.

\subsection{Algorithm}\label{sec:hamsim:algo}
\begin{figure}
    \centering
    \includegraphics[width=\linewidth]{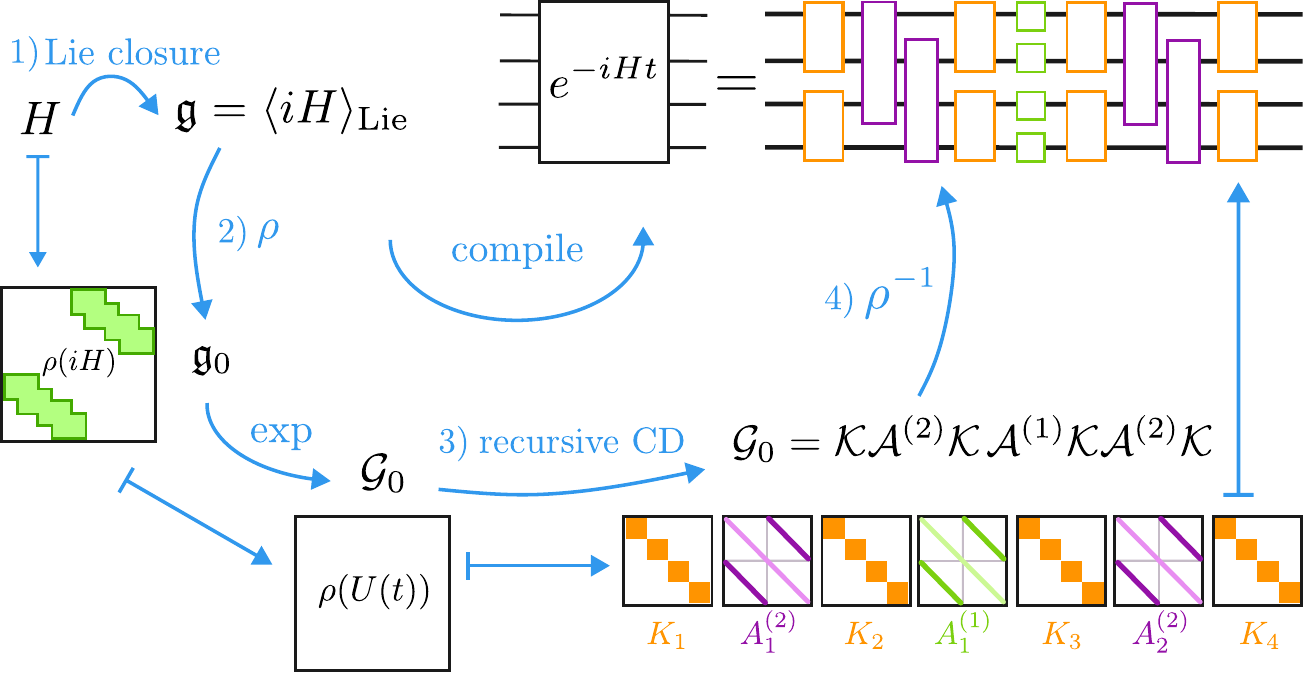}
    \caption{Our compilation algorithm for fixed-depth Hamiltonian simulation. For a given Hamiltonian and basis set, we obtain its \ac{DLA} $\mfg$ and find a mapping $\rho$ to the irreducible representation $\mfg_\circ$ (left).
    For fast-forwardable Hamiltonians, the mapped Hamiltonian $\rho(iH)$ can be exponentiated and the time evolution $\rho(U(t))$ can be broken down with a recursive \ac{CD} using the numerical techniques from Sec.~\ref{sec:numerical_decomps} (bottom). The matrix factors can then be mapped back to a quantum circuit in the original representation via $\rho^{-1}$ (right), for example by extracting their generators and mapping them to Pauli word generators (not shown here).}
    \label{fig:fdhs_algo}
\end{figure}
Our algorithm takes as input a Hamiltonian $H$ expressed in the Pauli basis and one or multiple evolution times $t$, as well as a desired recursive \ac{CD}; see step three below for details. It outputs a quantum circuit that implements the time evolution under $H$ exactly, in terms of Pauli rotation gates and their rotation angles for the different evolution times. The algorithm consists of four steps, visualized in Fig.~\ref{fig:fdhs_algo}.

First, the \ac{DLA} $\mfg$ of $H$ is computed with respect to the Pauli basis (see Def.~\ref{def:dla}), and an irreducible representation $\mfg_\circ$ of $\mfg$, which in turn most likely is reducible, is identified.
The Lie closure is straightforward to compute assuming proper conditioning of the involved coefficients\footnote{Straightforward does not imply efficient, it's just a reasonably simple task.}. Identifying $\mfg_\circ$ may be done in various ways, as we describe in App.~\ref{sec:fdhs_calculations:identify}.
In our example below, we will make use of a particularly simple consequence of the recent classification of all Pauli Lie algebras via their anticommutation graph classes~\cite{aguilar2024full}.

Second, an algebra isomorphism $\rho:\mfg\to\mfg_\circ$ is obtained and $iH$ is mapped to $\rho(iH)$. Note that we can also associate a group isomorphism to $\rho$, which we will denote with the same symbol in this section. For structured systems like spin chains, this can commonly be done manually. However, with the prospect of depending less on this type of manual work, we also provide an automation of this step for horizontal Hamiltonians in Sec.~\ref{sec:hamsim:variant}.
Then the target time evolution $\rho(U(t))=\exp(\rho(iH)t)$ is computed as a dense matrix. For the example below, we provide a physically motivated manual mapping, but also showcase the automated mapping algorithm.

Third, a recursive \ac{CD} of the algebra is specified, which can be done conveniently by defining the respective involutions.
For this, a set of compatible \acp{CD} (commuting involutions, c.f.~Prop.~\ref{prop:set_of_cds_to_grading}) and one of the recursions induced by the resulting Cartan grading may be chosen manually. The space of recursions could also be searched automatically\footnote{Meaning that the chain of involution \emph{types}, rather than concrete involutions, is searched. The basis choices can then be found in another step.}, optimizing a user-specified metric of the resulting circuit, like depth, CNOT count, or parameter count.
We will not outline such an automatic approach here but will consider the recursive \ac{CD} to be an additional input.
To execute the step, $\rho(U(t))$ from step two is decomposed with the chosen recursion. This yields a series of unitary operators $\{K_i\}$ and $\{A^{(j)}_k\}$ from the spaces $\groupK_r$ and $\{\groupA_j\}_{j=1}^r$, respectively, such that 
\begin{align}
    \rho(U(t)) = [K_1 A^{(r)}_1 K_2] A^{(r-1)}_1 [K_3 A^{(r)}_2 K_4]\cdots.
\end{align}

Fourth, the unitary operators are mapped back to the reducible representation by extracting their generators and applying $\rho^{-1}$ to them. This yields generators of quantum gates and thus a quantum circuit description for $U(t)$.
Note that the compiled circuit structure for $\exp(iHt)$ is fixed, but different parameter values are needed for different evolution times $t$.
In particular, the vertical group elements $K_i$ differ between evolution times.

\subsection{Variant for horizontal Hamiltonians}\label{sec:hamsim:variant}
Note that in the previous section we did not assume any particular relationship between the (recursive) \ac{CD} and the Hamiltonian defining the time evolution, which makes our algorithm more versatile and removes a constraint present in other approaches~\cite{kokcu2022fixed}.
However, if we are guaranteed that $\rho(iH)$ lies in the horizontal subspace of the first \ac{CD} in the recursion, we can exploit that the first \kak decomposition satisfies $K_2=K_1^\dagger$. Concretely, this implies that the time-evolution operator is given by
\begin{align}
    \rho(U(t))=K A(t) K^\dagger=K\exp(at)K^\dagger,
\end{align}
so that the compiled circuit can be used for all evolution times simply by rescaling the elements $at$ of the \ac{CSA}.

If we demand that $\rho(iH)$ be in the first horizontal space of the recursion, this is a requirement not only for $H$, but also for $\rho$ and the first involution: the terms in $H$ need to fit into the horizontal subspace\footnote{That is, they need to be isomorphic as an algebra subspace to a part of the horizontal space.} of a suitably chosen involution, and the mapping $\rho$ needs to respect this relationship by mapping $iH$ into the horizontal subspace. Prop.~\ref{prop:unique_CD_of_H} tells us that there can only be one such involution for fixed $H$, operator basis, and $\rho$. This interplay suggests to merge steps two and three into a single step.
An algorithm for automatically determining a mapping $\rho$—if it exists—that places a target Hamiltonian $H$ into the canonical horizontal space of the BDI involution is given for the free fermionic case in Sec.~\ref{sec:hgraph}.

\subsection{Example}\label{sec:fdhs:example}
Now we compile the time evolution under a free-fermionic Hamiltonian, using the algorithm variant from Sec.~\ref{sec:hamsim:variant}. We defer details of the following calculations to App.~\ref{sec:fdhs_calculations}.
Consider the transverse field XY model on a one-dimensional spin chain of length $n$ (i.e., acting on $n$ qubits), with open boundary conditions.
It is described by the Hamiltonian
\begin{align}\label{eq:hamsim:H}
    H=\sum_{i=1}^{n-1} \alpha^X_i X_i X_{i+1}+\alpha^Y_i Y_iY_{i+1} + \sum_{i=1}^n \beta_i Z_i,
\end{align}
where $X_i, Y_i, Z_i$ are the standard Pauli operators acting on the $i$th site, $\{\alpha^{X/Y}_i\}_i$ are coupling strengths, and $\{\beta_i\}_i$ are field strengths.

First, we compute the \ac{DLA} $\mfg=\langle H\rangle_\text{Lie}$. Following~\cite{kokcu2022fixed}, we consider the Pauli basis and find
\begin{align}
    \mfg 
    &= \left\langle\{\widehat{X_iX_j}, \widehat{X_iY_j}, \widehat{Y_i X_j}, \widehat{Y_i Y_j}\}_{1\leq i<j\leq n}\cup \{Z_i\}_{i=1}^n\right\rangle_{i\mbr},\nonumber
\end{align}
with $\widehat{A_iB_j}\coloneqq A_i Z_{i+1}\cdots Z_{j-1} B_j$.
$\mfg$ is isomorphic to $\mfso(2n)$, a classical simple algebra (App.~\ref{sec:fdhs_calculations:algebra}).
Note that obtaining a \ac{DLA} with polynomial size from a set of Pauli operators is not very common, and can only ever yield (copies of) $\mfso$~\cite{wiersema2023classification,kökcü2024classification,aguilar2024full}.
Polynomial \acp{DLA} with different structure and an ``almost-Pauli'' basis can be constructed via \acp{CD} themselves, by choosing Cartan involutions that are not diagonal in the Pauli basis (also c.f.~Sec.~\ref{sec:involutions_symmetries}).

For the next step, we choose the recursive \ac{CD} and the isomorphism $\rho$ synergistically, as anticipated in the horizontal variant of our algorithm.
We choose a repeated BDI decomposition $\mfso(p+q)\to\mfso(p)\oplus\mfso(q)$ for three reasons.
One, this recursion only has one type of \ac{CD} and a single dense matrix representation can be used for the full recursion, i.e.,~the canonical form of the BDI involution is self-consistent (see Sec.~\ref{sec:involutions_calculations:HOCD}).
Two, the full decomposition will be parameter-optimal if we choose $p$ and $q$ to differ at most by one at each recursion step (see Sec.~\ref{sec:overp_free_recursions}). That is, there will be exactly as many group elements in the compiled circuit as there are dimensions in $\mfg$ (namely $2n^2-n$).
Three, the Hamiltonian $H$ fits into the horizontal space of an initial BDI decomposition, allowing us to use the variant described in Sec.~\ref{sec:hamsim:variant} and to showcase the automated mapping from App.~\ref{sec:hgraph}.
According to Prop.~\ref{prop:unique_CD_of_H}, there can't be a DIII decomposition of $\mfg$ with respect to which $H$ is horizontal if we keep the Pauli basis fixed.

\begin{figure*}
    \centering
    \includegraphics[width=\linewidth]{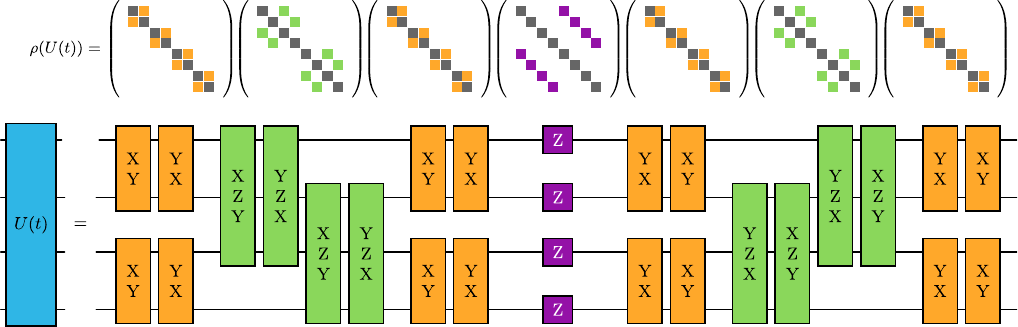}
    \caption{Compilation result for time evolution under the Hamiltonian in Eq.~(\ref{eq:hamsim:H}). After mapping the \ac{DLA} of $H$ to an irreducible representation via an isomorphism $\rho$, we obtain a recursive \kak decomposition (top) of the time-evolution operator $\rho(U(t))$ in the form of a matrix product. The matrix factors consist of commuting Givens rotations that represent Pauli rotations in the original representation, e.g.,~the middle factor above encodes Pauli-$Z$ rotations.
    Applying the inverse map $\rho^{-1}$ to the algebra elements that generate those abelian matrix factors then yields the quantum circuit (bottom) implementing $U(t)=\exp(iHt)$.}
    \label{fig:fdhs_compiled_circuit}
\end{figure*}

In addition to the \textit{type} of the decompositions, we need to fix the bases by means of concrete involutions. Due to the self-consistency of the canonical form of the BDI involution, we do not require intermediate basis changes and may simply use the canonical form on different block matrices.
At each recursion level, the involution is $\theta_j=\Ad_{\id_c\oplus I_{p,q}\oplus\id_d}$. $c$ and $d$ are padding dimensions that allow us to use the BDI involution on a $p+q$-dimensional block of the full matrix representation.

Next, $\mfg$ is mapped to an irreducible matrix representation on $\mbr^{2n}$. We derive $\rho$, which maps $iH$ into the horizontal space of the first BDI decomposition, manually in App.~\ref{sec:fdhs_calculations:mapping} and provide an implementation of the algorithm in Sec.~\ref{sec:hgraph} in code~\cite{symmetrycompilationrepo}.
As an example, for $n=3$ the manually mapped Hamiltonian takes the form
\begin{align}
    \rho(iH)=2\begin{pmatrix}
        0&0&0&\hamcolor{-\beta_1}&\hamcolor{\alpha^X_1}&0\\
        0&0&0&\hamcolor{\alpha^Y_1}&\hamcolor{-\beta_2}&\hamcolor{\alpha^X_2}\\
        0&0&0&0&\hamcolor{\alpha^Y_2}&\hamcolor{-\beta_3}\\
        \hamcolor{\beta_1}&\hamcolor{-\alpha^Y_1}&0&0&0&0\\
        \hamcolor{-\alpha^X_1}&\hamcolor{\beta_2}&\hamcolor{-\alpha^Y_2}&0&0&0\\
        0&\hamcolor{-\alpha^X_2}&\hamcolor{\beta_3}&0&0&0
    \end{pmatrix},
\end{align}
whereas one automatically obtained mapping swaps the fourth column (row) with the and fifth column (row) and flips the sign of the first and last columns and rows.
We will continue with the former of the two mappings.
Then we compute the time evolution $\rho(U(t))$ using a standard numerical matrix exponential and apply the recursive BDI decomposition; see App.~\ref{sec:fdhs_calculations:apply_decomp} for details.
Overall, this leaves us with $\mathcal{O}(n^2)$ Givens rotations, i.e.,~matrices from an embedding of $\grso(2)$, as well as $\mathcal{O}(n^2)$ \ac{CSG} elements that encode multiple such rotations, which commute and thus can be pulled apart.

Finally, we map the obtained decomposition of the $2n\times 2n$ matrix $\rho(U(t))$ back to the original qubit representation. This is particularly simple, because we are left with Givens rotations on the vector space $\mbr^{2n}$, from which the generators and rotation angles can be read out immediately. The former can be mapped to a Pauli word via $\rho^{-1}$, modifying the rotation angle by a prefactor of $\pm 2$; see App.~\ref{sec:fdhs_calculations:mapping_back} for details and an example.
This concludes our compilation of $U(t)=\exp(iHt)$ into Pauli rotation gates.

\subsubsection{Compiled circuit}

The Pauli rotations in the compiled quantum circuit for the time evolution $U(t)$ are single-qubit Pauli-$Z$ rotations (\ac{CSA} generators of the first decomposition), pairs of $XY$ and $YX$ rotations on disjoint neighbouring qubit pairs (vertical generators of the last decomposition), and commuting layers of non-local rotations generated by strings $\widehat{X_i Y_j}$ (intermediate \ac{CSA} generators), with larger gates appearing less often than smaller ones. 
An example circuit is shown together with a visualization of the generators in the irreducible representation in Fig.~\ref{fig:fdhs_compiled_circuit}.
Similar to the techniques shown in~\cite{kokcu2022fixed}, the output circuit can be optimized to reduce, e.g.,~the CNOT count. Groups of the multi-qubit gates benefit from gate cancellations, and the pairs of rotations generated by $X_iY_{i+1}$ and $Y_iX_{i+1}$ can be implemented with just two CNOTs~\cite{vidal2004universal}.

\subsubsection{Compilation performance}
Our algorithm scales polynomially in the dimension of the \ac{DLA} if a hardcoded mapping $\rho$ between the representations is used.
Our procedure to find $\rho$ automatically, on the other hand, uses a heuristic subroutine to solve an NP-complete problem, leading to a prohibitive worst-case computational cost in theory.
To assess the overall computational cost and the impact of this subroutine in practice, we run our compilation algorithm for various spin counts $n$ and report the runtimes at different levels of workflow automation in Fig.~\ref{fig:fdhs_performance}, including the worst-case NP-complete isomorphism finding. We provide the code and data at~\cite{symmetrycompilationrepo}.

\begin{figure}
    \centering
    \includegraphics[width=\linewidth]{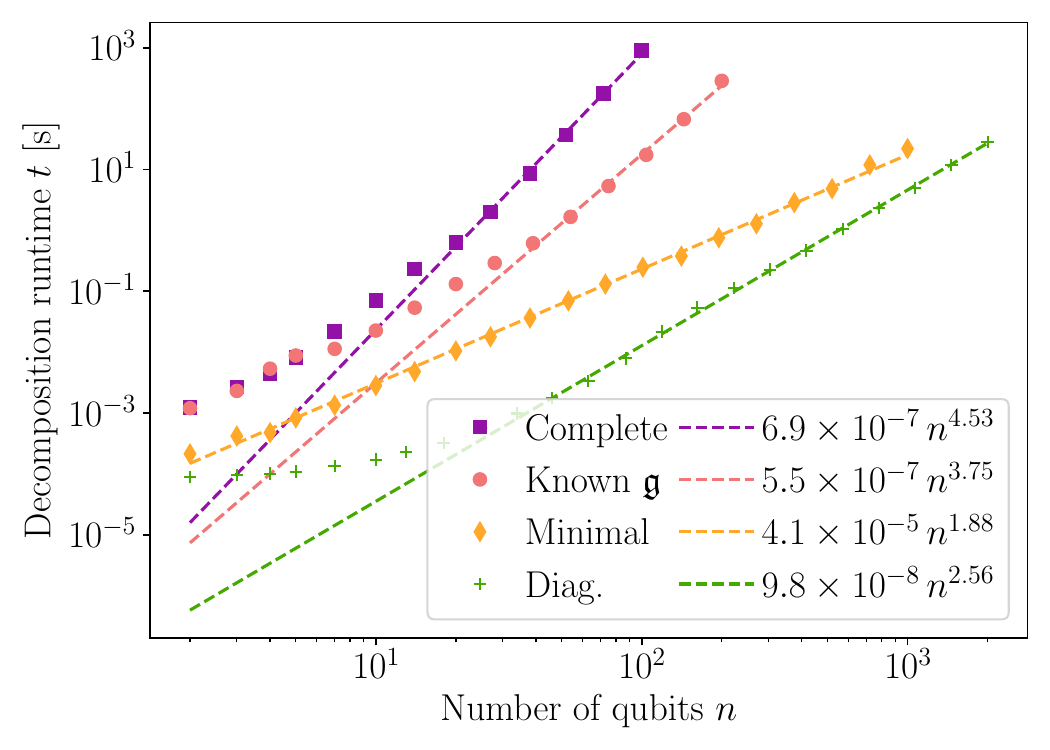}
    \caption{Time to compile the time evolution under the Hamiltonian $H$ of the transverse-field XY model, at different levels of automation.
    The complete workflow computes the \ac{DLA} $\mfg$ of $H$ and finds a mapping $\rho$ to a low-dimensional representation, solving a worst-case NP-complete problem (subgraph isomorphism).
    The intermediate workflow uses a hard-coded version of $\mfg$ instead, but still computes $\rho$ itself.
    The fastest, minimal workflow exploits hard-coded $\mfg$ and $\rho$ (see Apps.~\ref{sec:fdhs_calculations:algebra} and~\ref{sec:fdhs_calculations:mapping}). Computing the algebra is by far the most expensive step for the shown regime, and finding $\rho$ in turn is far more expensive than computing the recursive \ac{CD}, leading to a notable tradeoff between (manual) pre-processing and compilation runtime.
    Finally, we also show the runtime for the first decomposition step of the compiler alone, which computes the spectrum of $H$ (see Sec.~\ref{sec:diagonalization}).
    }
    \label{fig:fdhs_performance}
\end{figure}

The complete workflow takes the Hamiltonian, along with some abstract information about its \ac{DLA} $\mfg$, and the recursive \ac{CD} as inputs.
Accordingly, it has to both compute the qubit representation $\mfg$ and find the homomorphism $\rho$ to a small-dimensional irreducible representation before performing the decomposition itself. We find it to scale roughly as $\mathcal{O}(n^{4.5})$ up to $100$ qubits, for which it requires about $10^3$ seconds.

The intermediate workflow uses a hardcoded version of $\mfg$ instead of computing it automatically, but still computes  $\rho$ itself. We find that this already reduces the compilation time notably, to about $20$ seconds for $100$ qubits. This shows that in this regime computing $\mfg$, an evidently polynomially scaling task, is much more expensive in our implementation than the heuristic solver for the NP-hard mapping problem implemented in NetworkX~\cite{hagberg2008exploring}. These runtimes also expose that the linear fit is not a good approximation any longer, indicating that the heuristic indeed scales mildly super-polynomially.

The fastest, minimal workflow uses a hardcoded mapping for $\rho$ instead, namely that from App.~\ref{sec:fdhs_calculations:mapping}. We find that this again reduces the exponent of the polynomial scaling significantly, which in turn reduces the compilation time on $100$ qubits to $0.3$ seconds. \emph{Putting it all together, compiling the full Hamiltonian simulation for arbitrary $t$ on $10^3$ qubits into $\sim 2\cdot 10^6$ rotation gates takes about $22$ seconds.}
Surprisingly, the minimal workflow scales \emph{sub-linearly} with the dimension of the \ac{DLA}, $\dim(\mfg)=2n^2-n$, for the investigated regime. This indicates that the components of the algorithm that clearly scale at least quadratically, such as the \ac{CSD} at a fixed recursion level, contribute to the total runtime with very small prefactors. This is confirmed by the shown runtime for the very first decomposition step alone, which is used for diagonalization; see Sec.~\ref{sec:diagonalization}.
We suspect that an implementation that performs the required bookkeeping more efficiently than our Python implementation would then exhibit a quadratic scaling.

With this strong performance, our recursive \ac{CD} algorithm provides a scalable solution to compiling the time evolution of Hamiltonians $H$ with polynomial \acp{DLA}, and to obtain diagonalizing circuits for $H$.
The manual labour or additional knowledge required for the fastest variant of the algorithm can be traded-in for additional computational cost, until only a few conceptual aspects of the recursion are required as inputs for the complete workflow.

\subsection{Byproduct: Exact diagonalization}\label{sec:diagonalization}
The numerical implementation of (recursive) \acp{CD} together with automated mapping of horizontal Hamiltonians to small-dimensional representations allows us to diagonalize such Hamiltonians exactly.
Numerically, a single \ac{CD} is sufficient to obtain the \ac{CSG} element $A\in\groupA$ in this matrix representation, which then can be mapped back to commuting Paulis via $\rho^{-1}$.
This set of commuting Pauli operators may or may not be easy to diagonalize exactly. For models with single-qubit field operators\footnote{Note that a polynomial \ac{DLA} will only result from Hamiltonians with at most one type of single-qubit operators per site.}, the \ac{CSA} can be chosen to contain those field operators, facilitating the diagonalization.

Note that while we are not diagonalizing a new class of Hamiltonians exactly, this technique is a convenient byproduct of our time-evolution compiler, and allows us to solve some of the ``essentially easy" Hamiltonians that represent free fermions in an automated fashion. To showcase this technique, we generate random Hamiltonians of the transverse-field XY model from our compilation example above. Then we map them to the small-dimensional representation using the hard-coded isomorphism\footnote{As we use the fixed mapping, we skip representing the original Hamiltonian in the qubit basis.} $\rho$ from App.~\ref{sec:fdhs_calculations:mapping}, perform a single \kak decomposition of type BDI, and map the \ac{CSG} element $A$ back to the qubit representation. As we chose the \ac{CSA} $\langle\{iZ_j\}_j\rangle_{\mbr}$, we can then read off the eigenvalues of $H$ immediately from the coefficients of $A$ in this basis.
We report performance metrics of this method, which essentially reduces to the first decomposition step of the compilation algorithm, in Fig.~\ref{fig:fdhs_performance}.

\subsection{Other compilation techniques}
Above we compiled time evolution under a Hamiltonian with polynomially-sized \ac{DLA} using a modular, scalable algorithm. 
This compilation task has been researched intensely, and in this section we will discuss two closely related works in detail.

The first, by Kökcü et al.~\cite{kokcu2022fixed}, constructs a \kak decomposition of $U(t)$ variationally, by solving a non-convex optimization problem, and considers the transverse field XY model as an example, like we did in Sec.~\ref{sec:fdhs:example}. Precisely, it fixes a (non-recursive) \ac{CD} $\langle iH\rangle_\text{Lie}=\mfk\oplus\mfp$, a \ac{CSA} $\mfa\subset\mfp$, an element $iv\in \mfa$ such that $\exp(iv\mbr)$ is dense in $\exp(\mfa)$, and an ansatz $K(\phi)$ that parametrizes $\exp(\mfk)$ faithfully in a sufficiently large area.
The algorithm then optimizes the parameters $\phi$ by extremizing some cost function, which guarantees $a=K^\dagger(\phi_c)iHK(\phi_c)\in\mfa$ for any \emph{local} extremum $\phi_c$. This implies that
\begin{align}
    iH = K(\phi_c) a K^\dagger(\phi_c) \ \Rightarrow\ U(t) = K(\phi_c) \exp(at) K^\dagger(\phi_c),\nonumber
\end{align}
which is the desired time-evolution circuit.
The size of the ansatz circuit $K$, its parameter count $|\phi|$, and evaluation of $f$ and $\nabla f$ all scale with $\dim \langle iH\rangle_\text{Lie}$, making it feasible for polynomial \acp{DLA} only, as expected.

Our algorithm differs from that of~\cite{kokcu2022fixed} in multiple aspects, most of which derive from the fact that the above approach is variational, while ours is not.
The variational ansatz $K(\phi)$ allows to choose a favourable circuit structure with minimal parameter count immediately, which is only possible to a limited degree in our algorithm, via the choice of the recursive decomposition and the bases in the irreducible representation.
However, the runtime of the non-convex optimization task in the variational approach is hard to estimate, and in practice it turns out to require very many optimization steps, quickly making the compilation infeasible according to the authors~\cite{fdhsrepo}.
We reproduce the numerical experiments from~\cite{kokcu2022fixed} with PennyLane~\cite{bergholm2018pennylane} in~\cite{symmetrycompilationrepo} and confirm difficulties in achieving convergence beyond 10 qubits in our alternative implementation.
In contrast, our algorithm not only has a static runtime with respect to different instances of the same size, but this runtime also is very low in practice, allowing us to compile the time evolution on 1000 qubits in 22 seconds.
Finally, the variational approach requires $iH\in\mfp$, which is equivalent to $H$ having a $\mathbb{Z}_2$ antisymmetry that is a Cartan involution (see Sec.~\ref{sec:involutions_symmetries}) and fixes the \ac{CD} (Prop.~\ref{prop:unique_CD_of_H}). While our algorithm benefits from $iH\in\mfp$, providing a single compiled circuit for all evolution times $t$ (see Sec.~\ref{sec:hamsim:variant}), it can just as well be used for $iH\not\in\mfp$, as recompiling for different $t$ is feasible with our method.

Then, the related work~\cite{kokcu2022algebraic} is representative of various efforts in quantum circuit compression~\cite{bassman2022constant,peng2022quantum,ogunkoya2024qutrit} and also implements a compiler for the time evolution of Hamiltonians with polynomially-sized \ac{DLA}.
This method chains a Trotterization of $U(t)$ with an iterative compression of the resulting circuit.
Even though the Trotterization introduces a compilation error in principle, the cheap compression allows for very deep intermediate circuits, and in turn for Trotter step sizes that lead to numerically exact results.

The two key strengths of the compression algorithm lie 1) in its performance, allowing to compile time evolution on $10^3$ spins in minutes on a high-end desktop CPU, and 2) in the simple circuit structure it produces; it only uses gates generated by the individual terms in $H$, and thus inherits locality from the Hamiltonian, for example.
For the transverse field XY model, this compression approach produces favourable quantum circuits, but our recursive \ac{CD} algorithm is faster, running in seconds on a laptop for $10^3$ qubits and scaling quadratically rather than cubically. Admittedly, the cheaper quantum circuits are preferable over the reduced classical compilation runtime at the current state of quantum hardware.
A limitation of the compression algorithm is that it needs to recompile the circuits for different evolution times $t$, with the potential remedy to reuse circuits for fractions of the evolution time that were obtained during the iterative compression.
In addition, the compression cannot be used to diagonalize the Hamiltonian like we did with our method above.

\section{Conclusion}
\label{sec:conclusion}

This manuscript provides a comprehensive overview of recursive \acp{CD} for the task of unitary synthesis. We have elucidated the salient theoretical properties of such decompositions at both the group and algebra levels, and provided a unified approach for numerically computing any decompositions in practice. Leveraging this framework, we showcased three new recursive decomposition strategies and tackled a large-scale compilation problem related to Hamiltonian time-evolution.
The framework presented here can serve as a foundation for quantum compilation software which automates and streamlines the important task of unitary synthesis. 

Looking forward, some interesting directions for future work stand out. By knowing exactly which building blocks are fixed, and which degrees of freedom are still available, searches can more easily be undertaken to discover beneficial circuit decomposition strategies. Such searches could be performed manually, guided by expert knowledge, or carried out more comprehensively by automated software.
Another area for future investigation is incorporating additional optimizations into the decompositions. The most resource-efficient decompositions with respect to CNOT gate counts result from applying further special-purpose optimizations on top of (recursive) \acp{CD}, often by extracting certain gates from one part of the circuit and commuting them to be absorbed in another part~\cite{shende2005synthesis, krol2024beyond}. These optimizations may seem very specific to the settings where they appear, but we suspect that the special structure of \acp{CD}, in particular the connections to $\mathbb{Z}_2$ symmetries and the commutator relations Eqs.~(\ref{eq:subalg_prop})-(\ref{eq:symm_prop}), could be used to generalize these techniques and unearth further beneficial optimizations.

A further direction worth investigation is the application of the recursive \ac{CD} framework to problems with known symmetries, e.g., in geometric quantum machine learning~\cite{wiersema2025geometric, larocca2022group,skolik2022equivariant,meyer2022exploiting,west2024provably,sauvage2022building} or in the simulation of physical systems. In these cases, the symmetries of the system would be reflected in the symmetries of the chosen decomposition. As a simplified illustration, for a circuit where the initial state and the final observable are both invariant under some symmetry group $\groupK$, decomposing the circuit into a \kak form would allow the matrices $\mathrm{K}$ to be removed, leaving the problem in a fully-commuting form.
Finally, on the mathematical side, we note there is a generalized ``two-sided" notion of \acp{CD}, where the two matrices $K_1,K_2$ in \kak are chosen from different compatible subgroups~\cite{edelman2023fifty}. This generalization seems not to have been used before for quantum computing problems. It would be interesting to investigate this generalization more deeply, potentially incorporating it into the recursive \acp{CD} framework and providing a further tool for unitary synthesis.

\section*{Acknowledgements}
DW thanks Korbinian Kottmann for many helpful discussions.
MW and RTF were supported by the U.S. Department Of Energy through a quantum computing program sponsored by the Los Alamos National Laboratory Information Science \& Technology Institute. MC acknowledges support by the Laboratory Directed Research and Development (LDRD) program of LANL under project number 20230049DR. This work was also initially supported by the LANL's ASC Beyond Moore’s Law project.

\bibliography{ref}

\begin{thebibliography}{74}%
\makeatletter
\providecommand \@ifxundefined [1]{%
 \@ifx{#1\undefined}
}%
\providecommand \@ifnum [1]{%
 \ifnum #1\expandafter \@firstoftwo
 \else \expandafter \@secondoftwo
 \fi
}%
\providecommand \@ifx [1]{%
 \ifx #1\expandafter \@firstoftwo
 \else \expandafter \@secondoftwo
 \fi
}%
\providecommand \natexlab [1]{#1}%
\providecommand \enquote  [1]{``#1''}%
\providecommand \bibnamefont  [1]{#1}%
\providecommand \bibfnamefont [1]{#1}%
\providecommand \citenamefont [1]{#1}%
\providecommand \href@noop [0]{\@secondoftwo}%
\providecommand \href [0]{\begingroup \@sanitize@url \@href}%
\providecommand \@href[1]{\@@startlink{#1}\@@href}%
\providecommand \@@href[1]{\endgroup#1\@@endlink}%
\providecommand \@sanitize@url [0]{\catcode `\\12\catcode `\$12\catcode `\&12\catcode `\#12\catcode `\^12\catcode `\_12\catcode `\%12\relax}%
\providecommand \@@startlink[1]{}%
\providecommand \@@endlink[0]{}%
\providecommand \url  [0]{\begingroup\@sanitize@url \@url }%
\providecommand \@url [1]{\endgroup\@href {#1}{\urlprefix }}%
\providecommand \urlprefix  [0]{URL }%
\providecommand \Eprint [0]{\href }%
\providecommand \doibase [0]{https://doi.org/}%
\providecommand \selectlanguage [0]{\@gobble}%
\providecommand \bibinfo  [0]{\@secondoftwo}%
\providecommand \bibfield  [0]{\@secondoftwo}%
\providecommand \translation [1]{[#1]}%
\providecommand \BibitemOpen [0]{}%
\providecommand \bibitemStop [0]{}%
\providecommand \bibitemNoStop [0]{.\EOS\space}%
\providecommand \EOS [0]{\spacefactor3000\relax}%
\providecommand \BibitemShut  [1]{\csname bibitem#1\endcsname}%
\let\auto@bib@innerbib\@empty
\bibitem [{\citenamefont {Venturelli}\ \emph {et~al.}(2018)\citenamefont {Venturelli}, \citenamefont {Do}, \citenamefont {Rieffel},\ and\ \citenamefont {Frank}}]{venturelli2018compiling}%
  \BibitemOpen
  \bibfield  {author} {\bibinfo {author} {\bibfnamefont {D.}~\bibnamefont {Venturelli}}, \bibinfo {author} {\bibfnamefont {M.}~\bibnamefont {Do}}, \bibinfo {author} {\bibfnamefont {E.}~\bibnamefont {Rieffel}},\ and\ \bibinfo {author} {\bibfnamefont {J.}~\bibnamefont {Frank}},\ }\bibfield  {title} {\bibinfo {title} {Compiling quantum circuits to realistic hardware architectures using temporal planners},\ }\href {https://doi.org/10.1088/2058-9565/aaa331} {\bibfield  {journal} {\bibinfo  {journal} {Quantum Science and Technology}\ }\textbf {\bibinfo {volume} {3}},\ \bibinfo {pages} {025004} (\bibinfo {year} {2018})}\BibitemShut {NoStop}%
\bibitem [{\citenamefont {Khatri}\ \emph {et~al.}(2019)\citenamefont {Khatri}, \citenamefont {LaRose}, \citenamefont {Poremba}, \citenamefont {Cincio}, \citenamefont {Sornborger},\ and\ \citenamefont {Coles}}]{khatri2019quantum}%
  \BibitemOpen
  \bibfield  {author} {\bibinfo {author} {\bibfnamefont {S.}~\bibnamefont {Khatri}}, \bibinfo {author} {\bibfnamefont {R.}~\bibnamefont {LaRose}}, \bibinfo {author} {\bibfnamefont {A.}~\bibnamefont {Poremba}}, \bibinfo {author} {\bibfnamefont {L.}~\bibnamefont {Cincio}}, \bibinfo {author} {\bibfnamefont {A.~T.}\ \bibnamefont {Sornborger}},\ and\ \bibinfo {author} {\bibfnamefont {P.~J.}\ \bibnamefont {Coles}},\ }\bibfield  {title} {\bibinfo {title} {Quantum-assisted quantum compiling},\ }\href {https://doi.org/10.22331/q-2019-05-13-140} {\bibfield  {journal} {\bibinfo  {journal} {Quantum}\ }\textbf {\bibinfo {volume} {3}},\ \bibinfo {pages} {140} (\bibinfo {year} {2019})}\BibitemShut {NoStop}%
\bibitem [{\citenamefont {Sharma}\ \emph {et~al.}(2020)\citenamefont {Sharma}, \citenamefont {Khatri}, \citenamefont {Cerezo},\ and\ \citenamefont {Coles}}]{sharma2019noise}%
  \BibitemOpen
  \bibfield  {author} {\bibinfo {author} {\bibfnamefont {K.}~\bibnamefont {Sharma}}, \bibinfo {author} {\bibfnamefont {S.}~\bibnamefont {Khatri}}, \bibinfo {author} {\bibfnamefont {M.}~\bibnamefont {Cerezo}},\ and\ \bibinfo {author} {\bibfnamefont {P.~J.}\ \bibnamefont {Coles}},\ }\bibfield  {title} {\bibinfo {title} {Noise resilience of variational quantum compiling},\ }\href {https://doi.org/10.1088/1367-2630/ab784c} {\bibfield  {journal} {\bibinfo  {journal} {New Journal of Physics}\ }\textbf {\bibinfo {volume} {22}},\ \bibinfo {pages} {043006} (\bibinfo {year} {2020})}\BibitemShut {NoStop}%
\bibitem [{\citenamefont {Tucci}(1999)}]{tucci1999rudimentary}%
  \BibitemOpen
  \bibfield  {author} {\bibinfo {author} {\bibfnamefont {R.~R.}\ \bibnamefont {Tucci}},\ }\bibfield  {title} {\bibinfo {title} {A rudimentary quantum compiler (2cnd ed.)},\ }\href {https://arxiv.org/abs/quant-ph/9902062} {\bibfield  {journal} {\bibinfo  {journal} {arXiv preprint quant-ph/9902062}\ } (\bibinfo {year} {1999})}\BibitemShut {NoStop}%
\bibitem [{\citenamefont {M{\"o}tt{\"o}nen}\ \emph {et~al.}(2004)\citenamefont {M{\"o}tt{\"o}nen}, \citenamefont {Vartiainen}, \citenamefont {Bergholm},\ and\ \citenamefont {Salomaa}}]{mottonen2004quantum}%
  \BibitemOpen
  \bibfield  {author} {\bibinfo {author} {\bibfnamefont {M.}~\bibnamefont {M{\"o}tt{\"o}nen}}, \bibinfo {author} {\bibfnamefont {J.~J.}\ \bibnamefont {Vartiainen}}, \bibinfo {author} {\bibfnamefont {V.}~\bibnamefont {Bergholm}},\ and\ \bibinfo {author} {\bibfnamefont {M.~M.}\ \bibnamefont {Salomaa}},\ }\bibfield  {title} {\bibinfo {title} {Quantum circuits for general multiqubit gates},\ }\href {https://doi.org/10.1103/PhysRevLett.93.130502} {\bibfield  {journal} {\bibinfo  {journal} {{Phys. Rev. Lett.}}\ }\textbf {\bibinfo {volume} {93}},\ \bibinfo {pages} {130502} (\bibinfo {year} {2004})}\BibitemShut {NoStop}%
\bibitem [{\citenamefont {Bergholm}\ \emph {et~al.}(2005)\citenamefont {Bergholm}, \citenamefont {Vartiainen}, \citenamefont {M{\"o}tt{\"o}nen},\ and\ \citenamefont {Salomaa}}]{bergholm2005quantum}%
  \BibitemOpen
  \bibfield  {author} {\bibinfo {author} {\bibfnamefont {V.}~\bibnamefont {Bergholm}}, \bibinfo {author} {\bibfnamefont {J.~J.}\ \bibnamefont {Vartiainen}}, \bibinfo {author} {\bibfnamefont {M.}~\bibnamefont {M{\"o}tt{\"o}nen}},\ and\ \bibinfo {author} {\bibfnamefont {M.~M.}\ \bibnamefont {Salomaa}},\ }\bibfield  {title} {\bibinfo {title} {Quantum circuits with uniformly controlled one-qubit gates},\ }\href {https://doi.org/10.1103/PhysRevA.71.052330} {\bibfield  {journal} {\bibinfo  {journal} {{Phys. Rev. A}}\ }\textbf {\bibinfo {volume} {71}},\ \bibinfo {pages} {052330} (\bibinfo {year} {2005})}\BibitemShut {NoStop}%
\bibitem [{\citenamefont {Nakajima}\ \emph {et~al.}(2005)\citenamefont {Nakajima}, \citenamefont {Kawano},\ and\ \citenamefont {Sekigawa}}]{nakajima2005new}%
  \BibitemOpen
  \bibfield  {author} {\bibinfo {author} {\bibfnamefont {Y.}~\bibnamefont {Nakajima}}, \bibinfo {author} {\bibfnamefont {Y.}~\bibnamefont {Kawano}},\ and\ \bibinfo {author} {\bibfnamefont {H.}~\bibnamefont {Sekigawa}},\ }\bibfield  {title} {\bibinfo {title} {A new algorithm for producing quantum circuits using {KAK} decompositions},\ }\href {https://doi.org/10.48550/arXiv.quant-ph/0509196} {\bibfield  {journal} {\bibinfo  {journal} {arXiv preprint quant-ph/0509196}\ } (\bibinfo {year} {2005})}\BibitemShut {NoStop}%
\bibitem [{\citenamefont {M{\"o}tt{\"o}nen}\ and\ \citenamefont {Vartiainen}(2006)}]{mottonen2006decompositions}%
  \BibitemOpen
  \bibfield  {author} {\bibinfo {author} {\bibfnamefont {M.}~\bibnamefont {M{\"o}tt{\"o}nen}}\ and\ \bibinfo {author} {\bibfnamefont {J.~J.}\ \bibnamefont {Vartiainen}},\ }\href {https://doi.org/10.48550/arXiv.quant-ph/0504100} {\emph {\bibinfo {title} {Decompositions of general quantum gates}}}\ (\bibinfo  {publisher} {Nova Publishers},\ \bibinfo {year} {2006})\ p.\ \bibinfo {pages} {149}\BibitemShut {NoStop}%
\bibitem [{\citenamefont {Iten}\ \emph {et~al.}(2016)\citenamefont {Iten}, \citenamefont {Colbeck}, \citenamefont {Kukuljan}, \citenamefont {Home},\ and\ \citenamefont {Christandl}}]{iten2016quantum}%
  \BibitemOpen
  \bibfield  {author} {\bibinfo {author} {\bibfnamefont {R.}~\bibnamefont {Iten}}, \bibinfo {author} {\bibfnamefont {R.}~\bibnamefont {Colbeck}}, \bibinfo {author} {\bibfnamefont {I.}~\bibnamefont {Kukuljan}}, \bibinfo {author} {\bibfnamefont {J.}~\bibnamefont {Home}},\ and\ \bibinfo {author} {\bibfnamefont {M.}~\bibnamefont {Christandl}},\ }\bibfield  {title} {\bibinfo {title} {Quantum circuits for isometries},\ }\href {https://doi.org/10.1103/PhysRevA.93.032318} {\bibfield  {journal} {\bibinfo  {journal} {{Phys. Rev. A}}\ }\textbf {\bibinfo {volume} {93}},\ \bibinfo {pages} {032318} (\bibinfo {year} {2016})}\BibitemShut {NoStop}%
\bibitem [{\citenamefont {Khaneja}\ and\ \citenamefont {Glaser}(2001)}]{khaneja2001cartan}%
  \BibitemOpen
  \bibfield  {author} {\bibinfo {author} {\bibfnamefont {N.}~\bibnamefont {Khaneja}}\ and\ \bibinfo {author} {\bibfnamefont {S.~J.}\ \bibnamefont {Glaser}},\ }\bibfield  {title} {\bibinfo {title} {Cartan decomposition of ${SU}(2^{n})$ and control of spin systems},\ }\href {https://doi.org/10.1016/S0301-0104(01)00318-4} {\bibfield  {journal} {\bibinfo  {journal} {Chem. Phys.}\ }\textbf {\bibinfo {volume} {267}},\ \bibinfo {pages} {11} (\bibinfo {year} {2001})}\BibitemShut {NoStop}%
\bibitem [{\citenamefont {Vatan}\ and\ \citenamefont {Williams}(2004{\natexlab{a}})}]{vatan2004realization}%
  \BibitemOpen
  \bibfield  {author} {\bibinfo {author} {\bibfnamefont {F.}~\bibnamefont {Vatan}}\ and\ \bibinfo {author} {\bibfnamefont {C.~P.}\ \bibnamefont {Williams}},\ }\bibfield  {title} {\bibinfo {title} {Realization of a general three-qubit quantum gate},\ }\href {https://doi.org/10.48550/arXiv.quant-ph/0401178} {\bibfield  {journal} {\bibinfo  {journal} {arXiv preprint quant-ph/0401178}\ } (\bibinfo {year} {2004}{\natexlab{a}})}\BibitemShut {NoStop}%
\bibitem [{\citenamefont {Bullock}(2004)}]{bullock2004note}%
  \BibitemOpen
  \bibfield  {author} {\bibinfo {author} {\bibfnamefont {S.~S.}\ \bibnamefont {Bullock}},\ }\bibfield  {title} {\bibinfo {title} {Note on the {Khaneja Glaser} decomposition},\ }\href {https://doi.org/10.48550/arXiv.quant-ph/0403141} {\bibfield  {journal} {\bibinfo  {journal} {arXiv preprint quant-ph/0403141}\ } (\bibinfo {year} {2004})}\BibitemShut {NoStop}%
\bibitem [{\citenamefont {Mansky}\ \emph {et~al.}(2023)\citenamefont {Mansky}, \citenamefont {Castillo}, \citenamefont {Puigvert},\ and\ \citenamefont {Linnhoff-Popien}}]{mansky2023near}%
  \BibitemOpen
  \bibfield  {author} {\bibinfo {author} {\bibfnamefont {M.~B.}\ \bibnamefont {Mansky}}, \bibinfo {author} {\bibfnamefont {S.~L.}\ \bibnamefont {Castillo}}, \bibinfo {author} {\bibfnamefont {V.~R.}\ \bibnamefont {Puigvert}},\ and\ \bibinfo {author} {\bibfnamefont {C.}~\bibnamefont {Linnhoff-Popien}},\ }\bibfield  {title} {\bibinfo {title} {Near-optimal quantum circuit construction via {Cartan} decomposition},\ }\href {https://doi.org/10.1103/PhysRevA.108.052607} {\bibfield  {journal} {\bibinfo  {journal} {{Phys. Rev. A}}\ }\textbf {\bibinfo {volume} {108}},\ \bibinfo {pages} {052607} (\bibinfo {year} {2023})}\BibitemShut {NoStop}%
\bibitem [{\citenamefont {Shende}\ \emph {et~al.}(2005)\citenamefont {Shende}, \citenamefont {Bullock},\ and\ \citenamefont {Markov}}]{shende2005synthesis}%
  \BibitemOpen
  \bibfield  {author} {\bibinfo {author} {\bibfnamefont {V.~V.}\ \bibnamefont {Shende}}, \bibinfo {author} {\bibfnamefont {S.~S.}\ \bibnamefont {Bullock}},\ and\ \bibinfo {author} {\bibfnamefont {I.~L.}\ \bibnamefont {Markov}},\ }\bibfield  {title} {\bibinfo {title} {Synthesis of quantum logic circuits},\ }in\ \href {https://dl.acm.org/doi/10.1145/1120725.1120847} {\emph {\bibinfo {booktitle} {Proceedings of the 2005 Asia and South Pacific Design Automation Conference}}}\ (\bibinfo {year} {2005})\ pp.\ \bibinfo {pages} {272--275}\BibitemShut {NoStop}%
\bibitem [{\citenamefont {Drury}\ and\ \citenamefont {Love}(2008)}]{drury2008constructive}%
  \BibitemOpen
  \bibfield  {author} {\bibinfo {author} {\bibfnamefont {B.}~\bibnamefont {Drury}}\ and\ \bibinfo {author} {\bibfnamefont {P.}~\bibnamefont {Love}},\ }\bibfield  {title} {\bibinfo {title} {Constructive quantum {Shannon} decomposition from {Cartan} involutions},\ }\href {https://doi.org/10.1088/1751-8113/41/39/395305} {\bibfield  {journal} {\bibinfo  {journal} {{J. Phys. A-Math. Theor.}}\ }\textbf {\bibinfo {volume} {41}},\ \bibinfo {pages} {395305} (\bibinfo {year} {2008})}\BibitemShut {NoStop}%
\bibitem [{\citenamefont {De~Vos}\ and\ \citenamefont {De~Baerdemacker}(2016)}]{de2016block}%
  \BibitemOpen
  \bibfield  {author} {\bibinfo {author} {\bibfnamefont {A.}~\bibnamefont {De~Vos}}\ and\ \bibinfo {author} {\bibfnamefont {S.}~\bibnamefont {De~Baerdemacker}},\ }\bibfield  {title} {\bibinfo {title} {Block-${ZXZ}$ synthesis of an arbitrary quantum circuit},\ }\href {https://doi.org/10.1103/PhysRevA.94.052317} {\bibfield  {journal} {\bibinfo  {journal} {{Phys. Rev. A}}\ }\textbf {\bibinfo {volume} {94}},\ \bibinfo {pages} {052317} (\bibinfo {year} {2016})}\BibitemShut {NoStop}%
\bibitem [{\citenamefont {De~Vos}\ and\ \citenamefont {De~Baerdemacker}(2018)}]{de2018unified}%
  \BibitemOpen
  \bibfield  {author} {\bibinfo {author} {\bibfnamefont {A.}~\bibnamefont {De~Vos}}\ and\ \bibinfo {author} {\bibfnamefont {S.}~\bibnamefont {De~Baerdemacker}},\ }\bibfield  {title} {\bibinfo {title} {A unified approach to quantum computation and classical reversible computation},\ }in\ \href {https://doi.org/10.1007/978-3-319-99498-7_9} {\emph {\bibinfo {booktitle} {Reversible Computation: 10th International Conference, RC 2018, Leicester, UK, September 12-14, 2018, Proceedings 10}}}\ (\bibinfo {organization} {Springer},\ \bibinfo {year} {2018})\ pp.\ \bibinfo {pages} {133--143}\BibitemShut {NoStop}%
\bibitem [{\citenamefont {Führ}\ and\ \citenamefont {Rzeszotnik}(2018)}]{fuhr2018note}%
  \BibitemOpen
  \bibfield  {author} {\bibinfo {author} {\bibfnamefont {H.}~\bibnamefont {Führ}}\ and\ \bibinfo {author} {\bibfnamefont {Z.}~\bibnamefont {Rzeszotnik}},\ }\bibfield  {title} {\bibinfo {title} {A note on factoring unitary matrices},\ }\href {https://doi.org/https://doi.org/10.1016/j.laa.2018.02.017} {\bibfield  {journal} {\bibinfo  {journal} {{Linear Algebra Appl.}}\ }\textbf {\bibinfo {volume} {547}},\ \bibinfo {pages} {32} (\bibinfo {year} {2018})}\BibitemShut {NoStop}%
\bibitem [{\citenamefont {Krol}\ and\ \citenamefont {Al-Ars}(2024)}]{krol2024beyond}%
  \BibitemOpen
  \bibfield  {author} {\bibinfo {author} {\bibfnamefont {A.~M.}\ \bibnamefont {Krol}}\ and\ \bibinfo {author} {\bibfnamefont {Z.}~\bibnamefont {Al-Ars}},\ }\bibfield  {title} {\bibinfo {title} {Beyond quantum {S}hannon: Circuit construction for general n-qubit gates based on {B}lock {ZXZ}-decomposition},\ }\href {https://doi.org/10.48550/arXiv.2403.13692} {\bibfield  {journal} {\bibinfo  {journal} {arXiv preprint arXiv:2403.13692}\ } (\bibinfo {year} {2024})}\BibitemShut {NoStop}%
\bibitem [{\citenamefont {Barenco}\ \emph {et~al.}(1995)\citenamefont {Barenco}, \citenamefont {Bennett}, \citenamefont {Cleve}, \citenamefont {DiVincenzo}, \citenamefont {Margolus}, \citenamefont {Shor}, \citenamefont {Sleator}, \citenamefont {Smolin},\ and\ \citenamefont {Weinfurter}}]{barenco1995elementary}%
  \BibitemOpen
  \bibfield  {author} {\bibinfo {author} {\bibfnamefont {A.}~\bibnamefont {Barenco}}, \bibinfo {author} {\bibfnamefont {C.~H.}\ \bibnamefont {Bennett}}, \bibinfo {author} {\bibfnamefont {R.}~\bibnamefont {Cleve}}, \bibinfo {author} {\bibfnamefont {D.~P.}\ \bibnamefont {DiVincenzo}}, \bibinfo {author} {\bibfnamefont {N.}~\bibnamefont {Margolus}}, \bibinfo {author} {\bibfnamefont {P.}~\bibnamefont {Shor}}, \bibinfo {author} {\bibfnamefont {T.}~\bibnamefont {Sleator}}, \bibinfo {author} {\bibfnamefont {J.~A.}\ \bibnamefont {Smolin}},\ and\ \bibinfo {author} {\bibfnamefont {H.}~\bibnamefont {Weinfurter}},\ }\bibfield  {title} {\bibinfo {title} {Elementary gates for quantum computation},\ }\href {https://doi.org/10.1103/PhysRevA.52.3457} {\bibfield  {journal} {\bibinfo  {journal} {{Phys. Rev. A}}\ }\textbf {\bibinfo {volume} {52}},\ \bibinfo {pages} {3457} (\bibinfo {year} {1995})}\BibitemShut {NoStop}%
\bibitem [{\citenamefont {Cybenko}(2001)}]{cybenko2001reducing}%
  \BibitemOpen
  \bibfield  {author} {\bibinfo {author} {\bibfnamefont {G.}~\bibnamefont {Cybenko}},\ }\bibfield  {title} {\bibinfo {title} {Reducing quantum computations to elementary unitary operations},\ }\href {https://doi.org/10.1109/5992.908999} {\bibfield  {journal} {\bibinfo  {journal} {{Comput. Sci. Eng.}}\ }\textbf {\bibinfo {volume} {3}},\ \bibinfo {pages} {27} (\bibinfo {year} {2001})}\BibitemShut {NoStop}%
\bibitem [{\citenamefont {Aho}\ and\ \citenamefont {Svore}(2003)}]{aho2003compiling}%
  \BibitemOpen
  \bibfield  {author} {\bibinfo {author} {\bibfnamefont {A.~V.}\ \bibnamefont {Aho}}\ and\ \bibinfo {author} {\bibfnamefont {K.~M.}\ \bibnamefont {Svore}},\ }\bibfield  {title} {\bibinfo {title} {Compiling quantum circuits using the palindrome transform},\ }\href {https://doi.org/10.48550/arXiv.quant-ph/0311008} {\bibfield  {journal} {\bibinfo  {journal} {arXiv preprint quant-ph/0311008}\ } (\bibinfo {year} {2003})}\BibitemShut {NoStop}%
\bibitem [{\citenamefont {Vartiainen}\ \emph {et~al.}(2004)\citenamefont {Vartiainen}, \citenamefont {M{\"o}tt{\"o}nen},\ and\ \citenamefont {Salomaa}}]{vartiainen2004efficient}%
  \BibitemOpen
  \bibfield  {author} {\bibinfo {author} {\bibfnamefont {J.~J.}\ \bibnamefont {Vartiainen}}, \bibinfo {author} {\bibfnamefont {M.}~\bibnamefont {M{\"o}tt{\"o}nen}},\ and\ \bibinfo {author} {\bibfnamefont {M.~M.}\ \bibnamefont {Salomaa}},\ }\bibfield  {title} {\bibinfo {title} {Efficient decomposition of quantum gates},\ }\href {https://doi.org/10.1103/PhysRevLett.92.177902} {\bibfield  {journal} {\bibinfo  {journal} {{Phys. Rev. Lett.}}\ }\textbf {\bibinfo {volume} {92}},\ \bibinfo {pages} {177902} (\bibinfo {year} {2004})}\BibitemShut {NoStop}%
\bibitem [{\citenamefont {Jiang}\ \emph {et~al.}(2018)\citenamefont {Jiang}, \citenamefont {Sung}, \citenamefont {Kechedzhi}, \citenamefont {Smelyanskiy},\ and\ \citenamefont {Boixo}}]{jiang2018quantum}%
  \BibitemOpen
  \bibfield  {author} {\bibinfo {author} {\bibfnamefont {Z.}~\bibnamefont {Jiang}}, \bibinfo {author} {\bibfnamefont {K.~J.}\ \bibnamefont {Sung}}, \bibinfo {author} {\bibfnamefont {K.}~\bibnamefont {Kechedzhi}}, \bibinfo {author} {\bibfnamefont {V.~N.}\ \bibnamefont {Smelyanskiy}},\ and\ \bibinfo {author} {\bibfnamefont {S.}~\bibnamefont {Boixo}},\ }\bibfield  {title} {\bibinfo {title} {Quantum algorithms to simulate many-body physics of correlated fermions},\ }\href {https://doi.org/10.1103/PhysRevApplied.9.044036} {\bibfield  {journal} {\bibinfo  {journal} {{Phys. Rev. Appl.}}\ }\textbf {\bibinfo {volume} {9}},\ \bibinfo {pages} {044036} (\bibinfo {year} {2018})}\BibitemShut {NoStop}%
\bibitem [{\citenamefont {Rakyta}\ and\ \citenamefont {Zimbor{\'a}s}(2022)}]{rakyta2022approaching}%
  \BibitemOpen
  \bibfield  {author} {\bibinfo {author} {\bibfnamefont {P.}~\bibnamefont {Rakyta}}\ and\ \bibinfo {author} {\bibfnamefont {Z.}~\bibnamefont {Zimbor{\'a}s}},\ }\bibfield  {title} {\bibinfo {title} {Approaching the theoretical limit in quantum gate decomposition},\ }\href {https://doi.org/10.22331/q-2022-05-11-710} {\bibfield  {journal} {\bibinfo  {journal} {Quantum}\ }\textbf {\bibinfo {volume} {6}},\ \bibinfo {pages} {710} (\bibinfo {year} {2022})}\BibitemShut {NoStop}%
\bibitem [{\citenamefont {Arrazola}\ \emph {et~al.}(2022)\citenamefont {Arrazola}, \citenamefont {Di~Matteo}, \citenamefont {Quesada}, \citenamefont {Jahangiri}, \citenamefont {Delgado},\ and\ \citenamefont {Killoran}}]{arrazola2022universal}%
  \BibitemOpen
  \bibfield  {author} {\bibinfo {author} {\bibfnamefont {J.~M.}\ \bibnamefont {Arrazola}}, \bibinfo {author} {\bibfnamefont {O.}~\bibnamefont {Di~Matteo}}, \bibinfo {author} {\bibfnamefont {N.}~\bibnamefont {Quesada}}, \bibinfo {author} {\bibfnamefont {S.}~\bibnamefont {Jahangiri}}, \bibinfo {author} {\bibfnamefont {A.}~\bibnamefont {Delgado}},\ and\ \bibinfo {author} {\bibfnamefont {N.}~\bibnamefont {Killoran}},\ }\bibfield  {title} {\bibinfo {title} {Universal quantum circuits for quantum chemistry},\ }\href {https://doi.org/10.22331/q-2022-06-20-742} {\bibfield  {journal} {\bibinfo  {journal} {Quantum}\ }\textbf {\bibinfo {volume} {6}},\ \bibinfo {pages} {742} (\bibinfo {year} {2022})}\BibitemShut {NoStop}%
\bibitem [{\citenamefont {Cartan}(1926)}]{cartan1926sur}%
  \BibitemOpen
  \bibfield  {author} {\bibinfo {author} {\bibfnamefont {E.}~\bibnamefont {Cartan}},\ }\bibfield  {title} {\bibinfo {title} {Sur une classe remarquable d'espaces de {Riemann}},\ }\href {https://doi.org/10.24033/bsmf.1105} {\bibfield  {journal} {\bibinfo  {journal} {{Bull. Soc. Math. Fr.}}\ }\textbf {\bibinfo {volume} {54}},\ \bibinfo {pages} {214} (\bibinfo {year} {1926})}\BibitemShut {NoStop}%
\bibitem [{\citenamefont {Edelman}\ and\ \citenamefont {Jeong}(2023)}]{edelman2023fifty}%
  \BibitemOpen
  \bibfield  {author} {\bibinfo {author} {\bibfnamefont {A.}~\bibnamefont {Edelman}}\ and\ \bibinfo {author} {\bibfnamefont {S.}~\bibnamefont {Jeong}},\ }\bibfield  {title} {\bibinfo {title} {Fifty three matrix factorizations: A systematic approach},\ }\href {https://doi.org/10.1137/21M1416035} {\bibfield  {journal} {\bibinfo  {journal} {{SIAM J. Matrix Anal. Appl.}}\ }\textbf {\bibinfo {volume} {44}},\ \bibinfo {pages} {415} (\bibinfo {year} {2023})}\BibitemShut {NoStop}%
\bibitem [{\citenamefont {Bullock}\ and\ \citenamefont {Brennen}(2004)}]{bullock2004canonical}%
  \BibitemOpen
  \bibfield  {author} {\bibinfo {author} {\bibfnamefont {S.~S.}\ \bibnamefont {Bullock}}\ and\ \bibinfo {author} {\bibfnamefont {G.~K.}\ \bibnamefont {Brennen}},\ }\bibfield  {title} {\bibinfo {title} {Canonical decompositions of $n$-qubit quantum computations and concurrence},\ }\href {https://doi.org/10.1063/1.1723701} {\bibfield  {journal} {\bibinfo  {journal} {{J. Math. Phys.}}\ }\textbf {\bibinfo {volume} {45}},\ \bibinfo {pages} {2447} (\bibinfo {year} {2004})}\BibitemShut {NoStop}%
\bibitem [{\citenamefont {Sutton}(2009)}]{sutton2009computing}%
  \BibitemOpen
  \bibfield  {author} {\bibinfo {author} {\bibfnamefont {B.~D.}\ \bibnamefont {Sutton}},\ }\bibfield  {title} {\bibinfo {title} {Computing the complete {CS} decomposition},\ }\href {https://doi.org/10.48550/arXiv.0707.1838} {\bibfield  {journal} {\bibinfo  {journal} {{Numer. Algorithms}}\ }\textbf {\bibinfo {volume} {50}},\ \bibinfo {pages} {33} (\bibinfo {year} {2009})}\BibitemShut {NoStop}%
\bibitem [{\citenamefont {Hackl}\ and\ \citenamefont {Bianchi}(2021)}]{hackl2021bosonic}%
  \BibitemOpen
  \bibfield  {author} {\bibinfo {author} {\bibfnamefont {L.}~\bibnamefont {Hackl}}\ and\ \bibinfo {author} {\bibfnamefont {E.}~\bibnamefont {Bianchi}},\ }\bibfield  {title} {\bibinfo {title} {Bosonic and fermionic gaussian states from {K}{\"a}hler structures},\ }\href {https://scipost.org/10.21468/SciPostPhysCore.4.3.025} {\bibfield  {journal} {\bibinfo  {journal} {{SciPost Phys. Core}}\ }\textbf {\bibinfo {volume} {4}},\ \bibinfo {pages} {025} (\bibinfo {year} {2021})}\BibitemShut {NoStop}%
\bibitem [{\citenamefont {K{\"o}kc{\"u}}\ \emph {et~al.}(2022{\natexlab{a}})\citenamefont {K{\"o}kc{\"u}}, \citenamefont {Steckmann}, \citenamefont {Wang}, \citenamefont {Freericks}, \citenamefont {Dumitrescu},\ and\ \citenamefont {Kemper}}]{kokcu2022fixed}%
  \BibitemOpen
  \bibfield  {author} {\bibinfo {author} {\bibfnamefont {E.}~\bibnamefont {K{\"o}kc{\"u}}}, \bibinfo {author} {\bibfnamefont {T.}~\bibnamefont {Steckmann}}, \bibinfo {author} {\bibfnamefont {Y.}~\bibnamefont {Wang}}, \bibinfo {author} {\bibfnamefont {J.}~\bibnamefont {Freericks}}, \bibinfo {author} {\bibfnamefont {E.~F.}\ \bibnamefont {Dumitrescu}},\ and\ \bibinfo {author} {\bibfnamefont {A.~F.}\ \bibnamefont {Kemper}},\ }\bibfield  {title} {\bibinfo {title} {Fixed depth {H}amiltonian simulation via {C}artan decomposition},\ }\href {https://doi.org/10.1103/PhysRevLett.129.070501} {\bibfield  {journal} {\bibinfo  {journal} {{Phys. Rev. Lett.}}\ }\textbf {\bibinfo {volume} {129}},\ \bibinfo {pages} {070501} (\bibinfo {year} {2022}{\natexlab{a}})}\BibitemShut {NoStop}%
\bibitem [{\citenamefont {Chu}(2024)}]{chu2024lax}%
  \BibitemOpen
  \bibfield  {author} {\bibinfo {author} {\bibfnamefont {M.~T.}\ \bibnamefont {Chu}},\ }\bibfield  {title} {\bibinfo {title} {Lax dynamics for cartan decomposition with applications to hamiltonian simulation},\ }\href {https://doi.org/10.1093/imanum/drad018} {\bibfield  {journal} {\bibinfo  {journal} {IMA J. of Numer. Anal.}\ }\textbf {\bibinfo {volume} {44}},\ \bibinfo {pages} {1406} (\bibinfo {year} {2024})}\BibitemShut {NoStop}%
\bibitem [{\citenamefont {Wiersema}\ \emph {et~al.}(2025)\citenamefont {Wiersema}, \citenamefont {Kemper}, \citenamefont {Bakalov},\ and\ \citenamefont {Killoran}}]{wiersema2025geometric}%
  \BibitemOpen
  \bibfield  {author} {\bibinfo {author} {\bibfnamefont {R.}~\bibnamefont {Wiersema}}, \bibinfo {author} {\bibfnamefont {A.~F.}\ \bibnamefont {Kemper}}, \bibinfo {author} {\bibfnamefont {B.~N.}\ \bibnamefont {Bakalov}},\ and\ \bibinfo {author} {\bibfnamefont {N.}~\bibnamefont {Killoran}},\ }\bibfield  {title} {\bibinfo {title} {Geometric quantum machine learning with horizontal quantum gates},\ }\href {https://doi.org/10.1103/PhysRevResearch.7.013148} {\bibfield  {journal} {\bibinfo  {journal} {Physical Review Research}\ }\textbf {\bibinfo {volume} {7}},\ \bibinfo {pages} {013148} (\bibinfo {year} {2025})}\BibitemShut {NoStop}%
\bibitem [{\citenamefont {Hackl}\ \emph {et~al.}(2024)\citenamefont {Hackl}, \citenamefont {Kieburg},\ and\ \citenamefont {Maldonado}}]{hackl2024average}%
  \BibitemOpen
  \bibfield  {author} {\bibinfo {author} {\bibfnamefont {L.}~\bibnamefont {Hackl}}, \bibinfo {author} {\bibfnamefont {M.}~\bibnamefont {Kieburg}},\ and\ \bibinfo {author} {\bibfnamefont {J.}~\bibnamefont {Maldonado}},\ }\bibfield  {title} {\bibinfo {title} {Average mutual information for random fermionic gaussian quantum states},\ }\href {https://doi.org/10.48550/arXiv.2412.20244} {\bibfield  {journal} {\bibinfo  {journal} {arXiv preprint arXiv:2412.20244}\ } (\bibinfo {year} {2024})}\BibitemShut {NoStop}%
\bibitem [{\citenamefont {Gorodski}(2021)}]{gorodski2021introduction}%
  \BibitemOpen
  \bibfield  {author} {\bibinfo {author} {\bibfnamefont {C.}~\bibnamefont {Gorodski}},\ }\bibfield  {title} {\bibinfo {title} {An introduction to {R}iemannian symmetric spaces},\ }in\ \href {https://www.ime.usp.br/~gorodski/ps/symmetric-spaces.pdf} {\emph {\bibinfo {booktitle} {7th School and Workshop on Lie Theory}}}\ (\bibinfo {year} {2021})\ pp.\ \bibinfo {pages} {8--15}\BibitemShut {NoStop}%
\bibitem [{\citenamefont {Magnea}(2002)}]{magnea2002introduction}%
  \BibitemOpen
  \bibfield  {author} {\bibinfo {author} {\bibfnamefont {U.}~\bibnamefont {Magnea}},\ }\bibfield  {title} {\bibinfo {title} {An introduction to symmetric spaces},\ }\href {https://doi.org/10.48550/arXiv.cond-mat/0205288} {\bibfield  {journal} {\bibinfo  {journal} {arXiv preprint cond-mat/0205288}\ } (\bibinfo {year} {2002})}\BibitemShut {NoStop}%
\bibitem [{\citenamefont {D'Alessandro}\ and\ \citenamefont {Albertini}(2007)}]{dalesandro2007quantum}%
  \BibitemOpen
  \bibfield  {author} {\bibinfo {author} {\bibfnamefont {D.}~\bibnamefont {D'Alessandro}}\ and\ \bibinfo {author} {\bibfnamefont {F.}~\bibnamefont {Albertini}},\ }\bibfield  {title} {\bibinfo {title} {Quantum symmetries and {C}artan decompositions in arbitrary dimensions},\ }\href {https://doi.org/10.1088/1751-8113/40/10/013} {\bibfield  {journal} {\bibinfo  {journal} {{{J. Phys. A-Math. Theor.}}}\ }\textbf {\bibinfo {volume} {40}},\ \bibinfo {pages} {2439} (\bibinfo {year} {2007})}\BibitemShut {NoStop}%
\bibitem [{\citenamefont {Altland}\ and\ \citenamefont {Zirnbauer}(1997)}]{altland1997nonstandard}%
  \BibitemOpen
  \bibfield  {author} {\bibinfo {author} {\bibfnamefont {A.}~\bibnamefont {Altland}}\ and\ \bibinfo {author} {\bibfnamefont {M.~R.}\ \bibnamefont {Zirnbauer}},\ }\bibfield  {title} {\bibinfo {title} {Nonstandard symmetry classes in mesoscopic normal-superconducting hybrid structures},\ }\href {https://doi.org/10.1103/PhysRevB.55.1142} {\bibfield  {journal} {\bibinfo  {journal} {Phys. Rev. B}\ }\textbf {\bibinfo {volume} {55}},\ \bibinfo {pages} {1142} (\bibinfo {year} {1997})}\BibitemShut {NoStop}%
\bibitem [{\citenamefont {Ryu}\ \emph {et~al.}(2010)\citenamefont {Ryu}, \citenamefont {Schnyder}, \citenamefont {Furusaki},\ and\ \citenamefont {Ludwig}}]{ryu2010topological}%
  \BibitemOpen
  \bibfield  {author} {\bibinfo {author} {\bibfnamefont {S.}~\bibnamefont {Ryu}}, \bibinfo {author} {\bibfnamefont {A.~P.}\ \bibnamefont {Schnyder}}, \bibinfo {author} {\bibfnamefont {A.}~\bibnamefont {Furusaki}},\ and\ \bibinfo {author} {\bibfnamefont {A.~W.~W.}\ \bibnamefont {Ludwig}},\ }\bibfield  {title} {\bibinfo {title} {Topological insulators and superconductors: tenfold way and dimensional hierarchy},\ }\href {https://doi.org/10.1088/1367-2630/12/6/065010} {\bibfield  {journal} {\bibinfo  {journal} {{New J. Phys.}}\ }\textbf {\bibinfo {volume} {12}},\ \bibinfo {pages} {065010} (\bibinfo {year} {2010})}\BibitemShut {NoStop}%
\bibitem [{\citenamefont {Knapp}(2013)}]{knapp2013lie}%
  \BibitemOpen
  \bibfield  {author} {\bibinfo {author} {\bibfnamefont {A.~W.}\ \bibnamefont {Knapp}},\ }\href {https://doi.org/10.1007/978-1-4757-2453-0_1} {\emph {\bibinfo {title} {Lie Groups Beyond an Introduction}}},\ Vol.\ \bibinfo {volume} {140}\ (\bibinfo  {publisher} {Springer Science \& Business Media},\ \bibinfo {year} {2013})\BibitemShut {NoStop}%
\bibitem [{\citenamefont {Houde}\ \emph {et~al.}(2024)\citenamefont {Houde}, \citenamefont {McCutcheon},\ and\ \citenamefont {Quesada}}]{houde2024matrix}%
  \BibitemOpen
  \bibfield  {author} {\bibinfo {author} {\bibfnamefont {M.}~\bibnamefont {Houde}}, \bibinfo {author} {\bibfnamefont {W.}~\bibnamefont {McCutcheon}},\ and\ \bibinfo {author} {\bibfnamefont {N.}~\bibnamefont {Quesada}},\ }\bibfield  {title} {\bibinfo {title} {Matrix decompositions in quantum optics: Takagi/autonne, bloch–messiah/euler, iwasawa, and williamson},\ }\href {https://doi.org/10.1139/cjp-2024-0070} {\bibfield  {journal} {\bibinfo  {journal} {Can. J. Phys.}\ }\textbf {\bibinfo {volume} {102}},\ \bibinfo {pages} {497} (\bibinfo {year} {2024})},\ \Eprint {https://arxiv.org/abs/https://doi.org/10.1139/cjp-2024-0070} {https://doi.org/10.1139/cjp-2024-0070} \BibitemShut {NoStop}%
\bibitem [{\citenamefont {Da{\u{g}}l{\i}}\ \emph {et~al.}(2008)\citenamefont {Da{\u{g}}l{\i}}, \citenamefont {D'Alessandro},\ and\ \citenamefont {Smith}}]{dagli2008general}%
  \BibitemOpen
  \bibfield  {author} {\bibinfo {author} {\bibfnamefont {M.}~\bibnamefont {Da{\u{g}}l{\i}}}, \bibinfo {author} {\bibfnamefont {D.}~\bibnamefont {D'Alessandro}},\ and\ \bibinfo {author} {\bibfnamefont {J.~D.}\ \bibnamefont {Smith}},\ }\bibfield  {title} {\bibinfo {title} {A general framework for recursive decompositions of unitary quantum evolutions},\ }\href {http://doi.org/10.1088/1751-8113/41/15/155302} {\bibfield  {journal} {\bibinfo  {journal} {{J. Phys. A-Math. Theor.}}\ }\textbf {\bibinfo {volume} {41}},\ \bibinfo {pages} {155302} (\bibinfo {year} {2008})}\BibitemShut {NoStop}%
\bibitem [{\citenamefont {Wei}\ and\ \citenamefont {Di}(2012)}]{wei2012decomposition}%
  \BibitemOpen
  \bibfield  {author} {\bibinfo {author} {\bibfnamefont {H.-R.}\ \bibnamefont {Wei}}\ and\ \bibinfo {author} {\bibfnamefont {Y.-M.}\ \bibnamefont {Di}},\ }\bibfield  {title} {\bibinfo {title} {Decomposition of orthogonal matrix and synthesis of two-qubit and three-qubit orthogonal gates},\ }\href {https://doi.org/10.48550/arXiv.1203.0722} {\bibfield  {journal} {\bibinfo  {journal} {arXiv preprint arXiv:1203.0722}\ } (\bibinfo {year} {2012})}\BibitemShut {NoStop}%
\bibitem [{\citenamefont {de~Guise}\ \emph {et~al.}(2018)\citenamefont {de~Guise}, \citenamefont {Di~Matteo},\ and\ \citenamefont {S\'anchez-Soto}}]{deguise2018simple}%
  \BibitemOpen
  \bibfield  {author} {\bibinfo {author} {\bibfnamefont {H.}~\bibnamefont {de~Guise}}, \bibinfo {author} {\bibfnamefont {O.}~\bibnamefont {Di~Matteo}},\ and\ \bibinfo {author} {\bibfnamefont {L.~L.}\ \bibnamefont {S\'anchez-Soto}},\ }\bibfield  {title} {\bibinfo {title} {Simple factorization of unitary transformations},\ }\href {https://doi.org/10.1103/PhysRevA.97.022328} {\bibfield  {journal} {\bibinfo  {journal} {Phys. Rev. A}\ }\textbf {\bibinfo {volume} {97}},\ \bibinfo {pages} {022328} (\bibinfo {year} {2018})}\BibitemShut {NoStop}%
\bibitem [{\citenamefont {Guaita}\ \emph {et~al.}(2024)\citenamefont {Guaita}, \citenamefont {Hackl},\ and\ \citenamefont {Quella}}]{guaita2024representation}%
  \BibitemOpen
  \bibfield  {author} {\bibinfo {author} {\bibfnamefont {T.}~\bibnamefont {Guaita}}, \bibinfo {author} {\bibfnamefont {L.}~\bibnamefont {Hackl}},\ and\ \bibinfo {author} {\bibfnamefont {T.}~\bibnamefont {Quella}},\ }\bibfield  {title} {\bibinfo {title} {Representation theory of {G}aussian unitary transformations for bosonic and fermionic systems},\ }\href {https://doi.org/10.48550/arXiv.2409.11628} {\bibfield  {journal} {\bibinfo  {journal} {arXiv preprint arXiv:2409.11628}\ } (\bibinfo {year} {2024})}\BibitemShut {NoStop}%
\bibitem [{\citenamefont {Bullock}\ \emph {et~al.}(2005)\citenamefont {Bullock}, \citenamefont {Brennen},\ and\ \citenamefont {{O’L}eary}}]{bullock2005time}%
  \BibitemOpen
  \bibfield  {author} {\bibinfo {author} {\bibfnamefont {S.~S.}\ \bibnamefont {Bullock}}, \bibinfo {author} {\bibfnamefont {G.~K.}\ \bibnamefont {Brennen}},\ and\ \bibinfo {author} {\bibfnamefont {D.~P.}\ \bibnamefont {{O’L}eary}},\ }\bibfield  {title} {\bibinfo {title} {Time reversal and $n$-qubit canonical decompositions},\ }\href {https://doi.org/10.1063/1.1900293} {\bibfield  {journal} {\bibinfo  {journal} {{J. Math. Phys.}}\ }\textbf {\bibinfo {volume} {46}} (\bibinfo {year} {2005})}\BibitemShut {NoStop}%
\bibitem [{\citenamefont {D’Alessandro}\ and\ \citenamefont {Romano}(2006)}]{dalesandro2006decompositions}%
  \BibitemOpen
  \bibfield  {author} {\bibinfo {author} {\bibfnamefont {D.}~\bibnamefont {D’Alessandro}}\ and\ \bibinfo {author} {\bibfnamefont {R.}~\bibnamefont {Romano}},\ }\bibfield  {title} {\bibinfo {title} {Decompositions of unitary evolutions and entanglement dynamics of bipartite quantum systems},\ }\href {https://doi.org/10.1063/1.2245205} {\bibfield  {journal} {\bibinfo  {journal} {{J. Math. Phys.}}\ }\textbf {\bibinfo {volume} {47}} (\bibinfo {year} {2006})}\BibitemShut {NoStop}%
\bibitem [{\citenamefont {Albertini}\ and\ \citenamefont {D'Alessandro}(2006)}]{albertini2006analysis}%
  \BibitemOpen
  \bibfield  {author} {\bibinfo {author} {\bibfnamefont {F.}~\bibnamefont {Albertini}}\ and\ \bibinfo {author} {\bibfnamefont {D.}~\bibnamefont {D'Alessandro}},\ }\bibfield  {title} {\bibinfo {title} {Analysis and identification of quantum dynamics using {L}ie algebra homomorphisms and {C}artan decompositions},\ }\href {https://doi.org/10.48550/arXiv.quant-ph/0606057} {\bibfield  {journal} {\bibinfo  {journal} {arXiv preprint quant-ph/0606057}\ } (\bibinfo {year} {2006})}\BibitemShut {NoStop}%
\bibitem [{\citenamefont {Aguilar}\ \emph {et~al.}(2024)\citenamefont {Aguilar}, \citenamefont {Cichy}, \citenamefont {Eisert},\ and\ \citenamefont {Bittel}}]{aguilar2024full}%
  \BibitemOpen
  \bibfield  {author} {\bibinfo {author} {\bibfnamefont {G.}~\bibnamefont {Aguilar}}, \bibinfo {author} {\bibfnamefont {S.}~\bibnamefont {Cichy}}, \bibinfo {author} {\bibfnamefont {J.}~\bibnamefont {Eisert}},\ and\ \bibinfo {author} {\bibfnamefont {L.}~\bibnamefont {Bittel}},\ }\bibfield  {title} {\bibinfo {title} {Full classification of {P}auli {L}ie algebras},\ }\href {https://doi.org/10.48550/arXiv.2408.00081} {\bibfield  {journal} {\bibinfo  {journal} {arXiv preprint arXiv:2408.00081}\ } (\bibinfo {year} {2024})}\BibitemShut {NoStop}%
\bibitem [{\citenamefont {Wierichs}\ \emph {et~al.}(2025)\citenamefont {Wierichs}, \citenamefont {West}, \citenamefont {Forestano}, \citenamefont {Cerezo},\ and\ \citenamefont {Killoran}}]{symmetrycompilationrepo}%
  \BibitemOpen
  \bibfield  {author} {\bibinfo {author} {\bibfnamefont {D.}~\bibnamefont {Wierichs}}, \bibinfo {author} {\bibfnamefont {M.}~\bibnamefont {West}}, \bibinfo {author} {\bibfnamefont {R.}~\bibnamefont {Forestano}}, \bibinfo {author} {\bibfnamefont {M.}~\bibnamefont {Cerezo}},\ and\ \bibinfo {author} {\bibfnamefont {N.}~\bibnamefont {Killoran}},\ }\href {https://github.com/dwierichs/kak-tools} {\bibinfo {title} {Recursive {C}artan decompositions for unitary synthesis}} (\bibinfo {year} {2025}),\ \bibinfo {note} {{G}it{H}ub repository}\BibitemShut {NoStop}%
\bibitem [{\citenamefont {Wiersema}\ \emph {et~al.}(2023)\citenamefont {Wiersema}, \citenamefont {K{\"o}kc{\"u}}, \citenamefont {Kemper},\ and\ \citenamefont {Bakalov}}]{wiersema2023classification}%
  \BibitemOpen
  \bibfield  {author} {\bibinfo {author} {\bibfnamefont {R.}~\bibnamefont {Wiersema}}, \bibinfo {author} {\bibfnamefont {E.}~\bibnamefont {K{\"o}kc{\"u}}}, \bibinfo {author} {\bibfnamefont {A.~F.}\ \bibnamefont {Kemper}},\ and\ \bibinfo {author} {\bibfnamefont {B.~N.}\ \bibnamefont {Bakalov}},\ }\bibfield  {title} {\bibinfo {title} {Classification of dynamical lie algebras for translation-invariant 2-local spin systems in one dimension},\ }\href {https://arxiv.org/abs/2309.05690} {\bibfield  {journal} {\bibinfo  {journal} {arXiv preprint arXiv:2309.05690}\ } (\bibinfo {year} {2023})}\BibitemShut {NoStop}%
\bibitem [{\citenamefont {Kökcü}\ \emph {et~al.}(2024)\citenamefont {Kökcü}, \citenamefont {Wiersema}, \citenamefont {Kemper},\ and\ \citenamefont {Bakalov}}]{kökcü2024classification}%
  \BibitemOpen
  \bibfield  {author} {\bibinfo {author} {\bibfnamefont {E.}~\bibnamefont {Kökcü}}, \bibinfo {author} {\bibfnamefont {R.}~\bibnamefont {Wiersema}}, \bibinfo {author} {\bibfnamefont {A.~F.}\ \bibnamefont {Kemper}},\ and\ \bibinfo {author} {\bibfnamefont {B.~N.}\ \bibnamefont {Bakalov}},\ }\bibfield  {title} {\bibinfo {title} {Classification of dynamical {L}ie algebras generated by spin interactions on undirected graphs},\ }\href {https://doi.org/10.48550/arXiv.2409.19797} {\bibfield  {journal} {\bibinfo  {journal} {arXiv preprint arXiv:2409.19797}\ } (\bibinfo {year} {2024})}\BibitemShut {NoStop}%
\bibitem [{\citenamefont {Vidal}\ and\ \citenamefont {Dawson}(2004)}]{vidal2004universal}%
  \BibitemOpen
  \bibfield  {author} {\bibinfo {author} {\bibfnamefont {G.}~\bibnamefont {Vidal}}\ and\ \bibinfo {author} {\bibfnamefont {C.~M.}\ \bibnamefont {Dawson}},\ }\bibfield  {title} {\bibinfo {title} {Universal quantum circuit for two-qubit transformations with three controlled-{NOT} gates},\ }\href {https://doi.org/10.1103/PhysRevA.69.010301} {\bibfield  {journal} {\bibinfo  {journal} {Phys. Rev. A}\ }\textbf {\bibinfo {volume} {69}},\ \bibinfo {pages} {010301} (\bibinfo {year} {2004})}\BibitemShut {NoStop}%
\bibitem [{\citenamefont {Hagberg}\ \emph {et~al.}(2008)\citenamefont {Hagberg}, \citenamefont {Swart},\ and\ \citenamefont {Schult}}]{hagberg2008exploring}%
  \BibitemOpen
  \bibfield  {author} {\bibinfo {author} {\bibfnamefont {A.}~\bibnamefont {Hagberg}}, \bibinfo {author} {\bibfnamefont {P.~J.}\ \bibnamefont {Swart}},\ and\ \bibinfo {author} {\bibfnamefont {D.~A.}\ \bibnamefont {Schult}},\ }\bibfield  {title} {\bibinfo {title} {Exploring network structure, dynamics, and function using networkx},\ }in\ \href {https://www.osti.gov/biblio/960616} {\emph {\bibinfo {booktitle} {Scipy 08}}}\ (\bibinfo {organization} {Los Alamos National Laboratory, Los Alamos, NM (United States)},\ \bibinfo {year} {2008})\BibitemShut {NoStop}%
\bibitem [{\citenamefont {Kökcü}\ and\ \citenamefont {Steckmann}(2021)}]{fdhsrepo}%
  \BibitemOpen
  \bibfield  {author} {\bibinfo {author} {\bibfnamefont {E.}~\bibnamefont {Kökcü}}\ and\ \bibinfo {author} {\bibfnamefont {T.}~\bibnamefont {Steckmann}},\ }\href {https://github.com/kemperlab/cartan-quantum-synthesizer} {\bibinfo {title} {Cartan quantum synthesizer}} (\bibinfo {year} {2021})\BibitemShut {NoStop}%
\bibitem [{\citenamefont {Bergholm}\ \emph {et~al.}(2018)\citenamefont {Bergholm}, \citenamefont {Izaac}, \citenamefont {Schuld}, \citenamefont {Gogolin}, \citenamefont {Ahmed}, \citenamefont {Ajith}, \citenamefont {Alam}, \citenamefont {Alonso-Linaje}, \citenamefont {AkashNarayanan}, \citenamefont {Asadi} \emph {et~al.}}]{bergholm2018pennylane}%
  \BibitemOpen
  \bibfield  {author} {\bibinfo {author} {\bibfnamefont {V.}~\bibnamefont {Bergholm}}, \bibinfo {author} {\bibfnamefont {J.}~\bibnamefont {Izaac}}, \bibinfo {author} {\bibfnamefont {M.}~\bibnamefont {Schuld}}, \bibinfo {author} {\bibfnamefont {C.}~\bibnamefont {Gogolin}}, \bibinfo {author} {\bibfnamefont {S.}~\bibnamefont {Ahmed}}, \bibinfo {author} {\bibfnamefont {V.}~\bibnamefont {Ajith}}, \bibinfo {author} {\bibfnamefont {M.~S.}\ \bibnamefont {Alam}}, \bibinfo {author} {\bibfnamefont {G.}~\bibnamefont {Alonso-Linaje}}, \bibinfo {author} {\bibfnamefont {B.}~\bibnamefont {AkashNarayanan}}, \bibinfo {author} {\bibfnamefont {A.}~\bibnamefont {Asadi}}, \emph {et~al.},\ }\bibfield  {title} {\bibinfo {title} {Pennylane: Automatic differentiation of hybrid quantum-classical computations},\ }\href {https://doi.org/10.48550/arXiv.1811.04968} {\bibfield  {journal} {\bibinfo  {journal} {arXiv preprint arXiv:1811.04968}\ } (\bibinfo {year} {2018})}\BibitemShut {NoStop}%
\bibitem [{\citenamefont {K{\"o}kc{\"u}}\ \emph {et~al.}(2022{\natexlab{b}})\citenamefont {K{\"o}kc{\"u}}, \citenamefont {Camps}, \citenamefont {Bassman~Oftelie}, \citenamefont {Freericks}, \citenamefont {de~Jong}, \citenamefont {Van~Beeumen},\ and\ \citenamefont {Kemper}}]{kokcu2022algebraic}%
  \BibitemOpen
  \bibfield  {author} {\bibinfo {author} {\bibfnamefont {E.}~\bibnamefont {K{\"o}kc{\"u}}}, \bibinfo {author} {\bibfnamefont {D.}~\bibnamefont {Camps}}, \bibinfo {author} {\bibfnamefont {L.}~\bibnamefont {Bassman~Oftelie}}, \bibinfo {author} {\bibfnamefont {J.~K.}\ \bibnamefont {Freericks}}, \bibinfo {author} {\bibfnamefont {W.~A.}\ \bibnamefont {de~Jong}}, \bibinfo {author} {\bibfnamefont {R.}~\bibnamefont {Van~Beeumen}},\ and\ \bibinfo {author} {\bibfnamefont {A.~F.}\ \bibnamefont {Kemper}},\ }\bibfield  {title} {\bibinfo {title} {Algebraic compression of quantum circuits for {H}amiltonian evolution},\ }\href {https://doi.org/10.1103/PhysRevA.105.032420} {\bibfield  {journal} {\bibinfo  {journal} {{Phys. Rev. A}}\ }\textbf {\bibinfo {volume} {105}},\ \bibinfo {pages} {032420} (\bibinfo {year} {2022}{\natexlab{b}})}\BibitemShut {NoStop}%
\bibitem [{\citenamefont {Bassman~Oftelie}\ \emph {et~al.}(2022)\citenamefont {Bassman~Oftelie}, \citenamefont {Van~Beeumen}, \citenamefont {Younis}, \citenamefont {Smith}, \citenamefont {Iancu},\ and\ \citenamefont {de~Jong}}]{bassman2022constant}%
  \BibitemOpen
  \bibfield  {author} {\bibinfo {author} {\bibfnamefont {L.}~\bibnamefont {Bassman~Oftelie}}, \bibinfo {author} {\bibfnamefont {R.}~\bibnamefont {Van~Beeumen}}, \bibinfo {author} {\bibfnamefont {E.}~\bibnamefont {Younis}}, \bibinfo {author} {\bibfnamefont {E.}~\bibnamefont {Smith}}, \bibinfo {author} {\bibfnamefont {C.}~\bibnamefont {Iancu}},\ and\ \bibinfo {author} {\bibfnamefont {W.~A.}\ \bibnamefont {de~Jong}},\ }\bibfield  {title} {\bibinfo {title} {Constant-depth circuits for dynamic simulations of materials on quantum computers},\ }\href {https://doi.org/https://doi.org/10.1186/s41313-022-00043-x} {\bibfield  {journal} {\bibinfo  {journal} {J. Mater. Sci.}\ }\textbf {\bibinfo {volume} {6}},\ \bibinfo {pages} {13} (\bibinfo {year} {2022})}\BibitemShut {NoStop}%
\bibitem [{\citenamefont {Peng}\ \emph {et~al.}(2022)\citenamefont {Peng}, \citenamefont {Gulania}, \citenamefont {Alexeev},\ and\ \citenamefont {Govind}}]{peng2022quantum}%
  \BibitemOpen
  \bibfield  {author} {\bibinfo {author} {\bibfnamefont {B.}~\bibnamefont {Peng}}, \bibinfo {author} {\bibfnamefont {S.}~\bibnamefont {Gulania}}, \bibinfo {author} {\bibfnamefont {Y.}~\bibnamefont {Alexeev}},\ and\ \bibinfo {author} {\bibfnamefont {N.}~\bibnamefont {Govind}},\ }\bibfield  {title} {\bibinfo {title} {Quantum time dynamics employing the yang-baxter equation for circuit compression},\ }\href {https://doi.org/10.1103/PhysRevA.106.012412} {\bibfield  {journal} {\bibinfo  {journal} {Phys. Rev. A}\ }\textbf {\bibinfo {volume} {106}},\ \bibinfo {pages} {012412} (\bibinfo {year} {2022})}\BibitemShut {NoStop}%
\bibitem [{\citenamefont {Ogunkoya}\ \emph {et~al.}(2024)\citenamefont {Ogunkoya}, \citenamefont {Kim}, \citenamefont {Peng}, \citenamefont {\"Ozg\"uler},\ and\ \citenamefont {Alexeev}}]{ogunkoya2024qutrit}%
  \BibitemOpen
  \bibfield  {author} {\bibinfo {author} {\bibfnamefont {O.}~\bibnamefont {Ogunkoya}}, \bibinfo {author} {\bibfnamefont {J.}~\bibnamefont {Kim}}, \bibinfo {author} {\bibfnamefont {B.}~\bibnamefont {Peng}}, \bibinfo {author} {\bibfnamefont {A.~B.}\ \bibnamefont {\"Ozg\"uler}},\ and\ \bibinfo {author} {\bibfnamefont {Y.}~\bibnamefont {Alexeev}},\ }\bibfield  {title} {\bibinfo {title} {Qutrit circuits and algebraic relations: A pathway to efficient spin-1 hamiltonian simulation},\ }\href {https://doi.org/10.1103/PhysRevA.109.012426} {\bibfield  {journal} {\bibinfo  {journal} {Phys. Rev. A}\ }\textbf {\bibinfo {volume} {109}},\ \bibinfo {pages} {012426} (\bibinfo {year} {2024})}\BibitemShut {NoStop}%
\bibitem [{\citenamefont {Larocca}\ \emph {et~al.}(2022)\citenamefont {Larocca}, \citenamefont {Sauvage}, \citenamefont {Sbahi}, \citenamefont {Verdon}, \citenamefont {Coles},\ and\ \citenamefont {Cerezo}}]{larocca2022group}%
  \BibitemOpen
  \bibfield  {author} {\bibinfo {author} {\bibfnamefont {M.}~\bibnamefont {Larocca}}, \bibinfo {author} {\bibfnamefont {F.}~\bibnamefont {Sauvage}}, \bibinfo {author} {\bibfnamefont {F.~M.}\ \bibnamefont {Sbahi}}, \bibinfo {author} {\bibfnamefont {G.}~\bibnamefont {Verdon}}, \bibinfo {author} {\bibfnamefont {P.~J.}\ \bibnamefont {Coles}},\ and\ \bibinfo {author} {\bibfnamefont {M.}~\bibnamefont {Cerezo}},\ }\bibfield  {title} {\bibinfo {title} {Group-invariant quantum machine learning},\ }\href {https://doi.org/10.1103/PRXQuantum.3.030341} {\bibfield  {journal} {\bibinfo  {journal} {PRX Quantum}\ }\textbf {\bibinfo {volume} {3}},\ \bibinfo {pages} {030341} (\bibinfo {year} {2022})}\BibitemShut {NoStop}%
\bibitem [{\citenamefont {Skolik}\ \emph {et~al.}(2023)\citenamefont {Skolik}, \citenamefont {Cattelan}, \citenamefont {Yarkoni}, \citenamefont {B{\"a}ck},\ and\ \citenamefont {Dunjko}}]{skolik2022equivariant}%
  \BibitemOpen
  \bibfield  {author} {\bibinfo {author} {\bibfnamefont {A.}~\bibnamefont {Skolik}}, \bibinfo {author} {\bibfnamefont {M.}~\bibnamefont {Cattelan}}, \bibinfo {author} {\bibfnamefont {S.}~\bibnamefont {Yarkoni}}, \bibinfo {author} {\bibfnamefont {T.}~\bibnamefont {B{\"a}ck}},\ and\ \bibinfo {author} {\bibfnamefont {V.}~\bibnamefont {Dunjko}},\ }\bibfield  {title} {\bibinfo {title} {Equivariant quantum circuits for learning on weighted graphs},\ }\href {https://doi.org/10.1038/s41534-023-00710-y} {\bibfield  {journal} {\bibinfo  {journal} {npj Quantum Information}\ }\textbf {\bibinfo {volume} {9}},\ \bibinfo {pages} {47} (\bibinfo {year} {2023})}\BibitemShut {NoStop}%
\bibitem [{\citenamefont {Meyer}\ \emph {et~al.}(2023)\citenamefont {Meyer}, \citenamefont {Mularski}, \citenamefont {Gil-Fuster}, \citenamefont {Mele}, \citenamefont {Arzani}, \citenamefont {Wilms},\ and\ \citenamefont {Eisert}}]{meyer2022exploiting}%
  \BibitemOpen
  \bibfield  {author} {\bibinfo {author} {\bibfnamefont {J.~J.}\ \bibnamefont {Meyer}}, \bibinfo {author} {\bibfnamefont {M.}~\bibnamefont {Mularski}}, \bibinfo {author} {\bibfnamefont {E.}~\bibnamefont {Gil-Fuster}}, \bibinfo {author} {\bibfnamefont {A.~A.}\ \bibnamefont {Mele}}, \bibinfo {author} {\bibfnamefont {F.}~\bibnamefont {Arzani}}, \bibinfo {author} {\bibfnamefont {A.}~\bibnamefont {Wilms}},\ and\ \bibinfo {author} {\bibfnamefont {J.}~\bibnamefont {Eisert}},\ }\bibfield  {title} {\bibinfo {title} {Exploiting symmetry in variational quantum machine learning},\ }\href {https://doi.org/10.1103/PRXQuantum.4.010328} {\bibfield  {journal} {\bibinfo  {journal} {PRX Quantum}\ }\textbf {\bibinfo {volume} {4}},\ \bibinfo {pages} {010328} (\bibinfo {year} {2023})}\BibitemShut {NoStop}%
\bibitem [{\citenamefont {West}\ \emph {et~al.}(2024)\citenamefont {West}, \citenamefont {Heredge}, \citenamefont {Sevior},\ and\ \citenamefont {Usman}}]{west2024provably}%
  \BibitemOpen
  \bibfield  {author} {\bibinfo {author} {\bibfnamefont {M.~T.}\ \bibnamefont {West}}, \bibinfo {author} {\bibfnamefont {J.}~\bibnamefont {Heredge}}, \bibinfo {author} {\bibfnamefont {M.}~\bibnamefont {Sevior}},\ and\ \bibinfo {author} {\bibfnamefont {M.}~\bibnamefont {Usman}},\ }\bibfield  {title} {\bibinfo {title} {Provably trainable rotationally equivariant quantum machine learning},\ }\href {https://doi.org/10.1103/PRXQuantum.5.030320} {\bibfield  {journal} {\bibinfo  {journal} {PRX Quantum}\ }\textbf {\bibinfo {volume} {5}},\ \bibinfo {pages} {030320} (\bibinfo {year} {2024})}\BibitemShut {NoStop}%
\bibitem [{\citenamefont {Sauvage}\ \emph {et~al.}(2024)\citenamefont {Sauvage}, \citenamefont {Larocca}, \citenamefont {Coles},\ and\ \citenamefont {Cerezo}}]{sauvage2022building}%
  \BibitemOpen
  \bibfield  {author} {\bibinfo {author} {\bibfnamefont {F.}~\bibnamefont {Sauvage}}, \bibinfo {author} {\bibfnamefont {M.}~\bibnamefont {Larocca}}, \bibinfo {author} {\bibfnamefont {P.~J.}\ \bibnamefont {Coles}},\ and\ \bibinfo {author} {\bibfnamefont {M.}~\bibnamefont {Cerezo}},\ }\bibfield  {title} {\bibinfo {title} {Building spatial symmetries into parameterized quantum circuits for faster training},\ }\href {https://doi.org/10.1088/2058-9565/ad152e} {\bibfield  {journal} {\bibinfo  {journal} {Quantum Science and Technology}\ }\textbf {\bibinfo {volume} {9}},\ \bibinfo {pages} {015029} (\bibinfo {year} {2024})}\BibitemShut {NoStop}%
\bibitem [{\citenamefont {Vatan}\ and\ \citenamefont {Williams}(2004{\natexlab{b}})}]{vatan2004optimal}%
  \BibitemOpen
  \bibfield  {author} {\bibinfo {author} {\bibfnamefont {F.}~\bibnamefont {Vatan}}\ and\ \bibinfo {author} {\bibfnamefont {C.}~\bibnamefont {Williams}},\ }\bibfield  {title} {\bibinfo {title} {Optimal quantum circuits for general two-qubit gates},\ }\href {https://doi.org/10.1103/PhysRevA.69.032315} {\bibfield  {journal} {\bibinfo  {journal} {{Phys. Rev. A}}\ }\textbf {\bibinfo {volume} {69}},\ \bibinfo {pages} {032315} (\bibinfo {year} {2004}{\natexlab{b}})}\BibitemShut {NoStop}%
\bibitem [{\citenamefont {Desura}(2018)}]{SE_anticom}%
  \BibitemOpen
  \bibfield  {author} {\bibinfo {author} {\bibnamefont {Desura}},\ }\href {https://math.stackexchange.com/a/2807180} {\bibinfo {title} {On the anti-commuting matrices}} (\bibinfo {year} {2018})\BibitemShut {NoStop}%
\bibitem [{\citenamefont {Powell}(2011)}]{powell2011calculating}%
  \BibitemOpen
  \bibfield  {author} {\bibinfo {author} {\bibfnamefont {P.~D.}\ \bibnamefont {Powell}},\ }\bibfield  {title} {\bibinfo {title} {Calculating determinants of block matrices},\ }\href {https://doi.org/10.48550/arXiv.1112.4379} {\bibfield  {journal} {\bibinfo  {journal} {arXiv preprint arXiv:1112.4379}\ } (\bibinfo {year} {2011})}\BibitemShut {NoStop}%
\bibitem [{\citenamefont {Virtanen}\ \emph {et~al.}(2020)\citenamefont {Virtanen}, \citenamefont {Gommers}, \citenamefont {Oliphant}, \citenamefont {Haberland}, \citenamefont {Reddy}, \citenamefont {Cournapeau}, \citenamefont {Burovski}, \citenamefont {Peterson}, \citenamefont {Weckesser}, \citenamefont {Bright} \emph {et~al.}}]{virtanen2020scipy}%
  \BibitemOpen
  \bibfield  {author} {\bibinfo {author} {\bibfnamefont {P.}~\bibnamefont {Virtanen}}, \bibinfo {author} {\bibfnamefont {R.}~\bibnamefont {Gommers}}, \bibinfo {author} {\bibfnamefont {T.~E.}\ \bibnamefont {Oliphant}}, \bibinfo {author} {\bibfnamefont {M.}~\bibnamefont {Haberland}}, \bibinfo {author} {\bibfnamefont {T.}~\bibnamefont {Reddy}}, \bibinfo {author} {\bibfnamefont {D.}~\bibnamefont {Cournapeau}}, \bibinfo {author} {\bibfnamefont {E.}~\bibnamefont {Burovski}}, \bibinfo {author} {\bibfnamefont {P.}~\bibnamefont {Peterson}}, \bibinfo {author} {\bibfnamefont {W.}~\bibnamefont {Weckesser}}, \bibinfo {author} {\bibfnamefont {J.}~\bibnamefont {Bright}}, \emph {et~al.},\ }\bibfield  {title} {\bibinfo {title} {{SciPy} 1.0: fundamental algorithms for scientific computing in {P}ython},\ }\href {https://doi.org/10.1038/s41592-019-0686-2} {\bibfield  {journal} {\bibinfo  {journal} {{Nat. Methods}}\ }\textbf {\bibinfo {volume} {17}},\ \bibinfo {pages} {261} (\bibinfo {year} {2020})}\BibitemShut {NoStop}%
\bibitem [{\citenamefont {Rand}\ \emph {et~al.}(1988)\citenamefont {Rand}, \citenamefont {Winternitz},\ and\ \citenamefont {Zassenhaus}}]{rand1988identification}%
  \BibitemOpen
  \bibfield  {author} {\bibinfo {author} {\bibfnamefont {D.}~\bibnamefont {Rand}}, \bibinfo {author} {\bibfnamefont {P.}~\bibnamefont {Winternitz}},\ and\ \bibinfo {author} {\bibfnamefont {H.}~\bibnamefont {Zassenhaus}},\ }\bibfield  {title} {\bibinfo {title} {On the identification of a {L}ie algebra given by its structure constants. i. direct decompositions, {L}evi decompositions, and nilradicals},\ }\href {https://doi.org/https://doi.org/10.1016/0024-3795(88)90210-8} {\bibfield  {journal} {\bibinfo  {journal} {{Linear Algebra Appl.}}\ }\textbf {\bibinfo {volume} {109}},\ \bibinfo {pages} {197} (\bibinfo {year} {1988})}\BibitemShut {NoStop}%
\bibitem [{\citenamefont {de~Graaf}(1997)}]{degraaf1997algorithm}%
  \BibitemOpen
  \bibfield  {author} {\bibinfo {author} {\bibfnamefont {W.~A.}\ \bibnamefont {de~Graaf}},\ }\bibfield  {title} {\bibinfo {title} {An algorithm for the decomposition of semisimple {L}ie algebras},\ }\href {https://doi.org/https://doi.org/10.1016/S0304-3975(97)00060-1} {\bibfield  {journal} {\bibinfo  {journal} {{Theor. Comput. Sci.}}\ }\textbf {\bibinfo {volume} {187}},\ \bibinfo {pages} {117} (\bibinfo {year} {1997})}\BibitemShut {NoStop}%
\bibitem [{\citenamefont {Wan}\ \emph {et~al.}(2023)\citenamefont {Wan}, \citenamefont {Huggins}, \citenamefont {Lee},\ and\ \citenamefont {Babbush}}]{wan2022matchgate}%
  \BibitemOpen
  \bibfield  {author} {\bibinfo {author} {\bibfnamefont {K.}~\bibnamefont {Wan}}, \bibinfo {author} {\bibfnamefont {W.~J.}\ \bibnamefont {Huggins}}, \bibinfo {author} {\bibfnamefont {J.}~\bibnamefont {Lee}},\ and\ \bibinfo {author} {\bibfnamefont {R.}~\bibnamefont {Babbush}},\ }\bibfield  {title} {\bibinfo {title} {Matchgate shadows for fermionic quantum simulation},\ }\href {https://doi.org/10.1007/s00220-023-04844-0} {\bibfield  {journal} {\bibinfo  {journal} {{Commun. Math. Phys.}}\ }\textbf {\bibinfo {volume} {404}},\ \bibinfo {pages} {629} (\bibinfo {year} {2023})}\BibitemShut {NoStop}%
\bibitem [{\citenamefont {Diaz}\ \emph {et~al.}(2023)\citenamefont {Diaz}, \citenamefont {Garc{\'\i}a-Mart{\'\i}n}, \citenamefont {Kazi}, \citenamefont {Larocca},\ and\ \citenamefont {Cerezo}}]{diaz2023showcasing}%
  \BibitemOpen
  \bibfield  {author} {\bibinfo {author} {\bibfnamefont {N.~L.}\ \bibnamefont {Diaz}}, \bibinfo {author} {\bibfnamefont {D.}~\bibnamefont {Garc{\'\i}a-Mart{\'\i}n}}, \bibinfo {author} {\bibfnamefont {S.}~\bibnamefont {Kazi}}, \bibinfo {author} {\bibfnamefont {M.}~\bibnamefont {Larocca}},\ and\ \bibinfo {author} {\bibfnamefont {M.}~\bibnamefont {Cerezo}},\ }\bibfield  {title} {\bibinfo {title} {Showcasing a barren plateau theory beyond the dynamical {L}ie algebra},\ }\href {https://arxiv.org/abs/2310.11505} {\bibfield  {journal} {\bibinfo  {journal} {arXiv preprint arXiv:2310.11505}\ } (\bibinfo {year} {2023})}\BibitemShut {NoStop}%
\end{thebibliography}%
\clearpage
\newpage
\onecolumngrid
\appendix
\acresetall
\section{Riemannian symmetric spaces and their classification}\label{sec:symspace_classification}
\begin{figure}
    \centering
    \includegraphics[width=0.85\linewidth]{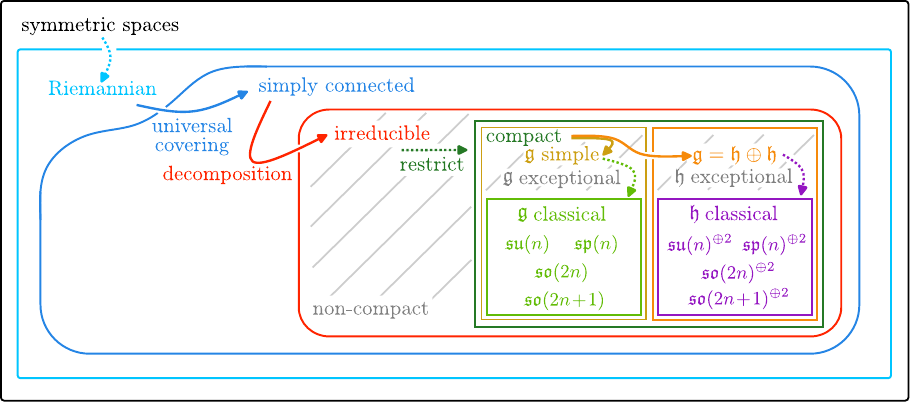}
    \caption{Classification procedure for Riemannian symmetric spaces (blue), a subclass of symmetric spaces.
    The procedure consists of multiple steps (solid arrows), which we interleave with restrictions to spaces of interest to us (dashed arrows).
    The total space of the quotient that yields the symmetric space takes the form $\groupG=\bigtimes_{i} \groupG_i$ with irreducible simply connected $\groupG_i$ (red) whose Lie algebras are one of $\mfsu(n)^{\oplus \nu}$, $\mfso(2n+1)^{\oplus \nu}$, $\mfsp(n)^{\oplus \nu}$, $\mfso(2n)^{\oplus \nu}$ with $n\geq 1$ and $\nu=1$ (green) or $\nu=2$ (purple). Here we manually excluded non-compact symmetric spaces and exceptional Lie algebras (grey).
    The last step is to classify the symmetric spaces that arrive from a given $\groupG_i$, which is not depicted here.
    }
    \label{fig:symspace_classification}
\end{figure}
Here we briefly outline the classification of (Riemannian) symmetric spaces, following~\cite{gorodski2021introduction}. We will make a number of restrictions to the mathematically rigorous classification procedure to adapt it to our needs.
We found it instructive to visualize the classification procedure together with our choices of restriction in Fig.~\ref{fig:symspace_classification}.

The first step is to consider simply connected symmetric spaces, which are (Riemannian) products of so-called \emph{irreducible} symmetric spaces.
If the space is, in fact, \emph{not} simply connected, we consider its universal covering, which is simply connected. 
Second, irreducible symmetric spaces come in four classes, two of which are compact and two of which are noncompact, with a duality between the cases.
We are interested in the compact case, which leaves us with two classes: the Lie algebra $\mfg$ of $\groupG$ must be a) real compact simple, or b) ``double" such an algebra, i.e.,~$\mfg=\mfh\oplus\mfh$ with $\mfh$ real compact simple.
The third step is to note that real compact simple algebras consists of just four infinite families of so-called \emph{classical} Lie algebras as well as five exceptional Lie algebras ($E_{6/7/8}$, $F_4$, $G_2$).
We restrict the mathematical classification to the classical Lie algebras alone, given by $A_n=\mfsu(n+1)$, $B_n=\mfso(2n+1)$, $C_n=\mfsp(n)$ and $D_n=\mfso(2n)$; also see Eqs.~(\ref{eq:def_sun})-(\ref{eq:def_spn}).
This means that the irreducible spaces we consider are created from groups $\groupG$ that either have a classical Lie algebra or double of a classical Lie algebra, depicted by the green and purple box in Fig.~\ref{fig:symspace_classification}, respectively.
In the last step, a detailed analysis of the remaining cases yields the subgroups $\groupK$ (or subalgebras $\mfk$) that give rise to symmetric spaces $\groupG/\groupK$ for a given $\groupG$. Filtering them for our restrictions, we arrive at a small number of Riemannian symmetric spaces of interest.

In addition to the manual restrictions performed above, we will implicitly extend the classification result to some non-semisimple algebras $\mfg$, using the fact that \acp{CD} are compatible with abelian phases; see Obs.~\ref{obs:phases}. Intuitively, abelian phases do not interact with any other part of the space, which gives them a Euclidean geometry. They can be divided out by including them in $\groupK$ or appended to $\groupP$ by allowing for such Euclidean components.
It turns out that the classification has particular notational clarity if we work with $\gru(n)=\grsu(n)\times \gru(1)$ instead of $\grsu(n)$.

We summarize the classification, combined with all manual restrictions and extensions, in Tab.~\ref{tab:symm_classif} in the main text, which is adapted from~\cite{edelman2023fifty,wiersema2025geometric}.

\section{AIII+A Cartan decompositions}\label{sec:unified_cartan_decomp_details}
In this section, we provide more details on the form of AIII+A \acp{CD} used for several popular unitary synthesis algorithms mentioned in Sec.~\ref{sec:recursive_kak_literature}, namely the Khaneja-Glaser~\cite{khaneja2001cartan, vatan2004realization, bullock2004note, mansky2023near}, Quantum Shannon~\cite{shende2005synthesis, mottonen2006decompositions, drury2008constructive}, and optimized Block-ZXZ~\cite{krol2024beyond} families. By doing so, we highlight the close similarity of all these approaches. The results are summarized in Tab.~\ref{tab:AIII_A_decompositions}.

\begin{table}[t]
    \centering
    \begin{tabular}{|p{3.5cm}|l|c|c|c|c|}
        \hline
        Family & Classes & $\mfg(N)$ & $\mfk(N)$ & $\mfp(N)$ & $\mfa(N)$ \\
        \hline

        \rc{tabpurple}
        \cellcolor{white} Quantum Shannon De-  &
        (i) AIII & 
        $\mfu(2^N)$ &
        $\{\id, \pauliz\}\otimes\mfu(2^{N-1})$ &
        $\{\paulix,\pauliy\}\otimes\mfu(2^{N-1})$ &
        $\pauliy\otimes\{\ketbra{j}{j}\}$ \\
        \rc{tabgreen}
        \cellcolor{white} composition \cite{shende2005synthesis, mottonen2006decompositions, drury2008constructive} &
        (ii) A & 
        $\{\id, \pauliz\}\otimes\mfu(2^{N-1})$ &
        $\id\otimes\mfu(2^{N-1})$ &
        $\pauliz\otimes\mfu(2^{N-1})$ &
        $\pauliz\otimes\{\ketbra{j}{j}\}$ \\
        \hline
        
        \rc{tabpurple}
        \cellcolor{white} Optimized Block-ZXZ &
        (i) AIII & 
        $\mfu(2^N)$ &
        $\{\id, \pauliz\}\otimes\mfu(2^{N-1})$ &
        $\{\paulix,\pauliy\}\otimes\mfu(2^{N-1})$ &
        $\paulix\otimes\{\ketbra{j}{j}\}$  \\
        \rc{tabgreen}
        \cellcolor{white} Decomposition \cite{krol2024beyond} &
        (ii) A & 
        $\{\id, \pauliz\}\otimes\mfu(2^{N-1})$ &
        $\id\otimes\mfu(2^{N-1})$ &
        $\pauliz\otimes\mfu(2^{N-1})$ &
        $\pauliz\otimes\{\ketbra{j}{j}\}$  \\
        \hline
        
        \rc{tabpurple}
        \cellcolor{white}Khaneja-Glaser &
        (i) AIII ($N>2$)&
        $\mfu(2^N)$ &
        $\{\id,\pauliz\}\otimes\mfu(2^{N-1})$ &
        $\{\paulix,\pauliy\}\otimes\mfu(2^{N-1})$ &
        Eqs.~(\ref{eq:KG_CSA_AIII_3})-(\ref{eq:KG_CSA_AIII_5}) \\
        \rc{tabgreen}
        \cellcolor{white} Decomposition \cite{khaneja2001cartan, bullock2004note, vatan2004realization,mansky2023near} &
        (ii) A ($N>2$) &
        $\{\id,\pauliz\}\otimes\mfu(2^{N-1})$ &
        $\id\otimes\mfu(2^{N-1})$ &
        $\pauliz\otimes\mfu(2^{N-1})$ &
        Eqs.~(\ref{eq:KG_CSA_A_3})-(\ref{eq:KG_CSA_A_5}) \\
        \rc{salmon1!50}
        \cellcolor{white} &
        (iii) AI ($N=2$) &
        $\mfu(4)$ &
        $\mfso(4)$ &
        $\mfso(4)^\perp$ &
        $\mfa_\mathrm{base}$ (Eq.~(\ref{eq:KG_CSA_base}))\\
        \hline
    \end{tabular}
  \caption{Explicit details for several decomposition algorithms in the literature, which all can be seen as recursive AIII+A \acp{CD}. $N$ indicates the number of qubits active in the decomposition. For brevity, imaginary spans of all listed elements are implied. 
  Cells with the same colour indicate that the same components are used in that decomposition step.
  Notice that for all classes, the AIII and A rows only differ in the choice of \ac{CSA} $\mfa$.
  The \ac{KGD} breaks this pattern at the $2$-qubit step.
  Here, $\mfso(4)$ and $\mfso(4)^\perp$ correspond to Pauli words on $2$ qubits with an odd/even number of $\pauliy$s, rotated to the so-called magic basis~\cite{vatan2004optimal}. This basis change makes it possible to decouple the last two qubit via a different decomposition than AIII+A, thanks to the exceptional isomorphism  $\mfso(4)\cong \mfsu(2)\oplus\mfsu(2)$.}
  \label{tab:AIII_A_decompositions}
\end{table}

For  clarity, in the following we will work with the unitary algebras $\mfu(n)$ with $n=2^N$ for $N$ qubits, rather than the special unitary algebras $\mfsu(n)$ that are historically seen for some of the above decompositions. The only difference between $\mfu(n)$ and $\mfsu(n)$ is whether the identity is included as an element. Since the identity plays nice (i.e., commutes) with everything else, this choice has no major mathematical consequences for the recursive decomposition, c.f. Obs.~\ref{obs:phases}. We find that using the unitary algebras leads to simpler derivations and a clearer picture of the various subalgebras and subspaces. It is also in line with our characterization of Riemannian symmetric spaces in the main text, Tab.~\ref{tab:symm_classif}. 

All decompositions cited above proceed by considering a basis of the unitary algebra consisting of Pauli words, 
\begin{align}
    \mfg(N)\coloneqq\mfu(n)=\spaniR\{P_1\otimes\cdots\otimes P_N \vert P_j\in\{\id, \paulix, \pauliy, \pauliz\}\},
\end{align}
and splitting off the $\id$ and $\pauliz$ components from the $\paulix$ and $\pauliy$ components for the first qubit. 
This gives a \ac{CD} $\mfg(N)=\mfk_{\rm AIII}(n)\oplus\mfp_{\rm AIII}(n)$ with 
\begin{align}\label{eq:k_aiii_literature}
    \mfk_{\rm AIII}(N) & :=\,\id\otimes\mfu(2^{N-1})\ \oplus\ \pauliz\otimes\mfu(2^{N-1}), \\
    \mfp_{\rm AIII}(N) & :=\paulix\otimes\mfu(2^{N-1})\ \oplus\ \pauliy\otimes\mfu(2^{N-1}),
\end{align}
which satisfies the symmetric space commutation conditions in Eqs.~(\ref{eq:subalg_prop})-(\ref{eq:symm_prop}) and clearly is induced by the involutions $\theta_{\rm AIII}=\Ad_{Z_0}$.
Notice that any other Pauli operator could be separated out instead of $\pauliz$.
However, separating out $\{\id, \pauliz\}$ is particularly illustrative because the subalgebra $\mfk_{\rm AIII}({N})$ becomes block diagonal.
This block diagonal structure makes it obvious how to introduce a type-A decomposition, which decomposes block diagonal matrices into their swap-symmetric ($\id$) and swap-antisymmetric ($\pauliz$) components with respect to the involution $\swapsymbol=\Ad_{X_0}$, 
\begin{align}
    \mfk_{\rm AIII}(N)&= \mfk_{A}({N})\oplus\mfp_A({N}), \\
    \mfk_A({N}) &= \id\otimes\mfu(2^{N-1}), \\
    \mfp_A({N}) &= \pauliz\otimes\mfu(2^{N-1})\label{eq:horizontal_typeA_qsd}.    
\end{align}
Since $\mfk_A({N})$ has fully decoupled the first qubit from the rest—i.e., $\mfk_A({N})\cong \mfu(2^{N-1})=\mfg(N-1)$ and all algebra elements are the identity on the first qubit—we can proceed recursively with AIII and A decompositions, until reaching a desired stopping point (typically, $\mfg(2)=\mfu(4)$ or $\mfg(1)=\mfu(2)$).

This recursive AIII+A decomposition strategy underlies—sometimes unknowingly—all the above-mentioned decompositions. The primary difference between the different decomposition families is in the choice of \ac{CSA}. This choice is mathematically arbitrary, but in practice it influences the final form of the decomposed circuit, which obfuscates the similarities between the decompositions. We will now outline how each decomposition fits this pattern. 

\subsection{Quantum Shannon Decomposition}
The first step of the \ac{QSD} is a straightforward usage of the \ac{CSD}, which can be seen to be an AIII \ac{CD} from Tab.~\ref{tab:numerical_overview}. For this decomposition, the \ac{CSG} element $A$ has the following form (in the computational basis):
\begin{align}
    A_{\rm AIII} = \begin{pmatrix}
        C & S \\
        -S & C
    \end{pmatrix},
\end{align}
where $C$/$S$ are diagonal matrices with entries $\cos(\theta_j)$/$\sin(\theta_j)$ for some $\theta_j\in\mathbb{R}$. It is easy to see that any such matrix is generated as $A_{\rm AIII}=\exp(i\sum_j\theta_j\pauliy\otimes\ketbra{j}{j})$; hence the initial \ac{CSA} for the \ac{QSD} can be identified as
\begin{align}
    \mfa_{\rm AIII}(N)=\spaniR\{\pauliy\otimes\ketbra{j}{j}\}_{j=0}^{2^{N-1}-1},
\end{align}
where $\ket{j}$ are the computational basis states of the remaining $N-1$ qubits. Note that the dimension of this \ac{CSA} is $n/2=2^{N-1}$, as expected for an AIII decomposition with $p=q=n/2$.

Next, the \ac{QSD} applies a ``demultiplexing" step—a type-A \ac{CD}—to the gates generated by $\mfk_{\rm AIII}(N)$. As discussed above, this produces a horizontal space given by Eq.~(\ref{eq:horizontal_typeA_qsd}).
The \ac{QSD} specifies the horizontal gates, i.e., the \ac{CSG} elements, for the demultiplexing step to be unitaries of the form $D\oplus D^\dagger$, where $D$ is diagonal, i.e., $D=\diag(\{e^{i\phi_j}\}_{j=0}^{2^{N-1}-1})$ with $\phi_j\in\mathbb{R}$.
Such matrices are easily seen to be generated by a \ac{CSA} of the form 
\begin{align}
    \mfa_{\rm A}({N})=\spaniR\{\pauliz\otimes\ketbra{j}{j}\}_{j=0}^{2^{N-1}-1},
\end{align}
with $\ket{j}$ again running through all computational basis states of the remaining $N-1$ qubits. Again, we note that the dimension of this \ac{CSA} is $2^{N-1}$, matching the rank of the type-A decomposition of $\mfu(2^{N-1})^{\oplus 2}$.

\subsection{Optimized Block-ZXZ Decomposition}

The (unoptimized) Block-ZXZ decomposition, first appearing in~\cite{de2016block, de2018unified}, breaks an $N$-qubit unitary down into a particular block-diagonal form,
\begin{align}
    \label{eq:Block_ZXZ_decomp}
    U = \begin{pmatrix}
        Q_0 & \\
          & Q_1
    \end{pmatrix}
    (H\otimes \id)
    \begin{pmatrix}
        \id &  \\
          & R
    \end{pmatrix}
    (H\otimes \id)
    \begin{pmatrix}
        \id &  \\
          & S
    \end{pmatrix},    
\end{align}
where $H$ is the Hadamard gate applied to the first qubit, and $Q_0$, $Q_1$, $R$, $S$ are arbitrary unitaries on the remaining $N-1$ qubits. This form looks similar to the \ac{CSD}, except (i) the inclusion of Hadamard gates, (ii) the block diagonal central matrix that is no longer diagonal, and (iii) two of the matrices having $\id$ on the upper blocks instead of generic unitaries.
To derive Eq.~(\ref{eq:Block_ZXZ_decomp}) from the \ac{CD} perspective, we first decompose $U$ via a type-AIII decomposition, choosing the \ac{CSA}\footnote{This \ac{CSA} also appears in the earlier work~\cite{nakajima2005new}.} $\mfa_\mathrm{BZXZ,0}(N)=\spaniR\{\paulix\otimes\ketbra{j}{j}\}_{j=0}^{2^{N-1}-1}$. This gives 
\begin{align}
    \label{eq:CSD_matrix_product}
    U & = \begin{pmatrix}
        V_0 & \\
          & V_1
    \end{pmatrix}
    \exp(i\paulix\otimes\sum_{j=0}^{2^{N-1}-1}\theta_j\ketbra{j}{j})
    \begin{pmatrix}
        W_0 &  \\
          & W_1
    \end{pmatrix} \\
&    = \begin{pmatrix}
        V_0 & \\
          & V_1
    \end{pmatrix}
    (H\otimes\id)
    \exp(i\pauliz\otimes\sum_{j=0}^{2^{N-1}-1}\theta_j\ketbra{j}{j})
    (H\otimes\id)
    \begin{pmatrix}
        W_0 &  \\
          & W_1
    \end{pmatrix} \\
&    = \begin{pmatrix}
        V_0 & \\
          & V_1
    \end{pmatrix}
    (H\otimes\id)
    \begin{pmatrix}
        D & \\
          & D^\dagger
    \end{pmatrix}
    (H\otimes\id)
    \begin{pmatrix}
        W_0 &  \\
          & W_1
    \end{pmatrix}, 
\end{align}
where $D:=\mathrm{diag}(\{e^{i\theta_j}\}_{j=0}^{2^{N-1}-1})$, and we have used the identity $X=HZH$.
So far we have a standard $U=K_1AK_2$ decomposition with the block-diagonal $K_1$, $K_2$ belonging to the subgroup $\groupK=\exp(\mfk_{\rm AIII}(N))$, and 
\begin{align}
   A=    
    (H\otimes\id)
    \begin{pmatrix}
        D & \\
          & D^\dagger
    \end{pmatrix}
    (H\otimes\id) 
\end{align}
is in the \ac{CSA} $\mfa_\mathrm{BZXZ,0}(N)$. 

Our next step is to insert the identity on each side of the \ac{CSA}, namely,
\begin{align}
U 
= & \begin{pmatrix}
        V_0 & \\
          & V_1
    \end{pmatrix}
    \left[ 
    \begin{pmatrix}
        W_0 &  \\
          & W_0
    \end{pmatrix}
    \begin{pmatrix}
        W_0^\dagger &  \\
          & W_0^\dagger
    \end{pmatrix}
    \right]
    (H\otimes\id)
    \begin{pmatrix}
        D & \\
          & D^\dagger
    \end{pmatrix}
    (H\otimes\id)
    \left[ 
    \begin{pmatrix}
        W_0 &  \\
          & W_0
    \end{pmatrix}
    \begin{pmatrix}
        W_0^\dagger &  \\
          & W_0^\dagger
    \end{pmatrix}
    \right]
    \begin{pmatrix}
        W_0 &  \\
          & W_1
    \end{pmatrix} \\
= & \begin{pmatrix}
        V_0W_0 & \\
          & V_1W_0
    \end{pmatrix}
    (H\otimes\id)
    \begin{pmatrix}
        W_0^\dagger D W_0 & \\
          & W_0^\dagger D^\dagger W_0
    \end{pmatrix}  
    (H\otimes\id)
    \begin{pmatrix}
        \id &  \\
          & W_0^\dagger W_1
    \end{pmatrix} \\
= & \begin{pmatrix}
        \widetilde{V_0} & \\
          & \widetilde{V_1}
    \end{pmatrix}
    (H\otimes\id)
    \begin{pmatrix}
        \widetilde{D} & \\
          & \widetilde{D}^\dagger
    \end{pmatrix}  
    (H\otimes\id)
    \begin{pmatrix}
        \id &  \\
          & S
    \end{pmatrix}, \label{eq:block_ZXZ_new_csa}
\end{align}
where $\widetilde{V_i}=V_iW_0$, $\widetilde{D}=W_0^\dagger D W_0$, $S=W_0^\dagger W_1$, and we have used that 
$\begin{pmatrix}
        W_0 & \\
          & W_0
\end{pmatrix} = \id\otimes W_0$ commutes with $H\otimes\id$. 
To view Eq.~(\ref{eq:block_ZXZ_new_csa}) as a $\kak$ decomposition, the $A$ must now be understood to be generated by the maximally commutative subalgebra $\mfa_\mathrm{BZXZ,1}(N)=(\id\otimes W_0^\dagger)\mfa_\mathrm{BZXZ,0}(N)(\id\otimes W_0)$. Fortunately, Eq.~(\ref{eq:block_ZXZ_new_csa}) is still a proper AIII Cartan $\kak$ decomposition, thanks to Eq.~(\ref{eq:CD_properties_at_mixed_level}). 

Our final step to arrive at Eq.~(\ref{eq:Block_ZXZ_decomp}) is to insert a second identity to the left—but not to the right—of the central matrix in the equation above:
\begin{align}
U 
= & \begin{pmatrix}
        \widetilde{V_0} & \\
          & \widetilde{V_1}
    \end{pmatrix}
    (H\otimes\id)
    \left[
    \begin{pmatrix}
        \widetilde{D} & \\
          & \widetilde{D}
    \end{pmatrix}    
    \begin{pmatrix}
        \widetilde{D}^\dagger & \\
          & \widetilde{D}^\dagger
    \end{pmatrix} 
    \right]
    \begin{pmatrix}
        \widetilde{D} & \\
          & \widetilde{D}^\dagger
    \end{pmatrix}  
    (H\otimes\id)
    \begin{pmatrix}
        \id &  \\
          & S
    \end{pmatrix} \\
= & \begin{pmatrix}
        \widetilde{V_0}\widetilde{D} & \\
          & \widetilde{V_1}\widetilde{D}
    \end{pmatrix} 
    (H\otimes\id)
    \begin{pmatrix}
        \id & \\
          & \widetilde{D}^{\dagger 2}
    \end{pmatrix}  
    (H\otimes\id)
    \begin{pmatrix}
        \id &  \\
          & S
    \end{pmatrix}. \label{eq:block_ZXZ_rewritten}
\end{align}
Setting $Q_i=\widetilde{V}_i\widetilde{D}=V_iDW_0$ and $R=\widetilde{D}^{\dagger 2}=W_0^\dagger D^{\dagger 2}W_0$, we arrive at the Block-ZXZ decomposition, Eq.~(\ref{eq:Block_ZXZ_decomp}). Although Eq.~(\ref{eq:block_ZXZ_rewritten}) has the form of a $\kak$ decomposition, we highlight that we can no longer read off the \ac{CD} from the block matrices directly. After our left-multiplication by $(\id\otimes\widetilde{D})(\id\otimes\widetilde{D}^\dagger)$, we absorbed the left matrix into the subgroup $\groupK$, and the right matrix into the \ac{CSA} $\mfa_\mathrm{BZXZ,1}$. While the subgroup is invariant under this operation, the horizontal subspace is not. Thus, Eq.~(\ref{eq:Block_ZXZ_decomp}) is equivalent to a type-AIII \ac{CD}, but some care should be taken to assign elements to correct horizontal and vertical subspaces.

Whether or not one performs the extra steps above to rewrite the type-AIII decomposition, the subgroup $\groupK$ is unchanged. At this stage, we have two paths for building a recursion. The first straightforwardly involves performing Block-ZXZ decompositions on each diagonal block in Eq.~(\ref{eq:Block_ZXZ_decomp}). On the other hand, we can intersperse our AIII recursions with a subsequent demultiplexing type-A decomposition, as is done in the \ac{QSD}. Employing the demultiplexing step, as noticed by~\cite{krol2024beyond}, allows extra optimizations from the \ac{QSD} to be applied, resulting in an optimized Block-ZXZ decomposition which has lower CNOT counts than the \ac{QSD}.

Finally, we note that~\cite{de2016block, de2018unified} also present a ``dual'' decomposition, the Block-XZX decomposition. Within our framework, this can be seen to be a type-AIII decomposition, where the subalgebra is chosen to be 
\begin{align}
    \mfk_{\rm AIII}=\id\otimes\mfu(2^{N-1})\ \oplus\ \paulix\otimes\mfu(2^{N-1}),
\end{align}
instead of that in Eq.~(\ref{eq:k_aiii_literature}), and the \ac{CSA} is chosen to be $\spaniR\{\pauliz\otimes\ketbra{j}{j}\}_{j=0}^{2^{N-1}-1}$. Such a choice remains compatible with a subsequent type-A decomposition where the antisymmetric part has an $\paulix$ instead of a $\pauliz$ on the first qubit.

\subsection{Khaneja-Glaser Decomposition}\label{sec:kgd_details}
The most important thing to note about the \ac{KGD}~\cite{khaneja2001cartan, bullock2004note, vatan2004realization,mansky2023near} is that, while it is a recursive AIII+A decomposition for $N>2$, it switches to an AI decomposition at $N=2$\footnote{This important detail, enabled by the exceptional isomorphism $\mfsu(2)\oplus\mfsu(2)\cong\mfso(4)$, is highly obfuscated within the existing literature, where it has been explicitly recognized at all. We make this identification based on the fact that the AI \ac{CD} is the only one which could have a \ac{CSA} whose dimension matches the \ac{CSA} $\mfa_\mathrm{base}$ which is given in~\cite{khaneja2001cartan}. Alternatively, one may recognize that the \ac{KGD} is the odd-even decomposition, which is of type AI on two qubits~\cite{dagli2008general}, or look at the chosen involution and characterize it via Tab.~\ref{tab:all_cartan_involutions}.}. 
Bullock~\cite{bullock2004note} identified the type-AIII decomposition correctly but—focusing on the connection to the \ac{CSD}—they missed the fact that there are different types of \acp{CD} at play, as would have been evident from the multi-stage approach in the original work~\cite{khaneja2001cartan}. Drury and Love~\cite{drury2008constructive}, as well as Da{\u{g}}l{\i} et al.~\cite{dagli2008general} acknowledge the second type of decomposition without identifying it.

All told, the \ac{KGD} employs the following recursion:

\begin{align}
    \gru(2^N)&\overset{\rm AIII}{\longrightarrow}\gru(2^{N-1})\times \gru(2^{N-1})
    \overset{\rm A}{\longrightarrow}\gru(2^{N-1})
    \cdots \nonumber 
    \overset{\rm A}{\longrightarrow}\gru(4) 
    \overset{{\rm AI}}{\longrightarrow}\grso(4).
\end{align}
After the AI decomposition, two-qubit gates take the form $G=K_1 A K_2$, with $K_1,K_2\in \grso(4)$ and $A\in\exp(\mfa_\mathrm{base})$, where we have defined the base \ac{CSA} for the \ac{KGD} to be
\begin{align}
    \mfa_\mathrm{base} := \spaniR
    \{
    \id\otimes\id,
    \paulix\otimes\paulix, 
    \pauliy\otimes\pauliy, 
    \pauliz\otimes\pauliz 
    \}.
    \label{eq:KG_CSA_base}
\end{align}
We highlight that this is a 4-dimensional \ac{CSA}, as one would expect from an AI decomposition that assigns the identity to the horizontal space, not the $2$-dimensional \ac{CSA} that one would have for an AIII decomposition\footnote{Though the final AI decomposition is a distinguishing characteristic of the \ac{KGD}, this choice is mathematically arbitrary.}.

The \acp{CSA} $\mfa_{\rm AIII}(N)$ and $\mfa_{\rm A}(N)$ for $N>2$ qubits are defined relative to this base AI \ac{CSA}. The AIII \acp{CSA} are defined as follows:
\begin{align}
    \mfa_{\rm AIII}(3) &= \phantom{\spanR\{}
    \paulix \otimes \mfa_\mathrm{base}\label{eq:KG_CSA_AIII_3},\\
    \mfa_{\rm AIII}(4) &= \spanR\{
    \paulix \otimes \{\id,\paulix\} \otimes \mfa_\mathrm{base}\} \label{eq:KG_CSA_AIII_4},\\
    \mfa_{\rm AIII}(5) &= \spanR\{
    \paulix \otimes \{\id,\paulix\} \otimes \{\id,\paulix\} \otimes \mfa_\mathrm{base}\}, \label{eq:KG_CSA_AIII_5}\\
    &\ \ \vdots \nonumber
\end{align}
It is easy to verify that each of these forms a commutative subalgebra within the corresponding horizontal subspace $\mfp_{\rm AIII}(N)=\paulix\otimes\mfu(2^{N-1})\ \oplus\ \pauliy\otimes\mfu(2^{N-1})$, and each has the correct (maximal) dimension $2^{N-1}$. 

The corresponding type-A involutions at each step have a similar form, with the initial $\paulix$ replaced by a $\pauliz$:
\begin{align}
    \mfa_{\rm A}(3) &= \phantom{\spanR\{}
    \pauliz \otimes \mfa_\mathrm{base} \label{eq:KG_CSA_A_3}, \\
    \mfa_{\rm A}(4) &= \spanR\{
    \pauliz \otimes \{\id,\paulix\} \otimes \mfa_\mathrm{base}\} \label{eq:KG_CSA_A_4},\\
    \mfa_{\rm A}(5) &= \spanR\{
    \pauliz \otimes \{\id, \paulix\} \otimes \{\id,\paulix\} \otimes \mfa_\mathrm{base}\} \label{eq:KG_CSA_A_5},\\
    &\ \ \vdots \nonumber 
\end{align}
These can again be confirmed to be valid \acp{CSA} with dimension $2^{N-1}$ within the horizontal subspace $\mfp_A(N)=\pauliz\otimes\mfu(2^{N-1})$.

One point worth further note is that in the $N=2$ AI decomposition, the subgroup that the operators $K_1$, $K_2$ belong to is $\grso(4)$, which naively would correspond to gates of the form $K_j=\exp(iH_j)$, with real symmetric generators $H_j$ that are non-local in general. However, it is common for the decomposition to include a basis change to the ``magic basis''~\cite{vatan2004optimal} which renders these generators local, namely $\mfk_{{\rm AI}}(2)=\spaniR\{\id\otimes\{\paulix, \pauliy, \pauliz \}, \{\paulix, \pauliy, \pauliz \}\otimes\id\}$. From this, we can see that $\mfk_{{\rm AI}}(2)$ is isomorphic to $\mfsu(2)\oplus\mfsu(2)\cong\mfso(4)$, i.e., the gates $K_j$ are tensor products of single-qubit rotations.

Dedicated readers who have thoroughly read the original \ac{KGD} proposal~\cite{khaneja2001cartan} may notice some small modifications in our treatment here. For one, we have chosen to include the identity component in $\mfa_\mathrm{base}$ (working with the unitary group instead of the special unitary group), as it makes the recursion relations much easier to phrase. In particular, the subalgebra of the AIII decomposition becomes $\mfk_{\rm AIII}(N)=\mfsu(2^{N-1})\oplus\mfsu(2^{N-1})\oplus \mfu(1)$ if global phases are excluded.
This requires an additional, rather manual step that splits off the $\mfu(1)$ component before applying the type-A decomposition, c.f.~\cite[Notation~4]{khaneja2001cartan}, which we do not need to worry about. For $N>2$, our $\mfa_{\rm AIII}(N)$ matches with that by Khaneja and Glaser because we include the identity generator in the vertical spaces.
Finally, we have reversed the order of the qubits from the \ac{KGD}, to better highlight the similarities with other decompositions. 

\section{Proofs for horizontally generated \acp{DLA}}\label{sec:horizontal_proofs}
Here we prove Prop.~\ref{prop:unique_CD_of_H} and Prop.~\ref{prop:CD_of_H_exists_Pauli} from the main text, which we restate for convenience.

\uniquecdofH*
\begin{proof}
    The \ac{DLA} $\mfg=\langle \basisB\rangle_\text{Lie}$ of $\basisB$ is obtained by iteratively computing commutators between its elements.
    Denote the set of $\kappa$th-order commutators for $\kappa\geq 0$ as 
    \begin{align}
        \basisB^{(\kappa)}= \ad^\kappa_{\basisB}(\basisB),\quad \text{i.e., }\quad \basisB^{(0)}=\basisB,\ \ \basisB^{(1)}=[\basisB,\basisB],\ \ \basisB^{(2)}=\big[\basisB,[\basisB,\basisB]\big],~\dots
    \end{align}
    Now assume there is a \ac{CD} $\mfg=\mfk\oplus\mfp$ with $\basisB\subset\mfp$.
    Then $\basisB ^{(2\kappa)}\subset\mfp$ and $\basisB ^{(2\kappa+1)}\subset\mfk$ for all $\kappa\geq 0$, which we may show by induction.
    By assumption we have $\basisB^{(0)}=\basisB\subset\mfp$, which implies $\basisB^{(1)}=[\basisB,\basisB]\subset[\mfp,\mfp]\subset\mfk$, so that the statement holds for $\kappa=0$.
    Now assume that the statement holds for some $\kappa$ and consider the case $\kappa+1$.
    Then we may repeat the above argument to show $\basisB^{(2(\kappa+1))}=[\basisB,\basisB^{(2\kappa+1)}]\subset [\mfp,\mfk]\subset \mfp$ as well as $\basisB^{(2(\kappa+1)+1)}=[\basisB,\basisB^{(2(\kappa+1))}]\subset [\mfp,\mfp]\subset \mfk$, proving the statement for $\kappa+1$.
    This then implies that $\spanR\left\{\bigoplus_{j}\basisB^{(2j)}\right\}\subset\mfp$ and $\spanR\left\{\bigoplus_{j}\basisB^{(2j+1)}\right\}\subset\mfk$, and in particular, the even-order and odd-order commutators span orthogonal vector spaces.
    As $\mfg$ is the span of (finite-order) commutators, the two spaces must cover all of $\mfp$ and $\mfk$ and thus sum up to $\mfg$. This means that, given the information that there is a \ac{CD} with $\basisB\subset \mfp$, we constructed the decomposition without any additional degrees of freedom, showing that it is unique.
\end{proof}

For the next proof, we prepare some observations about Pauli algebras:
\begin{lemma}\label{lemma:paulicom}
    Let $P, Q, R$, and $S$ be arbitrary Pauli words.
    Then
    \begin{align}
        [[iQ, iR], iR] &\propto 
        \begin{cases} 
            i Q & \text{ if }\ \{R,Q\}= 0\\ 0 & \text{ if }\ [R,Q] = 0
        \end{cases}\ \ ,\label{eq:paulicom_qrr}\\
        \textrm{If}~~iP\propto[iQ, iR] &\Rightarrow iQ\propto [iP, iR]\ \ \text{ and }\ iR\propto[iP,iQ],\label{eq:paulicom_is_cyclic}\\
        \textrm{If}~~[iP, iR]\propto[iQ, iR]\neq 0 &\Rightarrow iP\propto iQ,\label{eq:paulicom_equal_pauli_equal}\\
        [[iP, iQ], iR]&\propto\begin{cases}
            0 & \text{ if }\ [P,Q]=0\ \text{ or }\ [PQ, R]=0\\
            -[[iQ, iR], iP] & \text{ if }\ \{P,Q\}=0,\ [R,P]=0\ \text{ and }\ \{Q, R\}=0\\
            -[[iR, iP], iQ] & \text{ if }\ \{P,Q\}=0,\ [Q,R]=0\ \text{ and }\ \{R, P\}=0
        \end{cases}.\label{eq:paulicom_jacobi}
    \end{align}
    In addition, any $\kappa$th-order commutator $iS$ in which the same Pauli word $iR$ appears twice is proportional to a $(\kappa-2)$th-order commutator $i\tilde{S}$ of the remaining $\kappa-2$ Pauli words.
\end{lemma}
\begin{proof}
    We will mostly be computing this proof using the fact that Pauli words square to the identity and pairs of Pauli words either anticommute or commute.
    For Eq.~(\ref{eq:paulicom_qrr}), note that the second case is trivial and for the first we know $\{Q,R\}=\{QR,R\}=0$, so that $[[iQ, iR], iR]\propto i QRR=iQ$.
    For Eq.~(\ref{eq:paulicom_is_cyclic}), we apply Eq.~(\ref{eq:paulicom_qrr}) to the premise by applying the commutator with $iQ$ or with $iR$: 
    \begin{align}
        [iP,iQ]\propto[[iQ,iR],iQ]\overset{(\ref{eq:paulicom_qrr})}\propto iR\qquad 
        [iP,iR]\propto[[iQ,iR],iR]\overset{(\ref{eq:paulicom_qrr})}\propto iQ.
    \end{align}
    Similarly, Eq.~(\ref{eq:paulicom_equal_pauli_equal}) follows from the assumption by applying the commutator with $iR$:
    \begin{align}
        iP\overset{(\ref{eq:paulicom_qrr})}\propto [[iP, iR],iR]\propto[[iQ, iR],iR]\overset{(\ref{eq:paulicom_qrr})}\propto iQ.
    \end{align}
    To prove Eq.~(\ref{eq:paulicom_jacobi}), we will use the Jacobi identity, a defining property of Lie brackets.
    \begin{align}
        [[iP, iQ], iR]&\propto\begin{cases}
            0 & \text{ for }\ [P,Q]=0\\
            [-PQ,iR] & \text{ for }\ \{P,Q\}=0
        \end{cases}\\
        &\propto \begin{cases}
            0 & \text{ for }\ [P,Q]=0\ \text{ or }\ [PQ, R]=0\\
            -[[iQ, iR], iP]-[[iR, iP], iQ] & \text{ for }\ \{P,Q\}=0\ \text{ and }\ \{PQ, R\}=0 
        \end{cases}\\
        &\propto \begin{cases}
            0 & \text{ for }\ [P,Q]=0\ \text{or}\ [PQ, R]=0\\
            -[[iQ, iR], iP] & \text{ for }\ \{P,Q\}=0,~[R, P]=0,~\text{and}~\{Q, R\}=0\\
            -[[iR, iP], iQ] & \text{ for }\ \{P,Q\}=0,~[Q, R]=0,~\text{and}~\{R, P\}=0
        \end{cases}.
    \end{align}
    Here, the first step is trivial, the second step uses the Jacobi identity for the second case, and the third step uses the fact that if $PQ$ and $R$ anticommute, one of $P$ and $Q$ has to commute with $R$ and the other word has to anticommute with $R$.
    
    The final statement about $\kappa$th-order commutators can be proven using Eq.~(\ref{eq:paulicom_jacobi}). 
    Let $iR$ appear in such a commutator $iS$ at the $a$th and $(a+\kappa'+1)$th position. Then we may write $iS=[[[[iP,iR],\dots],iR], \dots]$ for some $(a-1)$-order commutator $iP$, the first dots implying $\kappa'$ commutator nesting levels and the second dots summarizing the remaining $\kappa''=\kappa-a-\kappa'-1$ nesting levels. We can then use Eq.~(\ref{eq:paulicom_jacobi}) $\kappa'$ times to move the outer $iR$ next to the inner $iR$, i.e.,~we bring the full commutator into the form $iS\propto[[[iP,iR],iR], \dots]$ where $iP$ is unchanged and the dots now summarize $\kappa''+\kappa'$ nesting levels of the commutator. Eq.~(\ref{eq:paulicom_qrr}) then implies $iS\propto[iP,\dots]=i\tilde{S}$ with the same dots as in the previous expression, which is the promised $(\kappa-2)$th order commutator ($iP$ has $(a-1)$ nesting levels and the dots summarize $\kappa''+\kappa'=\kappa-a-1$ levels), because we removed two commutator nestings from the $\kappa$th-order commutator $iS$.
\end{proof}

\cdofHexistsPauli*
\begin{proof}
    Let $\basisB=\{iP_j\}_j$ be a set of Pauli words and a minimal generating set for its \ac{DLA} $\mfg=\langle\basisB\rangle_\text{Lie}$.
    We will need to show all (higher-order) commutators of elements from $\basisB$, which will again be of the form $iP$, have a fixed ``commutator parity", i.e.,~they cannot arise as both some odd-order and some even-order commutator. Then it will be easy to show the existence of a \ac{CD}.

    First, we deal with generators from the center of $\basisB$.
    Note that any such center generator $iA$ is a zeroth-order commutator and will not appear as any higher-order commutator, because if we could write $iA\propto[iQ,iR]$, Eq.~(\ref{eq:paulicom_qrr}) in Lemma~\ref{lemma:paulicom} would imply $0=[iA,iR]\propto[[iQ,iR],iR]\propto iQ$ and thus $iA=[0,iR]=0$. Consequentially, we may assign the center of $\basisB$ to $\mfp$ and proceed with the remaining generating set, within which each Pauli word anticommutes with at least one other operator.
    Take the definition of $\basisB^{(\kappa)}$ from the proof of Prop.~\ref{prop:unique_CD_of_H} and denote scalar multiplication without linear combinations as $\basisB^{(\kappa)}_\mbr$, which satisfy the inclusion $\basisB_\mbr^{(\kappa)}\subset \basisB_\mbr^{(\kappa+2)}$.
    To see this, let $iQ\in\basisB^{(\kappa)}$ and $iR\in \basisB$ anticommute with $iQ$. Then Eq.~(\ref{eq:paulicom_qrr}) in Lemma~\ref{lemma:paulicom} implies $iQ\propto[iR, [iR, iQ]]\in\basisB^{(\kappa+2)}_\mbr$.

    Now we are prepared to show that the spaces $\basisB_\text{even}=\bigcup_{\kappa\geq 0} \basisB_{\mbr}^{(2\kappa)}$ and $\basisB_\text{odd}=\bigcup_{\kappa\geq 0} \basisB_{\mbr}^{(2\kappa+1)}$ are disjoint.
    To see this by contradiction, assume the contrary and let $\kappa_1,\kappa_2$ be the smallest integers such that there is a Pauli word $iP\in \basisB_{\mbr}^{(2\kappa_1)}\cap \basisB_{\mbr}^{(2\kappa_2+1)}$. Note that the inclusion from above implies that 
    \begin{align}
        \kappa_1\leq \kappa_2 &\Rightarrow iP\in \basisB_{\mbr}^{(2\kappa_1)}\subset\basisB_{\mbr}^{(2\kappa_1+2(\kappa_2-\kappa_1))}=\basisB_{\mbr}^{(2\kappa_2)},\\
        \kappa_1>\kappa_2 &\Rightarrow iP\in \basisB_{\mbr}^{(2\kappa_2+1)}\subset\basisB_{\mbr}^{(2\kappa_2+1+2(\kappa_1-\kappa_2-1))}=\basisB_{\mbr}^{(2\kappa_1-1)}.
    \end{align}
    Taken together, we find that $iP\in\basisB_{\mbr}^{(\kappa)}\cap \basisB_{\mbr}^{(\kappa+1)}$, where $\kappa=\max(2\kappa_1-1, 2\kappa_2)$ is the smallest number such that this statement holds.
    
    Now, $iP\in\basisB_{\mbr}^{(\kappa+1)}$ implies that there are $iQ\in\basisB^{(\kappa)}_\mbr, iR\in\basisB$ such that $iP \propto [iQ, iR]$.
    Meanwhile, $iP$ is also a nested commutator of $\kappa$ generators from $\basisB$.
    $iR$ cannot be among those generators, because then $iR$ would appear twice in the $(\kappa+1)$th-order commutator $[iP,iR]$, so that it is proportional to a $(\kappa-1)$th-order commutator $iP'$ (Lemma~\ref{lemma:paulicom}). Together with $iP\propto [iQ,iR]$, this would imply
    \begin{align}
        \basisB^{(\kappa)}_\mbr \ni iQ\overset{(\ref{eq:paulicom_qrr})}{\propto}[[iQ, iR], iR]=[iP,iR]\overset{\text{Lemma }\ref{lemma:paulicom}}{\propto} iP'\in\basisB^{(\kappa-1)}_\mbr,
    \end{align}
    and thus $iQ\in\basisB_{\mbr}^{(\kappa-1)}\cap \basisB_{\mbr}^{(\kappa)}$, contradicting the minimality of $\kappa$.
    Similarly, $iQ$ is a commutator of $\kappa$ generators which cannot contain $iR$.
    To see this by contradiction, assume that there exists $iQ'\in \basisB^{(\kappa-1)}$ such that $iQ\propto[iQ',iR]$.
    Then, by Eqs.~(\ref{eq:paulicom_is_cyclic}) and~(\ref{eq:paulicom_equal_pauli_equal}), we have $iQ'\propto iP$ so that $iP$ also is a $(\kappa-1)$th-order commutator. This implies $iP\in\basisB_{\mbr}^{(\kappa-1)}\cap \basisB_{\mbr}^{(\kappa)}$, again contradicting the minimality of $\kappa$.
    
    Eq.~(\ref{eq:paulicom_is_cyclic}) tells us that $iR\in\basisB$ can be obtained via $[iP,iQ]$, a nested commutator of elements from $\basisB\setminus\{iR\}$. This implies that $\basisB$ was not a minimal generating set, completing the contradiction.
    With this, we have shown that $\basisB_\text{odd}$ and $\basisB_\text{even}$ are disjoint, and thus contain bases for two orthogonal subspaces $\mfk$ and $\mfp$, respectively, due to orthonormality of the Pauli basis.
    
    To check the commutation relations, note that the commutator order is ``almost additive" under commutation\footnote{This can be shown by reordering commutators with the Jacobi identity and exploiting the simple structure that results for Pauli words.}:
    \begin{alignat}{4}
        & [\mfk,\mfk]= \sum_{\kappa,\kappa'\geq 0} [B_{\mbr}^{(2\kappa+1)},B_{\mbr}^{(2\kappa'+1)}]&\subset& \sum_{\kappa\geq 1} B_{\mbr}^{(2\kappa+1)}&=\mfk,\\
        \left[B_{\mbr}^{(\kappa)},B_{\mbr}^{(\kappa')}\right]\subset B_{\mbr}^{(\kappa+\kappa'+1)}\qquad\Rightarrow \qquad 
        & [\mfk,\mfp]= \sum_{\kappa,\kappa'\geq 0} [B_{\mbr}^{(2\kappa+1)},B_{\mbr}^{(2\kappa')}]&\subset& \sum_{\kappa\geq 1} B_{\mbr}^{(2\kappa)}&=\mfp,\\
        & [\mfp,\mfp]= \sum_{\kappa,\kappa'\geq 0} [B_{\mbr}^{(2\kappa)},B_{\mbr}^{(2\kappa')}]&\subset& \sum_{\kappa\geq 0} B_{\mbr}^{(2\kappa+1)}&=\mfk.
    \end{alignat}
    We thus have a \ac{CD} of $\mfg$ with $B\subset\mfp$, which furthermore is unique due to Prop.~\ref{prop:unique_CD_of_H}. 
    Clearly, adding even-order commutators of $B$ to the terms in the Hamiltonian will be compatible with this unique decomposition.
\end{proof}

\section{Cartan involutions and their interaction: calculations and proofs}\label{sec:involutions_calculations}
In this appendix we calculate the concrete form of all Cartan involutions acting on the defining irreducible representation of the classical Lie algebras $\mfsu$, $\mfso$, and $\mfsp$, as well as their doubles $\mfsu^{\oplus 2}$, $\mfso^{\oplus 2}$, and $\mfsp^{\oplus 2}$. Throughout App.~\ref{sec:involutions_calculations}, we will work with the special unitary groups/algebras, in contrast to the main text that uses their unitary counterparts. This is to highlight subtleties in the phases of composed and recursive \acp{CD}, which otherwise would be hidden in the global phases of $\mfu(n)$.

For convenience we recall the defining representations from Eqs.~(\ref{eq:def_sun})-(\ref{eq:def_spn}):
\begin{align}
    \mfsu(n)&=\left\{x\in \mbc^{n\times n} | x = -x^\dagger, \tr(x)=0\right\},\\
    \mfso(n)&=\left\{x\in \mbr^{n\times n} | x=-x^T\right\},\\
    \mfsp(n)&=\left\{x\in \mfsu(2n) | x=-J_nx^TJ_n^T\right\}.
\end{align}
Note that $\mfso(n)$ and $\mfsp(n)$ can be characterized as subspaces in $\mfsu(n)$ and $\mfsu(2n)$ on which $\ast=\rm{id}$ or $\ast=\Ad_{J_n}$, respectively.

We will then use the characterization of the involutions to compute two types of product tables for Cartan involutions; the first describes the involution resulting from composing two (commuting) involutions and the second captures the behaviour under recursion, in the sense of recursive \acp{CD} from Def.~\ref{def:recursive_decomps}.

\subsection{Characterization}\label{sec:involutions_calculations:characterize}
We will characterize all possible Cartan involutions on the classical Lie algebras and their doubles by combining the canonical form of the involution, which is commonly used to define them, with the degrees of freedom that capture changes of the basis in the representation. The result is summarized in Tab.~\ref{tab:all_cartan_involutions} in the main text. For AI and AII involutions, we find a generalization of Cor.~A3 in~\cite{wiersema2023classification}.

\textbf{A.}
The canonical form for an type-A involution on $\mfsu(n)\oplus\mfsu(n)$ is $\theta_0(x)= x^{\smallswap}$, where we use $\swapsymbol$ in analogy to $\ast$ to denote the swapping operation of the two algebra components.
In an arbitrary basis, captured by a unitary transformation $V_0\oplus V_1\in \grsu(n) \times \grsu(n)$, this leads to
\begin{align}
    \theta
    =\Ad_{V_0^\dagger\oplus V_1^\dagger}\circ\swapsymbol\circ \Ad_{V_0\oplus V_1}
    =\Ad_{(V_0^\dagger V_1) \oplus (V_1^\dagger V_0)}\circ\swapsymbol
    =\Ad_{X}\circ\swapsymbol,
\end{align}
with $X=(V_0^\dagger V_1)\oplus (V_1^\dagger V_0)\in \grsu(n)\times \grsu(n)$ and $X^\dagger=X^{\smallswap}$.

\textbf{AI.}
The canonical form for an AI involution on $\mfsu(n)$ is $\theta_0(x)=x^\ast$. In an arbitrary basis, captured by a unitary transformation $V_0\in \grsu(n)$, this leads to
\begin{align}
    \theta
    =\Ad_{V_0^\dagger}\circ\ast\circ\Ad_{V_0}
    =\Ad_{V_0^\dagger V_0^\ast}\circ\,\ast
    =\Ad_{V}\circ\,\ast,
\end{align}
with $V=V_0^\dagger V_0^\ast \in \grsu(n)$ and $V^T=(V_0^\dagger V_0^\ast)^T=V$.

\textbf{AII.}
The canonical form for an AII involution on $\mfsu(2n)$ is $\theta_0(x)=J_n x^\ast J_n^T$ with 
\begin{align}\label{eq:J_n_in_invol_app}
    J_n = \begin{pmatrix}
        0 & \mathbb{I}_n \\
        -\mathbb{I}_n & 0\\
    \end{pmatrix}\in \grsp(n)\cap \grso(2n).
\end{align}
In an arbitrary basis, captured by a unitary transformation $V_0\in \grsu(n)$, this leads to
\begin{align}
    \theta
    =\Ad_{V_0^\dagger}\circ\Ad_{J_n}\circ\ast\circ\Ad_{V_0}
    =\Ad_{V_0^\dagger J_n V_0^\ast}\circ\,\ast
    =\Ad_{W}\circ\,\ast,
\end{align}
with $W=V_0^\dagger J_n V_0^\ast \in \grsu(2n)$ and $W^T=(V_0^\dagger J_n V_0^\ast)^T=V_0^\dagger (-J_n) V_0^\ast=-W$.

\textbf{AIII.}
The canonical form for an AIII involution on $\mfsu(p+q)$ is $\theta_0(x)=I_{p,q} x I_{p,q}$
with 
\begin{align}\label{eq:I_pq_in_invol_app}
    I_{p,q} = \begin{pmatrix}
        \mathbb{I}_p & 0\\
        0 & -\mathbb{I}_q\\
    \end{pmatrix}\in \gro(p+q)\cap \mathrm{Sym}(p+q),
\end{align}
where $\mathrm{Sym}$ denotes symmetric matrices.
In an arbitrary basis, captured by a unitary transformation $V_0\in \grsu(p+q)$, this leads to
\begin{align}
    \theta
    =\Ad_{V_0^\dagger}\circ\Ad_{I_{p,q}}\circ\Ad_{V_0}
    =\Ad_{V_0^\dagger I_{p,q}V_0}
    =\Ad_{H},
\end{align}
with $H=V_0^\dagger I_{p,q}V_0 \in \gru(p+q)$ and $H^\dagger=(V_0^\dagger I_{p,q} V_0)^\dagger=H$ as well as $\det{H}=(-1)^q$. By using $-I_{p,q}$ in place of $I_{p,q}$, the determinant can be changed to $\det{H}=(-1)^p$, so that we may restrict to $H\in \grsu(p+q)$ as long as $p$ or $q$ is even.

\textbf{BD.}
The canonical form is as for the type-A involution, and accordingly we get for an arbitrary basis, captured by $Q_0\oplus Q_1\in \grso(n) \times \grso(n)$,
\begin{align}
    \theta=\Ad_{X}\circ\swapsymbol,
\end{align}
with $X=(Q_0^\dagger Q_1)\oplus (Q_1^\dagger Q_0)\in \grso(n)\times \grso(n)$ and $X^T=X^{\smallswap}$.

\textbf{BDI.}
The canonical form for a BDI involution on $\mfso(p+q)$ is $\theta_0(x)=I_{p,q} x I_{p,q}$
with $I_{p,q}$ from Eq.~(\ref{eq:I_pq_in_invol_app}).
In an arbitrary basis, captured by an orthogonal transformation $Q_0\in \grso(p+q)$, this leads to
\begin{align}
    \theta
    =\Ad_{Q_0^T}\circ\Ad_{I_{p,q}}\circ\Ad_{Q_0}
    =\Ad_{Q_0^T I_{p,q}Q_0}
    =\Ad_{R},
\end{align}
with $R=Q_0^T I_{p,q}Q_0 \in \gro(p+q)$ and $R^T=R$ as well as $\det{R}=(-1)^q$ (or $\det{R}=(-1)^p$), as for AIII.

\textbf{DIII.}
The canonical form for a DIII involution on $\mfso(2n)$ is $\theta_0(x)=J_n x J_n^T$
with $J_n$ from Eq.~(\ref{eq:J_n_in_invol_app}).
In an arbitrary basis, captured by an orthogonal transformation $Q_0\in \grso(2n)$, this leads to
\begin{align}
    \theta
    =\Ad_{Q_0^T}\circ\Ad_{J_n}\circ\Ad_{Q_0}
    =\Ad_{Q_0^T J_nQ_0}
    =\Ad_{L},
\end{align}
with $L=Q_0^T J_n Q_0 \in \grso(2n)$ (using $\det{J_n}=1$) and $L^T=(Q_0^T J_n Q_0)^T=-L$.

\textbf{C.}
The canonical form is as for the type-A involution, and accordingly we get for an arbitrary basis, captured by $S_0\oplus S_1\in \grsp(n) \times \grsp(n)$,
\begin{align}
    \theta=\Ad_{X}\circ\swapsymbol,
\end{align}
with $X=(S_0^\dagger S_1)\oplus (S_1^\dagger S_0)\in \grsp(n)\times \grsp(n)$ and $X^\dagger=X^{\smallswap}$.

\textbf{CI.}
The canonical form for a CI involution on $\mfsp(n)$ is $\theta_0(x)=J_nxJ_n^T$. In an arbitrary basis, captured by a unitary symplectic transformation $S_0\in \grsp(n)$, this leads to
\begin{align}
    \theta
    =\Ad_{S_0^\dagger}\circ\Ad_{J_n}\circ\Ad_{S_0}
    =\Ad_{S_0^\dagger J_n S_0}
    =\Ad_{S},
\end{align}
with $S=S_0^\dagger J_n S_0 \in \grsp(n)$ and $S^\dagger=(S_0^\dagger J_n S_0)^\dagger=-S$, where we used $J_n^\dagger=-J_n$.

A useful alternative description can be obtained by noticing that on $\grsp(n)$ (in its canonical basis with respect to $J_n$) we have $\Ad_{J_n}=\ast$, so that all CI involutions also take the form

\begin{align}
    \theta
    =\Ad_{S_0^\dagger}\circ\,\ast\circ\Ad_{S_0}
    =\Ad_{S_0^\dagger S_0^\ast}\circ\,\ast
    =\Ad_{\tilde{S}}\circ\,\ast,\qquad
    \tilde{S}^T = (S_0^\dagger S_0^\ast)^T=S_0^\dagger S_0^\ast=\tilde{S}.
\end{align}

\textbf{CII.}
The canonical form for a CII involution on $\mfsp(p+q)$ is $\theta_0(x)=K_{p,q}x K_{p,q}$ with $K_{p,q}$ from Eq.~(\ref{eq:def_I_pq_K_pq}).
In an arbitrary basis, captured by a unitary symplectic transformation $S_0\in \grsp(p+q)$, this leads to
\begin{align}
    \theta
    =\Ad_{S_0^\dagger}\circ\Ad_{K_{p,q}}\circ\Ad_{S_0}
    =\Ad_{S_0^\dagger K_{p,q} S_0}
    =\Ad_{P},
\end{align}
with $P=S_0^\dagger K_{p,q} S_0 \in \grsp(p+q)$ and $P^\dagger=(S_0^\dagger K_{p,q} S_0)^\dagger=P$.
Though we will not make use of it here, we mention an alternative way to phrase type-CII involutions, just like for type CI.
Using $\Ad_{J_n}=\Ad_{J_n^\dagger}=\ast$ on $\grsp(p+q)$, we write
\begin{align}
    \theta
    =\Ad_{S_0^\dagger}\circ\Ad_{K_{p,q}J_n}\circ\,\ast\circ\Ad_{S_0}
    =\Ad_{S_0^\dagger K_{p,q} J_n S_0^\ast}\circ\,\ast
    =\Ad_{\tilde{P}}\circ\,\ast,
\end{align}
with $\tilde{P}\in\grsp(p+q)$ and $\tilde{P}^T=(S_0^\dagger K_{p,q} J_n S_0^\ast)^T=-S_0^\dagger J_n K_{p,q} S_0^\ast=-\tilde{P}$, where we used $J_n^T=-J_n$, $K_{p,q}^T=K_{p,q}$ and $[K_{p,q},J_n]=0$.

\subsection{Composing Cartan involutions}\label{sec:involutions_calculations:compose}

In this section we report a composition table for Cartan involutions using the characterization in Tab.~\ref{tab:all_cartan_involutions} in the main text.

\begin{restatable}{lemma}{compositionofinvolutions}\label{lemma:composition_of_involutions}
    The composition $\theta_1\circ\theta_2$ of two Cartan involutions $\theta_{1,2}$ is a Cartan involution iff they commute. On a classical simple Lie algebra, the involution type resulting from the types of $\theta_{1,2}$ is given by
    \begin{center}
    {\footnotesize\emph{
    \begin{tabular}{c|ccc}
        $\mfsu$     & AI & AII & AIII \\
        \hline
        AI   & AIII & AIII & AI/AII \\
        AII  & AIII & AIII & AII/AI \\
        AIII & AI/AII & AII/AI & AIII
    \end{tabular}\hspace{6ex}
    \begin{tabular}{c|cc}
        $\mfso$     & BDI & DIII \\
        \hline
        BDI   & BDI/DIII & DIII/BDI \\
        DIII  & DIII/BDI & BDI/DIII
    \end{tabular}\hspace{6ex}
    \begin{tabular}{c|cc}
        $\mfsp$     & CI & CII \\
        \hline
        CI   & CII/CI & CI/CII \\
        CII  & CI/CII & CII/CI
    \end{tabular}
    }}
    \end{center}
    Which involutions commute is indicated by the commutation \emph{Criterion} $G_1G_2=\pm G_2G_1$ in Tab.~\ref{tab:involution_composition}, and the ambiguous entries above are determined by the sign in the criterion.
\end{restatable}

We provide the proof of this lemma below, using explicit calculations for each combination of involutions in the defining representation and leveraging our characterization in Tab.~\ref{tab:all_cartan_involutions}.
Note that even though we use a fixed representation to compute the table, it is valid for all representations because the representing maps are Lie algebra homomorphisms and thus preserve all necessary structure. The only adjustment required for differing representations is that the commutation criterion $G_1G_2=\pm G_2G_1$ needs to be inspected \emph{after} mapping to the defining representation, in order to determine the ambiguous entries in the Lemma.

Also notice that the composition of commuting involutions reproduces a special case of Lemma~\ref{prop:grading_to_set_of_cds}. The two involutions define two compatible \acp{CD} $\mfg=\mfk_i\oplus\mfp_i$, $i\in\{1,2\}$, which lead to a Cartan $2$-grading
\begin{align}
    \mfg=
    \underset{\mfg_{00}}{\underbrace{\mfk_1\cap \mfk_2}}
    \ \oplus\ \underset{\mfg_{01}}{\underbrace{\mfk_1\cap \mfp_2}}
    \ \oplus\ \underset{\mfg_{10}}{\underbrace{\mfp_1\cap \mfk_2}}
    \ \oplus\ \underset{\mfg_{11}}{\underbrace{\mfp_1\cap \mfp_2}}.
\end{align}
We may then take the full group $\Ztwo^2$ with the homomorphism $\varphi:\Ztwo^2\to\Ztwo, s_0s_1\mapsto s_0+s_1$, so that $\mfg=(\mfg_{00}\oplus \mfg_{11})\oplus (\mfg_{01}\oplus \mfg_{10})$ is a new \ac{CD}.
In words, we assign a positive sign to spaces that have the same sign under $\theta_1$ and $\theta_2$, and a negative sign to those that had differing sign.
This is exactly the behaviour under the composed involution $\theta_1\circ\theta_2$!

\textbf{Example.}
To illustrate the composition table above, we will generate Cartan involutions of a new type from a set of Cartan involutions. 
Consider three involutions $\theta_{1,2,3}$ of type AI on $\mfsu(4)$, with commutation relations of their conjugation operators given by $[V_1,V_2]=[V_2,V_3]=\{V_1,V_3\}=0$.
The composition table above tells us that $\theta_4\coloneqq\theta_1\circ\theta_2$ is of type AIII and $\theta_5\coloneqq\theta_3\circ\theta_4$ of type AII, where we used that $\{V_3,V_1V_2^\ast\}=0$ so that the second entry for AI$\circ$AIII has to be read out for $\theta_5$.

Concretely, consider $V_1=X_0X_1$, $V_2=X_0$, and $V_3=Z_1$, which lie in $\grsu(4)$ and satisfy the above commutation relations. The involutions then are
\begin{align*}
    \theta_1 &= \Ad_{X_0X_1}\circ\,\ast,\quad
    \theta_2 = \Ad_{X_0}\circ\,\ast,\quad
    \theta_3 = \Ad_{Z_1}\circ\,\ast,\\
    \theta_4 &= \Ad_{X_0X_1X^\ast_0}=\Ad_{X_1},\\
    \theta_5 &= \Ad_{Z_1X^\ast_1}\circ\,\ast=\Ad_{Y_1}\circ\,\ast,
\end{align*}
where we adjusted the phase in $\theta_{5}$ suitably.
Noting that $\det{X_1}$ and $\det{Y_1}$ are determinants of $4\times 4$ matrices, we find the results to match the properties of AIII and AII involutions, as anticipated:
\begin{align}
    X_1^\dagger=X_1,\ \det{X_1}=1,\qquad
    Y_1^T=-Y_1,\ \det{Y_1}=1.
\end{align}

Before proving Lemma~\ref{lemma:composition_of_involutions}, we prepare a small technical statement.

\begin{lemma}\label{lemma:ac_sym_ortho_mats}
    Let $R_1,R_2\in\grso(d)$ with $R_i^T=R_i$ and $R_1R_2=-R_2R_1$. Then $d=2n$ for some $n\in\mbn$, and $\det{R_1}=\det{R_2}=(-1)^{n}$.
\end{lemma}

\begin{proof}[Proof of Lemma~\ref{lemma:ac_sym_ortho_mats}]
    Let $R_1,R_2\in\grso(d)$ with $R_i^T=R_i$ and $R_1R_2=-R_2R_1$.
    We must have $\det{R_1R_2}=\det{-R_2R_1}$, which implies $1=(-1)^d$ and thus $d$ is even, i.e.,~$d=2n$ for some $n\in\mbn$.
    
    Let $v$ be an eigenvector of $R_1$ with eigenvalue $\lambda$, then $v=R_1^2v=\lambda^2 v$ implies that the spectrum of $R_1$ only contains $+1$ and $-1$.
    Next, note that anticommutativity implies $R_1R_2v=-R_2R_1 v=-R_2\lambda v=-\lambda R_2v$, so that $R_2v$ is an eigenvector of $R_1$ with the eigenvalue $-\lambda$. As $R_2$ is orthogonal and in particular invertible, it provides a bijection between eigenvectors of $R_1$ for the eigenvalue $1$ and those for $-1$, so that the respective multiplicity of both eigenvalues is $\frac{d}{2}=n$~\cite{SE_anticom}.
    The determinant of $R_1$ therefore is $\det{R_1}=(-1)^{n}$.
    All of the above statements hold when switching $R_1$ and $R_2$, so that $\det{R_2}=\det{R_1}$.
\end{proof}

We now turn to the proof of the composition lemma.

\begin{table*}
    \centering
    \begin{tabular}{ccccclcA}
        $T_1\circ T_2$ & $\theta_1$ & $\theta_2$ & Criterion & $\phi$ & $\ \ \quad\qquad\theta_3$ & $T_3$ & &\quad\text{Calculations} \\\midrule
        
        \belowrulesepcolor{tabgrey}\rc{tabgrey} &&&&&&& G^\dagger &= e^{i\frac{\phi}{2}}V_2^TV_1^\dagger=e^{i\frac{\phi}{2}}V_2V_1^\ast\overset{\star}{=}G \\
        \rc{tabgrey} \mr{AI$\circ$AI} & \mr{$\Ad_{V_1}\circ\,\ast$} & \mr{$\Ad_{V_2}\circ\,\ast$} & \mr{$V_1V_2^\ast=e^{i\phi}V_2V_1^\ast$} &
        \mr{$\frac{2\pi k}{n}$} & \mr{$\Ad_{G=e^{-i\frac{\phi}{2}}V_1V_2^\ast}$} & \mr{AIII} & \det{G}&=(-1)^k \quad(\ref{eq:phase_determinant})\\
        
        &&&&&&& G^\dagger&=e^{i\frac{\phi+\pi}{2}}W^TV^\dagger=e^{i\frac{\phi-\pi}{2}}WV^\ast\overset{\star}{=}G\\
        \mr{AI$\circ$AII} & \mr{$\Ad_{V}\circ\,\ast$} & \mr{$\Ad_{W}\circ\,\ast$} & \mr{$VW^\ast=e^{i\phi}WV^\ast$} & \mr{$\frac{2\pi k}{2n}$} & \mr{$\Ad_{G=e^{-i\frac{\phi+\pi}{2}}VW^\ast}$} & \mr{AIII} & \det{G}&=(-1)^{k+n} \quad(\ref{eq:phase_determinant}) \\
        
        \rc{tabgrey} &&& $VH^\ast=HV$ & $0$ && AI & G^T&=e^{-i\frac{\vartheta}{n}}H^\dagger V^T=e^{-i\frac{\vartheta}{n}}HV\overset{\star}{=}e^{i\phi} G \\
        \rc{tabgrey} \mr{AI$\circ$AIII} & \mr{$\Ad_{V}\circ\,\ast$} & \mr{$\Ad_{H}$} & $VH^\ast=-HV$ & $\pi$ & \mr{$\Ad_{G=e^{-i\frac{\vartheta}{n}}VH^\ast}\circ\,\ast$} & AII & \det{G}&=e^{-i\frac{\vartheta}{n}n}\det{H^\ast}=e^{-i(\vartheta+\vartheta)}=1 \\
        
        &&&&&&& G^\dagger&=e^{i\frac{\phi}{2}}W_2^TW_1^\dagger=e^{i\frac{\phi}{2}}W_2W_1^\ast\overset{\star}{=}G \\
        \mr{AII$\circ$AII} & \mr{$\Ad_{W_1}\circ\,\ast$} & \mr{$\Ad_{W_2}\circ\,\ast$} & \mr{$W_1W_2^\ast=e^{i\phi}W_2W_1^\ast$} & \mr{$\frac{2\pi k}{2n}$} & \mr{$\Ad_{G=e^{-i\frac{\phi}{2}}W_1W_2^\ast}$} & \mr{AIII} & \det{G}&=(-1)^k \quad(\ref{eq:phase_determinant}) \\
        
        \rc{tabgrey} &&& $WH^\ast=HW$ & $0$ && AII & G^T&=-e^{-i\frac{\vartheta}{2n}}HW\overset{\star}{=}-e^{i\phi} G \\
        \rc{tabgrey}\mr{AII$\circ$AIII} & \mr{$\Ad_{W}\circ\,\ast$} & \mr{$\Ad_{H}$} & $WH^\ast=-HW$ & $\pi$ & \mr{$\Ad_{G=e^{-i\frac{\vartheta}{2n}}WH^\ast}\circ\,\ast$} & AI & \det{G}&=e^{-i\frac{\vartheta}{2n}2n}\det{H^\ast}=e^{-i(\vartheta+\vartheta)}=1 \\ 

        &&&&&&& G^\dagger &= e^{i\frac{\phi}{2}}H_2^\dagger H_1^\dagger=e^{i\frac{\phi}{2}}H_2H_1\overset{\star}{=}G \\
        \mr{AIII$\circ$AIII} & \mr{$\Ad_{H_1}$} & \mr{$\Ad_{H_2}$} & \mr{$H_1H_2=e^{i\phi}H_2H_1$} &
        \mr{$\frac{2\pi k}{n}$} & \mr{$\Ad_{G=e^{-i\frac{\phi}{2}}H_1H_2}$} & \mr{AIII} & \det{G}&=(-1)^k \quad(\ref{eq:phase_determinant}) \\
        \midrule
        
        \belowrulesepcolor{tabgrey}\rc{tabgrey} & & & $R_1R_2=R_2R_1$ & $0$ & & BDI & G^T&=R_2^TR_1^T=R_2R_1\overset{\star}{=} e^{i\phi} G\\
        \rc{tabgrey} \mr{BDI$\circ$BDI} &\mr{$\Ad_{R_1}$} & \mr{$\Ad_{R_2}$} & $R_1R_2=-R_2R_1$ & $\pi$ & \mr{$\Ad_{G=R_1R_2}$} & DIII & \det{G}&=\det{R_1}\det{R_2}\quad (\text{Lemma}~\ref{lemma:ac_sym_ortho_mats})\\
        
        & & & $RL=LR$ & $0$ & & DIII & G^T&=L^TR^T=-LR\overset{\star}{=} -e^{i\phi} G\\
        \mr{BDI$\circ$DIII} & \mr{$\Ad_{R}$} & \mr{$\Ad_{L}$} & $RL=-LR$ & $\pi$ & \mr{$\Ad_{G=RL}$} & BDI & \det{G}&=\det{R}\quad (\ref{eq:bdi_diii_determinant}) \\
        
        \rc{tabgrey} & & & $L_1L_2=L_2L_1$ & $0$ & & BDI & G^T&=L_2^TL_1^T=L_2L_1\overset{\star}{=} e^{i\phi} G\\
        \rc{tabgrey} \mr{DIII$\circ$DIII} & \mr{$\Ad_{L_1}$} & \mr{$\Ad_{L_2}$} & $L_1L_2=-L_2L_1$ & $\pi$ & \mr{$\Ad_{G=L_1L_2}$} & DIII & \det{G}&=1 \\\aboverulesepcolor{tabgrey}\midrule
        
        & & & $S_1S_2=S_2S_1$ & $0$ & & CII & G^\dagger&=S_2^\dagger S_1^\dagger=S_2S_1\overset{\star}{=}e^{i\phi} G\\
        \mr{CI$\circ$CI} & \mr{$\Ad_{S_1}$} & \mr{$\Ad_{S_2}$} & $S_1S_2=-S_2S_1$ & $\pi$ & \mr{$\Ad_{G=S_1S_2}$} & CI & \det{G}&=1\\
        
        \rc{tabgrey}& & & $SP=PS$ & $0$ & & CI & G^\dagger&=P^\dagger S^\dagger=-PS\overset{\star}{=}-e^{i\phi} G\\
        \rc{tabgrey}\mr{CI$\circ$CII} & \mr{$\Ad_{S}$} & \mr{$\Ad_{P}$} & $SP=-PS$ & $\pi$ & \mr{$\Ad_{G=SP}$} & CII & \det{G}&=1\\
        
        & & & $P_1P_2=P_2P_1$ & $0$ & & CII & G^\dagger&=P_2^\dagger P_1^\dagger=P_2P_1\overset{\star}{=}e^{i\phi} G\\
        \mr{CII$\circ$CII} & \mr{$\Ad_{P_1}$} & \mr{$\Ad_{P_2}$} & $P_1P_2=-P_2P_1$ & $\pi$ & \mr{$\Ad_{G=P_1P_2}$} & CI & \det{G}&=1\\
    \end{tabular}
    \caption{The composition table of Cartan involutions, based on the characterization in Tab.~\ref{tab:all_cartan_involutions}. Given two involution types with concrete involutions $\theta_{1,2}$ (left), their composition yields a new Cartan involution if a commutativity criterion holds, which involves a phase $e^{i\phi}$ (mid left). The new involution $\theta_3$ and its type $T_3=T_1\circ T_2$ are fixed up to a phase of the conjugation operator, which in turn is dictated by $\phi$ (mid right). We verify $T_3$ with brief calculations of its conjugation operator (transformation behaviour and determinant, right). For AI$\circ$AIII and AII$\circ$AIII we defined $e^{i\vartheta}=\det{H}$. If a pair of types can give rise to different new types, the calculations apply to both cases.
    Each equality marked with a $\star$ exploits the commutativity criterion, and the equalities preceding each of these use properties of the conjugation operators of $\theta_{1,2}$. For details on the phase and determinant calculations, see the main text.
    }
    \label{tab:involution_composition}
\end{table*}

\begin{proof}[Proof of Lemma~\ref{lemma:composition_of_involutions}]
As $\theta_{1,2}$ are elements of the automorphism group of $\mfg$, we know that $\theta_{3}\coloneqq\theta_1\circ\theta_2$ is an automorphism as well.
Thus, $\theta_{3}$ is a Cartan involution if and only if $\theta_{3}^2=\id_{\mfg}$, which implies
\begin{align}
    \id_{\mfg}
    =\theta_{3}^2=\theta_1\circ\theta_2\circ\theta_1\circ\theta_2
    \quad \Leftrightarrow\quad \theta_1\circ\theta_2=\theta_2\circ\theta_1,
\end{align}
where we used $\theta_{1,2}^2=\id_\mfg$.

To compute the composition tables, we start with a useful observation regarding the adjoint action of $\gru(n)$.
Given two group elements $G, H$ from $SU$, $SO$ or $Sp$, $\Ad_G=\Ad_H$ implies $G=e^{i\phi} H$.
This is because, by Schur's lemma, the commutants $\gru^{\mfg}$ for $\mfg\in\{\mfsu,\mfso,\mfsp\}$ all are $\gru(1)$.
For $G,H\in \grsu(n)$, the phase above is constrained to $\phi=\frac{2\pi k}{n},\ k\in\mathbb{Z}$ due to $\det{e^{i\phi}\id_n}=\det{G}/\det{H}=1$. In this scenario, we will make use of
\begin{align}\label{eq:phase_determinant}
    \det{e^{i\frac{\phi}{2}}\id_n}=e^{i\frac{2\pi k}{2n}n}=(-1)^k.
\end{align}
For $G,H\in \grso(n)$ and $G,H\in \grsp(n)$ the phase is instead constrained by $\phi\in\{0, \pi\}$ because $e^{i\phi}\in\mathbb{R}$ and $({e^{i\phi}}\id_{2n})^TJ_n(e^{i\phi}\id_{2n})=J_n$ if and only if $e^{2i\phi}=1$, respectively.

For the individual scenarios, we will go through the same computational steps for each pair of involution types that act on the same classical Lie algebra. 
Given such a pair, the first step is to make the involutions concrete, using Tab.~\ref{tab:all_cartan_involutions} from the main text. Second, we derive a commutativity criterion for the conjugation operators of the involutions, together with a constraint from above on the global phase in the criterion.
Third, we give a specific functional form for $\theta_3=\theta_1\circ\theta_2$ and its type, and lastly prove that the conjugation operator of $\theta_3$ has the correct properties if the commutativity criterion is satisfied. 
We denote the type of $\theta_3$ as $T_3=T_1\circ T_2$, remarking that $T_3$ is not fixed by $T_1$ and $T_2$ alone, but additionally depends on the phase in the commutativity criterion. Further, note that we only consider the scenario of commuting involutions from here on, so that the ordering of $T_1$ and $T_2$ does not matter.
Given the number of calculations we have to perform and their similarity, we present them in short form in Tab.~\ref{tab:involution_composition}. We perform the calculation explicitly for AI$\circ$AI as an example below, and only give complementary comments for the other cases where necessary.

\textbf{AI$\circ$AI.}
First, according to Tab.~\ref{tab:all_cartan_involutions}, two involutions of type AI take the form $\theta_i=\Ad_{V_i}\circ\,\ast$ with $V_i^T=V_i$. Second, their composition reads 
\begin{align}
    \theta_1\circ\theta_2=\Ad_{V_1}\circ\,\ast\,\circ\Ad_{V_2}\circ\,\ast=\Ad_{V_1V_2^\ast},
\end{align}
and analogously for $\theta_2\circ\theta_1$, so that the commutativity criterion is $V_1V_2^\ast=e^{i\phi} V_2V_1^\ast$ with $\phi=\frac{2\pi k}{n}$, as discussed in the observations above.
Third, we observe the absence of complex conjugation from the composition, so that the candidate involution type is AIII and we should test for the corresponding criterion $G^\dagger\overset{?}{=}G$ from Tab.~\ref{tab:all_cartan_involutions}. We do this for a global phase ansatz $G=e^{i\alpha}V_1V_2^\ast$:
\begin{align}
    G^\dagger = e^{-i\alpha}V_2^TV_1^\dagger
    =e^{-i\alpha}V_2V_1^\ast
    =e^{-i(\alpha+\phi)}V_1V_2^\ast
    \overset{?}{=}e^{i\alpha}V_1V_2^\ast,
\end{align}
where we used $V_i^T=V_i$ and the commutativity criterion. To satisfy $G^\dagger=G$, we simply set $\alpha=-\phi/2$. Finally, we need to check the determinant $\det{G}=\det{e^{-i\frac{\phi}{2}}\id_n}\det{V_1}\det{V_2^\ast}=(-1)^k$, which is in line with the criterion for AIII involutions, and tells us that $p$ or $q$ for the involution must have the same parity as $k$, which derives from the commutativity criterion.

As mentioned above, we will only give some additional comments where necessary for the other pairs.

\textbf{AI$\circ$AIII and AII$\circ$AIII.}
For these pairs, we can (and need to) restrict the global phase $\phi$ further than for AI$\circ$AI or AI$\circ$AII.
We consider AI$\circ$AIII here but the discussion for AII$\circ$AIII is analogous.
Note that $VH^\ast=e^{i\phi}HV$ implies
\begin{align}
    e^{i\phi}\id=HVH^\ast V^\ast
    =(A+iB)(A-iB)
    =A^2+B^2-i[A,B],
\end{align}
for some real matrices $A,B$, and that $[A,B]$ is traceless, so that it cannot contribute to the LHS term.
Consequentially, $e^{i\phi}\in\mbr$ and thus we may restrict the phase to $\phi\in\{0,\pi\}$.
The phase then determines whether $\theta_3$ is of type AI or AII, as can be seen from $G^T=e^{i\phi}G=\pm G$.
As for the phase of $G$, it is chosen to compensate the determinant of $H^\ast$, where $\det{H}=e^{i\vartheta}\in\{1,-1\}$.

\textbf{BDI$\circ$BDI.}
This case is straightforward, except for the determinant of $G$.
Note that $R_1$ and $R_2$ may carry a determinant of $-1$, so that $\det{G}=\det{R_1}\det{R_2}$ might differ from $+1$.
This is not an issue for the case $R_1R_2=R_2R_1$, as $\theta_3$ is of type BDI then, and $\det{G}=-1$ is allowed.
However, for $R_1R_2=-R_2R_1$, we need to prove that $n$ is even (because DIII is only defined on $\mfso(2n)$) and that $\det{G}=1$, which is exactly what the prepared Lemma~\ref{lemma:ac_sym_ortho_mats} guarantees.

\textbf{BDI$\circ$DIII.}
This case again is straightforward up to the determinant.
Note that $L$ and $R$ must be $2n\times 2n$ matrices for DIII to apply.
For $RL=LR$, we compute that $G=RL$ is skew-symmetric and orthogonal, i.e.,~$G^2=-\id$. This implies, for an eigenvalue $\lambda$ of $G$ with eigenvector $v$, that 
\begin{align}
    v=-G^2v=-\lambda^2v \quad\Rightarrow\quad\lambda=\pm i.
\end{align}
Furthermore, the eigenvalues of a skew-symmetric matrix come in pairs $(\lambda, -\lambda)$, leading to the same multiplicity for $i$ and $-i$ and thus to $\det{G}=i^n(-i)^n=1$.
For $RL=-LR$, we obtain a type-BDI involution, so that the determinant need not be fixed beyond $\det{G}\in\{\pm1\}$.
Overall, we find
\begin{align}\label{eq:bdi_diii_determinant}
    \det{G}=\det{R}\ \underset{1}{\underbrace{\det{L}}} =\begin{cases}
        1 & \text{ for }\ RL=LR\\
        \det{R} & \text{ for }\ RL=-LR.\\
    \end{cases}
\end{align}

\textbf{CI$\circ$CI.}
The cases for symplectic matrices are straightforward, we only remark that we use the characterizations $\theta=\Ad_S,\ S^\dagger=-S$ for CI and $\theta=\Ad_{P},\ P^\dagger=P$ for CII.
\end{proof}

\textbf{Composition table structure.}
Note that the composition tables for $\mfso$ and $\mfsp$ look like the multiplication table of the sign group $C_2$ (or addition table of $\mathbb{Z}_2$), with identity elements BDI and CII for $G_1G_2=G_2G_1$ or DIII and CI for $G_1G_2=-G_2G_1$.
For $\mfsu(n)$, the composition operation on the set of types is commutative, non-associative, has an identity (only for $G_1G_2=G_2G_1$), and non-unique inverses. This gives $\{\text{AI, AII, AIII}\}$ together with composition the structure of a commutative (unital) magma. While it is composed of invertible elements only, it is not a quasi-group due to non-uniqueness of the inverses.
We stress that these are merely observations relating Cayley tables of groupoids to the composition tables above, and there might not be a deeper meaning behind this; we will not pursue this line of investigation further.

\subsection{Proofs for higher-order Cartan decompositions}\label{sec:involutions_calculations:HOCD}
In this section, we prove Props.~\ref{prop:set_of_cds_to_grading},~\ref{prop:grading_to_set_of_cds} and Thm.~\ref{thm:grading_recursion_equivalence}.

\setofcdstograding*
\begin{proof}
    First, note that the compatibility criterion implies that 
    \begin{align}
        \mfg = \mfk_1\oplus\mfp_1
        =(\mfk_1\cap\mfk_2)\oplus(\mfk_1\cap\mfp_2)\oplus(\mfp_1\cap\mfk_2)\oplus(\mfp_1\cap\mfp_2)
        =\dots
        =\bigoplus_{s\in\Ztwo^c}\mfg_s.
    \end{align}
    Next, let $x\in\mfg_s$ and $y\in\mfg_t$ for $s,t\in\Ztwo^c$.
    Note that this implies, for all $1\leq j\leq c$, $x\in\mfk_j$ if $s_j=0$ and $x\in\mfp_j$ if $s_j=1$, and similarly for $y$ and $t$.
    The commutation relations of the \acp{CD} then imply
    \begin{align}
        \forall\ 1\leq j\leq c:\ [x,y]\in\begin{cases}
            [\mfk_j,\mfk_j] & \text{ if }\ s_j=t_j=0\\
            [\mfk_j,\mfp_j] & \text{ if }\ s_j=0, t_j=1\\
            [\mfp_j,\mfk_j] & \text{ if }\ s_j=1, t_j=0\\
            [\mfp_j,\mfp_j] & \text{ if }\ s_j=t_j=1
        \end{cases} \ \ =\begin{cases}
            \mfk_j & \text{ if }\ s_j+t_j=0\\   
            \mfp_j & \text{ if }\ s_j+t_j=1
        \end{cases}
        \quad\Rightarrow\quad [x,y]\in\mfg_{s+t}.
    \end{align}
    This proves the proposition. 
    
    We also want to show that the compatibility requirement cannot be skipped as claimed in~\cite{dagli2008general}.
    Concretely, consider the algebra $\mfsu(3)$ and its defining representation
    \begin{align}
        \mfg=\{x\in \mbc^{3\times 3}|x^\dagger=-x, \tr(x)=0\},
    \end{align}
    together with the involutions
    \begin{align}
        \theta_1(x)=x^\ast,\quad \theta_2(x)=\Ad_W;~W=\begin{pmatrix}
            0 & e^{i\frac{\pi}{4}} & 0 \\
            e^{-i\frac{\pi}{4}} & 0 & 0 \\
            0 & 0 & 1
        \end{pmatrix}.
    \end{align}
    $\theta_1$ clearly is of type AI, and $\theta_2$ is conjugation by a Hermitian operator $W$ from $\gru(3)$ with $\det{W}=-1$, and thus of type AIII (c.f.~Tab.~\ref{tab:all_cartan_involutions}).
    The involutions do not commute,
    \begin{align}
        \theta_1\circ\theta_2\circ\theta_1 = \Ad_{W^\ast}\neq\Ad_{W},
        \quad\text{using} \quad
        W^\ast\not\propto W,
    \end{align}
    so that the \acp{CD} are not compatible. Indeed we find the subspaces
    \begin{alignat}{5}
        \mfk_1&=\{x\in\mfg | x^\ast=x\},\quad
        &\mfk_2=&\left\{
            \begin{pmatrix}
                ia & (i-1)b & \varphi \\ 
                (i+1)b & ia & \varphi (1-i) \\ 
                -\varphi^\ast & -\varphi^\ast (1+i) & -2ia 
            \end{pmatrix}
            \bigg | a, b \in\mbr, \varphi\in\mbc\right\},
        \\
        \mfp_1&=\{x\in\mfg | x^\ast=-x\},\quad
        &\mfp_2=&\left\{
            \begin{pmatrix} 
                ia & (i+1)b & \varphi \\ 
                (i-1)b & -ia & -\varphi (1-i) \\
                -\varphi^\ast & \varphi^\ast (1+i) & 0
            \end{pmatrix}
            \bigg | a, b \in\mbr, \varphi\in\mbc\right\},
    \end{alignat}
    which lead to the intersections
    \begin{align}
        \mfg_{00}=\mfg_{01}=\{0\}, \mfg_{10}=\{\diag(ia, ia, -2ia)|a\in \mbr\},\mfg_{11}=\{\diag(ia, -ia, 0)|a\in \mbr\},
    \end{align}
    spanning a two dimensional subspace of $\mfg$ only. In particular, as the $\mfg_s,\ s\in\Ztwo^2$ do not sum to $\mfg$, we do not obtain a Cartan grading.
\end{proof}

\gradingtosetofcds*

\begin{proof}
    For the first statement, let $Q\subseteq\Ztwo^p$ be a subgroup and note that
    \begin{align}
        [\mfg_Q,\mfg_Q]\subset \bigoplus_{s,t\in Q}\mfg_{s+t}=\bigoplus_{r\in Q}\mfg_r=\mfg_Q,
    \end{align}
    where we used the grading property and then that $Q$ is a group, with equality following from $0+t=t$.

    Next, let $\varphi:Q\to\Ztwo$ be a surjective group homomorphism and denote its preimages as $P_i\coloneqq\varphi^{-1}(i)$. Note that
    \begin{align}
        \varphi(P_i+P_j)=\varphi(P_i)+\varphi(P_j)=i+j,
    \end{align}
    so that $P_i+P_j\subseteq P_{i+j}$. This implies
    \begin{align}
        [\mfg_{P_i}, \mfg_{P_j}]
        =\bigoplus_{s\in P_i, t\in P_j}[\mfg_s,\mfg_t]
        \subseteq \bigoplus_{s\in P_i, t\in P_j} \mfg_{s+t}
        \subseteq \bigoplus_{s\in P_{i+j}} \mfg_s\nonumber
        =\mfg_{P_{i+j}},
    \end{align}
    which is equivalent to the commutation relations of a \ac{CD}:
    \begin{align}
        [\mfk_\varphi,\mfk_\varphi]\subseteq \mfk_\varphi,\quad
        [\mfk_\varphi,\mfp_\varphi]\subseteq \mfp_\varphi,\quad
        [\mfp_\varphi,\mfp_\varphi]\subseteq \mfk_\varphi.
    \end{align}
    
    For the last statement we simply note that $\varphi_j:s\mapsto s_j$ is a homomorphism, because addition on $\Ztwo^c$ is defined element-wise, and surjective because $\Ztwo^c$ contains elements with $s_j=0$ and elements with $s_j=1$.
\end{proof}

\gradingrecursionequivalence*
\begin{proof}
    The first part is easy to prove using Prop.~\ref{prop:grading_to_set_of_cds}. Note that the vector space structure is correct:
    \begin{alignat}{3}
        \mfk_0&=
        \mfg=
        \bigoplus_{s\in\Ztwo^{c}} \mfg_{s}=
        \bigoplus_{s\in\Ztwo^{c-1}} \mfg_{0 \mathlarger{s}}\ \oplus
        \bigoplus_{s\in\Ztwo^{c-1}} \mfg_{1 \mathlarger{s}}
        &&=\mfk_{1}\oplus\mfp_{1},\\
        \mfk_{\ell+1}&=
        \bigoplus_{s\in\Ztwo^{c-\ell-1}} \mfg_{0^{\ell} 0\mathlarger{s}}
        =\bigoplus_{s\in\Ztwo^{c-\ell-2}} \mfg_{0^{\ell+1} 0\mathlarger{s}}\oplus \bigoplus_{s\in\Ztwo^{c-\ell-2}} \mfg_{0^{\ell+1} 1\mathlarger{s}}
        &&=\mfk_{\ell+2}\oplus \mfp_{\ell+2}.
    \end{alignat}
    Next, note that $\{0^\ell s|s\in\Ztwo^{c-\ell}\}\subset\Ztwo^c$ is a subgroup and that $\varphi: 0^\ell s\mapsto s_{1}$, for each $\ell$, is a group homomorphism from this subgroup to $\Ztwo$ with $\mfk_{\varphi}=\mfk_{\ell+1}$ and $\mfp_{\varphi}=\mfp_{\ell+1}$.
    Prop.~\ref{prop:grading_to_set_of_cds} then proves that the vector space decompositions above are in fact \acp{CD}.

    For the second part we will have to work more.
    We will use a proof by induction on the recursion depth $r$.
    The induction hypothesis is that an $r$-recursive \ac{CD} that can be made homogeneous induces a Cartan $r$-grading of $\mfg$.
    Accordingly, we first need to prove the hypothesis for $r=2$ (the base case) by transforming—or lifting—a $2$-recursive \ac{CD} to a pair of \acp{CD} on $\mfg$, which then yields a Cartan $2$-grading via Prop.~\ref{prop:set_of_cds_to_grading}.
    We defer this proof to further below and continue with the induction for now.
    
    Assume that the hypothesis holds for some $r$ and consider the case $r+1$.
    That is, we have an $(r+1)$-recursive \ac{CD} $T_1\to T_2\to\dots \to T_{r+1}$, which we assume has been made homogeneous. This recursion contains a homogeneous $r$-recursive \ac{CD} $T_2\to \dots\to T_{r+1}$ of $\mfk_1$, the vertical space of the very first \ac{CD} of type $T_1$.
    The induction hypothesis for $r$ then implies that there are $r$ \acp{CD} (of types $T'_2,\dots,T_{r+1}'$, say) of $\mfk_1$ that yield a Cartan $r$-grading of $\mfk_1$ which reproduces the $r$-recursive \ac{CD} of $\mfk_1$.
    Then we have $r$ new $2$-recursive \acp{CD} $T_1\to T_i'$, for $2\leq i\leq r+1$, of $\mfg$.
    We force these to be homogeneous again by choosing the same lift on branches of the given $r$-recursive decomposition, which we illustrate below.
    Applying the base case to these $2$-recursive \acp{CD}, we find a pair of ``standard" \acp{CD} $\{T_1, T_i''\}$ of $\mfg$ for each $2\leq i\leq r+1$. Now, the first \ac{CD} is the same for all $i$, so that we obtain $r+1$ compatible \acp{CD} $\{T_1, T_2'',\dots T_{r+1}''\}$ of $\mfg$ overall. To conclude, we note that those \acp{CD} again yield a Cartan $(r+1)$-grading via Prop.~\ref{prop:set_of_cds_to_grading}.
    
    Before completing the proof by proving the base case, we want to illustrate the interplay of homogeneity and the induction step from above.
    For this, consider a $3$-recursive \ac{CD}  $T_1\to T_2\to T_3$ as an example. 
    During the induction step, we first perform the lift $(T_2\to T_3)\nearrow\{T_2,T_3'\}$ via the induction hypothesis. This yields two $2$-recursive \acp{CD} $\{T_1\to T_2, T_1\to T_3'\}$, which then are lifted as $(T_1\to T_2)\nearrow\{T_1,T_2''\}$ and $(T_1\to T_3')\nearrow \{T_1,T_3''\}$, respectively.
    Overall, we obtain the three \acp{CD} of types $T_1$, $T_2''$, and $T_3''$.
    If any of the types AIII, BDI, and CII are contained in the chain, they lead to additional simple components of the intermediate Lie algebras. These are then decomposed independently of each other in the next recursion steps, with the exception of types A, BD, and C.
    This leads to a tree-like structure (with types A, BD, and C causing branches to merge), as discussed in the context of Def.~\ref{def:homogeneity}. This tells us that the recursive structure can be ``resolved" locally, lifting involutions on the simple components independently of each other.
    As the recursion is homogeneous, i.e.,~the decomposition types across all branches are the same at a given recursion depth, this local resolution of the recursion is self-consistent, and we only ever encounter $r=2$ recursions from our explicit calculations above.
    
    For lifts that are not unique in the type they produce, we assume that compatible choices are made between branches. Take, e.g.,~the recursion AIII$\to$(AI$\oplus$AI)$\to$(DIII$\oplus$DIII), which is homogeneous. In order to resolve the recursion into a set of \acp{CD}, we need to first lift each branch, created by the type-AIII \ac{CD}, to the same type of involution, e.g.,~to $\{$AI$\oplus$AI,AIII$\oplus$AIII$\}$. Afterwards, the resulting homogeneous $2$-recursive \acp{CD} AIII$\to$(AI$\oplus$AI) and AIII$\to$(AIII$\oplus$AIII) can be lifted. As a counter-example, the lift ((AI$\oplus$AI)$\to$(DIII$\oplus$DIII))$\nearrow\{$(AI$\oplus$AI),(AIII$\oplus$AII)$\}$ would be invalid.
    
    \textbf{Proof of the base case $r=2$.}
    For the base case $r=2$, we unfortunately did not find a general proof, but a type-based one. Similar to our derivation of involution compositions in App.~\ref{sec:involutions_calculations:compose}, we will use the characterization of all Cartan involutions of classical simple Lie algebras in their defining representation to show that the second involution, defined on the $+1$ eigenspace of the first, can be extended to a Cartan involution on the total space that is compatible with the first involution.  We first discuss the general procedure for a $2$-recursive decomposition in detail and then briefly perform the calculations for all specific type combinations.
    
    We are given two Cartan involutions $\theta_{1,2}$ expressed in the defining representation of their domains. $\theta_1$ is an involution on $\mfg$ whereas $\theta_2$ is an involution on $\mfk$, the defining representation of the $+1$ eigenspace $\tilde{\mfk}_1$ of $\theta_1$. Note that $\mfg$ or $\mfk$ (or both) may be a direct sum of two classical algebras and potentially an abelian center $\mfu(1)$. In small dimensions, the classical algebras might not be simple, e.g.,~$\mfso(2)\cong\mfu(1)$, which is abelian, or $\mfso(4)\cong\mfsu(2)\oplus\mfsu(2)$, which is semisimple.
    Denoting the generic form (see Tab.~\ref{tab:all_cartan_involutions}) of $\theta_1$ by $\theta_0$, we may express it as
    \begin{align}
        \theta_1=\Ad_{G^{-1}} \circ~\theta_0\circ\Ad_{G},\quad G\in\exp(\mfg);
    \end{align}
    also see the characterization in Sec.~\ref{sec:involutions_calculations:characterize}.
    Using an isomorphism $\phi$ from the $+1$ eigenspace of $\theta_0$ to its defining representation $\mfk$, we modify $\theta_2$ to act on $\tilde{\mfk}_1\subset\mfg$:
    \begin{align}
        \theta'_2 = \Ad_{G^{-1}} \circ\phi^{-1}\circ\theta_2\circ \phi \circ\Ad_G.
    \end{align}
    The extension from $\tilde{\mfk}_1$ to $\mfg$ then is as simple as to change the domain of $\theta'_2$, which we then prove to be a valid Cartan involution on $\mfg$ by identifying it with one of the entries in Tab.~\ref{tab:all_cartan_involutions}.
    Next we discuss the required isomorphisms $\phi$ and afterwards we provide an example calculation and additional details on more involved cases.

    \textbf{Isomorphisms $\phi$.}
    We will only use two distinct isomorphisms (up to domain changes) for our case-based proof: the identity and a specific isomorphism $\phi_\times$ on $\mfu(n)$, which we discuss in the following.

    We begin with the isomorphism for DIII and CI involutions, which will turn out to also be useful for types A, BD, and C later on.
    The defining representation of $\mfso(2n)$ is
    \begin{align}
        \mfso(2n)=\left\{\begin{pmatrix}x & y \\ -y^T & z \end{pmatrix}\bigg|x,y,z\in\mbr^{n\times n}, x^T=-x, z^T=-z\right\}.
    \end{align}
    The DIII Cartan involution in canonical form on this representation is $\Ad_{J_n}$, which has the $+1$ eigenspace
    \begin{align}
        \tilde{\mfk}_1=\left\{\begin{pmatrix}x & y \\ -y & x \end{pmatrix}\bigg|x,y\in\mbr^{n\times n} x^T=-x, y^T=y\right\}\cong\mfu(n).
    \end{align}
    Before we define the isomorphism $\phi_\times$, note that the CI Cartan involution in its canonical form, defined on
    \begin{align}
        \mfsp(n)=\left\{\begin{pmatrix}x & y \\ -y^\ast & x^\ast \end{pmatrix}\bigg|x,y\in\mbc^{n\times n}, x^\dagger=-x, y^T=y\right\},
    \end{align}
    is $\Ad_{J_n}$ as well, which equals $\ast$ on this representation. The $+1$ eigenspace for the CI involution thus works out to be the same $\tilde{\mfk}_1$ as from the DIII involution above. This is why we only require one non-trivial isomorphism, which we define now.
    
    Given a DIII or CI involution $\theta_1$ expressed in the definining representation of its domain, an isomorphism from its $+1$ eigenspace $\tilde{\mfk}_1$ to the defining representation of the unitary algebra $\mfk=\mfu(n)=\{x\in\mbc^{n\times n}|x^\dagger =-x\}$ is given by
    \begin{align}\label{eq:def_phi_times}
        \phi_\times:\tilde{\mfk}_1\to\mfk_1,\quad
        \begin{pmatrix}x & y \\ -y & x \end{pmatrix} \mapsto x+iy.
    \end{align}
    The inverse of $\phi_\times$ simply copies the real and imaginary parts into a $2\times 2$ block matrix. To check that it is a vector space isomorphism, note that it is linear and surjective,
    \begin{align}
        v\in\mbc^{n\times n}, v^\dagger=-v\ \ &\Rightarrow\ \  
        \real(v)^T=\real(v^\dagger)=-\real(v),\quad
        \imag(v)^T=\imag(-v^\dagger)=\imag(v),
    \end{align}
    and thus bijective because $\dim \tilde{\mfk}_1=\dim \mfk_1$. Further, it is a Lie algebra isomorphism because
    \begin{align}
        \left[\phi_\times\begin{pmatrix}x & y \\ -y & x \end{pmatrix},\phi_\times\begin{pmatrix}z & w \\ -w & z \end{pmatrix}\right]
        &=[x+iy,z+iw]
        =[x,z]-[y,w]+i[y,z]+i[x,w]\\
        &=\phi_\times\begin{pmatrix}[x,z]-[y,w] & [y,z]+[x,w] \\ -[y,z]-[x,w] & [x,z]-[y,w] \end{pmatrix}
        =\phi_\times\left(\left[\begin{pmatrix}x & y \\ -y & x \end{pmatrix},\begin{pmatrix}z & w \\ -w & z \end{pmatrix} \right]\right).
    \end{align}

    Next, we would like to know the transformation behaviour of an involution $\theta_2=\Ad_G\circ\,\ast^\tau$, $\tau\in\{0,1\}$, on $\mfu(n)$ (or $\mfsu(n)$) when pulled back to $\tilde{\mfk}_1$:
    \begin{align}
        \phi_\times^{-1}\circ\theta_2\circ\phi_\times \begin{pmatrix}x & y \\ -y & x \end{pmatrix}
        &=\phi_\times^{-1}\left(G(x+(-1)^\tau iy)G^\dagger\right)\\
        &=\begin{pmatrix}
            \real(GxG^\dagger)-(-1)^\tau\imag(GyG^\dagger) &
            \imag(GxG^\dagger)+(-1)^\tau \real(GyG^\dagger) \\
            -\imag(GxG^\dagger)-(-1)^\tau \real(GyG^\dagger) &
            \real(GxG^\dagger)-(-1)^\tau\imag(GyG^\dagger) \end{pmatrix}\nonumber\\
        \Rightarrow\qquad\phi_\times^{-1}\circ\theta_2\circ\phi_\times &= \Ad_{\varphi_\tau(G)},
        \quad \varphi_\tau(G)\coloneqq\begin{pmatrix}
        \real(G) & (-1)^\tau \imag(G) \\
        -\imag(G) & (-1)^\tau \real(G) 
        \end{pmatrix}.\label{eq:def_varphi_tau}
    \end{align}
    For our recursion base cases we are interested in the scenarios $G^T=\pm G,\tau=1$ (AI and AII) and $G^\dagger=G, \tau=0$ (AIII). Note that they imply
    \begin{align}
        \varphi_\tau(G)^T
        &=\begin{pmatrix}
        \real(G)^T & - \imag(G)^T \\
        (-1)^\tau\imag(G)^T & (-1)^\tau \real(G)^T 
        \end{pmatrix}
        =\begin{cases}
            \varphi_\tau(G) & \text{ for } G^T=G, \tau=1,\\
            -\varphi_\tau(G) & \text{ for } G^T=-G, \tau=1,\\
            \varphi_\tau(G) & \text{ for } G^\dagger=G, \tau=0.
        \end{cases}
    \end{align}
    Here we used $\real(G)^T=\pm\real(G)$, $\imag(G)^T=\pm\imag(G)$ for $G^T=\pm G, \tau=1$ and $\real(G)^T=\real(G)$, $\imag(G)^T=-\imag(G)$ for $G^\dagger=G, \tau=0$.
    In addition, we compute
    \begin{align}
        \det{\varphi_\tau(G)}&=\det{(-1)^\tau\real(G)^2-\imag(G)^2}
        =\begin{cases}
            (-1)^n & \text{ for } G^T=G, \tau=1,\\
            1 & \text{ for } G^T=-G, \tau=1,\\
            1 & \text{ for } G^\dagger=G, \tau=0,
        \end{cases}
    \end{align}
    using $[\real(G),\imag(G)]=0$, $\real(G)^2+\imag(G)^2=\pm\id$ for $G^T=\pm G, \tau=1$ and $\{\real(G),\imag(G)\}=0$, $\real(G)^2-\imag(G)^2=\id$ for $G^\dagger=G, \tau=0$, which all follow from $G^\dagger G=\id$. These commutation rules then allow us to compute the determinant~\cite{powell2011calculating}.

    We also note that $\phi_\times$ can be used for involutions of types A, BD, and C as well. We illustrate this with a type-A involution, the other cases follow analogously.
    The $+1$ eigenspace of the type-A Cartan involution in canonical form is
    \begin{align}
        \tilde{\mfk}_2=\left\{ x\oplus x | x\in\mfsu(n)\right\}\cong\mfsu(n).
    \end{align}
    Changing the domain of $\phi_\times$ from $\tilde{\mfk}_1$ to $\tilde{\mfk}_2$ yields the isomorphism $x\oplus x\mapsto x$ onto the defining representation of $\mfsu(n)$. We omit the proof that this is still an isomorphism, as its inverse is a canonical reducible representation of $\mfsu(n)$.

    \textbf{Calculation example.}
    Let us go through the calculation for the $2$-recursion AII$\to$CI in detail, which is staged on $\mfsu(2n)$.
    The canonical involution for AII is $\theta_0=\Ad_{J_n}\circ\,\ast$, which needs to be combined with a basis choice $U\in \grsu(2n)$ to obtain any other Cartan involution. The involution for CI in any basis of the defining representation of $\mfsp(n)$ is $\theta_2=\Ad_S\circ\,\ast$ with $S\in \grsp(n)$ and $S^T=S$; see the alternative characterization in Sec.~\ref{sec:involutions_calculations:characterize}. 
    As the $+1$ eigenspace of $\theta_0$ on the defining representation of $\mfsu(2n)$ matches the defining representation of $\mfsp(n)$\footnote{That is to say, unitary symplectic generators are defined as unitary block matrices with additional constraints imposed on the blocks.}, we may use the trivial isomorphism $\phi=\mathrm{id}$.
    Together, this yields the pulled-back map 
    \begin{align}
        \theta'_2
        &=\Ad_{U^\dagger}\circ\mathrm{id}^{-1}\circ[\Ad_S\circ\,\ast]\circ\mathrm{id}\circ\Ad_U
        =\Ad_{U^\dagger S U^\ast}\circ\,\ast=\Ad_{S'}\circ\,\ast,\\
        \text{with}\ \ S'^T&=U^\dagger S^T U^\ast=S', \ \ \det{S'}=\det{U^\dagger}\ \det{S}\ \det{U}=1,
    \end{align}
    which is an involution of type AI.
    To understand that lifts are not unique, consider the original expression $\theta_2=\Ad_S$ with $S^\dagger=-S$. Repeating the above calculation, we find
    \begin{align}
        \theta'_2
        &=\Ad_{U^\dagger}\circ\mathrm{id}^{-1}\circ[\Ad_S]\circ\mathrm{id}\circ\Ad_U
        =\Ad_{U^\dagger S U}=\Ad_{iU^\dagger S U}=\Ad_{S'},\\
        \text{with}\ \ S'^\dagger&=-iU^\dagger S^\dagger U=S', \ \ \det{S'}=i^{2n}\det{U^\dagger}\ \det{S}\ \det{U}=(-1)^n,
    \end{align}
    indicating that this extension $\theta'_2$ is of type AIII.

    For all other cases, we abstract the calculation away into a tabular form in Tab.~\ref{tab:2_recursive_lifts}, commenting on noteworthy details in the following.
    
    \begin{table*}
        \centering
        \resizebox{0.95\textwidth}{!}{%
        \begin{tabular}{llccccclc}
            $T_1\rightarrow T_2$ & & &
            $\theta_0$ & $\theta_2$ & \\
            $\qquad\nearrow \{T_1,T_2'\}$ & \mr{$\mfg$} & \mr{basis $G\in\exp(\mfg)$} & $\theta'_2$ &
            Calculations & \mr{$\phi$} \\\midrule
            
            \belowrulesepcolor{tabgrey}\rc{tabgrey} A$\to$AI & & &
            $\swapsymbol$ & $\Ad_V\circ\,\ast, V\in \grsu(n), V^T=V$ & \\
            \rc{tabgrey} $\qquad\nearrow\{$A,AI$^{\oplus 2}\}$ & \mr{$\mfsu(n)^{\oplus 2}$} & \mr{$U\in \grsu(n)^{\times 2}$} &  $\Ad_{V'=U^\dagger V^{\oplus 2}U^\ast}\circ \ast$ &
            $V'^T=U^\dagger (V^T)^{\oplus 2} U^\ast=V', \det{V'}=1$ & \mr{$\phi_\times$} \\
            
            A$\to$AII & & &
            $\swapsymbol$ & $\Ad_W\circ\,\ast, W\in \grsu(2n), W^T=-W$ & \\
            $\qquad\nearrow\{$A,AII$^{\oplus 2}\}$ &\mr{$\mfsu(2n)^{\oplus 2}$} & \mr{$U\in \grsu(2n)^{\times 2}$}&  $\Ad_{W'=U^\dagger W^{\oplus 2}U^\ast}\circ \ast$ &
            $W'^T=U^\dagger (W^T)^{\oplus 2} U^\ast=-W', \det{W'}=1$ & \mr{$\phi_\times$} \\

            \rc{tabgrey} A$\to$AIII & & &
            $\swapsymbol$ & $\Ad_H, H\in \gru(n), H^\dagger=H, \det{H}\in\{\pm 1\}$ & \\
            \rc{tabgrey} $\qquad\nearrow\{$A,AIII$^{\oplus 2}\}$ &\mr{$\mfsu(n)^{\oplus 2}$} & \mr{$U\in \grsu(n)^{\times 2}$}&  $\Ad_{H'=U^\dagger H^{\oplus 2}U}$ &
            $H'^\dagger=U^\dagger (H^\dagger)^{\oplus 2} U=H', \det{H'}=1$ & \mr{$\phi_\times$} \\ 
            
            AI$\to$BDI & & &
            $\ast$ & $\Ad_R, R\in \gro(n), R^T=R$ &\\
            $\qquad\nearrow\{$AI,AIII$\}$ & \mr{$\mfsu(n)$} & \mr{$U\in \grsu(n)$} &  $\Ad_{H=U^\dagger R U}$ &
            $H^\dagger=U^\dagger R^\dagger U=H, \det{H}=\det{R}\in\{\pm1\}$ & \mr{id} \\
            
            \rc{tabgrey} AI$\to$DIII & & &
            $\ast$ & $\Ad_L, L\in \grso(n), L^T=-L$ & \\
            \rc{tabgrey} $\qquad\nearrow\{$AI,AIII$\}$ & \mr{$\mfsu(2n)$} & \mr{$U\in \grsu(2n)$} &  $\Ad_{H=iU^\dagger L U}$ &
            $H^\dagger=-iU^\dagger L^\dagger U=H, \det{H}=(-1)^n$ & \mr{id} \\
            
            AII$\to$CI && &
            $\Ad_{J_n}\circ\,\ast$ & $\Ad_S\circ\,\ast, S\in \grsp(n), S^T=S$ & \\
            $\qquad\nearrow\{$AII,AI$\}$ & \mr{$\mfsu(2n)$} & \mr{$U\in \grsu(2n)$} &  $\Ad_{V=U^\dagger S U^\ast}\circ\,\ast$ &
            $V^T=U^\dagger S^T U^\ast=V, \det{V}=1$ & \mr{id} \\
            
            \rc{tabgrey} AII$\to$CII  & & &
            $\Ad_{J_n}\circ\,\ast$ & $\Ad_P, P\in \grsp(n), P^\dagger=P$ & \\
            \rc{tabgrey} $\qquad\nearrow\{$AII,AIII$\}$ &\mr{$\mfsu(2n)$} & \mr{$U\in \grsu(2n)$} &  $\Ad_{H=U^\dagger P U}$ &
            $H^\dagger=U^\dagger P^\dagger U=H, \det{H}=1$ & \mr{id} \\
            
            AIII$\to$A & & &
            $\Ad_{I_{n,n}}$ & $\Ad_{X}\circ\swapsymbol, X\in \grsu(n)^{\times 2}, X^\dagger=X^{\smallswap}$ & \\
            $\qquad\nearrow\{$AIII,AIII$\}$ & \mr{$\mfsu(2n)$} & \mr{$U\in \grsu(2n)$} &  $\Ad_{H=U^\dagger X U^{\smallswap}\smallswap}$ &
            $H^\dagger=\swapsymbol {U^{\smallswap}}^\dagger X^\dagger U=H, \det{H}=(-1)^n$ & \mr{id} \\
            
            \rc{tabgrey} AIII$\to$AI$^{\oplus 2}$ & & &
            $\Ad_{I_{p,q}}$ & $\Ad_{V}\circ\,\ast, V\in \grsu(p)\times \grsu(q), V^T=V$ & \\
            \rc{tabgrey} $\qquad\nearrow\{$AIII,AI$\}$ & \mr{$\mfsu(n)$} & \mr{$U\in \grsu(n)$} &  $\Ad_{V'=U^\dagger V U^\ast}\circ\,\ast$ &
            $V'^T=U^\dagger V^T U^\ast=V', \det{V'}=1$ & \mr{id} \\
            
            AIII$\to$AII$^{\oplus 2}$ & & &
            $\Ad_{I_{2p,2q}}$ & $\Ad_{W}\circ\,\ast, W\in \grsu(2p)\times \grsu(2q), W^T=-W$ & \\
            $\qquad\nearrow\{$AIII,AII$\}$ & \mr{$\mfsu(2n)$} & \mr{$U\in \grsu(2n)$} &  $\Ad_{W'=U^\dagger W U^\ast}\circ\,\ast$ &
            $W'^T=U^\dagger W^T U^\ast=-W', \det{W'}=1$ & \mr{id} \\
            
            \rc{tabgrey} AIII$\to$AIII$^{\oplus 2}$ & & &
            $\Ad_{I_{p,q}}$ & $\Ad_{H}, H\in \gru(p)\times \gru(q), H^\dagger=H,\det{H}\in\{\pm 1\}$ & \\
            \rc{tabgrey} $\qquad\nearrow\{$AIII,AIII$\}$ & \mr{$\mfsu(n)$} & \mr{$U\in \grsu(n)$} &  $\Ad_{H'=U^\dagger H U}$ &
            $H'^\dagger=U^\dagger H^\dagger U=H', \det{H'}\in\{\pm 1\}$ & \mr{id} \\
            \aboverulesepcolor{tabgrey}\midrule
            
            BD$\to$BDI & & &
            $\swapsymbol$ & $\Ad_R, R\in \gro(n), R^T=R$ & \\
            $\qquad\nearrow\{$BD,BDI$^{\oplus 2}\}$ & \mr{$\mfso(n)^{\oplus 2}$} & \mr{$Q\in \grso(n)^{\times 2}$} &  $\Ad_{R'=Q^T R^{\oplus 2}Q}$ &
            $R'^T=Q^T (R^T)^{\oplus 2} Q=R', \det{R'}=1$ & \mr{$\phi_\times$} \\
            
            \rc{tabgrey} BD$\to$DIII & & &
            $\swapsymbol$ & $\Ad_L, L\in \grso(2n), L^T=-L$ & \\
            \rc{tabgrey} $\qquad\nearrow\{$BD,DIII$^{\oplus 2}\}$) & \mr{$\mfso(2n)^{\oplus 2}$} & \mr{$Q\in \grso(2n)^{\times 2}$} &  $\Ad_{L'=Q^T L^{\oplus 2}Q}$ &
            $L'^T=Q^T (L^T)^{\oplus 2} Q=-L', \det{L'}=1$ & \mr{$\phi_\times$} \\
            
            BDI$\to$BD & & &
            $\Ad_{I_{n,n}}$ & $\Ad_{X}\circ\swapsymbol, X\in \grso(n)^{\times 2}, X^T=X^{\smallswap}$ & \\
            $\qquad\nearrow\{$BDI,BDI$\}$ & \mr{$\mfso(2n)$} & \mr{$Q\in \grso(2n)$} & $\Ad_{V=Q^T X Q^{\smallswap}\smallswap}$ &
            $V^T=\swapsymbol {Q^{\smallswap}}^T X^T Q=V, \det{V}=(-1)^n$ & \mr{id} \\            
            \rc{tabgrey} BDI$\to$BDI$^{\oplus 2}$ & & &
            $\Ad_{I_{p,q}}$ & $\Ad_{R}, R\in \gro(p)\times \gro(q), R^T=R$ & \\
            \rc{tabgrey} $\qquad\nearrow\{$BDI,BDI$\}$ & \mr{$\mfso(n)$} & \mr{$Q\in \grso(n)$} &  $\Ad_{R'=Q^T R Q}$ &
            $R'^T=Q^T R^T Q=R', \det{R'}\in\{\pm 1\}$ & \mr{id} \\
            
            BDI$\to$DIII$^{\oplus 2}$ & & &
            $\Ad_{I_{2p,2q}}$ & $\Ad_{L}, L\in \grso(2p)\times \grso(2q), L^T=-L$ & \\
            $\qquad\nearrow\{$BDI,DIII$\}$ & \mr{$\mfso(2n)$} & \mr{$Q\in \grso(2n)$} &  $\Ad_{L'=Q^T L Q}$ &
            $L'^T=Q^T L^T Q=-L', \det{L'}=1$ & \mr{id} \\
            
            \rc{tabgrey} DIII$\to$AI & & &
            $\Ad_{J_n}$ & $\Ad_{V}\circ\,\ast, V\in \grsu(n), V^T=V$ & \\
            \rc{tabgrey} $\qquad\nearrow\{$DIII,BDI$\}$ & \mr{$\mfso(2n)$} & \mr{$Q\in \grso(2n)$} &  $\Ad_{R=Q^T \varphi_1(V) Q}$ &
            $R^T=Q^T\varphi_1(V)^T Q=R, \det{R}=(-1)^n$ & \mr{$\phi_\times$} \\
            
            DIII$\to$AII & & &
            $\Ad_{J_{2n}}$ & $\Ad_{W}\circ\,\ast, W\in \grsu(2n), W^T=-W$ & \\
            $\qquad\nearrow\{$DIII,DIII$\}$ & \mr{$\mfso(4n)$} & \mr{$Q\in \grso(4n)$} & $\Ad_{L=Q^T \varphi_1(W) Q}$ &
            $L^T=Q^T\varphi_1(W)^T Q=-L, \det{L}=1$ & \mr{$\phi_\times$} \\
            
            \rc{tabgrey} DIII$\to$AIII & & &
            $\Ad_{J_n}$ & $\Ad_{H}, H\in \gru(n), H^\dagger=H, \det{H}\in\{\pm1\}$ & \\
            \rc{tabgrey} $\qquad\nearrow\{$DIII,BDI$\}$ & \mr{$\mfso(2n)$} & \mr{$Q\in \grso(2n)$} &  $\Ad_{R=Q^T \varphi_0(H) Q}$ &
            $R^T=Q^T \varphi_0(H)^T Q=R, \det{R}=1$ & \mr{$\phi_\times$} \\
            \aboverulesepcolor{tabgrey}\midrule
            
            C$\to$CI & & &
            $\swapsymbol$ & $\Ad_S, S\in \grsp(n), S^\dagger=-S$ &\\
            $\qquad\nearrow\{$C,CI$^{\oplus 2}\}$ &\mr{$\mfsp(n)^{\oplus 2}$} & \mr{$Z\in \grsp(n)^{\times 2}$} &  $\Ad_{S'=Z^\dagger S^{\oplus 2}Z}$ &
            $S'^\dagger=Z^\dagger (S^\dagger)^{\oplus 2} Z=-S', \det{S'}=1$ & \mr{$\phi_\times$} \\
            
            \rc{tabgrey} C$\to$CII & & &
            $\swapsymbol$ & $\Ad_P, P\in \grsp(n), P^\dagger=P$ & \\
            \rc{tabgrey} $\qquad\nearrow\{$C,CII$^{\oplus 2}\}$ & \mr{$\mfsp(n)^{\oplus 2}$} & \mr{$Z\in \grsp(n)^{\times 2}$} &  $\Ad_{P'=Z^\dagger P^{\oplus 2}Z}$ &
            $P'^\dagger=Z^\dagger (P^\dagger)^{\oplus 2} Z=P', \det{P'}=1$ & \mr{$\phi_\times$} \\
            
            CI$\to$AI & & & 
            $\Ad_{J_{2n}}$ & $\Ad_{V}\circ\,\ast, V\in \grsu(2n), V^T=V$ & \\
            $\qquad\nearrow\{$CI,CII$\}$ & \mr{$\mfsp(2n)$} & \mr{$Z\in \grsp(2n)$} &  $\Ad_{P=Z^\dagger \varphi_1(V) Z}$ &
            $P^\dagger=Z^\dagger \varphi_1(V)^\dagger Z=P, \det{P}=1$ & \mr{$\phi_\times$} \\
            
            \rc{tabgrey} CI$\to$AII & & &
            $\Ad_{J_{2n}}$ & $\Ad_{W}\circ\,\ast, W\in \grsu(2n), W^T=-W$ & \\
            \rc{tabgrey} $\qquad\nearrow\{$CI,CI$\}$ & \mr{$\mfsp(2n)$} & \mr{$Z\in \grsp(2n)$} &  $\Ad_{S=Z^\dagger \varphi_1(W)Z}$ &
            $S^\dagger=Z^\dagger \varphi_1(W)^\dagger Z=-S, \det{S}=1$ & \mr{$\phi_\times$} \\
            
            CI$\to$AIII & & & 
            $\Ad_{J_n}$ & $\Ad_{H}, H\in \gru(n), H^\dagger=H, \det{H}\in\{\pm1\}$ & \\
            $\qquad\nearrow\{$CI,CII$\}$ & \mr{$\mfsp(n)$} & \mr{$Z\in \grsp(n)$} &  $\Ad_{P=Z^\dagger \varphi_0(H)Z}$ &
            $P^\dagger=Z^\dagger \varphi_0(H)^\dagger Z=P, \det{P}=1$ & \mr{$\phi_\times$} \\
            
            \rc{tabgrey} CII$\to$C & & &
            $\Ad_{K_{n,n}}$ & $\Ad_X\circ\swapsymbol, X\in \grsp(n)^{\times 2}, X^\dagger=X^{\smallswap}$ & \\
            \rc{tabgrey} $\qquad\nearrow\{$CII,CII$\}$ & \mr{$\mfsp(2n)$} & \mr{$Z\in \grsp(2n)$} &  $\Ad_{P=Z^\dagger X Z^{\smallswap}\smallswap}$ &
            $P^\dagger=\swapsymbol {Z^{\smallswap}}^\dagger X^\dagger Z=P, \det{P}=(-1)^{2n}=1$ & \mr{id} \\
            
            CII$\to$CI$^{\oplus 2}$ & & &
            $\Ad_{K_{p,q}}$ & $\Ad_S, S\in \grsp(p)\times \grsp(q),S^\dagger=-S$ & \\
            $\qquad\nearrow\{$CII,CI$\}$ & \mr{$\mfsp(n)$} & \mr{$Z\in \grsp(n)$} &  $\Ad_{S'=Z^\dagger S Z}$ &
            $S'^\dagger=Z^\dagger S^\dagger Z=-S', \det{S'}=1$ & \mr{id} \\
            
            \rc{tabgrey} CII$\to$CII$^{\oplus 2}$ & & &
            $\Ad_{K_{p,q}}$ & $\Ad_P, P\in \grsp(p)\times \grsp(q),P^\dagger=P$ & \\
            \rc{tabgrey} $\qquad\nearrow\{$CII,CII$\}$ & \mr{$\mfsp(n)$} & \mr{$Z\in \grsp(n)$} &  $\Ad_{P'=Z^\dagger P Z}$ &
            $P'^\dagger=Z^\dagger P^\dagger Z=P', \det{P'}=1$ & \mr{id} \\
        \end{tabular}%
        }
        \caption{How to lift any (homogeneous) $2$-recursive \ac{CD}, $(T_1\to T_2)\nearrow \{T_1,T_2'\}$, proving the base case for Thm.~\ref{thm:grading_recursion_equivalence}.~The proof requires us to verify that for compatible $\theta_{1}=\Ad_{G^{-1}}\theta_0 \Ad_G$ and $\theta_2$, where $G\in\exp(\mfg)$ expresses the basis choice of $\theta_1$,  the extension $\theta'_2$ of $\theta_2$ is a \ac{CD} on $\mfg$, which we calculate here.
        $\phi_\times$ is defined in Eq.~(\ref{eq:def_phi_times}), $\varphi_\tau$ for $\tau=0,1$ in Eq.~(\ref{eq:def_varphi_tau}).}
        \label{tab:2_recursive_lifts}
    \end{table*}
    
    \textbf{Details on remaining cases.}
    For AI$\to$DIII our calculation would yield the conjugation operator $H_0=U^\dagger L U$ and $\theta'_2=\Ad_{H_0}$, which satisfies $H_0^\dagger=-H$. However, there is no Cartan involution of this form on $\mfsu(n)$. We note that $\Ad_{iH_0}=\Ad_{H_0}$ and $(iH_0)^\dagger=iH_0$, so that we may use $H'=iH_0$ as conjugation operator and find $\theta'_2$ to be of type AIII.
    For AIII$\to$A (and similarly for BDI$\to$BD and CII$\to$C), our computation is
    \begin{align}
        \theta'_2=\Ad_{U^\dagger} \circ \Ad_X \circ \swapsymbol \circ \Ad_U
        = \Ad_{U^\dagger X U^{\smallswap}\smallswap}=\Ad_H,\quad H^\dagger = \swapsymbol {U^{\smallswap}}^\dagger X^\dagger U = U^\dagger {X^\dagger}^{\smallswap} U^{\smallswap}\swapsymbol=H,
    \end{align}
    where we used $X^\dagger=X^{\smallswap}$ and denoted in a slight misuse of notation the conjugation operator that gives rise to the map $\swapsymbol$ as $\swapsymbol$ again. Note that this SWAP operator has determinant $\det{\swapsymbol}=(-1)^n$.
    For cases starting with CI or DIII, refer to the explicit calculations for $\phi_\times$ above.  
\end{proof}

\textbf{Inhomogeneous recursions.}
In the proof above we calculated a lift for all homogeneous $r=2$-recursive \acp{CD}.
Here we briefly comment on the inhomogeneous case, which is excluded in the statement of Thm.~\ref{thm:grading_recursion_equivalence}.
There are three inhomogeneous involution types on $\mfsu(p)\oplus\mfsu(q)$, which emerges as the subalgebra of an AIII involution: AI$\oplus$AIII, AI$\oplus$AII, and AII$\oplus$AIII (note that AIII includes the identity involution for $p=n, q=0$ or vice versa).
The inhomogeneity implies that the former involutions cannot be extended to the (simple) algebra $\mfsu(p+q)$, because we know the form that all involutions on simple algebras take, and a partial complex conjugation or conjugation by a partially symmetric and partially antisymmetric operator are not among them.
The same holds for the orthogonal (symplectic) cases, with the indefinite involution type given by BDI$\oplus$DIII (CII$\oplus$CI) and the semisimple algebra arising from a BDI (CII) decomposition.
With this, we know that the following base cases for $r=2$ cannot be lifted to a pair of compatible \acp{CD} on the total algebra:
AIII$\to($AI$\oplus$AIII$)$,
AIII$\to($AII$\oplus$AIII$)$,
AIII$\to($AI$\oplus$AII$)$,
BDI$\to($DIII$\oplus$BDI$)$, and
CII$\to($CI$\oplus$CII$)$.

Conversely, there are in fact inhomogeneous recursions that can be lifted to a grading. Take, for example, the $3$-recursive \ac{CD} AIII$\to$(AI$\oplus$AI)$\to$(DIII$\oplus$BDI), which is inhomogeneous in the last involution.
If we lift the DIII and BDI decompositions to AIII decompositions, we find the set of two \acp{CD} $\{\text{AIII}\to\text{(AI$\oplus$AI)}, \text{AIII}\to\text{(AIII$\oplus$AIII)}\}$, which can be lifted further to a set of three \acp{CD} with types $\{\text{AIII},\text{AI},\text{AIII}\}$.

\subsection{Grading from Khaneja-Glaser decomposition}\label{sec:grading_for_literature}
Here we use our results from the previous section to obtain a grading from the \ac{KGD}.
This question has been considered in~\cite[Sec.~4.1]{dagli2008general}, but without considering the fact that the final step in the \ac{KGD} differs in type from the other steps in the recursion; see App.~\ref{sec:kgd_details}.
Specifically, the \ac{KGD} for $N$ qubits performs $N-2$ times the sequence AIII$\to$A and concludes with a single type-AI decomposition of $\gru(4)$, arriving at a $(2N-3)$-recursive \ac{CD}.
This means that we first need to lift A$\to$AI, resulting in AI$^{\oplus 2}$ alongside the original type-A decomposition of $\gru(4)\times\gru(4)$, i.e.,~two $(2N-4)$-recursive \acp{CD}.
Next, we lift AIII$\to$A and AIII$\to$AI$^{\oplus 2}$, resulting in AIII and AI, respectively. We obtain three $(2N-5)$-recursive \acp{CD} and observe that reducing the recursion depth by two produced two type-AIII \acp{CD} and reproduced the AI decomposition.
The last thing we need to note then is that lifting A$\to$AIII results in AIII$^{\oplus 2}$, which in turn is lifted to AIII.
This happens alongside the lift of the AI decomposition, which we can draw schematically as
\begin{align}
    \left\{(\mathrm{AIII}\to\mathrm{A})^{\to N-2}\to \mathrm{AI}\right\}
    \!\!\overset{\text{lift by }2}{\nearrow}\!\!
    \left\{\!\begin{array}{l} 
    (\mathrm{AIII}\to\mathrm{A})^{\to N-3}\to \mathrm{AI},\\
    (\mathrm{AIII}\to\mathrm{A})^{\to N-3}\to \mathrm{AIII},\\
    (\mathrm{AIII}\to\mathrm{A})^{\to N-3}\to \mathrm{AIII}
    \end{array}\!\right\}
    \!\!\overset{\text{lift by }2}{\nearrow}\!\!
    \left\{\!\begin{array}{l} 
    (\mathrm{AIII}\to\mathrm{A})^{\to N-4}\to \mathrm{AI},\\
    (\mathrm{AIII}\to\mathrm{A})^{\to N-4}\to \mathrm{AIII},\\
    (\mathrm{AIII}\to\mathrm{A})^{\to N-4}\to \mathrm{AIII},\\
    (\mathrm{AIII}\to\mathrm{A})^{\to N-4}\to \mathrm{AIII},\\
    (\mathrm{AIII}\to\mathrm{A})^{\to N-4}\to \mathrm{AIII}
    \end{array}\!\right\}
    \overset{\cdots}{\nearrow}
    \left\{\!\begin{array}{c} 
    \mathrm{AI},\\
    \mathrm{AIII},\\
    \vdots\\
    \mathrm{AIII}
    \end{array}\!\right\},\nonumber
\end{align}
with one type-AI and $(2N-4)$ type-AIII decompositions on the RHS.
The total number of \acp{CD} differs from the categorization in~\cite{dagli2008general} by $2$, because it assumes the recursion $(\text{AIII}\to\text{A})^{\to (N-1)}\to\text{AIII}$.
However, we may still use their first $2N-4$ \acp{CD} and simply replace their remaining three type-AIII \acp{CD} by one type-AI \ac{CD}.
The rotated AI decomposition in the magic basis is described easiest through its involution,
\begin{align}
    \theta_{\rm AI} = \Ad_{Y_{N-1}Y_{N}}\circ\,\ast.
\end{align}
It can trivially be extended from $\mfu(4)$ to the total algebra $\mfu(2^N)$ and maintains its type to be AI, as expected from our discussion above.
The \ac{CD} of $\mfu(2^N)$ to replace the three type-AIII \acp{CD} thus is
\begin{align}
    \mfk_{\rm AI} &= \spanR\left\{\mfu\!\left(2^{N-1}\right)\otimes \id, \mfu\!\left(2^{N-2}\right)\otimes \id\otimes\mfsu(2)\right\},\\
    \mfp_{\rm AI} &= \spanR\left\{\mfu\!\left(2^{N-2}\right)\otimes \{X, Y, Z\}\otimes \{X, Y, Z\}\right\}.
\end{align}

\section{Details on numerical \kak decompositions}\label{sec:numerical_details}
In this section, we discuss our unified perspective on numerical \kak decomposition algorithms from Sec.~\ref{sec:numerical_decomps} and collect important details for the practical implementation provided in~\cite{symmetrycompilationrepo}.
The key statements are summarized in Thm.~\ref{thm:abstract_numerical_non_general}, the proof of which tells us how to use an \ac{EVD} of the relative complex structure $\Delta$ to construct a \kak decomposition of the original group element $G\in\groupG$.
Before we prove the theorem, we will introduce some helper objects and show some of their useful properties.
We begin with some computations with involutions in the standard representation.

\begin{lemma}\label{lemma:conjugate_eigenvector}
    Let $\groupP=\groupG/\groupK$ be a symmetric space from Tab.~\ref{tab:symm_classif} with involution $\Theta$ in its defining representation. Then $\Theta=\Ad_U\circ\ \ast^\tau$, with $\tau\in\{0,1\}$ and we define $\varphi:\mbc^n\to\mbc^n, v\mapsto U v^{\ast^\tau}$.
    Then for any group element $K\in\groupK$ ($P\in\groupP$) with eigenvectors $v_j$ for the eigenvalues $\lambda_j$, $\varphi(v_j)$ are eigenvectors as well, with eigenvalues $\lambda_j^{\ast^\tau}$ ($\lambda_j^{\ast^{1-\tau}}$).
    Note that $U\in\groupG$ for types AI, AII, DIII, CI and CII, $U\in\gru(p+q)$ with $\det{U}=\pm 1$ for type AIII, $U\in\gro(p+q)$ for BDI, and $U=U'\swapsymbol$ with $U'\in\groupG$ for types A, BD and C.
\end{lemma}
\begin{proof}
    First, note that, for any $v,w\in\mbc^n$ and any orthonormal basis $\{a_i\}_i$,
    \begin{align}
        \tr[\Theta(G)]&=\tr[U G^{\ast^\tau} U^\dagger]=\tr[G]^{\ast^\tau},\label{eq:trace_and_Theta}\\
        \varphi(v)\varphi(w)^\dagger&=\Theta(vw^\dagger), \label{eq:connection_varphi_Theta}\\
        \varphi(a_i)^\dagger\varphi(a_j)
        &=\tr[\varphi(a_j)\varphi(a_i)^\dagger]
        \overset{(\ref{eq:connection_varphi_Theta})}{=}\tr[\Theta(a_ja_i^\dagger)]
        \overset{(\ref{eq:trace_and_Theta})}{=}\tr[a_ja_i^\dagger]^{\ast^\tau}
        =(a_i^\dagger a_j)^{\ast^\tau}
        =\delta_{ij}.\label{eq:varphi_is_isometry}
    \end{align}
    Note as well that $\Theta$ of this form is a group homomorphism, i.e.,~$\Theta(AB) = \Theta(A)\Theta(B)$.
    We perform the proof simultaneously for $G\in\groupK$ ($G\in\groupP$).
    $G\in\groupK$ ($G\in\groupP$) implies $\Theta(G)=G$ ($\Theta(G)=G^\dagger$) and thus, for an orthonormal eigenbasis $\{v_j\}$ with eigenvalues $\{\lambda_j\}_j$,
    \begin{align}
        \varphi(v_i)^\dagger G \varphi(v_j)
        &=\tr[G \varphi(v_j)\varphi(v_i)^\dagger]
        \overset{(\ref{eq:connection_varphi_Theta})}{=}\tr[\Theta\left(\Theta(G) v_jv_i^\dagger\right)]
        \overset{(\ref{eq:trace_and_Theta})}{=}\tr[G^{(\dagger)} v_jv_i^\dagger]^{\ast^\tau}
        \overset{Gv_j=\lambda_j v_j}{=}\tr[\lambda^{(\ast)} v_jv_i^\dagger]^{\ast^\tau}\nonumber\\
        &=\lambda^{\ast^{\tau(+1)}}(v_i^\dagger v_j)^{\ast^\tau}.
    \end{align}
    The orthonormality of $\{v_j\}_j$ and $\{\varphi(v_j)\}_j$ (via Eq.~(\ref{eq:varphi_is_isometry})) then yields that $G \varphi(v_j)=\lambda^{\ast^{\tau(+1)}}\varphi(v_j)$, i.e.,~$\varphi(v_j)$ is again an eigenvector of $G$, with eigenvalue $\lambda_j^{\ast^{\tau}}$ ($\lambda_j^{\ast^{\tau+1}}=\lambda_j^{\ast^{1-\tau}}$).
\end{proof}

Next, we will prove two lemmas that we will use throughout the construction of the numerical \kak decompositions below.
We begin by re-proving~\cite[Lemma~5.2]{fuhr2018note}, in order to connect the general perspective in the lemma above to a specific involution.

\begin{lemma}\label{lemma:ai_orthogonal}
    Let $\Delta$ be symmetric unitary, i.e.,~$\Delta\in\gru(n)$ and $\Delta^T=\Delta$. Then a real orthogonal eigenbasis can be obtained from any eigenbasis of $\Delta$.
\end{lemma}
\begin{proof}
    Note that $\Delta^T=\Delta$ implies $\Delta^\ast=\Delta^\dagger$, so that $\Delta\in\groupP$ with respect to the generic type-AI involution $\Theta=\ast$. Thus, the ``$\groupP$-case" of Lemma~\ref{lemma:conjugate_eigenvector} can be applied to $\Delta$, telling us that for any eigenvector $v$ of $\Delta$ with eigenvalue $\lambda$, $v^\ast$ is an eigenvector with eigenvalue $\lambda$ as well. In the following, assume that $v$ is normalized.
    
    If $v^\ast\propto v$, let $\tilde{v}=v/v_0$, where $v_0$ is the first non-zero entry of $v$.
    Then $\tilde{v}^\ast$ is an eigenvector with $\tilde{v}^\ast\propto \tilde{v}$, still, which actually implies $\tilde{v}^\ast=\tilde{v}$, so that $\tilde{v}$ is evidently real.
    If $v^\ast\not\propto v$, we define $\tilde{v}\coloneqq e^{i\alpha}v$, with $\alpha$ a yet-to-be-specified phase, and calculate the overlap of $\tilde{v}$ with its complex conjugate,
    \begin{align}
        (\tilde{v}^\ast)^\dagger \tilde{v}=\tilde{v}^T \tilde{v}=e^{2i\alpha}v^Tv \eqqcolon e^{2i\alpha}(a+ib).
    \end{align}
    Then, we select the phase to be $\alpha=-\tfrac{1}{2}\arctan(b/a)$ for $a\neq 0$ and to $\tfrac{\pi}{4}$ for $a=0$, so that the above overlap is purely real, i.e., $\imag(\tilde{v}^T \tilde{v})=0$. The following new vectors then form an orthonormal basis of the two-dimensional eigenspace $\spanC\{v,v^\ast\}$ with eigenvalue $\lambda$,
    \begin{align}
        v'_\pm =\sqrt{\pm 1} \frac{\tilde{v} \pm \tilde{v}^\ast}{\|\tilde{v}\pm \tilde{v}^\ast\|_2},
        \quad \|\tilde{v}'_\pm\|_2=\frac{\|\tilde{v}\pm \tilde{v}^\ast\|_2}{\|\tilde{v}\pm \tilde{v}^\ast\|_2}=1,\quad (\tilde{v}'_+)^\dagger \tilde{v}'_-\propto\underset{1}{\underbrace{\tilde{v}^\dagger \tilde{v}}}\underset{2 \imag(\tilde{v}^T \tilde{v})=0}{\underbrace{-\tilde{v}^\dagger \tilde{v}^\ast+\tilde{v}^T\tilde{v}}}-\underset{1}{\underbrace{\tilde{v}^T \tilde{v}^\ast}}=0.
    \end{align}
    Further, these vectors are real-valued,
    \begin{align}
        {v'}^\ast_\pm=\pm\sqrt{\pm 1}\frac{\tilde{v}^\ast \pm \tilde{v}}{\|\tilde{v}\pm \tilde{v}^\ast\|_2}=v'_\pm,
    \end{align}
    where we used $\sqrt{\pm 1}^\ast=\pm\sqrt{\pm 1}$.
    Overall, we thus need to filter a generic eigenbasis for eigenvectors that additionally are linearly independent under complex conjugation, apply global phases to them, and recombine any $v$ that has $v^\ast\not\propto v$ with its complex conjugate $v^\ast$.
\end{proof}

Note that while the next lemma makes a statement about $\groupK$ and $\groupP$ in the symplectic case, the parentheses are swapped, i.e.,~the expressions in parentheses correspond to $\groupK$, not $\groupP$ like in Lemma~\ref{lemma:conjugate_eigenvector}.

\begin{lemma}\label{lemma:aii_symplectic}
    Let $\Delta$ be a $J_n$-symmetric (a $J_n$-anti-symmetric) unitary, i.e.,~$\Delta\in\gru(n)$ and $J_n\Delta^T J_n^T=\Delta$ ($J_n\Delta^T J_n^T=\Delta^\dagger$) . Then a symplectic eigenbasis can be obtained from any eigenbasis of $\Delta$.
\end{lemma}
\begin{proof}
    We begin with the $J_n$-symmetric case.
    Note that $J_n\Delta^T J_n^T=\Delta$ implies $J_n\Delta^\ast J_n^T=\Delta^\dagger$, so that for the involution $\Theta=\Ad_{J_n}\circ\, \ast$ (of type AII), we have $\Delta\in\groupP$, and we can use Lemma~\ref{lemma:conjugate_eigenvector}.
    It tells us that for any eigenvector $v$ with eigenvalue $\lambda$, $J_nv^\ast$ is an eigenvector with eigenvalue $\lambda$ as well.
    Writing $v=u\oplus w$, compute
    \begin{align}\label{eq:orthogonal_conjugate_vector_sp}
        v^\dagger J_n v^\ast=u^\dagger w^\ast - w^\dagger u^\ast=\sum_{j=1}^n u_j^\ast w_j^\ast-w_j^\ast u_j^\ast=0,
    \end{align}
    so that the conjugate eigenvectors $\{v, J_nv^\ast\}$ with eigenvalue $\lambda$ form an orthonormal basis.
    If we arrange all pairs as unitary matrix $(J_nv^\ast_1, \dots, J_nv^\ast_n, v_1, \dots, v_n)$, it is additionally symplectic with respect to $J_n$,
    \begin{align}
        \begin{pmatrix}
        \rowrule & v^\dagger_1 J_n^T & \rowrule \\ 
        & \vdots & \\
        \rowrule & v^\dagger_n J_n^T & \rowrule \\ 
        \rowrule & v_1^T & \rowrule \\ 
        & \vdots & \\
        \rowrule & v_n^T & \rowrule \\ 
        \end{pmatrix}
        J_n
        \begin{pmatrix}
            | &  &| & | &  & | \\
            J_nv^\ast_1 & \dots & J_nv^\ast_n & v_1 & \dots & v_n \\
            | &  &| & | & & | 
        \end{pmatrix}
        =
        \begin{pmatrix}
            v_i^\dagger J_n^T J_n J_n v_j^\ast & v_i^\dagger J_n^T J_n v_j \\
            v_i^T J_n J_n v^\ast_j & v_i^T J_n v_j \\
        \end{pmatrix}
        =
        \begin{pmatrix}
            0 & \id_n \\
            -\id_n & 0 \\
        \end{pmatrix}
        =
        J_n,
    \end{align}
    where we used Eq.~(\ref{eq:orthogonal_conjugate_vector_sp}) for the block-diagonal and $J_n^TJ_n=\id=-J_n^2$ as well as $v_i^\dagger v_j=\delta_{ij}$ for the block-off-diagonal entries.

    For the $J_n$-anti-symmetric case, note that $J_n\Delta^T J_n^T=\Delta^\dagger$ implies $J_n\Delta^\ast J_n^T=\Delta$, so that Lemma~\ref{lemma:conjugate_eigenvector} applies as before, but with $\Delta\in\groupK$. This implies that for $v$ an eigenvector with eigenvalue $\lambda$, $J_nv^\ast$ is an eigenvector with eigenvalue $\lambda^\ast$. Everything else remains the same, so that we can use the same symplectic basis as above.

\end{proof}

Throughout the proof below, we will refer to the \emph{relative complex structure} $\Delta\coloneqq G\Theta(G)^\dagger$ of a group element $G\in\groupG$ and a Cartan involution (on the group), $\Theta$.
It will be relevant that $\Delta$ extracts the horizontal part of $G$, in the sense of a KP decomposition $G=PK$ (c.f.~Thm.~\ref{thm:kp_decomp}),
\begin{align}
    \Delta
    =G \Theta(G)^\dagger
    =PK \Theta(PK)^\dagger
    =PK(P^\dagger K)^\dagger
    =P^2.
\end{align}

The following is a useful generalization of a computation in~\cite{fuhr2018note}:
\begin{lemma}\label{lemma:horizontal_to_kak_decomp}
    Let $\groupP=\groupG/\groupK$ be a symmetric space from Tab.~\ref{tab:symm_classif} with \ac{CSG} $\groupA$, and let $G\in\groupG$ with relative complex structure $\Delta\in\groupP$.
    Then a horizontal \kak decomposition (c.f.~Prop.~\ref{prop:horizontal_kak_decomp}) $\Delta=K_1A^2K_1^\dagger$ with $K_1\in\groupK$ and $A\in\groupA$ implies a \kak decomposition, $G=K_1AK_2$, with $K_2=A^\dagger K_1^\dagger G\in\groupK$.
\end{lemma}
\begin{proof}
    Given a horizontal decomposition $\Delta=K_1A^2K_1^\dagger$, first note that we can compute a (non-unique) square root $A$ of $A^2$ for any of the \acp{CSG} in Tab.~\ref{tab:numerical_overview}. Such a root can be mapped to other bases and other representations of $\groupA$ via basis changes and group isomorphisms, respectively.
    Next, compute for $K_2=A^\dagger K_1^\dagger G$,
    \begin{align}
        \Theta(K_2)^\dagger K_2
        &=\big(\underset{A}{\underbrace{\Theta(A^\dagger)}}\underset{K_1^\dagger}{\underbrace{\Theta(K_1^\dagger)}}\Theta(G) \big)^\dagger A^\dagger K_1^\dagger G
        =\Theta(G)^\dagger \underset{(K_1A^2K_1^\dagger)^\dagger=\Delta^\dagger}{\underbrace{(K_1 (A^2)^\dagger K_1^\dagger)}} G
        =\Theta(G)^\dagger \underset{\Theta(G) G^\dagger}{\underbrace{\Delta^\dagger}} G
        =\id.
    \end{align}
    This implies that $\Theta(K_2)^\dagger=K_2^\dagger$, and thus $\Theta(K_2)=K_2$.
    Next, define $\overline{\groupK}\coloneqq\{G\in\groupG|\Theta(G)=G\}$ so that $K_2\in\overline{\groupK}$. We have $\groupK\subset\overline{\groupK}$, but in general $\groupK\neq\overline{\groupK}$, as $\overline{\groupK}$ can consist of multiple connected components, each diffeomorphic to $\groupK$.
    The spaces $\groupP$, $\groupG$ and $\groupK$ are (path-)connected and contain the identity, allowing us to create smooth paths $t\mapsto U^t$ for $t\in[0, 1]$ and $U$ from any of those three spaces.
    Multiplying such paths together, we define $\phi: [0, 1]\to \overline{\groupK}, t\mapsto A^{-t}K_1^{(1-2t)}G^t$.
    We have $\phi(0)=K_1$ and $\phi(1)=A^\dagger K_1^\dagger G=K_2$, so that $\phi$ connects $K_1$ and $K_2$, implying that they both must be in the same connected component of $\overline{\groupK}$, namely $\groupK$.
\end{proof}

With these general facts in our hands, we are ready to prove the main theorem, which we restate for convenience.

\abstractnumericalnongeneral*
\begin{proof}
    Throughout, we assume that eigenvalue decompositions are performed over the complex numbers and obtained eigenvectors are orthonormal.
    
    \textbf{A.} We have $\groupG=\gru(n)\times\gru(n)$, $\groupK=\{K\oplus K|K\in\gru(n)\}$ and $\groupA=\{D\oplus D^\dagger|D\in\grudiag(n)\}$.\\\noindent
    This decomposition is constructed in the proof of~\cite[Thm.~12]{shende2005synthesis}. As we will adapt it for types BD and C below, we recite it here.
    Given a group element $G=U\oplus U'\in\groupG$, compute $\delta=UU'^\dagger\in\gru(n)$\footnote{Note that $\Delta=(U\oplus U')(U'\oplus U)^\dagger=\delta\oplus\delta^\dagger$.}.
    A standard \ac{EVD} of $\delta$ yields $\delta=U_1 D^2 U_1^\dagger$, with $U_1\in\gru(n)$ and $D^2\in\grudiag(n)$.
    We compute a (non-unique) square root $D$ from $D^2$ and set $U_2=DU_1^\dagger U'$. Then
    \begin{align}\label{eq:num_type_A_proof}
        (U_1\oplus U_1) (D\oplus D^\dagger)(U_2\oplus U_2)
        =(\underset{\delta=UU'^\dagger}{\underbrace{U_1DDU_1^\dagger}} U')\oplus(U_1D^\dagger DU_1^\dagger U')
        =U\oplus U'=G,
    \end{align}
    which is the sought-after \kak decomposition of $G$.
    
    \textbf{AI.} We have $\groupG=\gru(n)$, $\groupK=\grso(n)$ and $\groupA=\grudiag(n)$.\\\noindent
    This proof very closely follows~\cite[Thm.~5.1]{fuhr2018note}, but restricts $\groupK$ from $\gro(n)$ to $\grso(n)$ and puts it into a broader context using $\Delta$ and Lemma.~\ref{lemma:horizontal_to_kak_decomp}.
    Given $U\in\gru(n)$, we first extract a global phase and set $\tilde{U}=\sqrt[n]{\det{U}}^{-1}U\in\grsu(n)$, where $\sqrt[n]{\det{U}}$ is any $n$th root of the determinant.
    Then, compute $\Delta=\tilde{U}\Theta(\tilde{U})^\dagger=\tilde{U}\tilde{U}^T$, where we used that $\Theta=\ast$ in the generic case.
    Next, we compute the eigenvectors of $\Delta$, which we could assume to be real-valued because $\Delta^T=\Delta$ is symmetric. However, if we only assume access to generic eigenvalue decompositions, an explicit orthogonalization is provided in Lemma~\ref{lemma:ai_orthogonal} above.
    Now we have a horizontal decomposition of $\Delta=O_1 \tilde{D}^2 O_1^\dagger$ with $O_1\in\grso(n)$, where we can guarantee $\det{O_1}=1$ by flipping the sign of any eigenvector, if necessary. After computing a (non-unique) square root $\tilde{D}$ of $\tilde{D}^2$ with $\det{\tilde{D}}=1$\footnote{Note that $\det{\tilde{D}^2}=\det{\Delta}=\det{\tilde{U}}^2=1$.}, Lemma~\ref{lemma:horizontal_to_kak_decomp} tells us that $O_2\coloneqq \tilde{D}^\dagger O_1^\dagger \tilde{U} \in\grso(n)$. Reintroducing the global phase via $D=\sqrt[n]{\det{U}} \tilde{D}$, we obtain the desired \kak decomposition,
    \begin{align}
        O_1 D O_2=\sqrt[n]{\det{U}} O_1 \tilde{D} \tilde{D}^\dagger O_1^\dagger \tilde{U}=\sqrt[n]{\det{U}}\tilde{U}=U.
    \end{align}
    
    Our perspective on this proof not only makes the subgroup more precise by setting $\groupK=\grso(n)\subset\gro(n)$ and connects it to other involution types, but it also allows us to easily generalize to type-AI decompositions in any other basis of $\gru(n)$.
    
    \textbf{AII.} We have $\groupG=\gru(2n)$, $\groupK=\grsp(n)$ and $\groupA=\{D\oplus D|D\in\grudiag(n)\}$.\\\noindent
    A type-AII decomposition was constructed in \cite[Sec.~IV~A]{bullock2005time}, using a decomposition of a $J_n$-anti-symmetric Lie algebra element, employing a QR decomposition along the way. Our proof provides a different path to the decomposition, following that for type AI instead.
    
    Given $U\in\gru(2n)$, we again obtain $\tilde{U}\in\grsu(2n)$ by extracting $\sqrt[2n]{\det{U}}$ and compute $\Delta=\tilde{U}\Theta(\tilde{U})^\dagger=\tilde{U}J_n\tilde{U}^TJ_n^T$, which follows from $\Theta=\Ad_{J_n}\circ\,\ast$.
    Next, we compute the eigenvectors of $\Delta$ and make them symplectic via Lemma~\ref{lemma:aii_symplectic}.
    Then we have a horizontal decomposition, $\Delta=S_1 (\tilde{D}^2\oplus \tilde{D}^2) S_1^\dagger$ with $S_1\in\grsp(n)$ and compute a (non-unique) square root $\tilde{D}$ of $\tilde{D}^2$ with $\det{\tilde{D}}=1$.
    Applying Lemma~\ref{lemma:horizontal_to_kak_decomp} tells us that $S_2\coloneqq (\tilde{D}\oplus \tilde{D})^\dagger S_1^\dagger \tilde{U} \in\grsp(n)$. Reintroducing the global phase via $D=\sqrt[2n]{\det{U}} \tilde{D}$, we obtain the desired \kak decomposition,
    \begin{align}
        S_1 (D\oplus D) S_2=\sqrt[2n]{\det{U}} S_1 (\tilde{D}\oplus \tilde{D}) (\tilde{D}\oplus\tilde{D})^\dagger S_1^\dagger \tilde{U}=\sqrt[2n]{\det{U}}\tilde{U}=U.
    \end{align}
    
    \textbf{AIII.} We have $\groupG=\gru(p+q)$, $\groupK=\gru(p)\times\gru(q)$ and $\groupA=\grcs(p,q)$.\\\noindent
    As was noted repeatedly in the literature~\cite{bullock2005time,fuhr2018note,edelman2023fifty}, the type-AIII \kak decomposition can be implemented with a \ac{CSD}, which we here consider an elementary building block.
    We just remark that the implementation \texttt{scipy.linalg.cossin} in SciPy~\cite{virtanen2020scipy} based on~\cite{sutton2009computing}, which we use in our code, requires a minor modification for the case $p>q$ to adapt it to our choice of \ac{CSG}. Given $U=K_1 FK_2$, this is achieved by exchanging the first $q$ with the following $p-q$ columns (rows / rows and columns) of $K_1$ ($K_2$ / $F$), which leaves the structure of $K_1$ and $K_2$ intact and moves $F$ into our \ac{CSG}.

    \textbf{BD.} We have $\groupG=\grso(n)\times\grso(n)$, $\groupK=\{K\oplus K|K\in\grso(n)\}$ and $\groupA=\{\mu\oplus \mu^T|\mu\in\grschur(n)\}$.\\\noindent
    We will proceed analogously to the construction for type A, but using a Schur decomposition instead of an \ac{EVD}.
    Given $G=O\oplus O'\in\groupG$, compute $\delta=OO'^T$ and perform a real-valued Schur decomposition $\delta=O_1 \mu^2 O_1^\dagger$ with $O_1\in \gro(n)$ and $\mu^2\in\grschur(n)$.
    If $\det{O_1}=-1$, swap two columns of $O_1$ and the according rows \emph{and} columns of $\mu$, so that $O_1\in\grso(n)$ and $\mu\in\grschur(n)$.
    Then, compute a (non-unique) square root $\mu$ of $\mu^2$ and set $O_2=\mu O_1^\dagger O'$, so that $\det{O_2}=\det{\mu}\det{O_1}^{-1}\det{O'}=1$, because Schur matrices have determinant $1$, $\det{O_1}=1$ was guaranteed manually, and $O'\in\grso(n)$. Analogously to Eq.~(\ref{eq:num_type_A_proof}), we find the \kak decomposition $(O_1\oplus O_1)(\mu\oplus\mu^T)(O_2\oplus O_2)$ of $G=O\oplus O'$.
    
    \textbf{BDI.} We have $\groupG=\grso(p+q)$, $\groupK=\grso(p)\times\grso(q)$ and $\groupA=\grcs(p,q)$.\\\noindent
    The construction is analogous to that for AIII, simply using a \ac{CSD}. The available implementation in SciPy returns real-valued, i.e.,~orthogonal matrices if provided with a real-valued input.
    Similar to AIII, we reorder the blocks to adapt to our choice of \ac{CSA}, finding
    \begin{align}
        O = K_1 F K_2 = (O_{1,p}\oplus O_{1,q}) F (O_{2,p}\oplus O_{2,q}),
    \end{align}
    with $O_{i,p}\in\gro(p)$, $O_{i,q}\in\gro(q)$ and $F\in\grcs(p,q)$.
    To conclude, we need to adjust the determinants of $O_{i,p}$ and $O_{i,q}$ to one.
    Happily, this is not difficult to do. First, we note that we have
    \begin{equation}\label{eq:detconst}
        1 = \det{G} = \det{O_{1,p}}\det{O_{1,q}}\det{F}\det{O_{2,p}}\det{O_{2,q}},
    \end{equation}
    with $\det{F}=1$ because $F\in\grcs(p,q)$, so that the only thing that can go wrong is that pair(s) of the above matrices $O_{i,r}$ have determinant $-1$.
    We may fix those by multiplying the first column (row) of $O_{1,r}$ ($O_{2,r}$) by the determinant of the respective $O_{i,r}$ and applying the same sign changes to $F$:
    \begin{align}
        G = (O_{1,p}\oplus O_{1,q}) d_1^2 F d_2^2(O_{2,p}\oplus O_{2,q})
        =(\widetilde{O}_{1,p}\oplus \widetilde{O}_{1,q}) \widetilde{F} (\widetilde{O}_{2,p}\oplus \widetilde{O}_{2,q}),
    \end{align}
    where $d_i=\diag (d_{i,p},1,\dots, 1,d_{i,q},1,\dots, 1)$, which squares to $\id$, with the determinants $d_{i,r}=\det{O_{i,r}}$ at positions $\ell_0=1$ and $\ell_1=\max(p,q)+1$.
    Clearly, we achieved $\det{\widetilde{O}_{i,r}}=1$, and the transformed \ac{CSG} element satisfies $\widetilde{F}\in\grcs(p,q)$ as only its submatrix corresponding to rows/columns $(\ell_0,\ell_1)$ is transformed via
    \begin{align}
        \widetilde{F}\big|_{(\ell_0,\ell_1)}=
        \begin{pmatrix} 
            d_{1,p}d_{2,p}\cos(\theta_1) & d_{1,p}d_{2,q}\sin(\theta_1) \\ -d_{1,q}d_{2,p}\sin(\theta_1) & d_{1,q}d_{2,q}\cos(\theta_1) 
        \end{pmatrix}
        =\begin{pmatrix} 
            d_{1,p}d_{2,p}\cos(\theta_1) & d_{1,p}d_{2,q}\sin(\theta_1) \\ -d_{1,p}d_{2,q}\sin(\theta_1) & d_{1,p}d_{2,p}\cos(\theta_1) 
        \end{pmatrix}
        =\begin{pmatrix} 
            \cos(\widetilde{\theta}_1) & \sin(\widetilde{\theta}_1) \\ -\sin(\widetilde{\theta}_1) & \cos(\widetilde{\theta}_1) 
        \end{pmatrix},\nonumber
    \end{align}
    with $\widetilde{\theta}_1=\tfrac{1}{2}(1-d_{1,p}d_{2,p})\pi +d_{1,p}d_{1,q}\theta_1$.
    
    \textbf{DIII.} We have $\groupG=\grso(2n)$, $\groupK=\gru(n)$ and $\groupA=\{\mu\oplus \mu^T|\mu\in\grschur(n)\}$.\\\noindent
    Looking at the canonical DIII involution $\Theta = \Ad_{J_n}$, we see that this decomposition is analogous to the type-AII decomposition, but restricted to $\grso(2n)$, so that we need to construct a symplectic eigenbasis like for AII while staying within the orthogonal group.
    Given $O\in\grso(2n)$, compute $\Delta=O\Theta(O)^T=OJ_nO^TJ_n^T$. 
    Applying Lemma~\ref{lemma:aii_symplectic}, we find a symplectic basis $\{v_k,J_nv_k^\ast\}_k$ of eigenvectors of $\Delta$ with eigenvalues $\lambda_k$.
    As $\Delta^\ast=\Delta$, Lemma~\ref{lemma:conjugate_eigenvector} implies that these pairs come with (mutually orthogonal) conjugate eigenvectors $\{v^\ast_{k}, J_n v_k\}_k$ with eigenvalues $\lambda_k^\ast$.
    In the following, we will create a real-valued symplectic basis, distinguishing the cases $v^\ast_k\not\propto v_k$ and $v_k^\ast\propto v_k$.
    
    If $v_k^\ast\not\propto v_k$, we may make the vectors $v_k$ and $v_k^\ast$ orthogonal, as we did in Lemma~\ref{lemma:ai_orthogonal}, and combine the quadruple $(v_k, J_n v_k^\ast, v_k^\ast, J_n v_k)$ into a \emph{real} symplectic basis of the spanned four-dimensional space,
    \begin{align}\label{eq:recombine_DIII}
        \begin{array}{rl}
            w_{k,\alpha}&=\begin{cases}
                \tfrac{1}{\sqrt{2}}(J_nv_k^\ast+J_nv_k) & \text{ for }\ \alpha=0\\
                \tfrac{i}{\sqrt{2}}(J_nv_k^\ast-J_nv_k) & \text{ for }\ \alpha=1\\
                \tfrac{1}{\sqrt{2}}(v_k^\ast+v_k) & \text{ for }\ \alpha=2\\
                \tfrac{i}{\sqrt{2}}(v_k^\ast-v_k) & \text{ for }\ \alpha=3\\
            \end{cases}\ \ ,\\
            & \\ & \\
        \end{array}
        \quad
        \begin{array}{rl}
            w_{k,\alpha}^\ast
            &=w_{k,\alpha},\\
            w_{k,\alpha}^\dagger w_{\ell,\beta}
            &=\delta_{k\ell}\delta_{\alpha\beta},\\
            w_{k,\alpha}^T J_n w_{\ell,\beta}
            &=(-1)^{\beta\leq 1} w_{k,\alpha}^\dagger w_{\ell,\bar\beta}\\
            &=(-1)^{\beta\leq 1} \delta_{k\ell}\delta_{\alpha,\bar\beta}
            =\delta_{k\ell} \begin{pmatrix}
                0 & 0 & 1 & 0 \\
                0 & 0 & 0 & 1 \\
                -1 & 0 & 0 & 0 \\
                0 & -1 & 0 & 0 \\
            \end{pmatrix},
        \end{array}
    \end{align}
    where we used orthonormality of $\{v_k,v_k^\ast,J_nv_k,J_nv_k^\ast\}$ and calculated $J_nw_{k,\beta}=(-1)^{\beta\leq 1}w_{k,\bar\beta}$ with $\bar\beta=(\beta+2)\!\!\mod 4$.
    We denote the number of eigenspaces for which $v_k\perp v_k^\ast$ as $d$, with $d\leq\frac{n}{2}$.
    
    On the other hand, if $v_k^\ast \propto v_k$, we have $J_nv_k\propto J_nv_k^\ast$ as well, and like in the proof of Lemma~\ref{lemma:ai_orthogonal}, we may adjust the global phase of $v_k$ to obtain two orthogonal real symplectic vectors $v'_\ell, J_n v'_\ell$.
    Note that this can only happen if $\lambda_k^\ast=\lambda_k$, i.e.,~$\lambda_k\in\{\pm 1\}$, because eigenvectors of differing eigenvalues are orthogonal. We sort the vectors $v'_\ell$ such that $\Delta v'_\ell=-v'_\ell$ for $1\leq \ell\leq f$, i.e.,~the $f$ eigenvectors with eigenvalue $-1$ appear first. We note that such degenerate pairs in turn must come in eigenvalue pairs, i.e.,~for each $(v'_\ell, J_n v'_\ell)$ with $\lambda_\ell=\pm 1$, there is another pair $(v'_m,J_nv_m')$ with the same eigenvalue\footnote{This can be seen from the fact that if we slightly perturb the original matrix such that the eigenvectors $v_\ell$ and $v_\ell^\ast$ are no longer degenerate, we will produce a quadruple that has complex conjugate eigenvalues. Continuity of the eigenvalue computation then tells us that the eigenvalues of these pairs must match if we move back to the degenerate point.}. This guarantees that $f$ is even.
    
    Symplecticity with respect to $J_n$ is then achieved by sorting the basis vectors as 
    \begin{align}
        (w_{1,0}, w_{1, 1}, w_{2,0},\dots,w_{d,1},\ |\ v'_1,\dots v'_{n-2d},\ \ |\ \ w_{1,2}, w_{1,3}, w_{2,2},\dots,w_{d,3},\ |\  J_nv'_1,\dots J_nv'_{n-2d}).
    \end{align}
    
    The $w_{k,\alpha}$ form a $4d$-dimensional subspace. For each fixed $k$, $\Delta$ takes the form of the following $4\times4$ matrix:
    \begin{align}
        \Delta_{k,\alpha\beta}
        &=w_{k,\alpha}^\dagger \Delta w_{k,\beta}
        =w_{k,\alpha}^\dagger \real(\lambda_k) w_{k,\beta}+w_{k,\alpha}^\dagger \imag(\lambda_k) (-1)^{\beta+(\beta\geq 2)}w_{k,\beta'}\nonumber\\
        &=
        \begin{pmatrix}
            \real(\lambda_k) & -\imag(\lambda_k) & 0 & 0 \\
            \imag(\lambda_k) & \real(\lambda_k) & 0 & 0 \\
            0 & 0 & \real(\lambda_k) & \imag(\lambda_k)\\
            0 & 0 & -\imag(\lambda_k) & \real(\lambda_k)\\
        \end{pmatrix}_{\alpha\beta},
    \end{align}
    where we implied $(-1)^{\beta\geq 2}$ to be $-1$ if $\beta\geq 2$ and 1 otherwise, and we computed $\Delta w_{k,\beta}=\real(\lambda_k)w_{k,\beta}+\imag(\lambda_k)(-1)^{\beta+(\beta\geq 2)}w_{k,\beta'}$ with $\beta'=(1-\beta)\!\!\mod 4$.
    Note that the two non-zero $2\times 2$ block matrices above are the transpose of each other and have the structure of a Schur block matrix because $\lambda_k=\exp(i\theta_k)$ for some angle $\theta_k$, and thus $\real{\lambda_k}=\cos(\theta_k)$ as well as $\imag{\lambda_k}=\sin(\theta_k)$.

    Now, investigating the remaining basis vectors given by $\{v'_\ell, J_nv'_\ell\}$, we will have $n-2d$ blocks with shape $2\times2$. $\Delta$ has the eigenvalue $\lambda_\ell=-1$ on $f$ of these blocks, and $\lambda_\ell=1$ on the other blocks. The diagonal entries with $-1$ ($1$) can be interpreted as $\cos(\pi)$ ($\cos(0)=1$), with corresponding off-diagonal entries $\sin(\pi)=0$ ($\sin(0)=0$).
    As we found $f$ to be even, these cosine entries can be paired up into Schur blocks, which evidently are symmetric as they are real diagonal.
    
    Note that our arrangement of the basis guarantees that the first and second $n\times n$ block matrices are the transpose of each other. We made this explicit in the subspace spanned by the $\{w_{k,\alpha}\}$ and manually sorted the $\{v'_\ell, J_nv'_\ell\}$ such that this is true.
    
    Putting the pieces together, we thus have a horizontal \kak decomposition $\Delta=U_1 \left(\mu^2\oplus{\mu^2}^T\right) U_1^\dagger$ with $U_1\in\grso(2n)\cap\grsp(n)\cong\gru(n)$ and $\mu^2\in\grschur(n)$.
    Using a (non-unique) square root $\mu$ of $\mu^2$, Lemma~\ref{lemma:horizontal_to_kak_decomp} then implies the sought-after \kak decomposition $U_1 (\mu\oplus\mu^T) U_2$ with $U_2=(\mu^T\oplus\mu) U_1^\dagger O$.

    \textbf{C.} We have $\groupG=\grsp(n)\times\grsp(n)$, $\groupK=\{K\oplus K|K\in\grsp(n)\}$ and $\groupA=\{D\oplus D^\dagger|D\in\grspdiag(n)\}$.\\\noindent
    We will proceed analogously to the construction for type A, but use a symplectic \ac{EVD} provided by Lemma~\ref{lemma:aii_symplectic}.
    Given $G=S\oplus S'\in\groupG$, compute $\delta=SS'^\dagger$, which is symplectic, and apply Lemma~\ref{lemma:aii_symplectic} to obtain an \ac{EVD} $\delta=S_1 \fsl{D}^2S_1^\dagger$ with $S_1\in\grsp(n)$ and $\fsl{D}^2\in\grspdiag(n)$.
    Then, compute a (non-unique) square root $\fsl{D}$ of $\fsl{D}^2$ and set $S_2=\fsl{D}S_1^\dagger S'$, which is in $\grsp(n)$.
    Eq.~(\ref{eq:num_type_A_proof}) then holds, telling us that we found the desired \kak decomposition of $G=S\oplus S'$.
    
    \textbf{CI.} We have $\groupG=\grsp(n)$, $\groupK=\gru(n)$ and $\groupA=\grspdiag(n)$.\\\noindent
    This decomposition is similar to type AI, but restricted to $\grsp(n)$. It is to AI what DIII is to AII, and in particular $\groupK$ is the same.
    Accordingly, we note that for a given $S\in\groupG$, $\Delta=S\Theta(S)^\dagger=SS^T$ comes with a symplectic eigenbasis that can additionally be made real-valued, just like for DIII.
    More precisely, for each eigenvector $v_k$ of $\Delta$ with eigenvalue $\lambda_k$, $J_nv_k$ is an eigenvector with eigenvalue $\lambda_k^\ast$ (Lemma~\ref{lemma:aii_symplectic}) and $v_k^\ast$ ($J_nv_k^\ast$) is an eigenvector with eigenvalue $\lambda_k$ ($\lambda^\ast_k$) ($\groupP$-case of Lemma~\ref{lemma:ai_orthogonal}).
    The recombination in Eq.~(\ref{eq:recombine_DIII}) happens within an eigenspace, so that $\Delta$ remains diagonal in the new basis. Using the same basis vector ordering as for DIII, we additionally find that the diagonalization is skew-repeat-diagonal, i.e.,~$\Delta=U_1(D^2\oplus {D^2}^\dagger)U_1^\dagger$ with $D\in\grudiag(n)$ so that $\fsl{D}^2\coloneqq D^2\oplus {D^2}^\dagger\in\grspdiag(n)$.
    Using a (non-unique) square root $\fsl{D}$ of $\fsl{D}^2$, Lemma~\ref{lemma:horizontal_to_kak_decomp} then implies the sought-after \kak decomposition $S=U_1 \fsl{D} U_2$ with $U_2=\fsl{D}^\dagger U_1^\dagger S$.
    
    \textbf{CII.} We have $\groupG=\grsp(p+q)$, $\groupK=\grsp(p)\times\grsp(q)$ and $\groupA=\{F\oplus F^T| F \in \grcs(p,q)\}$.\\\noindent
    This case involves a number of permutations, leading to heavy notation.
    In essence, the symplecticity of $S\in\groupG$ guarantees that we may reorder a generic unitary \ac{CSD} in such a way that all its components become symplectic as well, while maintaining the product structure of $\gru(p)\times\gru(q)$, achieving $K_i\in\grsp(p)\times\grsp(q)$ overall.

    We will assume $p\geq q$ for simplicity.
    First, we make the involved matrix representations explicit, and introduce the objects $\chi$ and $\eta$,
    \begin{align}
        \chi &= \begin{pmatrix}
            \id_q & 0 & 0 & 0 & 0 & 0\\
            0 & \id_{p-q} & 0 & 0 & 0 & 0\\
            0 & 0 & 0 & 0 & \id_q & 0\\
            0 & 0 & \id_{q} & 0 & 0 & 0\\
            0 & 0 & 0 & \id_{p-q} & 0 & 0\\
            0 & 0 & 0 & 0 & 0 & \id_q\\
        \end{pmatrix},\quad
        \eta = \begin{pmatrix}
            \id_q & 0 & 0 & 0 & 0 & 0\\
            0 & 0 & \id_{p-q} & 0 & 0 & 0\\
            0 & \id_q & 0 & 0 & 0 & 0\\
            0 & 0 & 0 & \id_{p-q} & 0 & 0\\
            0 & 0 & 0 & 0 & \id_q & 0\\
            0 & 0 & 0 & 0 & 0 & \id_q\\
        \end{pmatrix},\\
        \grsp(p)\times\grsp(q)&\cong\left\{\chi (S_1\oplus S_2)\chi^\dagger | S_1\in\grsp(p), S_2\in\grsp(q)\right\}\\
        &=\left\{\begin{pmatrix}
            A_1 & 0 & B_1 & 0 \\
            0 & A_2 & 0 & B_2 \\
            C_1 & 0 & D_1 & 0 \\
            0 & C_2 & 0 & D_2\\
        \end{pmatrix}\bigg| \begin{pmatrix}
            A_1 & B_1 \\ C_1 & D_1
        \end{pmatrix}\in\grsp(p), \begin{pmatrix}
            A_2 & B_2 \\ C_2 & D_2
        \end{pmatrix}\in\grsp(q)\right\}.\label{calc:0}
    \end{align}
    Now let $S\in\grsp(p+q)$, calculate $S'=\eta^\dagger \chi^\dagger S \chi\eta$, and apply a generic complex-valued \ac{CSD} (see type AIII) to $S'$ to obtain
    \begin{align}
        S'=
        \left(U_{1, 1}\oplus U_{1,2}\right) F_0 \left(U_{2,1}\oplus U_{2,2}\right)
        \quad\Rightarrow\quad
        S=\chi\eta S' \eta^\dagger \chi^\dagger
        =V_1 \bar{F} V_2.
    \end{align}
    Here we introduced $V_i=\chi\eta\left(U_{i,1}\oplus U_{i,2}\right)\eta^\dagger\chi^\dagger$, which is of the form $\chi(A\oplus B)\chi^\dagger$ for some $A\in\gru(2p),B\in\gru(2q)$, as well as
    \begin{align}
        \bar F 
        \coloneqq \chi\eta F_0\eta^\dagger\chi^\dagger
        = \chi\eta \begin{pmatrix}
            C_1 & 0 & 0 & 0 & S_1 & 0 \\
            0 & C_2 & 0 & 0 & 0 & S_2 \\
            0 & 0 & \id_{p-q} & 0 & 0 & 0 \\
            0 & 0 & 0 & \id_{p-q} & 0 & 0 \\
            -S_1 & 0 & 0 & 0 & C_1 & 0 \\
            0 & -S_2 & 0 & 0 & 0 & C_2 \\
        \end{pmatrix}\eta^\dagger\chi^\dagger
        =\begin{pmatrix}
            C_1 & 0 & S_1 & 0 & 0 & 0 \\
            0 & \id_{p-q} & 0 & 0 & 0 & 0 \\
            -S_1 & 0 & C_1 & 0 & 0 & 0 \\
            0 & 0 & 0 & C_2 & 0 & S_2 \\
            0 & 0 & 0 & 0 & \id_{p-q} & 0 \\
            0 & 0 & 0 & -S_2 & 0 & C_2 \\
        \end{pmatrix}.\label{calc:2}
    \end{align}
    Note that $K_{p,q} V_i K_{p,q}^\dagger=V_i$ and $K_{p,q} \bar F K_{p,q}^\dagger=\bar F ^\dagger$, so that 
    \begin{align}
        \Delta
        = S K_{p,q} S^\dagger K_{p,q}^\dagger
        = V_1 \bar F V_2 K_{p,q} V_2^\dagger \bar F^\dagger V_1^\dagger K_{p,q}^\dagger
        = V_1 \bar F ^2 V_1^\dagger.
    \end{align}
    In addition, $K_{p,q}\in\grsp(p+q)$, so that $\Delta\in\grsp(p+q)$ as well. Lemma~\ref{lemma:aii_symplectic} then allows us to reorder the columns in $V_1$ to obtain a new matrix $K_1\in\grsp(p+q)$\footnote{Strictly speaking, the lemma applies to an eigenbasis. However, a basis change between $\bar F^2$ and a proper diagonalization $\bar D^2$ is compatible with the procedure, so that we do not detail this additional basis change here.}, 
    and we can maintain the property $\Theta_{\mathrm{CII}}(V_1)=V_1$ while reordering, because $J_{p+q}$ and $K_{p,q}$ commute and thus share an eigenbasis\footnote{More precisely, Lemma~\ref{lemma:conjugate_eigenvector} (applied twice) tells us that eigenvectors of $V_1$ come in quadruples, because it is in the vertical space $\grsp(p+q)$ of the generic type-AII involution and in the vertical space of $\Theta_{\mathrm{CII}}$.}.
    The reordering accordingly is applied to the rows and columns of $\bar F^2$ to obtain ${\mathring{F}}^2$, which still must satisfy $\Theta_{\mathrm{CII}}({\mathring{F}}^2)={{\mathring{F}}^2}{}^\dagger$ because $\Theta_\mathrm{CII}(\Delta)=\Delta^\dagger$ and $\Theta_\mathrm{CII}(K_1)=K_1$. Additionally, ${\mathring{F}}^2$ is in $\grsp(p+q)$ because $K_1$ and $\Delta$ are. Overall, this gives ${\mathring{F}}^2$ the structure $F\oplus F^T$ for some $F\in\grcs(p,q)$, because this is the only compatible structure with symplecticity and being horizontal with respect to $\Theta_\mathrm{CII}$\footnote{Strictly speaking, the diagonal structure of the cosine and sine blocks might not be achieved automatically. However, simultaneously reordering rows and columns within the dimensions $(p+1,\dots,p+q)$ and $(2p+q+1,\dots,2p+2q)$ preserves all mentioned properties and achieves this diagonal structure.}.
    Now we have $\Delta=K_1 {\mathring{F}}^2 K_1^\dagger$ with $K_1\in\grsp(p+q)\cap\left(\grsp(p)\times\grsp(q)\right)=\groupK$ and ${\mathring{F}}^2\in\groupP$, so that Lemma~\ref{lemma:horizontal_to_kak_decomp} provides a \kak decomposition of $S$.
\end{proof}

\section{Details on fixed-depth Hamiltonian simulation}\label{sec:fdhs_calculations}

\subsection{Computing the dynamical Lie algebra}\label{sec:fdhs_calculations:algebra}
The considered Hamiltonian is 
\begin{align}\label{eq:ff_ham}
    H=\sum_{i=1}^{n-1} \alpha^X_i X_i X_{i+1}+\alpha^Y_i Y_iY_{i+1} + \sum_{i=1}^n \beta_i Z_i,
\end{align}
and we want to obtain its \ac{DLA} $\mfg=\langle iH\rangle_\text{Lie}$ with respect to the Pauli basis. 
It has been shown, e.g.,~in~\cite{kokcu2022fixed}, to be
\begin{align}
    \mfg 
    &= \left\langle\{\widehat{X_iX_j}, \widehat{X_iY_j}, \widehat{Y_i X_j}, \widehat{Y_i Y_j}\}_{1\leq i<j\leq n}\cup \{Z_i\}_{i=1}^n\right\rangle_{i\mbr}.\label{eq:ff_dla}
\end{align}

We verify this calculation by testing a) that all multi-qubit operators $i\widehat{A_iB_j}$ with $A,B\in\{X,Y\}$ can be reached with nested commutators of Hamiltonian terms, and b) that no other operators can be created. For ease of reference, we call such strings of Pauli operators $XY$-strings.
For step a), we may create the $2\ell$-qubit $XY$-strings 
\begin{align}
    &\bigcirc_{k=1}^{\ell-1} \left(\ad_{iX_{i+2k}X_{i+2k+1}} \circ \ad_{iY_{i+2k-1}Y_{i+2k}}\right)\ \ (iX_i X_{i+1}) \\
    =&[iX_{i+2\ell-2}X_{i+2\ell-1}, [\cdots, [iX_{i+2}X_{i+3}, [iY_{i+1}Y_{i+2}, iX_iX_{i+1}]]]] \nonumber \\
    \propto&i\widehat{X_iX_{i}}_{+2\ell-1},
\end{align}

for $1\leq \ell \leq \frac{n+1-i}{2}$, as well as the $(2\ell-1)$-qubit $XY$-strings $i\widehat{X_iY_{i}}_{+2\ell-2}$ for $2\leq \ell \leq \frac{n+2-i}{2}$ that result from skipping the outermost adjoint map above.
Then, all other $XY$-strings can be obtained by applying $\ad_{iZ_i}$ to one or both ends of the created strings.

For step b), we note that the $Z_i$ commute and that commutators between $i\widehat{A_i B_j}$ and $iZ_k$ vanish or yield an $XY$-string of the same size.
This leaves us to check commutators of the form $[i\widehat{A_iB_j},i\widehat{C_kD_l}]$, where $i<j$, $k<l$, and w.l.o.g. $i\leq k$.
Those commutators vanish unless at least one of the edges coincide, i.e.,~unless at least one of the conditions $\{i=k, j=k, j=l\}$ is satisfied.
Thus, we need to compute the commutators (where differing index variables are assumed to differ in value)
\begin{alignat}{4}
    [i\widehat{A_iB_j},i\widehat{C_iD_l}]&\propto\begin{cases}
        i\widehat{\overline{B}_j D_l} & \text{ for } A=C, j<l\\
        i\widehat{\overline{D}_l B_j} & \text{ for } A=C, l<j\\
        0 & \text{ else}
    \end{cases},&\quad 
    [i\widehat{A_iB_j},i\widehat{C_iD_j}]&\propto\begin{cases}
        i Z_j & \text{ for } A=C, B\neq D\\
        i Z_i & \text{ for } A\neq C, B=D\\
        0 &\text{ else}
    \end{cases},\\
    [i\widehat{A_iB_j},i\widehat{C_jD_l}]&\propto \begin{cases}
        i \widehat{A_i D_l} & \text{ for } B\neq C\\
        0 & \text{ else}
    \end{cases},& \quad
    [i\widehat{A_iB_j},i\widehat{C_kD_j}]&\propto \begin{cases}
        i \widehat{A_i \overline{C}_k} & \text{ for } B=D\\
        0 & \text{ else}
    \end{cases}.
\end{alignat}
Here $\overline{B}$ denotes $X$ if $B=Y$ and vice versa.
We find that commutators of the constructed $XY$-strings indeed only ever reproduce other $XY$-strings or single-qubit $Z$ operators. This concludes the second step and shows that the Lie closure given above is correct.

\subsection{Identifying the dynamical Lie algebra}\label{sec:fdhs_calculations:identify}
Next, we would like to identify the algebra as a representation of one of the classical Lie algebras, or as a direct sum thereof.
There is a whole toolbox available to perform this identification in general.
A first indication is the dimension of the algebra, limiting the type of algebra to a finite number of options.
This can be refined by analyzing the connected components of the non-commutation graph individually, and by considering the center of the algebra/its components. If the basis elements are Pauli words, the recent full classification of Pauli algebras in~\cite{aguilar2024full} limits the options further, and for complete families of algebras their scaling with the qubit count may indicate the class of algebra in Thm.~1 therein.
For a fully automated identification, more advanced algorithms can be found in~\cite{rand1988identification,degraaf1997algorithm}.

For our concrete algebra, we take an easier path and exploit the simple structure of $H$ and the \ac{DLA} ``manually".
For this, we note that the terms in $H$ are not a minimal generating set, but that the set $\{iX_jX_{j+1}\}_{j=1}^{n-1}\cup\{iZ_j\}_{j=1}^n$ generates the same algebra, and that this new set indeed is a minimal generating set\footnote{Note that the missing terms $\{Y_jY_{j+1}\}_j$ from the Hamiltonian are second-order commutators of the minimal generating set.}.
Furthermore, note that the anticommutation graph of this new set is particularly simple, as it is a one-dimensional chain with $2n-1$ nodes:

\begin{figure}[h!]
    \centering
\includegraphics[width=.4\linewidth]{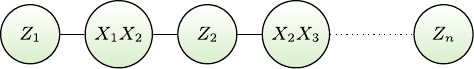}
\end{figure}

This means that the anticommutation graph is in class~A from Thm.~1 in~\cite{aguilar2024full} with $n_L=2n-1$ and $n_c=0$. Their Thm.~2 then implies that $\mfg\cong\mfso(2n)$.
As a sanity check, we count dimensions. $\mfg$ has $4\frac{n^2-n}{2}$ $XY$-strings and $n$ single-qubit field operators, summing up to $2n^2-n$ dimensions. The classical algebra $\mfso(2n)$ has $\frac{(2n)^2-2n}{2}=2n^2-n$ dimensions, so that the dimension counts match.

\subsection{Mapping to an irreducible representation}\label{sec:fdhs_calculations:mapping}
The above analysis implies that the defining representation of $\mfso(2n)$ (i.e., as traceless antisymmetric matrices acting on $\mbc^{2n}$) is a faithful representation of the \ac{DLA} Eq.~(\ref{eq:ff_dla}).
Our strategy will be to map a target Hamiltonian of the form Eq.~(\ref{eq:ff_ham}) to this representation, and then use our explicit numerical algorithms (Section~\ref{sec:numerical_decomps} and App.~\ref{sec:numerical_details}) to construct \kak decompositions by manipulating only $(2n\times 2n)$-dimensional matrices. \\

Our realization of the isomorphism will go via \textit{Majorana fermions}, i.e., a set $\{c_\mu\}_{1\leq\mu\leq 2n}$ of operators satisfying the canonical anticommutation relations
\begin{equation}\label{eq:cac}
    \{c_\mu, c_\nu\}=2\delta_{\mu,\nu}.
\end{equation}
There is some freedom in how exactly to pick the Majoranas, captured by the group of \emph{fermionic Gaussian Cliffords}~\cite{wan2022matchgate}; here we 
will make the choice (corresponding to a Jordan-Wigner transformation~\cite{diaz2023showcasing})
\begin{alignat}{5}\label{eq:majos}
c_1&=YIII\ldots,\quad&c_2&=ZYII\ldots,\quad&c_n&=ZZ\ldots ZY,    \\
c_{n+1}&=XIII\ldots,\quad&c_{n+2}&=ZXII\ldots,\quad&c_{2n}&=ZZ\ldots ZX.  
\end{alignat}
The connection to the \ac{DLA} comes from the observation that each of its Pauli words  can be written as a product $c_\mu c_\nu$ of two Majorana fermions. Given a Pauli words $iP=c_\mu c_\nu\in\mfg$ the isomorphism can then be taken to be (the linear extension of)
\begin{equation}
    \varphi:\mfg\to\mfso(2n);\ iP=c_\mu c_\nu\mapsto 2(E_{\mu,\nu}-E_{\nu,\mu}),
\end{equation}
where $E_{\mu,\nu}$ denotes an elementary matrix which is zero except for a $1$ at position $(\mu,\nu)$ (we omit the mechanical details of verifying that this mapping is indeed an isomorphism). 
For $n=4$, the isomorphism can be visualized as
\begin{align}
    \left(
    \begin{array}{cccc|cccc}
        0&iXYII&iXZYI&iXZZY&\hamcolor{-iZIII}&\hamcolor{iXXII}&iXZXI&iXZZX\\
        &0&iIXYI&iIXZY&\hamcolor{iYYII}&\hamcolor{-iIZII}&\hamcolor{iIXXI}&iIXZX\\
        &&0&iIIXY&iYZYI&\hamcolor{iIYYI}&\hamcolor{-iIIZI}&\hamcolor{iIIXX}\\
        &&&0&iYZZY&iIYZY&\hamcolor{iIIYY}&\hamcolor{-iIIIZ}\\
        \hline
        &&&&0&-iYXII&-iYZXI&-iYZZX\\
        &&&&&0&-iIYXI&-iIYZX\\
        &&&&&&0&-iIIYX \\
        &&&&&&&0
    \end{array}
    \right),
\end{align}
where a Pauli word $\pm iP$ denoted in position $(\mu, \nu)$ (with $\mu<\nu$) of this matrix implies that $iP$ is mapped to the element $\pm 2(E_{\mu,\nu}-E_{\nu,\mu})$ in the irreducible representation. The entries $(\mu, \nu)$ with $\mu>\nu$ are not mapped independently and we skip them in the visualization. The highlighted Pauli operators are the terms of the Hamiltonian of Eq.~(\ref{eq:ff_ham}); we see immediately that they are \textit{horizontal} with respect to the $p=q=4$ BDI decomposition of $\mfso(8)$ in its canonical form. The fact that our target Hamiltonian is canonically horizontal with respect to the isomorphism $\varphi$ clearly depends on both the choice and ordering of the Majoranas in Eq.~(\ref{eq:majos}); an algorithm for obtaining an ordering that makes a target Hamiltonian horizontal is presented in Sec.~\ref{sec:hgraph}

\subsection{Recursive BDI decomposition}\label{sec:fdhs_calculations:apply_decomp}
The recursive decomposition chosen to compile the time-evolution circuit of the transverse field XY model is a repeated BDI decomposition with its involution in canonical form.
That is, any occurring block matrix $\grso(p+q)$ is decomposed according to the involution $\theta=\Ad_{I_{p,q}}$, using the \ac{CSD} as described in Sec.~\ref{sec:numerical_decomps}.

At the $j$th recursion level, starting to count with $j=0$ for the initial decomposition step, $4^j$ special orthogonal matrices need to be decomposed. These matrices either have the shape $d_j\times d_j$ or $(d_j+1)\times(d_j+1)$, with 
\begin{align}
    d_j=\lfloor (2n)/2^j\rfloor.
\end{align}
The unequal splitting occurs for qubit counts $n$ that are not powers of $2$. At the $j$th recursion level, $(2n)\!\!\mod 2^j$ matrices are of shape $(d_j+1)\times(d_j+1)$, the others are of shape $d_j\times d_j$. Each matrix is split into five special orthogonal matrices, two block submatrices for each matrix $K$ in $\kak$, and one commuting matrix for the CSG element $A$. The block submatrices have shape $d_{j+1}\times d_{j+1}$ or $(d_{j+1}+1)\times(d_{j+1}+1)$, while $A$ stays at the previous shape.
The recursion is performed until the individual \acp{CD} yield $K$'s from $\grso(2)\times \grso(2)$ or $\grso(1)\times \grso(2)\cong \grso(2)$, i.e.,~until all non-abelian blocks of sizes larger than $2$ have been decomposed.
This takes $\lceil\log_2 n\rceil$ steps, leading to $\mathcal{O}\left(4^{\lceil\log_2 n\rceil}\right)=\mathcal{O}(n^2)$ $2\times 2$ matrices embedded in the irreducible representation on $\mbr^{2n}$. Each of these matrices performs a standard Givens rotation in this space.
In addition, the decomposition creates $\mathcal{O}(n^2)$ elements from the \acp{CSG} $\gru(1)^{\times d_{j+1}}\cong \grso(2)^{\times d_{j+1}}$, which may be pulled apart into $d_{j+1}$ Givens rotations again.
Note that because the BDI decomposition with $|p-q|\leq 1$ is parameter-optimal, we know that the total number of Givens rotations matches the dimension of the total group or algebra, $2n^2-n$.

\begin{figure}
    \centering
    \includegraphics[width=.3\linewidth]{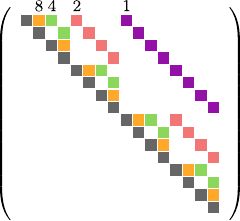}
    \caption{Visualization of the generators of Givens rotations encountered in the recursive BDI decomposition used to compile Hamiltonian simulation. The separate matrix elements on the secondary diagonal (orange) correspond to vertical generators of the $K$'s, whereas the other off-diagonal stripes (green, salmon, purple) mark \ac{CSA} elements (the main diagonal is indicated in grey for reference). The numbers at the top indicate how many layers of the respective generators occur in the recursive decomposition. For simplicity we here assume $n$ to be a power of $2$, leading to a particularly regular structure. Other $n$ lead to relative offsets by $1$ in the block sizes.}
    \label{fig:fdhs_decomp_generators}
\end{figure}

The subgroup elements $K$ exclusively are Givens rotations generated by (half of the) secondary diagonal elements, $E_{2i,2i+1}-E_{2i+1,2i}$. The \ac{CSG} elements at the $j$th recursion step are generated by elements from the $d_j$th diagonal. We visualize the occurring generators in the recursive decomposition in Fig.~\ref{fig:fdhs_decomp_generators}.

\subsection{Mapping back to qubit representation}\label{sec:fdhs_calculations:mapping_back}
The simple mapping from step three of the algorithm (see Sec.~\ref{sec:fdhs_calculations:mapping}) allows us to easily create the inverse map $\rho^{-1}$ back to the reducible representation on qubits.

As discussed in the previous section, the recursive decomposition produces $\grso(2)$ Givens rotations in the canonical basis $\{E_{i,j}\}$.
The rotation angle can simply be read out from the matrix, and after accounting for the prefactor $\pm2$ in the mapping $\rho$, we obtain the Pauli rotation angle for the time-evolution quantum circuit.

As an example, consider the \ac{CSG} element $A$ from the $j=0$th level of the recursion for $n=3$, given by
\begin{align}
    A&=\begin{pmatrix}
        c_1&0&0&s_1&0&0\\
        0&c_2&0&0&s_2&0\\
        0&0&c_3&0&0&s_3\\
        -s_1&0&0&c_1&0&0\\
        0&-s_2&0&0&c_2&0\\
        0&0&-s_3&0&0&c_3\\
    \end{pmatrix}
    =\prod_{i=1}^n \exp(\eta_i (E_{i,n+i}-E_{n+i,i}))
    \in \grso(2)^{\times n},
\end{align}
where we abbreviated $c_i=\cos(\eta_i)$ and $s_i=\sin(\eta_i)$.
We can easily read out the rotation angles $\eta_i$ from the matrix elements $c_i$ and $s_i$, and only need to take into account that $iZ_i$ was mapped to $-2(E_{i,n+i}-E_{n+i,i})$, so that the Pauli rotation angle becomes $-\eta_i/2$. We find the qubit representation of $A$ to be
\begin{align}
    \rho^{-1}(A)=\prod_{i=1}^n \exp(-\frac{i\eta_i}{2}Z_i).
\end{align}

The \ac{CSG} elements from subsequent recursion levels lead to more complex Pauli rotations. As mentioned in the previous section, they are generated by elements $E_{i, i+d_j}-E_{i+d_j,i}$ for $j>0$, which are mapped back to $XY$-strings of length $n/2^j$ by $\rho^{-1}$.
Consequentially, the most expensive Pauli rotation gates result from the \ac{CSG} elements at $j=1$ with size $n/2$.

The vertical subgroup elements $K$ are generated by matrices of the form $E_{i,i+1}-E_{i+1,i}$, which are mapped to two-qubit operators $X_iY_{i+1}$ for $i<n$ and $Y_{i-n}X_{i+1-n}$ for $i\geq n$. In addition, only every other of these matrices, of the form $E_{2i,2i+1}-E_{2i+1,2i}$, appears in the decomposition, so that all $K$ commute and with exception of the pairs $(X_iY_{i+1},Y_iX_{i+1})$, no two $K$'s act on overlapping qubits.

\begin{figure}
    \includegraphics[width=0.5\linewidth]{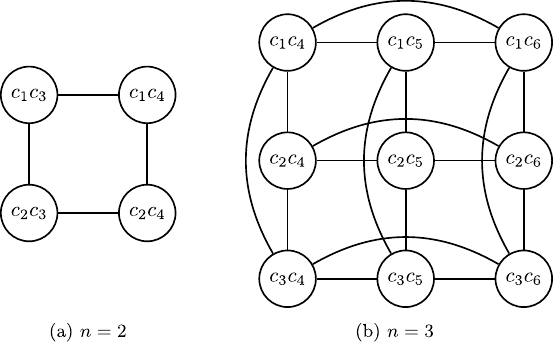}
    \caption{Frustration graphs of a generic element of the horizontal subspace of the canonical BDI decomposition for $n=2,3$; the canonical anticommutation relations Eq.~(\ref{eq:cac}) of the Majoranas ensure that they take a particularly nice form. The vertices of the graph sit on the integer-valued coordinates of a square lattice of length $n$, and are connected by an edge if and only if they are within either the same row or column (with no self-edges). }
    \label{fig:frust}
\end{figure}

\subsection{Finding a horizontal mapping}\label{sec:hgraph}
In Sec.~\ref{sec:fdhs_calculations:mapping} we produced a choice Eq.~(\ref{eq:majos}) of Majorana fermions with respect to which our target Hamiltonian Eq.~(\ref{eq:ff_ham}) was horizontal under a canonical BDI decomposition. A different ordering of the Majoranas would have produced a different set of horizontal operators; conversely, for a different target Hamiltonian to be horizontal one would need to choose a different ordering of the Majoranas. The goal of this section is to, for a target Hamiltonian, obtain such an ordering. 
To this end, suppose we have a Hamiltonian
\begin{equation}
    H = \sum_{(\mu,\nu)\in M} ih_{\mu,\nu}c_\mu c_\nu,
\end{equation}
for some subset $M$ of the pairs $(\mu,\nu)$ of integers between $1$ and $2n$ with, say, $\mu<\nu$. We need to produce a permutation $\pi:[2n]\to[2n]$ such that, for all $(\mu,\nu)\in M$, $(\pi(\mu), \pi(\nu))$ satisfies $\pi(\mu)<n$ and $ \pi(\nu)>n$. We begin by investigating the question of the existence of such a $\pi$.
Let us start by recalling the definition of the
\textit{frustration graph} $\graph{G}$ of a Hamiltonian $H$, which has vertices corresponding to each Pauli word that appears in $H$, and edges between vertices whose corresponding operators anticommute. As Pauli words either commute or anticommute, the existence of a Lie-algebraic structure-preserving bijection between two sets of Pauli words can be translated to a question of whether or not their frustration graphs are isomorphic. In particular, we see that  we can find a choice of Majoranas that place $H$ into the $p=q=n$ BDI induced horizontal space if and only if  the frustration graph of $H$ is isomorphic to a subgraph of the graph $\graph{F}$ that has vertices at every integral pair $(a,b)$ with $1\leq a,b\leq n$ and edges between a pair of vertices $(a_1,b_1),\ (a_2,b_2)$ if and only if $a_1=a_2$ or $b_1=b_2$ (this follows immediately from the canonical anticommutation relations Eq.~\ref{eq:cac} of the Majoranas; see Fig.~\ref{fig:frust}). Given such a subgraph, choose Majoranas such that the string at $(a,b)$ is (up to $\pm i$) given by $c_a c_{b+n}$. By construction, $H$ is horizontal. The situation for more general values of $p$ and $q$ is analogous,   but with the constraints on the ``vertex coordinates'' now reading $1\leq a\leq p,\ 1\leq b\leq q$.

Finally, we note that (conditioned on us having found an isomorphism) the above plan to ``choose Majoranas such that the string at $(a,b)$ is (up to $\pm i$) given by $c_a c_{b+n}$'' (that is, determine $\pi$) can be  implemented in time $\mathcal{O}(n^2)$. First, note that every row (column) of the graph Fig.~\ref{fig:frust}  within which our frustration graph is embedded is characterized by a single Majorana operator. We begin by looping through the rows (columns) of the graph. If for any row (column) we find that two (or more) vertices are ``occupied'' by the subgraph, then we assign to that row (column) the Majorana operator shared by the multiple Paulis corresponding to the multiple vertices (which, as these Paulis are mutually anticommuting products of two Majoranas, is guaranteed both to exist and be unique). Next we again loop through the rows (columns) of the graph. If a given row (column) contains only a single occupied vertex, and one of the corresponding Majorana operators has already been assigned elsewhere, the remaining Majorana must characterize the row (column). We loop through the rows (columns) one more time, assigning the Majorana corresponding to the first (second) factor of that Pauli to the row (column) whenever we find a singly-occupied row (column). The isomorphism assures us that this can be done without introducing contradictions. At this point we have an incomplete ordering of the Majoranas with respect to which our Hamiltonian is horizontal; the remaining slots can be filled arbitrarily. An implementation of this algorithm is available at~\cite{symmetrycompilationrepo}.

\section{K{\"a}hler structures}\label{sec:kahler}
The purpose of this appendix is to comment on the structure underlying some of the similarities between the various numerical algorithms of Section~\ref{sec:numerical_decomps}. The discussion is strongly based on Ref.~\cite{hackl2021bosonic}.
In the following, we will use an endomorphism $J$, which can, but does not have to, be $J_n$ from Eq.~(\ref{eq:def_J_n}).
To begin, we define the notion of a \textit{K{\"a}hler space}:
\begin{definition}
    We call a $2N$-dimensional real vector space $V$ a K{\"a}hler space if it is equipped with 
    \begin{enumerate}
        \item A symmetric, positive-definite bilinear form $G$ 
        \item An antisymmetric, non-degenerate bilinear form $\Omega$ 
        \item An endomorphism $J$ satisfying $J^2=-\id$
    \end{enumerate}
    that satisfy the consistency conditions\footnote{For brevity, we will suppress covariant and contravariant indices in the following; for example by $-G\Omega^{-1}=J$ we mean (Einstein-summing over $c$) $-G^{ac}(\Omega^{-1})_{cb} =J^a_{\hspace{2mm} b}$ (see Ref.~\cite{hackl2021bosonic} for details).}  $-G\Omega^{-1} =J = \Omega G^{-1}$. We call the triple $(G,\Omega,J)$ a \emph{K{\"a}hler structure}.
\end{definition}
Suppose we have a $2N$-dimensional real vector space $V$ equipped with  a K{\"a}hler structure $(G,\Omega,J)$. Three subgroups of $\grgl(2N,\mbr)$ are naturally picked out:
\begin{align}
    \gro(2N,\mbr) &= \{ M\in \grgl(2N,\mbr)\ |\ MGM^T=G\}, \label{eq:ko} \\
    \grsp(2N,\mbr) &= \{ M\in \grgl(2N,\mbr)\ |\ M\Omega M^T=\Omega\},\label{eq:ksp}\\
    \grgl(N,\mbc) &= \{ M\in \grgl(2N,\mbr)\ |\ MJ = JM\}.\label{eq:kj}
\end{align}
Note that $\grsp(2N,\mbr)$ is \emph{not} the symplectic group $\grsp(N)$ from the main text.
One can show that the consistency conditions force the intersection of these three groups to be isomorphic to $\gru(N)$; further, for any $M\in \groupH$, where $\groupH$ is either of the first two groups on the list, $M\in\gru(N)$ if and only if it commutes with $J$. At the Lie algebra level, this induces a \ac{CD} $\mfh=\mfu\oplus\mfu^\perp$ of the Lie algebra $\mfh$ of $\groupH$; from Tab.~\ref{tab:all_cartan_involutions} we identify it as of type DIII\footnote{Note that the Lie algebras of $\gro$ and $\grso$ are identical.} (CI) for $\groupH=\gro$ ($\grsp$). Conveniently, we have: 
\begin{lemma}\label{lem:ko}
 The orthogonal (with respect to the Killing form) complement $\mfu^\perp$ of the above-defined $\mfu$  consists exactly of the $X\in \mfh$ that anticommute with $J$~\cite[Prop.~1]{hackl2021bosonic}.
\end{lemma}
\begin{proof}
First we show that any $X\in\mfh$ which anticommutes with $J$ is orthogonal to any $Y\in\mfu$. Noting that the Killing form is proportional to the trace form for $\mfso(2N,\mathbb{R})$ and $\mfsp(2N,\mathbb{R})$, we have 
\begin{equation*}
\tr[XY]=\tr[XJ(-J)Y]=\tr[(-J)X(-J)Y]=\tr[X(-J)Y(-J)]=\tr[X(-J)^2Y]=-\tr[XY],
\end{equation*}
where we have used $J^2=-\id$, the cyclicity of the trace, and the (anti)commutation of $J$ with $(X)\ Y$. From this, we can conclude $\tr[XY]=0$. In the other direction, following Ref.~\cite{hackl2021bosonic}, suppose $K\perp \mfu$. Then, as $[J,K-JKJ]=0$, we have $K-JKJ\in\mfu$ and
\begin{equation*}
    \frac{1}{2}\tr[(K-JKJ)^2]
    =\frac{1}{2}\tr[K^2-KJKJ-JKJK+JKJJKJ]
    =\tr[K(K-JKJ)]
    =0
\end{equation*}
where we expanded the square, applied the cyclicity of the trace, and then used $K\perp \mfu$ together with $K-JKJ\in\mfu$.
By the non-degeneracy of the trace form on $\mfu\ni K-JKJ$, we conclude $K-JKJ=0$, i.e., $\{K,J\}=0$.
\end{proof} 

The unified picture of the type-CI and type-DIII decompositions offered by the K{\"a}hler structure explains several of the similarities in the numerical algorithms of Section~\ref{sec:numerical_decomps}.  Suppose we wish to find a \ac{CD} of $M\in \groupH$, where $\groupH$ continues to denote either $\gro$ or $\grsp$, and in either case we quotient by $\gru$. Let us fix some K{\"a}hler structure $(G,\Omega,J)$.
$M$ defines a new  K{\"a}hler structure 
\begin{equation}
     (G_M,\Omega_M,J_M) =  (MGM^T,M\Omega M^T,MJM^{-1}).
\end{equation}
Right-multiplying $M$ by an element $e^u\in \gru$ does not change the obtained K{\"a}hler structure; we can therefore think of the \ac{CD} $M=e^{p}e^u$ (with $p\in\mfu^\perp$) as fixing a representative of $M$ in the equivalence class thus defined on $\groupH$ by multiplication by $\gru$\footnote{More technically, the \ac{CD} naturally induces a notion of $\groupH$ as a \textit{principal fiber bundle} with base manifold $\groupH/\gru$. The fiber of an element $M\in \groupH$ corresponds to the set $\{Mu\ |\ u\in\gru\}$ of elements which can be assigned the same ``polar component'' $P=\sqrt{\Delta}$ by the \ac{CD}. This perspective is explored in Ref.~\cite{guaita2024representation}.}.  \\

Now, given a \ac{CD} $M=PU=e^{p}e^{u}$, with $p\in\mfu^\perp,\ u\in\mfu$, the complex structure $J_M$ of $M$ satisfies
\begin{align}
    J_M&=MJM^{-1}\\
    &=e^{p}e^{u}Je^{-u}e^{-p}\\
    &=e^{p}e^{u}e^{-u}Je^{-p}\\
    &=e^{p}Je^{-p}\\
    &=e^{2p}J,
\end{align}
where we have used $[u,J]=0$ and $\{p,J\}=0$ (Lemma~\ref{lem:ko}). We conclude that $P^2=\Delta$, with $\Delta=-J_MJ$ the so-called \textit{relative complex structure}~\cite{hackl2021bosonic}. In the ``bosonic" case ($M\in\grsp$) we have 
\begin{align}
    \Delta_{\grsp} &= -J_{M}J\\ 
    &= -MJM^{-1}J\\
    &= MG\Omega^{-1} M^{-1}\Omega G^{-1}\\
    &= MG\Omega^{-1} \Omega M^{T}G^{-1}\\
    &= MGM^{T}G^{-1},
\end{align}
where we have used $M\Omega M^T=\Omega$.
Similarly, in the ``fermionic'' case ($M\in\gro$) we have 
\begin{align}
    \Delta_{\gro} &= -J_{M}J\\ 
    &= -MJM^{-1}J\\
    &= M\Omega G^{-1}M^{-1}G\Omega^{-1} \\
    &= M\Omega G^{-1} G M^{T}\Omega^{-1}\\
    &= M\Omega M^{T}\Omega^{-1},
\end{align}
where we have used that in this case $MG M^T=G$.
Up to taking square roots and logarithms (and noting $\Omega^{-1}=-\Omega$) we see the origin of the use of $\Delta$ in the proof of Thm.~\ref{thm:abstract_numerical_non_general} in the main text from the K{\"a}hler space point of view.

\end{document}